\documentclass[UKenglish]{article}
\usepackage[T1]{fontenc}
\usepackage{lmodern}
\usepackage[british,UKenglish,USenglish,english,american]{babel}
\usepackage[latin9]{inputenc}
\usepackage{geometry}
\usepackage{float}
\usepackage{amsthm}
\usepackage{amsmath}
\usepackage{amsfonts}
\usepackage{listings}

\floatstyle{ruled}
\newfloat{algorithm}{tbp}{loa}
\providecommand{\algorithmname}{Algorithm}
\floatname{algorithm}{\protect\algorithmname}

\newenvironment{lyxcode}
{\par\begin{list}{}{
\setlength{\rightmargin}{\leftmargin}
\setlength{\listparindent}{0pt}% needed for AMS classes
\raggedright
\setlength{\itemsep}{0pt}
\setlength{\parsep}{0pt}
\normalfont\ttfamily}%
 \item[]}
{\end{list}}

\floatstyle{ruled}
\newfloat{algorithm}{tbp}{loa}
\providecommand{\algorithmname}{Algorithm}
\floatname{algorithm}{\protect\algorithmname}
\newtheorem{theorem}{Theorem}[section]
\newtheorem{corollary}[theorem]{Corollary}
\newtheorem{prop}[theorem]{Proposition}
\newtheorem{definition}[theorem]{Definition}
  \theoremstyle{remark}
  \newtheorem{remark}{Remark}
  \theoremstyle{plain}
\newtheorem{lemma}[theorem]{Lemma}

\usepackage{microtype}%if unwanted, comment out or use option "draft"
\usepackage{tikz}
\usetikzlibrary{calc}
\usepackage{verbatim}

\usepackage{tikz}
\usepackage{verbatim}
\usetikzlibrary{calc}

\usepackage{lipsum}

\makeatother

\begin{document}
\title{A $\frac{9}{7}$-Approximation Algorithm for Graphic TSP in Cubic
Bipartite Graphs}

\author{Jeremy A. Karp, R. Ravi}
\maketitle

\begin{abstract}
We prove new results for approximating Graphic TSP. Specifically,
we provide a polynomial-time $\frac{9}{7}$-approximation algorithm
for cubic bipartite graphs and a $(\frac{9}{7}+\frac{1}{21(k-2)})$-approximation
algorithm for $k$-regular bipartite graphs, both of which are improved
approximation factors compared to previous results. Our approach involves
finding a cycle cover with relatively few cycles, which we are able
to do by leveraging the fact that all cycles in bipartite graphs are
of even length along with our knowledge of the structure of cubic
graphs.
 \end{abstract}

\section{Introduction}

\subsection{Motivation and Related Work}

The traveling salesman problem (TSP) is one of most well known problems
in combinatorial optimization, famous for being hard to solve precisely. 
In this problem, given a complete undirected graph $G=(V,E)$ with
vertex set $V$ and edge set $E$, with non-negative edge costs $c\in\mathbb{R}^{|E|}$,
$c\neq0$, the objective is to find a Hamiltonian cycle in $G$ of
minimum cost. In its most general form, TSP cannot be approximated
in polynomial time unless $P=NP$. In order to successfully find approximate
solutions for TSP, it is common to require that instances of the problem
have costs that satisfy the triangle inequality ($c_{ij}+c_{jk}\geq c_{ik}\;\forall i,j,k\in V$).
This is the Metric TSP problem. The Graphic TSP problem is a special
case of the Metric TSP, where instances are restricted to those where
$\forall i,j\in E$, the cost of edge $(i,j)$ in the complete graph
$G$ are the lengths of the shortest paths between nodes $i$ and
$j$ in an unweighted, undirected graph, on the same vertex set.

One value related to the ability to approximate TSP is the integrality
gap, which is the worst-case ratio between the optimal solution for
a TSP instance and the solution to a linear programming relaxation called
the subtour relaxation \cite{dantzig1954solution}.
A long-standing conjecture (see, e.g.,\cite{goemans1995worst}) for Metric TSP is that the integrality
gap is $\frac{4}{3}$. One source of motivation for studying Graphic TSP is that
the family of graphs with two vertices connected by three paths of
length $k$ has an integrality gap that approaches $\frac{4}{3}$.
This family of graphs demonstrates that
Graphic TSP captures much of the complexity of the more general Metric
TSP problem.

For several decades, Graphic TSP did not have any approximation algorithms
that achieved a better approximation than Christofides' classic $\frac{3}{2}$-approximation
algorithm for Metric TSP \cite{christofides1976worst}, further motivating the study
of this problem. However, a wave of recent papers
\cite{Gamarnik:2005:IUB:2308887.2309049,aggarwal20114,boyd2011tsp,gharan2011randomized,momke2011approximating,correa2012tsp,sebHo2012shorter} have provided significant
improvements in approximating Graphic TSP. Currently, the best known
approximation algorithm for Graphic TSP is due to Seb\H{o} and Vygen \cite{sebHo2012shorter}, with an approximation factor
of $\frac{7}{5}$.

Algorithms with even smaller approximation factors have also been
found for Graphic TSP instances generated by specific subclasses of
graphs. In particular, algorithms for Graphic TSP in cubic graphs
(where all nodes have degree $3$) have drawn significant interest
as this appears to be the simplest class of graphs that has many of
the same challenges as the general case. Currently, the best approximation
algorithm for Graphic TSP in cubic graphs is due to Correa, Larr\'e,
and Soto \cite{correa2012tsp}, whose algorithm achieves an approximation factor
of $(\frac{4}{3}-\frac{1}{61236})$ for $2$-edge-connected cubic graphs. Similarly, a $\frac{4}{3}$-approximation was recently obtained for instances of sub-quartic graphs \cite{DBLP:conf/esa/Newman14}.
Progress in approximating Graphic TSP in cubic graphs also relates
to traditional graph theory, as Barnette's conjecture \cite{barnette1969conjecture} states
that all bipartite, planar, $3$-connected, cubic graphs are
Hamiltonian. This conjecture suggests that instances of Graph TSP
on Barnette graphs could be easier to approximate, and conversely,
approximation algorithms for Graphic TSP in Barnette graphs may lead
to the resolution of this conjecture. Indeed, Correa, Larr\'e, and Soto
\cite{2013arXiv1310.1896C} provided a $(\frac{4}{3}-\frac{1}{18})$-approximation
algorithm for Barnette graphs. Along these lines, Aggarwal, Garg,
and Gupta \cite{aggarwal20114} were able to obtain a $\frac{4}{3}$-approximation
algorithm for $3$-edge-connected cubic graphs before any $\frac{4}{3}$-approximation
algorithms were known for all cubic graphs. In this paper, we examined
graphs that are cubic and bipartite, another class of graphs that
includes all Barnette graphs. An improved approximation for this class
of graphs is the primary theoretical contribution of this paper.
\begin{theorem}
\label{t1}
Given a cubic bipartite connected graph $G$ with $n$ vertices, there is a polynomial
time algorithm that computes a spanning Eulerian multigraph $H$ in
$G$ with at most $\frac{9}{7}n$ edges.
\end{theorem}
\begin{corollary} \label{k-reg}
Given a $k$-regular bipartite connected graph $G$ with $n$ vertices where $k\geq4$,
there is a polynomial time algorithm that computes a spanning Eulerian
multigraph $H$ in $G$ with at most $(\frac{9}{7}+\frac{1}{21(k-2)})n-2$
edges.
\end{corollary}
 This extension complements results~\cite{vishnoi2012permanent,FRS14} which provide guarantees for $k$-regular graphs
in the asymptotic regime. Corollary \ref{k-reg} improves on these guarantees for small values of $k$. Note that even for $k=4$ Corollary \ref{k-reg} yields a solution with fewer then $\frac{4}{3}n$ edges.

\subsection{Overview}

In this paper, we will present an algorithm to solve Graphic TSP,
which guarantees a solution with at most $\frac{9}{7}n$ edges in
cubic bipartite graphs. The best possible solution to Graphic TSP
is a Hamiltonian cycle, which has exactly $n$ edges, so this algorithm
has an approximation factor of $\frac{9}{7}$.

A corollary of Petersen's theorem is that every cubic bipartite graph
contains three edge-disjoint perfect matchings. The union of any $2$
of these matchings forms a $2$-factor. The following proposition
demonstrates the close relationship between $2$-factors and Graphic
TSP tours in connected graphs.
\begin{prop} \label{t3}
Any $2$-factor with $k$ cycles in a connected graph can be extended
into a spanning Eulerian multigraph with the addition of exactly $2(k-1)$
edges. This multigraph contains exactly $n+2(k-1)$ edges in total.
\end{prop}
Proposition \ref{t3} can be implemented algorithmically by compressing each cycle in the $2$-factor into a single node
and then finding a spanning tree in this compressed graph. We then add two copies of the edges
from this spanning tree to the $2$-factor.
We present an algorithm, BIGCYCLE, which begins by finding a $2$-factor
with at most $\frac{n}{7}$ cycles. Then, it applies Proposition \ref{t3}
to generate a spanning Eulerian subgraph from this $2$-factor containing
at most $n+2\times(\frac{n}{7}-1)=\frac{9}{7}n-2$ edges.

BIGCYCLE first shrinks every $4$-cycle in the graph, then it generates
a $2$-factor in the condensed graph. If the resulting $2$-factor
has no $6$-cycles, then we can expand the $4$-cycles and this will
be our solution. If the $2$-factor does have a $6$-cycle, then the
algorithm contracts either this $6$-cycle or a larger subgraph that
includes this $6$-cycle. We are able to iterate this process until
we find a $2$-factor in the compressed graph with no $6$-cycles that have contracted portions. We will later define the term ``organic'' to describe these cycles.
At this point, the algorithm is able to expand the compressed graph
back to its original state, maintaining a $2$-factor with relatively
few cycles. Theorem \ref{t13} in Section 3.5 proves that this $2$-factor
has at most $\frac{n}{7}$ cycles.

\section{A $\frac{9}{7}$-Approximation Algorithm for Graphic TSP in Cubic Bipartite
Graphs}

\subsection{Overview}

In a graph with no $4$-cycles (squares), all $2$-factors will have
an average cycle length of at least $6$, so all $2$-factors will
have at most $\frac{n}{6}$ cycles, which results in a $\frac{4}{3}$-approximation after applying Proposition \ref{t3}.
In order to improve our approximation guarantee, we need to target
$6$-cycles, as well as $4$-cycles. The algorithm we present finds
a square-free $2$-factor in which every $6$-cycle can be put in correspondence with
a distinct cycle of size $8$ or larger. Then, we can find a $2$-factor
in which every large cycle and its corresponding $6$-cycles have
average cycle length of at least $7$ via an amortized analysis
over the compressing iterations
(Lemma \ref{t12} in Section 3.4). We
then show that this is enough to conclude that the $2$-factor contains
at most $\frac{n}{7}$ cycles (Theorem \ref{t13} in Section 3.5). This is
primary contribution of this paper.

A method used throughout this paper is to systematically replace certain
subgraphs containing $4$-cycles and $6$-cycles with other subgraphs.
We will refer to these replacement subgraphs as ``gadgets''. To keep
track of portions of the graph that have not been altered by these
gadgets, we define the term ``organic'' as follows.
\begin{definition} \label{t4}
A subgraph is organic if it consists entirely of nodes and edges contained
in the original graph. For a single edge to be organic, both its end-nodes
must be organic.
\end{definition}
We also give a formal definition of the term ``gadget'':
\begin{definition} \label{t5}
A gadget is a subgraph that is inserted into the graph by the BIGCYCLE
algorithm in place of a different subgraph. Examples of gadgets are
shown in Section \ref{gadgetsection}. Gadgets are used to replace other
subgraphs containing $4$- or $6$-cycles.
\end{definition}
In Section \ref{gadgetsection}, we introduce the gadgets used in the BIGCYCLE
algorithm. The BIGCYCLE algorithm is defined in Section \ref{algorithmsection}. When our $2$-factor contains $4$-cycles and organic $6$-cycles,this algorithm uses the gadgets from Section \ref{gadgetsection} to condense our graph and remove these cycles.
BIGCYCLE repeats this process (condensing $4$-cycles that appear along
the way) and compute a new $2$-factor in the condensed graph until
we obtain a $2$-factor with no organic $6$-cycles. In Section \ref{accounting} we examine expansions that can introduce $6$-cycles into our $2$-factor and show that while expanding the graph can create a small number of new $6$-cycles in
our $2$-factor, we are able to account for them, ensuring that
the bounds described in the previous paragraph must hold.

\subsection{Gadgets} \label{gadgetsection}

In this section, we present the subgraphs that will
be replaced with gadgets by the algorithm. In total, there are $3$ gadgets
to replace $4$-cycles and $6$ gadgets to replace $6$-cycles. We
will give these configurations the names $S_{1}$, $S_{2}$, $S_{3}$,
$H_{1}$, $H_{2}$, $H_{3}$, $H_{4}$, $H_{5}$, and $H_{6}$. The
gadget that replaces a configuration $X$ will be called $X'$.

First, we introduce the gadget we use to replace squares whose outgoing
edges are incident on four distinct vertices.

\begin{figure}[ht]
\begin{minipage}[t]{0.45\columnwidth}%
\begin{center}
\tikzstyle{node}=[circle, draw, fill=black!50,                         inner sep=0pt, minimum width=4pt]
\begin{tikzpicture}[thick,scale=0.18] 	\node [node] (a) at (-2,2) {}; 	\node [node] (b) at (2,2) {}; 	\node [node] (c) at (2,-2) {}; 	\node [node] (d) at (-2,-2) {}; 	\draw (a) -- (b) -- (c) -- (d) -- (a); 	\coordinate (1) at  ($ (a)+ (-1,1) $); 	 \coordinate (2) at  ($ (b)+ (1,1) $); 	\coordinate (3) at  ($ (c)+ (1,-1) $); 	\coordinate (4) at  ($ (d)+ (-1,-1) $); 	\draw (a) to node [label={[xshift=.2cm]1}] {} (1); 	\draw (b) to node [label={[xshift=-.2cm]2}] {} (2); 	\draw (c) to node [label={[xshift=-.2cm,yshift=-.5cm]3}] {} (3); 	\draw (d) to node [label={[xshift=.2cm,yshift=-.5cm]4}] {} (4);
\end{tikzpicture}
\par\end{center}

\caption{A square with four distinct neighbors: $S_{1}$}
\end{minipage}\hfill{}%
\begin{minipage}[t]{0.45\columnwidth}%
\begin{center}
\tikzstyle{node}=[circle, draw, fill=black!50,                         inner sep=0pt, minimum width=4pt]
\begin{tikzpicture}[auto,thick, scale=.18]	 	\node [node] (a) at (0,2) {}; 	\node [node] (b) at (0,-2) {}; 	\draw (a) -- (b); 	 \coordinate (1) at  ($ (a)+ (-1,1) $); 	\draw (a) to node [label={[xshift=0cm]1}] {} (1); 	\coordinate (3) at  ($ (a)+ (1,1) $); 	 \draw (a) to node [label={[xshift=.45cm, yshift=-.4cm]3}] {} (3); 	\coordinate (2) at  ($ (b)+ (-1,-1) $); 	\draw (b) to node [label={[xshift=-.4cm, yshift = .25cm]2}] {} (2); 	\coordinate (4) at  ($ (b)+ (1,-1) $); 	\draw (b) to node [label={[yshift=-.2cm]4}] {} (4); \end{tikzpicture}
\par\end{center}

\begin{center}
\caption{The gadget that replaces this configuration: $S'_{1}$}

\par\end{center}%
\end{minipage}
\end{figure}

Next, we introduce the $S_{2}$, used to replace squares with two exiting edges connected to a common vertex.

\begin{figure}[H]
\begin{minipage}[t]{0.45\columnwidth}%
\begin{center}
\tikzstyle{node}=[circle, draw, fill=black!50,                         inner sep=0pt, minimum width=4pt]
\begin{tikzpicture}[thick,scale=0.2] 	\node [node] (a) at (-2,2) {}; 	\node [node] (b) at (2,2) {}; 	\node [node] (c) at (2,-2) {}; 	\node [node] (d) at (-2,-2) {}; 	\node [node] (e) at (0,4) {}; 	\draw (a) -- (b) -- (c) -- (d) -- (a); 	\draw (a) -- (e); 	\draw (c) to[out=0,in=50, distance=5cm] (e); 	\coordinate (1) at  ($ (e)+ (0,1.5) $); 	\coordinate (2) at  ($ (b)+ (1,1) $); 	\coordinate (4) at  ($ (d)+ (-1,-1) $); 	\draw (e) to node [label={[xshift=-.2cm]1}] {} (1); 	\draw (b) to node [label={[xshift=-.10cm]2}] {} (2); 	\draw (d) to node [label={[xshift=.1cm,yshift=-.5cm]3}] {} (4);
\end{tikzpicture}
\par\end{center}

\caption{A square with three distinct neighbors: $S_{2}$}
\end{minipage}\hfill{}%
\begin{minipage}[t]{0.45\columnwidth}%
\begin{center}
\tikzstyle{node}=[circle, draw, fill=black!50,                         inner sep=0pt, minimum width=4pt]
\begin{tikzpicture}[thick,scale=.8] 	\node [node] (a) at (-1,.5) {}; 	\coordinate (1) at  ($ (a)+ (-.75,.75) $); 	\coordinate (2) at  ($ (a)+ (.75,.75) $); 	\coordinate (3) at  ($ (a)+ (0,-.8) $); 	\draw (a) to node [label={[xshift=.05cm]1}] {} (1); 	\draw (a) to node [label={[xshift=-.1cm]2}] {} (2); 	\draw (a) to node [label={[xshift=-.2cm,yshift=-.3cm]3}] {} (3); \end{tikzpicture} 
\par\end{center}

\begin{center}
\caption{The gadget which replaces this configuration: $S'_2$}

\par\end{center}%
\end{minipage}
\end{figure}

We also introduce the $S_{3}$, which is used to replace squares whose outgoing edges are incident on only two vertices.

\begin{figure}[H]

\begin{minipage}[t]{0.45\columnwidth}%
\begin{center}
\tikzstyle{node}=[circle, draw, fill=black!50,                         inner sep=0pt, minimum width=4pt]
\begin{tikzpicture}[thick,scale=0.125] 	\node [node] (a) at (-2,2) {}; 	\node [node] (b) at (2,2) {}; 	\node [node] (c) at (2,-2) {}; 	\node [node] (d) at (-2,-2) {}; 	\node [node] (e) at (0,4) {}; 	\node [node] (f) at (0,-4) {};   	\draw (a) -- (b) -- (c) -- (d) -- (a); 	\draw (a) -- (e); 	\draw (d) -- (f); 	\draw (c) to[out=0,in=50, distance=5cm] (e); 	\draw (b) to[out=-50,in=0, distance=5cm] (f); 	\node [node] (1) at  ($ (e)+ (0,3) $) {A}; 	\node [node] (2) at  ($ (f)+ (0,-3) $) {B};  \coordinate (a1) at  ($ (1)+ (-2,2) $);   \coordinate (a2) at  ($ (1)+ (2,2) $);	\draw (e) to (1); 	\draw (f) to (2);
	\draw (1) to node [label={1}] {} (a1); 	\draw (1) to  node [label={2}] {} (a2); 	\draw (2) to  node [label={[yshift=-.5cm]3}] {} ($ (2)+ (-2,-2) $); 	\draw (2) to  node [label={[yshift=-.5cm]4}] {} ($ (2)+ (2,-2) $);
\end{tikzpicture}
\par\end{center}

\caption{A square with two distinct neighbors: $S_{3}$}
\end{minipage}\hfill{}%
\begin{minipage}[t]{0.45\columnwidth}%
\begin{center}
\tikzstyle{node}=[circle, draw, fill=black!50,                         inner sep=0pt, minimum width=4pt]
\begin{tikzpicture}[thick,scale=0.125]

 	\node [node] (A) at (0,5.5) {A}; 	\node [node] (B) at (0,-5.5) {B}; 	\draw (A) to (B);
	\draw (A) to node [label={2}] {} ($ (A) + (2,2) $); 	\draw (A) to node [label={1}] {} ($ (A) + (-2,2) $); 	\draw (B) to node [label={[yshift=-.5cm]4}] {} ($ (B) + (2,-2) $); 	\draw (B) to node [label={[yshift=-.5cm]3}] {} ($ (B) + (-2,-2) $); \end{tikzpicture}
\par\end{center}

\begin{center}
\caption{The super-edge which replaces this configuration: $S'_{3}$ \label{s3p}}

\par\end{center}%
\end{minipage}
\end{figure}

The first gadget used to replace $6$-cycles is two super-vertices which replace
a simple $6$-cycle, $H_{1}$.
\begin{figure}[ht]
\begin{minipage}[t]{0.45\columnwidth}%
\begin{center}
\tikzstyle{node}=[circle, draw, fill=black!50,                         inner sep=0pt, minimum width=4pt]
\begin{tikzpicture}[thick,scale=.4] 	\node [node] (a) at (-1,.5) {}; 	\node [node] (b) at (0,1.2) {}; 	\node [node] (c) at (1,.5) {}; 	\node [node] (d) at (1,-.5) {}; 	\node [node] (e) at (0,-1.2) {}; 	\node [node] (f) at (-1,-.5) {};
	\draw (a) -- (b) -- (c) -- (d) -- (e) -- (f) -- (a); 	\coordinate (1) at  ($ (a)+ (-.75,.75) $); 	\coordinate (2) at  ($ (b)+ (0,.8) $); 	\coordinate (3) at  ($ (c)+ (.75,.75) $); 	\coordinate (4) at  ($ (d)+ (.75,-.75) $); 	\coordinate (5) at  ($ (e)+ (0,-.8) $); 	\coordinate (6) at  ($ (f) + (-.75,-.75) $); 	\draw (a) to node [label={[xshift=.1cm]1}] {} (1); 	\draw (b) to node [label={[xshift=-.0cm]2}] {} (2); 	\draw (c) to node [label={[xshift=-.1cm]3}] {} (3); 	\draw (d) to node [label={[xshift=.3cm,yshift=-.2cm]4}] {} (4); 	\draw (e) to node [label={[xshift=-.2cm,yshift=-.3cm]5}] {} (5); 	\draw (f) to node [label={[xshift=-.1cm]6}] {} (6); \end{tikzpicture}
\par\end{center}

\caption{A simple $6$-cycle: $H_{1}$}
\end{minipage}\hfill{}%
\begin{minipage}[t]{0.45\columnwidth}%
\begin{center}
\tikzstyle{node}=[circle, draw, fill=black!50,                         inner sep=0pt, minimum width=4pt]
\begin{tikzpicture}[thick,scale=.6] 	\node [node] (a) at (-1,.5) {};
	\node [node] (b) at (1,.5) {};
	\coordinate (1) at  ($ (a)+ (-.75,.75) $); 	\coordinate (2) at  ($ (b)+ (-.75,.75) $); 	\coordinate (3) at  ($ (a)+ (.75,.75) $); 	\coordinate (4) at  ($ (b)+ (.75,.75) $); 	\coordinate (5) at  ($ (a)+ (0,-.8) $); 	\coordinate (6) at  ($ (b) + (0,-.8) $); 	\draw (a) to node [label={[xshift=.05cm]1}] {} (1); 	\draw (b) to node [label={[xshift=.05cm]2}] {} (2); 	 \draw (a) to node [label={[xshift=-.1cm]3}] {} (3); 	\draw (b) to node [label={[xshift=-.1cm]4}] {} (4); 	\draw (a) to node [label={[xshift=-.2cm,yshift=-.3cm]5}] {} (5); 	\draw (b) to node [label={[xshift=-.2cm,yshift=-.3cm]6}] {} (6); \end{tikzpicture}
\par\end{center}

\begin{center}
\caption{The gadget which replaces the $6$-cycle: $H'_{1}$}

\par\end{center}%
\end{minipage}
\end{figure}

The remaining gadgets are special cases of $6$-cycles. Note
that every $H_{2}$ contains a $H_{1}$, all $H_{3}$s contain a $H_{2}$,
and $H_{4}$s, $H_{5}$s, and $H_{6}$s are special cases of $H_{3}$s.

The motivation to use these additional gadgets comes out of necessity,
to prevent large numbers of $6$-cycles from being introduced into
the $2$-factor during the expansion phase of the algorithm. For example,
Figures \ref{badex1} and \ref{badex2} in Section \ref{accounting} document an expansion that turns an
$x+y+4$-cycle in the cycle cover passing through a gadget which replaced a $H_{1}$
into two cycles of lengths $x+3$ and $y+5$. In Section \ref{algorithmsection} we specify that the algorithm
will condense $H_{2}$s before $H_{1}$s.  This ensures that $y$,
the length of a path, is at least $3$, meaning that the $y+5$-cycle
is not a $6$-cycle. The motivation for introducing the remaining
specialized gadgets is similar.

\begin{figure}[ht]
\begin{minipage}[t]{0.45\columnwidth}%
\begin{center}
\tikzstyle{node}=[circle, draw, fill=black!50,                         inner sep=0pt, minimum width=4pt]
\begin{tikzpicture}[thick,scale=.5] 	\node [node] (a) at (-1,.5) {}; 	\node [node] (b) at (0,1.2) {}; 	\node [node] (c) at (1,.5) {}; 	\node [node] (d) at (1,-.5) {}; 	\node [node] (e) at (0,-1.2) {}; 	\node [node] (f) at (-1,-.5) {}; 	\node [node] (g) at (0,.5) {}; 	\node [node] (h) at (0,-.5) {};
	\draw (a) -- (b) -- (c) -- (d) -- (e) -- (f) -- (a); 	\draw (b) -- (g) -- (h) -- (e); 	\coordinate (1) at  ($ (a)+ (-.75,.75) $); 	\coordinate (2) at  ($ (g)+ (-1.2, 1.2) $); 	\coordinate (3) at  ($ (c)+ (.75,.75) $); 	\coordinate (4) at  ($ (d)+ (.75,-.75) $); 	\coordinate (5) at  ($ (h)+ (1.2,-1.2) $); 	\coordinate (6) at  ($ (f) + (-.75,-.75) $); 	\draw (a) to node [label={[xshift=.05cm]1}] {} (1); 	\draw (g) to node [label={[xshift=.2cm]2}] {} (2); 	\draw (c) to node [label={[xshift=-.2cm]3}] {} (3); 	\draw (d) to node [label={[xshift=.3cm,yshift=-.2cm]4}] {} (4); 	\draw (h) to node [label={[xshift=-.1cm, yshift=-.4cm]5}] {} (5); 	\draw (f) to node [label={[xshift=-.1cm]6}] {} (6); \end{tikzpicture}
\par\end{center}

\caption{Two $6$-cycles with $3$ common edges: $H_{2}$}
\end{minipage}\hfill{}%
\begin{minipage}[t]{0.45\columnwidth}%
\begin{center}
\tikzstyle{node}=[circle, draw, fill=black!50,                         inner sep=0pt, minimum width=4pt]
\begin{tikzpicture}[thick,scale=.6] 	\node [node] (a) at (-1,.5) {};
	\node [node] (b) at (1,.5) {};
	\coordinate (1) at  ($ (a)+ (-.75,.75) $); 	\coordinate (2) at  ($ (b)+ (-.75,.75) $); 	\coordinate (3) at  ($ (a)+ (.75,.75) $); 	\coordinate (4) at  ($ (b)+ (.75,.75) $); 	\coordinate (5) at  ($ (a)+ (0,-.8) $); 	\coordinate (6) at  ($ (b) + (0,-.8) $); 	\draw (a) to node [label={[xshift=.05cm]1}] {} (1); 	\draw (b) to node [label={[xshift=.05cm]4}] {} (2); 	 \draw (a) to node [label={[xshift=-.1cm]2}] {} (3); 	\draw (b) to node [label={[xshift=-.1cm]5}] {} (4); 	\draw (a) to node [label={[xshift=-.2cm,yshift=-.3cm]3}] {} (5); 	\draw (b) to node [label={[xshift=-.2cm,yshift=-.3cm]6}] {} (6); \end{tikzpicture}
\par\end{center}

\begin{center}
\caption{The gadget which replaces the $H_2$: $H'_{2}$}

\par\end{center}%
\end{minipage}
\end{figure}

\begin{figure}[ht]
\begin{minipage}[t]{0.45\columnwidth}%
\begin{center}
\tikzstyle{node}=[circle, draw, fill=black!50,                         inner sep=0pt, minimum width=4pt]
\begin{tikzpicture}[thick,scale=.5] 	\node [node] (a) at (-1,.5) {}; 	\node [node] (b) at (0,1.2) {}; 	\node [node] (c) at (1,.5) {}; 	\node [node] (d) at (1,-.5) {}; 	\node [node] (e) at (0,-1.2) {}; 	\node [node] (f) at (-1,-.5) {}; 	\node [node] (g) at (0,.5) {}; 	\node [node] (h) at (0,-.5) {}; 	\node [node] (i) at  ($ (a)+ (.5,1.5) $) {}; 	\node [node] (j) at  ($ (d)+ (-.5,2.5) $) {};
	\draw (a) -- (b) -- (c) -- (d) -- (e) -- (f) -- (a); 	\draw (b) -- (g) -- (h) -- (e); 	\draw (a) -- (i) -- (j); 	\draw (d) to[out=0,in=0, distance=.5cm] (j); 	\coordinate (1) at  ($ (j) + (.75,.75) $); 	\coordinate (2) at  ($ (g)+ (-1.2, 1.2) $); 	 \coordinate (3) at  ($ (c)+ (.75,.75) $); 	\coordinate (4) at  ($ (i)+ (-.75,.75) $); 	\coordinate (5) at  ($ (h)+ (1.2,-1.2) $); 	\coordinate (6) at  ($ (f) + (-.75,-.75) $); 	\draw (i) to node [label=6] {} (4); 	\draw (j) to node [label=1] {} (1); 	\draw (g) to node [label={[xshift=-.35cm,yshift=-.05cm]2}] {} (2); 	\draw (c) to node [label={[xshift=0cm]3}] {} (3); 	 \draw (h) to node [label={[xshift=.15cm, yshift=-.15cm]5}] {} (5); 	\draw (f) to node [label={[xshift=-.1cm]4}] {} (6); \end{tikzpicture}
\par\end{center}

\caption{A specialized configuration containing three overlapping $6$-cycles:
$H_{3}$}
\end{minipage}\hfill{}%
\begin{minipage}[t]{0.45\columnwidth}%
\begin{center}
\tikzstyle{node}=[circle, draw, fill=black!50,                         inner sep=0pt, minimum width=4pt]
\begin{tikzpicture}[thick,scale=0.5] 	\coordinate (4) at (-2,1); 	\coordinate (3) at (-2,-1); 	\draw (4) to node [label={[xshift=.15cm,yshift=.3cm]4}] {} (3); 	\draw (4) to node [label={[xshift=.15cm,yshift=-.7cm]3}] {} (3); 	\coordinate (5) at (0,1); 	\coordinate (1) at (0,-1); 	\draw (5) to node [label={[xshift=.15cm,yshift=.3cm]5}] {} (1); 	\draw (5) to node [label={[xshift=.15cm,yshift=-.7cm]1}] {} (1); 	\coordinate (6) at (2,1); 	\coordinate (2) at (2,-1); 	\draw (6) to node [label={[xshift=.15cm,yshift=.3cm]6}] {} (2); 	\draw (6) to node [label={[xshift=.15cm,yshift=-.7cm]2}] {} (2); \end{tikzpicture}
\par\end{center}

\begin{center}
\caption{The gadget which replaces the $H_3$: $H'_{3}$}

\par\end{center}%
\end{minipage}
\end{figure}

\begin{figure}[H]
\begin{minipage}[t]{0.45\columnwidth}%
\begin{center}
\tikzstyle{node}=[circle, draw, fill=black!50,                         inner sep=0pt, minimum width=4pt]
\begin{tikzpicture}[thick,scale=.5] 	\node [node] (a) at (-1,.5) {}; 	\node [node] (b) at (0,1.2) {}; 	\node [node] (c) at (1,.5) {}; 	\node [node] (d) at (1,-.5) {}; 	\node [node] (e) at (0,-1.2) {}; 	\node [node] (f) at (-1,-.5) {}; 	\node [node] (g) at (0,.5) {}; 	\node [node] (h) at (0,-.5) {}; 	\node [node] (i) at  ($ (a)+ (.5,1.5) $) {}; 	\node [node] (j) at  ($ (d)+ (-.5,2.5) $) {}; 	\node [node] (w1) at ($ (j) + (2.5,-.7) $) {}; 	\node [node] (w2) at ($ (i) + (-2.5,-.7) $) {};
	\draw (a) -- (b) -- (c) -- (d) -- (e) -- (f) -- (a); 	\draw (b) -- (g) -- (h) -- (e); 	\draw (a) -- (i) -- (j) -- (w1) -- (g); 	\draw (i) -- (w2) -- (h); 	\draw (d) to[out=0,in=0, distance=.5cm] (j); 	\coordinate (1) at  ($ (w2) + (-1,0) $); 	\coordinate (2) at  ($ (w1)+ (1, 0) $); 	\coordinate (3) at  ($ (c)+ (1.5,-.75) $); 	\coordinate (4) at  ($ (f)+ (-1.5,.0) $); 	\draw (w2) to node [label=1] {} (1); 	\draw (w1) to node [label=2] {} (2); 	\draw (c) to node [label={[xshift=.1cm]3}] {} (3); 	\draw (f) to node [label={[xshift=-.1cm]4}] {} (4); \end{tikzpicture}
\par\end{center}

\caption{The $H_{4}$}
\end{minipage}\hfill{}%
\begin{minipage}[t]{0.45\columnwidth}%
\begin{center}
\tikzstyle{node}=[circle, draw, fill=black!50,                         inner sep=0pt, minimum width=4pt]
\begin{tikzpicture}[auto,thick, scale=.3]	 	\node [node] (a) at (0,2) {}; 	\node [node] (b) at (0,-2) {}; 	\draw (a) -- (b); 	\coordinate (1) at  ($ (a)+ (-1,1) $); 	\draw (a) to node [label={[xshift=0cm]1}] {} (1); 	\coordinate (3) at  ($ (a)+ (1,1) $); 	\draw (a) to node [label={[xshift=.5cm, yshift=-.5cm]3}] {} (3); 	\coordinate (2) at  ($ (b)+ (-1,-1) $); 	\draw (b) to node [label={[xshift=-.5cm, yshift = .5cm]2}] {} (2); 	\coordinate (4) at  ($ (b)+ (1,-1) $); 	\draw (b) to node [label={4}] {} (4); \end{tikzpicture} 
\par\end{center}

\begin{center}
\caption{The gadget which replaces $H_{4}$: $H'_{4}$}

\par\end{center}%
\end{minipage}
\end{figure}

\begin{figure}[H]
\begin{minipage}[t]{0.45\columnwidth}%
\begin{center}
\tikzstyle{node}=[circle, draw, fill=black!50,                         inner sep=0pt, minimum width=4pt]
\begin{tikzpicture}[thick,scale=.45] 	\node [node] (a) at (-1,.5) {}; 	\node [node] (b) at (0,1.2) {}; 	\node [node] (c) at (1,.5) {}; 	\node [node] (d) at (1,-.5) {}; 	\node [node] (e) at (0,-1.2) {}; 	\node [node] (f) at (-1,-.5) {}; 	\node [node] (g) at (0,.5) {}; 	\node [node] (h) at (0,-.5) {}; 	\node [node] (i) at  ($ (a)+ (.5,1.5) $) {}; 	\node [node] (j) at  ($ (d)+ (-.5,2.5) $) {}; 	\node [node] (w1) at ($ (j) + (2.5,-.7) $) {}; 	\node [node] (w2) at ($ (i) + (-2.5,-.7) $) {};
	\draw (a) -- (b) -- (c) -- (d) -- (e) -- (f) -- (a); 	\draw (b) -- (g) -- (h) -- (e); 	\draw (a) -- (i) -- (j) -- (w1) -- (g); 	\draw (i) -- (w2) -- (h); 	\draw (d) to[out=0,in=0, distance=.5cm] (j);
	\coordinate (2) at  ($ (w1)+ (1, 0) $); 	\coordinate (3) at  ($ (c)+ (1.5,-.75) $); 	\coordinate (4) at  ($ (f)+ (-1.5,.0) $);
	\draw (w1) to node [label=2] {} (2); 	\draw (f) to node [label={[xshift=-.1cm]3}] {} (4);
	\node [node] (1-3) at ($ (f) + (0,-1) $) {}; 	\coordinate (1) at  ($ (1-3) + (0,-1) $);
	\draw (w2) to[out=180,in=180] (1-3) to (1-3) to[out=-30,in=0, distance=2cm] (c); 	\draw (1-3) to node [label={[xshift=-.2cm,yshift=-.3cm]1}] {} (1);
\end{tikzpicture}
\par\end{center}

\caption{The $H_{5}$}
\end{minipage}\hfill{}%
\begin{minipage}[t]{0.45\columnwidth}%
\begin{center}
\tikzstyle{node}=[circle, draw, fill=black!50,                         inner sep=0pt, minimum width=4pt]
\begin{tikzpicture}[thick,scale=1] 	\node [node] (a) at (-1,.5) {}; 	\coordinate (1) at  ($ (a)+ (-.75,.75) $); 	\coordinate (2) at  ($ (a)+ (.75,.75) $); 	\coordinate (3) at  ($ (a)+ (0,-.8) $); 	\draw (a) to node [label={[xshift=.05cm]1}] {} (1); 	\draw (a) to node [label={[xshift=-.1cm]2}] {} (2); 	\draw (a) to node [label={[xshift=-.2cm,yshift=-.3cm]3}] {} (3); \end{tikzpicture} 
\par\end{center}

\begin{center}
\caption{The gadget which replaces $H_{5}$: $H'_{5}$}

\par\end{center}%
\end{minipage}
\end{figure}

\begin{figure}[H]
\begin{minipage}[t]{0.45\columnwidth}%
\begin{center}
\tikzstyle{node}=[circle, draw, fill=black!50,                         inner sep=0pt, minimum width=4pt]
\begin{tikzpicture}[thick,scale=.45] 	\node [node] (a) at (-1,.5) {}; 	\node [node] (b) at (0,1.2) {}; 	\node [node] (c) at (1,.5) {}; 	\node [node] (d) at (1,-.5) {}; 	\node [node] (e) at (0,-1.2) {}; 	\node [node] (f) at (-1,-.5) {}; 	\node [node] (g) at (0,.5) {}; 	\node [node] (h) at (0,-.5) {}; 	\node [node] (i) at  ($ (a)+ (.5,1.5) $) {}; 	\node [node] (j) at  ($ (d)+ (-.5,2.5) $) {}; 	\node [node] (w1) at ($ (j) + (2.5,-.7) $) {}; 	\node [node] (w2) at ($ (i) + (-2.5,-.7) $) {};
	\draw (a) -- (b) -- (c) -- (d) -- (e) -- (f) -- (a); 	\draw (b) -- (g) -- (h) -- (e); 	\draw (a) -- (i) -- (j) -- (w1) -- (g); 	\draw (i) -- (w2) -- (h); 	\draw (d) to[out=0,in=0, distance=.5cm] (j);
	\node [node] (left) at ($ (w2) + (0,-1) $) {}; 	\node [node] (right) at ($ (w1) + (0,-1) $) {};
	\draw (w2) to (left) to[out=290,in=-30,distance=4cm] (c); 	\draw (w1) to (right) to[out=280,in=250,distance=2.5cm] (f);
	\coordinate (1) at ($ (left) + (-1,0) $); 	\draw (left) to node [label={[xshift=.05cm]1}] {} (1);
	\coordinate (2) at ($ (right) + (1,0) $); 	\draw (right) to node [label={[xshift=.05cm]2}] {} (2); \end{tikzpicture}
\par\end{center}

\caption{The $H_{6}$}
\end{minipage}\hfill{}%
\begin{minipage}[t]{0.45\columnwidth}%
\begin{center}
\tikzstyle{node}=[circle, draw, fill=black!50,                         inner sep=0pt, minimum width=4pt]
\begin{tikzpicture}[thick,scale=1	] 	\coordinate (1) at (0,1); 	\coordinate (2) at (0,-1); 	\draw (1) to node [label={[xshift=.15cm,yshift=1cm]1}] {} (2); 	\draw (1) to node [label={[xshift=.15cm,yshift=-1.4cm]2}] {} (2);
\end{tikzpicture}
\par\end{center}

\begin{center}
\caption{The gadget which replaces $H_{6}$: $H'_{6}$}

\par\end{center}%
\end{minipage}
\end{figure}

In the next subsection, we present a detailed description of the algorithm.

\subsection{The Algorithm} \label{algorithmsection}

Listing 1 presents pseudocode for the BIGCYCLE algorithm. The remainder
of this section explains the details of the algorithm, broken up into
three subroutines, and presents motivation for the operations performed
by the algorithm. The COMPRESS, EXPAND, and DOUBLETREE subroutines called by BIGCYCLE
are described in the following three subsections.

\begin{algorithm}[H] 
\caption{BIGCYCLE($G$)\label{bigcyclealgo}}

\begin{lyxcode}
Input:~An~undirected,~unweighted,~cubic,~bipartite~graph,~$G=(V,E)$~\\
$F_{compressed}\leftarrow$COMPRESS($G$)\\
$F\leftarrow$EXPAND($F_{compressed}$)\\
$TSP\leftarrow$DOUBLETREE($G$,$F$)\\
Return~$TSP$\end{lyxcode}
\end{algorithm}

\subsubsection{Finding a ``good'' 2-factor in the condensed graph $G_{k}$}

We start the algorithm by receiving a connected cubic bipartite graph.
Call this graph $G_{0}$. If $G_{0}$ is a $K_{3,3}$ then we compute a $2$-factor in this
graph, which will be a Hamiltonian cycle, and return this cycle as
our solution. Otherwise, we search for $4$-cycles that are not contained in $K_{3,3}$s and replace them
with their corresponding gadgets until we are returned a graph with
no squares except possibly inside of $K_{3,3}$s. We replace $S_3$ subgraphs first, followed by $S_2$s and $S_1$s so as to replace the most specialized subgraphs first. Let $i$ be the
number of square compressions made, and let $G_{i}$ be the compressed
graph at the end of this process. Next, construct a $2$-factor, $F_{i}$,
in $G_{i}$. When we construct $2$-factors throughout the algorithm,
we do so by decomposing the graph into $3$ edge-disjoint perfect
matchings and taking the union of the two perfect matching containing
the fewest $S'_{3}$ gadgets, shown in Figure \ref{s3p}.
These two perfect matchings form a $2$-factor with limited
potential to introduce organic $6$-cycles of the type shown in Figure
\ref{s3p6cycle} (Section 3.2). If $F_{i}$ contains no organic $6$-cycles, then
we advance to the next phase of the algorithm, described in the next subsection. In this case, $k=i$.

If $F_{i}$ does contain an organic $6$-cycle, $C$, then we check
if the current compressed graph $G_{i}$ contains organic subgraphs
that can be replaced by gadgets in the following order (ordered from
most specialized to most general): $H_{6}$, $H_{5}$, $H_{4}$, $H_{3}$,
$H_{2}$, $H_{1}$. We choose the first organic configuration on the
list (the most specialized configuration) we can find in $G_{i}$
and replace this configuration with the corresponding gadget, outputting
graph $G_{i+1}$ to reflect this change. The order of choosing subgraphs to replace is useful in accounting for the average length
of the cycles in the final $2$-factor, as shown in the proof of Lemma~\ref{lem-corr}. We then search for $4$-cycles that are not contained in $K_{3,3}$s
and replace them with their corresponding gadgets until we have removed
any $4$-cycles generated as a consequence of replacing a subgraph
with one of our gadgets, obtaining a new compressed graph $G_{j}$, where $j-i+1$ is the number of $4$-cycles compressed.
We construct a new $2$-factor $F_{j}$ and repeat the process in
this paragraph until we have a $2$-factor $F_{k}$ with no organic
$6$-cycles, in a condensed graph $G_{k}$, where $k$ is the total
number of gadget replacement operations performed during this phase
of the algorithm. This process is performed by the COMPRESS subroutine
in Listing 1.

\subsubsection{Expanding a 2-factor in $G_{k}$ into a 2-factor in $G_{0}$}

We will describe the process of expanding $F_{k}$ and $G_{k}$ so that we get back
to the original graph $G_{0}$ with a desirable $2$-factor $F_{0}$ in more detail.

We will reverse the process described in the previous subsection by replacing
our gadgets in compressed graph $G_{i}$ with the original configuration
from the earlier graph $G_{i-1}$ in the reverse order of that in
which we replaced the configurations. In other words, the gadgets
we inserted last are those which we first replace with their original configuration.
We call this process ``expanding'' because each one of these operations
adds vertices and edges to the graph. After we have made each replacement
to expand the graph, $F_{i}$ is no longer a $2$-factor in $G_{i-1}$
because the new nodes added by the most recent expansion step are
not covered by $F_{i}$. However, we can add edges to $F_{i}$ so
that it becomes a $2$-factor, $F_{i-1}$ in the graph after this
expansion step. It may not be immediately clear that this is always
possible. In fact, one of the bigger challenges in developing this
algorithm was choosing a set of gadgets where this property holds.
Figures \ref{gadgetex1} and \ref{gadgetex2} show an example of how this process works.

\begin{figure}[ht]
\begin{minipage}[t]{0.45\columnwidth}%
\begin{center}
\tikzstyle{node}=[circle, draw, fill=black!50,                         inner sep=0pt, minimum width=4pt]
\begin{tikzpicture}[thick,scale=.7] 	\node [node] (a) at (-1,.5) {};
	\node [node] (b) at (1,.5) {};
	\coordinate (1) at  ($ (a)+ (-.75,.75) $); 	\coordinate (2) at  ($ (b)+ (-.75,.75) $); 	\coordinate (3) at  ($ (a)+ (.75,.75) $); 	\coordinate (4) at  ($ (b)+ (.75,.75) $); 	\coordinate (5) at  ($ (a)+ (0,-.8) $); 	\coordinate (6) at  ($ (b) + (0,-.8) $); 	\draw (a) to node [label={[xshift=.05cm,yshift=-.05cm]1}] {} (1); 	\draw (b) to node [label={[xshift=.05cm,yshift=-.05cm]2}] {} (2); 	\draw (a) to node [label={[xshift=-.1cm,yshift=-.05cm]3}] {} (3); 	\draw (b) to node [label={[xshift=-.1cm,yshift=-.05cm]4}] {} (4); 	\draw (a) to node [label={[xshift=-.2cm,yshift=-.3cm]5}] {} (5); 	 \draw (b) to node [label={[xshift=-.2cm,yshift=-.3cm]6}] {} (6);
	\draw [line width=.08cm] (a) to (1); 	\draw [line width=.08cm] (a) to (3); 	\draw [line width=.08cm] (b) to (2); 	\draw [line width=.08cm] (b) to (4); 	\draw [dashed, line width=.08cm] (1) to[distance=1cm] (3); 	\draw [dashed, line width=.08cm] (2) to[distance=1cm] (4); \end{tikzpicture}
\par\end{center}

\caption{A pair of super-vertices in $G_{i}$. The bold edges are included
in $2$-factor $F_{i}$. The dashed bold edges represent a path, included
in $F_{i}$. \label{gadgetex1}}
\end{minipage}\hfill{}%
\begin{minipage}[t]{0.45\columnwidth}%
\begin{center}
\tikzstyle{node}=[circle, draw, fill=black!50,                         inner sep=0pt, minimum width=4pt]
\begin{tikzpicture}[thick,scale=.5] 	\node [node] (a) at (-1,.5) {}; 	\node [node] (b) at (0,1.2) {}; 	\node [node] (c) at (1,.5) {}; 	\node [node] (d) at (1,-.5) {}; 	\node [node] (e) at (0,-1.2) {}; 	\node [node] (f) at (-1,-.5) {};
	\draw (a) -- (b) -- (c) -- (d) -- (e) -- (f) -- (a); 	\coordinate (1) at  ($ (a)+ (-.75,.75) $); 	\coordinate (2) at  ($ (b)+ (0,.8) $); 	\coordinate (3) at  ($ (c)+ (.75,.75) $); 	\coordinate (4) at  ($ (d)+ (.75,-.75) $); 	\coordinate (5) at  ($ (e)+ (0,-.8) $); 	\coordinate (6) at  ($ (f) + (-.75,-.75) $); 	\draw (a) to node [label={[xshift=.1cm,yshift=-.1cm]1}] {} (1); 	\draw (b) to node [label={[xshift=-.2cm,yshift=-.2cm]2}] {} (2); 	\draw (c) to node [label={[xshift=-.2cm]3}] {} (3); 	\draw (d) to node [label={[xshift=.1cm,yshift=-.0cm]4}] {} (4); 	\draw (e) to node [label={[xshift=-.2cm,yshift=-.3cm]5}] {} (5); 	\draw (f) to node [label={[xshift=-.1cm]6}] {} (6);
	\draw [line width=.08cm] (a) to (1); 	\draw [line width=.08cm] (b) to (2); 	\draw [line width=.08cm] (c) to (3); 	\draw [line width=.08cm] (d) to (4); 	\draw [line width=.08cm] (b) to (c); 	\draw [line width=.08cm] (a) to (f) to (e) to (d);
	\draw [dashed, line width=.08cm] (1) to[distance=2cm] (3); 	\draw [dashed, line width=.08cm] (2) to[in=50,distance=2.5cm] (4); \end{tikzpicture}
\par\end{center}

\begin{center}
\caption{$2$-factor $F_{i-1}$ after expanding the super-vertices from Fig.
\ref{gadgetex1}. \label{gadgetex2}}

\par\end{center}%
\end{minipage}
\end{figure}

At each expansion, we are able to extend $F_{i}$ into a set of edges
$F_{i-1}$ , which will be a $2$-factor in the expanded graph, $G_{i-1}$.
In order to optimize the performance of $BIGCYCLE$, we must impose
one extra operation in this phase of the algorithm. After each expansion
of a $H_{1}$ that introduces an organic $6$-cycle, $C_{1}$, into
the $2$-factor, we will perform a local search to see if $C_{1}$
and the nearby edges of the newly expanded $2$-factor $F_{i-1}$
contained in the surrounding portion of the graph are in the position
shown in Figure \ref{localimprove1}. If they are, then we update $F_{i-1}$ so that
it covers this region with one fewer cycle, as shown in Figure \ref{localimprove2}.
 In addition to being an effective heuristic to reduce the
number of cycles in our final $2$-factor, this operation allows us
to improve our approximation factor by eliminating an otherwise troubling
corner case. (see Remark \ref{remark1}, Figure \ref{remarkfigure2}, and Lemma \ref{lem-corr}).

\begin{figure}[H]
\begin{minipage}[t]{0.45\columnwidth}%
\begin{center}
\tikzstyle{node}=[circle, draw, fill=black!50,                         inner sep=0pt, minimum width=4pt] \tikzstyle{text_box}=[circle, fill=black!0,                         inner sep=0pt, minimum width=4pt]
\begin{tikzpicture}[thick,scale=.6] 	\node [node] (a) at (-1,.5) {}; 	\node [node] (b) at (0,1.2) {}; 	\node [node] (c) at (1,.5) {}; 	\node [node] (d) at (1,-.5) {}; 	\node [node] (e) at (0,-1.2) {}; 	\node [node] (f) at (-1,-.5) {};
	\draw (a) -- (b) -- (c) -- (d) -- (e) -- (f) -- (a);

	\node [node] (3) at  ($ (c)+ (.75,.75) $) {}; 	\coordinate (4) at  ($ (d)+ (.75,-.75) $); 	\coordinate (5) at  ($ (e)+ (0,-.8) $); 	\node [node] (6) at  ($ (f) + (-.75,-.75) $) {}; 	\draw (c) to node [label={[xshift=-.2cm]}] {} (3); 	\draw (d) to node [label={[xshift=.3cm,yshift=-.2cm]}] {} (4); 	\draw (e) to node [label={[xshift=-.2cm,yshift=-.3cm]}] {} (5); 	\draw (f) to node [label={[xshift=-.1cm]}] {} (6);
	%\draw plot [smooth cycle] coordinates { ($ (b) +(0,.4) $) (.5,.2) ($ (d) + (.3,-.2) $) (0,-.4) ($ (f) + (-.3,-.2) $) (-.5,.2) }; 	%\node [text_box] (txt) at (0,.2) {$B_1$};
	\node [node] (p1) at ($ (5) + (0,-.7) $) {}; 	\node [node] (p2) at ($ (4) + (.6,-.5) $) {};
	\draw [line width=.08cm] (p2) to node [label={[xshift=.2cm]}] {} (3) to (c) to node [label={[xshift=.3cm,yshift=-.25cm]}] {} (d) to node [label={[xshift=.3cm,yshift=-.5cm]}] {} (e) to node [label={[xshift=-.1cm,yshift=-.5cm]}] {} (f) to (6) to node [label={[yshift=-.6cm]}] {} (p1) to node [label={[yshift=-.6cm]}] {} (p2); 	\draw [line width=.08cm, dashed] (a) to[out=170,in=60, distance=3cm] node [label={[xshift=.15cm,yshift=-.6cm]$C_1$}] {} (b); 	\draw [line width=.08cm] (a) to (b);
	\node [node] (e1) at ($ (a) + (-1.5,1.5) $) {}; 	\node [node] (e2) at ($ (b) + (0,2) $) {}; 	\node [node] (1) at ($ (e1) + (-.3,.8) $) {}; 	\node [node] (2) at ($ (e2) + (-.3,.8) $) {}; 	\draw [line width=.08cm] (1) to node [label={[xshift=-.3cm,yshift=-.2cm]}] {} (e1) to (e2) to node [label={[xshift=.2cm,yshift=-.1cm]}] {} (2); 	\draw [line width=.08cm, dashed] (1) to[out=170,in=60, distance=1cm] node [label={[xshift=.0cm]}] {} (2); 
	\draw (6) to (e1) to (e2) to (3);
	\draw (1) to ($ (1) + (-.7,-.2) $); 	\draw (2) to ($ (2) + (.7,.2) $);
	\draw (p1) to ($ (p1) + (0,-.5) $); 	\draw (p2) to ($ (p2) + (.20,-.5) $);
\end{tikzpicture}
\par\end{center}

\caption{The configuration $BIGCYCLE$ searches for after expanding any $H_{1}$
that introduces an organic $6$-cycle. \label{localimprove1}}
\end{minipage}\hfill{}%
\begin{minipage}[t]{0.45\columnwidth}%
\begin{center}
\tikzstyle{node}=[circle, draw, fill=black!50,                         inner sep=0pt, minimum width=4pt] \tikzstyle{text_box}=[circle, fill=black!0,                         inner sep=0pt, minimum width=4pt]
\begin{tikzpicture}[thick,scale=.6] 	\node [node] (a) at (-1,.5) {}; 	\node [node] (b) at (0,1.2) {}; 	\node [node] (c) at (1,.5) {}; 	\node [node] (d) at (1,-.5) {}; 	\node [node] (e) at (0,-1.2) {}; 	\node [node] (f) at (-1,-.5) {};
	\draw (a) -- (b) -- (c) -- (d) -- (e) -- (f) -- (a);

	\node [node] (3) at  ($ (c)+ (.75,.75) $) {}; 	\coordinate (4) at  ($ (d)+ (.75,-.75) $); 	\coordinate (5) at  ($ (e)+ (0,-.8) $); 	\node [node] (6) at  ($ (f) + (-.75,-.75) $) {}; 	\draw (c) to node [label={[xshift=-.2cm]}] {} (3); 	\draw (d) to node [label={[xshift=.3cm,yshift=-.2cm]}] {} (4); 	\draw (e) to node [label={[xshift=-.2cm,yshift=-.3cm]}] {} (5); 	\draw (f) to node [label={[xshift=-.1cm]}] {} (6);
	%\draw plot [smooth cycle] coordinates { ($ (b) +(0,.4) $) (.5,.2) ($ (d) + (.3,-.2) $) (0,-.4) ($ (f) + (-.3,-.2) $) (-.5,.2) }; 	%\node [text_box] (txt) at (0,.2) {$B_1$};
	\node [node] (p1) at ($ (5) + (0,-.7) $) {}; 	\node [node] (p2) at ($ (4) + (.6,-.5) $) {};
	\draw [line width=.08cm] (b) to (c) to node [label={[xshift=.3cm,yshift=-.25cm]}] {} (d) to node [label={[xshift=.3cm,yshift=-.5cm]}] {} (e) to node [label={[xshift=-.1cm,yshift=-.5cm]}] {} (f) to (a); 	\draw [line width=.08cm, dashed] (a) to[out=170,in=60, distance=3cm] node [label={[xshift=.15cm,yshift=-.6cm]$C_1$}] {} (b);
	\node [node] (e1) at ($ (a) + (-1.5,1.5) $) {}; 	\node [node] (e2) at ($ (b) + (0,2) $) {}; 	\node [node] (1) at ($ (e1) + (-.3,.8) $) {}; 	\node [node] (2) at ($ (e2) + (-.3,.8) $) {}; 	\draw [line width=.08cm] (1) to node [label={[xshift=-.3cm,yshift=-.2cm]}] {} (e1); 	\draw [line width=.08cm] (e2) to node [label={[xshift=.2cm,yshift=-.1cm]}] {} (2); 	\draw [line width=.08cm, dashed] (1) to[out=170,in=60, distance=1cm] node [label={[xshift=.0cm]}] {} (2); 
	\draw (e1) to (e2);
	\draw (1) to ($ (1) + (-.7,-.2) $); 	\draw (2) to ($ (2) + (.7,.2) $);
	\draw (p1) to ($ (p1) + (0,-.5) $); 	\draw (p2) to ($ (p2) + (.20,-.5) $);
	\draw [line width=.08cm] (e1) to (6) to node [label={[yshift=-.6cm]}] {} (p1) to node [label={[yshift=-.6cm]}] {} (p2) to (3) to (e2);
\end{tikzpicture}
\par\end{center}

\begin{center}
\caption{The updated $2$-factor $F_{i-1}$ after the configuration in Figure
\ref{localimprove1} is corrected. \label{localimprove2}}

\par\end{center}%
\end{minipage}
\end{figure}

At this point, we can repeat the process of replacing gadgets with
their original configurations and adding edges to the $2$-factor
until we have expanded the graph back to the original input $G_{0}$
and have a $2$-factor, $F_{0}$, in this graph. This process is performed
by the EXPAND subroutine in Algorithm \ref{bigcyclealgo}.

\subsubsection{Obtaining a good final solution by adding edges to $F_{0}$}

We now have a $2$-factor $F_{0}$, which contains at most $k$ cycles. We compress
each cycle into a single node and compute a spanning tree in this compressed graph.
This spanning tree has $k-1$ edges. Then, we add two copies of the edges in this spanning tree
to our $2$-factor $F_{0}$ to obtain
a solution with $n+2(k-1)$ edges. In Section 3 we prove that $F_{0}$
has at most $\frac{n}{7}$ cycles, so this gives us a solution of
at most $\frac{9}{7}n-2$ edges. This process is performed by the
DOUBLETREE subroutine in Listing 1.

\section{Accounting for $6$-Cycles}
\label{accounting}

In the proof of our approximation guarantee, the limitation on producing
a lower approximation factor comes from the possibility that some
proportion of our final $2$-factor's cycles will be of length $6$.
Most operations the algorithm performs while expanding the $2$-factor
from the condensed to the original graph result in
cycles of length $8$ or larger, so in this section we will look at all
operations that create organic $6$-cycles in detail. To account
for $6$-cycles, we show that every organic $6$-cycle can be put
in correspondence with some long cycle of length $8$ or longer. Then,
Lemma \ref{t12} demonstrates that the average cycle length of any long cycle
and its corresponding set of $6$-cycles is sufficiently long to ensure
that our final cycle cover has relatively few cycles, even if some
of them are $6$-cycles.

Figures \ref{badex1} and \ref{badex2}, taking the dashed lines to be paths of lengths
$x$ and $y$, demonstrate how a $(y+7)$-cycle can turn into a
$6$-cycle and a $(y+5)$-cycle after an expansion if $x=3$.

\begin{figure}[ht]
\begin{minipage}[t]{0.45\columnwidth}%
\begin{center}
\tikzstyle{node}=[circle, draw, fill=black!50,                         inner sep=0pt, minimum width=4pt]
\begin{tikzpicture}[thick,scale=.7] 	\node [node] (a) at (0,2) {$SV_{1}$};
	\node [node] (b) at (0,-2) {$SV_{2}$};
	\node [node] (1) at  ($ (a)+ (-.75,-.75) $) {}; 	\node [node] (2) at  ($ (b)+ (-.75,.75) $) {}; 	\node [node] (3) at  ($ (a)+ (.75,-.75) $) {$A_1$};
	\coordinate (4) at  ($ (b)+ (.75,.75) $); 	\coordinate (5) at  ($ (a)+ (0,-.8) $); 	\node [node] (6) at  ($ (b) + (.75,.75) $) {$B_1$};		 	\draw (a) to node [label={[xshift=-.1cm]1}] {} (1); 	\draw (b) to node [label={[xshift=.1cm]2}] {} (2); 	\draw (a) to node [label={[xshift=.2cm]3}] {} (3); 	\draw (b) to node [label={[xshift=-.1cm]6}] {} (6);
	\draw [line width=.08cm] (a) to (1); 	\draw [line width=.08cm] (a) to (3); 	\draw [line width=.08cm] (b) to (2); 	\draw [line width=.08cm] (b) to (6); 	\draw [dashed, line width=.08cm] (1) to node [label={[xshift=-.2cm]x}] {} (2); 	\draw [dashed, line width=.08cm] (3) to node [label={[xshift=.2cm,yshift=.05cm]y}] {} (6);
\coordinate (5) at ($ (a) + (0,1) $); \draw (5) to node [label={[xshift=-.2cm,yshift=-.1cm]5}] {} (a);
\coordinate (4) at ($ (b) + (0,-1) $); \draw (4) to node [label={[xshift=-.2cm,yshift=-.3cm]4}] {} (b);  \end{tikzpicture}
\par\end{center}

\begin{center}
\caption{A cycle in $F_{i}$, a $2$-factor over the condensed graph $G_{i}$.
$SV_{1}$ and $SV_{2}$ are two super-vertices which replaced a standard
hexagon. The dashed lines represent paths of length 3. \label{badex1}}

\par\end{center}%
\end{minipage}\hfill{}%
\begin{minipage}[t]{0.45\columnwidth}%
\begin{center}
\tikzstyle{node}=[circle, draw, fill=black!50,                         inner sep=0pt, minimum width=4pt]
\begin{tikzpicture}[thick,scale=.7]
\clip (-3,-3) rectangle (3.5,3);
\node [node] (a) at (-1,.5) {}; 	\node [node] (b) at (0,1.2) {}; 	\node [node] (c) at (1,.5) {$B_{2}$}; 	\node [node] (d) at (1,-.5) {}; 	\node [node] (e) at (0,-1.2) {}; 	\node [node] (f) at (-1,-.5) {$A_{2}$};
	\draw (a) -- (b) -- (c) -- (d) -- (e) -- (f) -- (a); 	\node [node] (1) at  (-2,1.2) {}; 	\node [node] (2) at  (0,2.2) {}; 	 \node [node] (3) at  ($ (c)+ (.75,.75) $) {$A_{1}$}; 	\coordinate (4) at  ($ (d)+ (.75,-.75) $); 	\coordinate (5) at  ($ (e)+ (0,-.8) $); 	\node [node] (6) at  ($ (f) + (-.75,-.75) $) {$B_{1}$}; 	\draw (a) to node [label={[xshift=.1cm]1}] {} (1); 	 \draw (b) to node [label={[xshift=-.2cm,yshift=-.2cm]2}] {} (2); 	\draw (c) to node [label={[xshift=-.2cm]3}] {} (3); 	\draw (d) to node [label={[xshift=.3cm,yshift=-.2cm]4}] {} (4); 	\draw (e) to node [label={[xshift=-.2cm,yshift=-.3cm]5}] {} (5); 	 \draw (f) to node [label={[xshift=-.1cm]6}] {} (6);
	\draw [line width=.08cm] (1) to (a) to (b) to (2); 	\draw [line width=.08cm] (3) to (c) to node [label={[xshift=-.2cm,yshift=-.25cm]$P$}] {} (d) to node [label={[xshift=-.1cm,yshift=-.0cm]$P$}] {} (e) to node [label={[xshift=.1cm,yshift=-.1cm]$P$}] {} (f) to (6);
	\draw [dashed, line width=.08cm] (3) to[out=360,in=270,distance=4cm] node [label={[xshift=.2cm,yshift=-.5cm]y}] {} (6); 	 \draw [dashed, line width=.08cm] (1) to[out=90,distance=.8cm] node [label=x] {} (2); \end{tikzpicture}
\par\end{center}

\begin{center}
\caption{The cycle from Figure \ref{badex1}, after expanding $SV_{1}$ and $SV_{2}$ \label{badex2}}

\par\end{center}%
\end{minipage}
\end{figure}

In Sections \ref{expandh1}-\ref{expandsquare} we will carefully analyze this and several other cases that form the
bottleneck in our analysis, which occur when expanding our graph back
to its original state. In Section \ref{analyzingworstcase} we will account for organic $6$-cycles by
creating a correspondence from every $6$-cycle in the final $2$-factor
$F_{0}$ to some larger cycle in $F_{0}$. This way, if we can show
that every large cycle of length $l\geq8$ in $F_{0}$ is affiliated
with at most $f(l)$ $6$-cycles, then the average cycle length is
at least $\min_{l\geq8}\frac{6\times f(l)+l}{f(l)+1}$. Once we have
placed a lower-bound on the average cycle length in this manner (Lemma
\ref{t12}), we can easily determine our approximation factor (Theorems \ref{t1} and
\ref{t13}).

We now state some definitions about ``protected edges'',
organic paths which help us formalize the correspondence between $6$-cycles and larger cycles in our $2$-factor. We use this term because protected edges cannot be separated
from each other in the $2$-factor during subsequent expansion operations.
\begin{definition} \label{t6}
Protected edges are edges contained in maximal paths that are organic,
included in a cycle of length at least 8 in a $2$-factor $F_{i}$, and part of an organic subgraph
that was previously contracted and then expanded. Protected edges are identified during the ``expanding''
phase of the algorithm (described in Section 2.3), when the expansion
operation introduces an organic $6$-cycle.
The edges labeled ``$P$'' in Figure \ref{badex2} are an example of a set
of protected edges.\end{definition}

\begin{definition} \label{t8}
For a given $6$-cycle, $C$, in the final $2$-factor, $F_{0}$,
consider the value $i$ such that $C$ is a cycle of $F_{i}$ but
not of $F_{i+1}$. Then, it was the expansion operation from $G_{i+1}$
to $G_{i}$ which ``finalized'' this cycle. Then, the protected
edges identified during this ``finalizing'' operation are defined
to be the protected edges corresponding to $C$.
\end{definition}

\begin{definition} \label{t9}
For any cycle, $C_{i}$, of length at least $8$, in the $2$-factor
$F$ we will define a set of $6$-cycles, $S_{C_{i}}$ which correspond
to $C_{i}$. For each $6$-cycle, $C$, in $F$, we say $C$ is an
element of $S_{C_{i}}$ if $C$'s protected edges are in $C_{i}$
or if $C$'s protected edges are in another $6$-cycle $C'$ whose
protected edges are in $C_{i}$.
\end{definition}

\subsection{Expanding $H_1$ gadgets\label{expandh1}}

Appendix B in Section \ref{apdxh1} documents, in detail for all cases, the process
of winding the $2$-factor through a $H_{1}$ after the algorithm
has expanded a $H'_{1}$, the $H_{1}$'s replacement gadget. An examination
of this appendix confirms that the example shown in Figures \ref{badex1} and
\ref{badex2} is the only type of operation involving the $H'_{1}$ gadget that
can introduce an organic $6$-cycle into the $2$-factor during an
expansion. There are two other expansions with a similar outcome which
involve the $H'_{2}$ and $H'_{3}$ gadgets, respectively. The analysis
to account for these expansions is very similar to the analysis of
this case. Furthermore, these ``bad'' $H_{2}$ and $H_{3}$ expansions
introduce larger sets of protected edges than the $H_{1}$ expansion
in Figures \ref{badex1} and \ref{badex2}, so this $H_{1}$ expansion is what limits the
approximation factor we obtain in our analysis.
\begin{remark}  \label{remark1}
If an expansion operation of the type shown in Figures \ref{badex1} and \ref{badex2} were
to occur, then it is not possible for the nodes corresponding to $A_{1}$
and $B_{1}$ in these figures to be the super-vertices of a $H_{1}$'s
gadget whose expansion introduces an organic $6$-cycle into the $2$-factor.
Section 6.6 examines the expansions of $H_1$s  where the gadget is covered by a single cycle.%FIX%
There are only two expansions in this section that introduce organic $6$-cycles, and they are isomorphic to each other.
If $A_{1}$ and $B_{1}$ were the super-vertices of a $H_{1}$'s
gadget whose expansion introduces an organic $6$-cycle into the $2$-factor, we must consider two cases. The path that goes through nodes $A_2$ and $B_2$ could end up either in the $6$-cycle or the longer cycle after the expansion. In the first case, the graph
$G_{i-1}$ would have contained a $H_{2}$ and would have been compressed
differently at the time when $A_{1}$ and $B_{1}$ would have been
created during a compression (see Figure \ref{remarkfigure1}). Then, at this stage,
the algorithm would have have replaced a $H_{i}$ for some $i\geq2$
at this time, not a $H_{1}$, which would have been necessary to create
$A_{1}$ and $B_{1}$ as specified. In the second case, upon expanding,
the graph $G_{i-1}$ and $2$-factor $F_{i-1}$ would be in the configuration
shown in Figure \ref{localimprove1}, prompting a local improvement so that this expansion
no longer introduces an organic $6$-cycle. Figure \ref{remarkfigure1} shows the first case
where $A_2$ and $B_2$ are in the $6$-cycle after the expansion, and
Figure \ref{remarkfigure2} shows the second case where $A_2$ and $B_2$ are in the longer cycle after the expansion.
\end{remark}

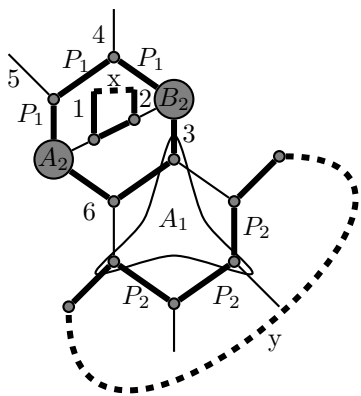
\begin{figure}[H]
\begin{centering}
\tikzstyle{node}=[circle, draw, fill=black!50,                         inner sep=0pt, minimum width=4pt] \tikzstyle{text_box}=[circle, fill=black!0,                         inner sep=0pt, minimum width=4pt]
\begin{tikzpicture}[thick,scale=.8]
\clip (-3,-3) rectangle (3.5,4);
\node [node] (a) at (-1,.5) {}; 	\node [node] (b) at (0,1.2) {}; 	\node [node] (c) at (1,.5) {}; 	\node [node] (d) at (1,-.5) {}; 	\node [node] (e) at (0,-1.2) {}; 	\node [node] (f) at (-1,-.5) {};
	\draw (a) -- (b) -- (c) -- (d) -- (e) -- (f) -- (a); 	\node [node] (1) at  (-2,1.2) {$A_{2}$}; 	\node [node] (2) at  (0,2.2) {$B_{2}$}; 	\node [node] (d1) at ($ (2) + (-2,0) $) {}; 	\node [node] (d2) at ($ (a) + (0,2.4) $) {}; 	\node [node] (n1) at ($ (1)+ (.67,.33) $) {}; 	\node [node] (n2) at ($ (n1)+ (.67,.33) $) {}; 	\coordinate (o1) at ($ (n1) + (0,.8) $); 	\draw (n1) to node [label={[xshift=-.2cm,yshift=-.2cm]1}] {} (o1); 	\coordinate (o2) at ($ (n2) + (0,.5) $); 	\draw (n2) to node [label={[xshift=.15cm,yshift=-.2cm]2}] {} (o2); 	\draw [line width=.08cm] (o1) to (n1) to (n2) to (o2); 	\draw [dashed, line width=.08cm] (o1) to node [label={[yshift=-.1cm]x}] {} (o2); 	\draw (1) to (n1) to (n2) to (2); 	\coordinate (e1) at  ($ (d1)+ (-.75,.75) $); 	\coordinate (e2) at  ($ (d2)+ (0,.8) $); 	\draw (d1) to node [label={[xshift=-.2cm,yshift=-.3cm]5}] {} (e1); 	\draw (d2) to node [label={[xshift=-.2cm,yshift=-.3cm]4}] {} (e2);

	\node [node] (3) at  ($ (c)+ (.75,.75) $) {}; 	\coordinate (4) at  ($ (d)+ (.75,-.75) $); 	\coordinate (5) at  ($ (e)+ (0,-.8) $); 	\node [node] (6) at  ($ (f) + (-.75,-.75) $) {}; 	\draw (a) to node [label={[xshift=0cm,yshift=-.6cm]6}] {} (1); 	\draw (b) to node [label={[xshift=.2cm,yshift=-.2cm]3}] {} (2); 	\draw (c) to node [label={[xshift=-.2cm]}] {} (3); 	\draw (d) to node [label={[xshift=.3cm,yshift=-.2cm]}] {} (4); 	\draw (e) to node [label={[xshift=-.2cm,yshift=-.3cm]}] {} (5); 	\draw (f) to node [label={[xshift=-.1cm]}] {} (6);
	\draw [line width=.08cm] (1) to (a) to (b) to (2); 	\draw [line width=.08cm] (1) to node [label={[xshift=-.3cm,yshift=-.2cm]$P_1$}] {} (d1) to node [label={[xshift=-.1cm,yshift=-.07cm]$P_1$}] {} (d2) to node [label={[xshift=.17cm,yshift=-.1cm]$P_1$}] {} (2);
	\draw [line width=.08cm] (3) to (c) to node [label={[xshift=.3cm,yshift=-.25cm]$P_2$}] {} (d) to node [label={[xshift=.3cm,yshift=-.5cm]$P_2$}] {} (e) to node [label={[xshift=-.1cm,yshift=-.5cm]$P_2$}] {} (f) to (6);
	\draw [dashed, line width=.08cm] (3) to[out=360,in=270,distance=4cm] node [label={[xshift=.1cm,yshift=-.5cm]y}] {} (6);
	\draw plot [smooth cycle] coordinates { ($ (b) +(0,.4) $) (.5,.2) ($ (d) + (.3,-.2) $) (0,-.4) ($ (f) + (-.3,-.2) $) (-.5,.2) }; 	\node [text_box] (txt) at (0,.2) {$A_1$}; \end{tikzpicture}
\par\end{centering}

\caption{The cycles from Figure \ref{badex2}, after expanding $A_{1}$ and $B_{1}$,
if these nodes had been a $H_{1}$'s gadget whose expansion introduced
a $6$-cycle, where the path through $A_2$ and $B_2$ is in a $6$-cycle after the expansion. We can see in this figure that nodes $A_{2}$ and
$B_{2}$ are part of an organic $H_{2}$. Then, $A_{1}$ and $B_{1}$
cannot be a $H_{1}$'s gadget in this configuration, otherwise the
$BIGCYCLE$ algorithm would have performed different operations, compressing
this $H_{2}$ or some other $H_{i}$, for some $i\geq2$ instead of
the $H_{1}$ that $A_{1}$ and $B_{1}$ replaced. \label{remarkfigure1}}
\end{figure}

\begin{figure}[H]
\begin{centering}
\tikzstyle{node}=[circle, draw, fill=black!50,                         inner sep=0pt, minimum width=4pt] \tikzstyle{text_box}=[circle, fill=black!0,                         inner sep=0pt, minimum width=4pt]
\begin{tikzpicture}[thick,scale=.8] 	\node [node] (a) at (-1,.5) {}; 	\node [node] (b) at (0,1.2) {}; 	\node [node] (c) at (1,.5) {}; 	\node [node] (d) at (1,-.5) {}; 	\node [node] (e) at (0,-1.2) {}; 	\node [node] (f) at (-1,-.5) {};
	\draw (a) -- (b) -- (c) -- (d) -- (e) -- (f) -- (a);

	\node [node] (3) at  ($ (c)+ (.75,.75) $) {$B_2$}; 	\coordinate (4) at  ($ (d)+ (.75,-.75) $); 	\coordinate (5) at  ($ (e)+ (0,-.8) $); 	\node [node] (6) at  ($ (f) + (-.75,-.75) $) {$A_2$}; 	\draw (c) to node [label={[xshift=-.2cm]3}] {} (3); 	\draw (d) to node [label={[xshift=.3cm,yshift=-.2cm]}] {} (4); 	\draw (e) to node [label={[xshift=-.2cm,yshift=-.3cm]}] {} (5); 	\draw (f) to node [label={[xshift=-.1cm]6}] {} (6);
	\draw plot [smooth cycle] coordinates { ($ (b) +(0,.4) $) (.5,.2) ($ (d) + (.3,-.2) $) (0,-.4) ($ (f) + (-.3,-.2) $) (-.5,.2) }; 	\node [text_box] (txt) at (0,.2) {$B_1$};
	\node [node] (p1) at ($ (5) + (0,-.7) $) {}; 	\node [node] (p2) at ($ (4) + (.6,-.5) $) {};
	\draw [line width=.08cm] (p2) to node [label={[xshift=.2cm]$P_1$}] {} (3) to (c) to node [label={[xshift=.3cm,yshift=-.25cm]$P_2$}] {} (d) to node [label={[xshift=.3cm,yshift=-.5cm]$P_2$}] {} (e) to node [label={[xshift=-.1cm,yshift=-.5cm]$P_2$}] {} (f) to (6) to node [label={[yshift=-.6cm]$P_1$}] {} (p1) to node [label={[yshift=-.6cm]$P_1$}] {} (p2); 	\draw [line width=.08cm, dashed] (a) to[out=170,in=60, distance=3cm] node [label={[xshift=.15cm,yshift=-.6cm]y}] {} (b); 	\draw [line width=.08cm] (a) to (b);
	\node [node] (e1) at ($ (a) + (-1.5,1.5) $) {}; 	\node [node] (e2) at ($ (b) + (0,2) $) {}; 	\node [node] (1) at ($ (e1) + (-.3,.8) $) {}; 	\node [node] (2) at ($ (e2) + (-.3,.8) $) {}; 	\draw [line width=.08cm] (1) to node [label={[xshift=-.3cm,yshift=-.2cm]1}] {} (e1) to (e2) to node [label={[xshift=.2cm,yshift=-.1cm]2}] {} (2); 	\draw [line width=.08cm, dashed] (1) to[out=170,in=60, distance=1cm] node [label={[xshift=.0cm]x}] {} (2); 
	\draw (6) to (e1) to (e2) to (3); 	\draw (p1) to node [label={[xshift=-.2cm,yshift=-.2cm]5}] {} ($ (p1) + (0,-.5) $);
\draw (p2) to node [label={[xshift=-.2cm,yshift=-.4cm]4}] {} ($ (p2) + (.2,-.4) $); 	\draw (1) to ($ (1) + (-.5,-.2) $); 
\draw (2) to ($ (2) + (.5,.2) $);
\end{tikzpicture}
\par\end{centering}

\caption{The cycles from Figure \ref{badex2}, after expanding $A_{1}$ and $B_{1}$,
if these nodes had been a $H_{1}$'s gadget whose expansion introduced
a $6$-cycle, where the path through $A_2$ and $B_2$ is in the longer cycle after the expansion. We can see in this figure the cycle containing path
$y$ corresponds to $6$-cycle $C_{1}$ in the configuration shown
in Figure \ref{localimprove1}. Then, the algorithm would have updated the $2$-factor
to the configuration shown in Figure \ref{localimprove2}, making path $y$ part of
a $10$-cycle. Then, because of this local correction, this expansion
step would not have introduced an organic $6$-cycle into the $2$-factor. \label{remarkfigure2}}
\end{figure}
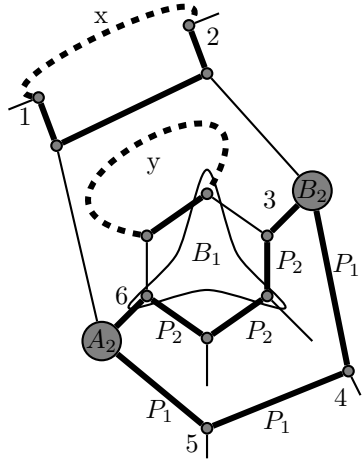

\subsection{Expanding $H_2$ gadgets\label{expandh2}}

Appendix C in Section \ref{apdxh2} documents, in detail for all cases, the process
of winding the $2$-factor through a $H_{2}$ after the algorithm
has expanded a $H_{2}'$, the $H_{2}$'s replacement gadget. An examination
of this appendix confirms that the example shown in Figures \ref{h2badex1} and
\ref{h2badex2} is the only type of operation involving the $H_{2}$ configuration
that can introduce an organic $6$-cycle into the $2$-factor during
an expansion.

\begin{figure}[H]
\begin{minipage}[t]{0.45\columnwidth}%
\begin{center}
\tikzstyle{node}=[circle, draw, fill=black!50,                         inner sep=0pt, minimum width=4pt]
\begin{tikzpicture}[thick,scale=1] 	\node [node] (a) at (0,2) {$SV_{1}$};
	\node [node] (b) at (0,-2) {$SV_{2}$};
	\node [node] (1) at  ($ (a)+ (-.75,-.75) $) {}; 	\node [node] (4) at  ($ (b)+ (-.75,.75) $) {}; 	\node [node] (3) at  ($ (a)+ (.75,-.75) $) {$A_1$};
	\coordinate (5) at  ($ (a)+ (0,-.8) $); 	\node [node] (6) at  ($ (b) + (.75,.75) $) {$B_1$};		 	\draw (a) to node [label={[xshift=0cm]1}] {} (1); 	\draw (b) to node [label={[xshift=0cm]4}] {} (4); 	\draw (a) to node [label={[xshift=0cm]3}] {} (3); 	\draw (b) to node [label={[xshift=0cm]6}] {} (6);
	\draw [line width=.08cm] (a) to (1); 	\draw [line width=.08cm] (a) to (3); 	\draw [line width=.08cm] (b) to (4); 	\draw [line width=.08cm] (b) to (6); 	\draw [dashed, line width=.08cm] (1) to node [label={[xshift=-.2cm]x}] {} (4); 	\draw [dashed, line width=.08cm] (3) to node [label={[xshift=.2cm]y}] {} (6); 
\draw (a) to node [label={[xshift=-.2cm]5}] {} ($ (a) + (0,.7) $);
\draw (b) to node [label={[xshift=-.2cm,yshift=-.2cm]2}] {} ($ (b) + (0,-.7) $);
\end{tikzpicture}
\par\end{center}

\begin{center}
\caption{A cycle in $F_{i}$, a $2$-factor over the condensed graph $G_{i}$.
$SV_{1}$ and $SV_{2}$ are two super-vertices which replaced a $H_{2}$. \label{h2badex1}}

\par\end{center}%
\end{minipage}\hfill{}%
\begin{minipage}[t]{0.45\columnwidth}%
\begin{center}
\tikzstyle{node}=[circle, draw, fill=black!50,                         inner sep=0pt, minimum width=4pt]
\begin{tikzpicture}[thick,scale=1] 	\node [node] (a) at (-1,.5) {}; 	\node [node] (b) at (0,1.2) {}; 	\node [node] (c) at (1,.5) {$B_{2}$}; 	\node [node] (d) at (1,-.5) {$A_2$}; 	\node [node] (e) at (0,-1.2) {}; 	\node [node] (f) at (-1,-.5) {}; 	\node [node] (g) at (0,.5) {}; 	\node [node] (h) at (0,-.5) {};
	\draw (a) -- (b) -- (c) -- (d) -- (e) -- (f) -- (a); 	\draw (b) -- (g) -- (h) -- (e); 	\node [node] (1) at  ($ (a)+ (-.75,.75) $) {}; 	\coordinate (2) at  ($ (g)+ (1,1) $); 	\node [node] (3) at  ($ (c)+ (.75,.75) $) {$A_{1}$}; 	\node [node] (6) at  ($ (d)+ (.75,-.75) $) {$B_1$}; 	\coordinate (5) at  ($ (h)+ (1,-1) $); 	\node [node] (4) at  ($ (f) + (-.75,-.75) $) {}; 	\draw (a) to node [label={[xshift=.1cm]1}] {} (1); 	\draw (g) to node [label={[xshift=-.2cm]2}] {} (2); 	\draw (c) to node [label={[xshift=-.2cm]3}] {} (3); 	\draw (d) to node [label={[xshift=.3cm,yshift=-.2cm]6}] {} (6); 	\draw (h) to node [label={[xshift=-.1cm,yshift=-.5cm]5}] {} (5); 	\draw (f) to node [label={[xshift=-.1cm]4}] {} (4);
	\draw [line width=.08cm] (1) to (a) to (f) to (4); 	\draw [line width=.08cm] (3) to (c) to node [label={[xshift=.0cm,yshift=-.6cm]$P$}] {} (b) to node [label={[xshift=-.2cm,yshift=-.3cm]$P$}] {} (g) to node [label={[xshift=-.2cm,yshift=-.15cm]$P$}] {} (h) to node [label={[xshift=-.2cm,yshift=-.15cm]$P$}] {} (e) to node [label={[xshift=.0cm,yshift=.0cm]$P$}] {} (d) to (6);
	\draw [dashed, line width=.08cm] (1) to[out=200,distance=1.5cm] node [label={[xshift=-.3cm,yshift=-.2cm]x}] {} (4); 	\draw [dashed, line width=.08cm] (3) to[out=-20, in=0,distance=1.5cm] node [label={[xshift=-.3cm,yshift=-.2cm]y}] {} (6); \end{tikzpicture}
\par\end{center}

\begin{center}
\caption{The cycle from Figure \ref{h2badex1}, after expanding $SV_{1}$ and $SV_{2}$ \label{h2badex2}}

\par\end{center}%
\end{minipage}
\end{figure}

The difference in this case is that now five edges are ``protected'',
rather than three for the $H_{1}$, when the graph is expanded. To
see this, consider Figure \ref{h2badex2}, where we can see that the nodes on the
dark path from $B_{2}$ to $A_{2}$ cannot be super-vertices. These nodes cannot be super-vertices, as they are part of an $H_2$ that was compressed, which could only happen if these nodes are organic.  
For the nodes along this path to become part of a $6$-cycle from expanding
$H_{1}$ or $H_{2}$ gadgets, they must be in between two corresponding
super-vertices located exactly five edges apart on one of the cycles
of the $2$-factor. The path between $A_{2}$ and $B_{2}$ in Figure
\ref{h2badex2} is five edges, so any sufficiently long path that contains these
nodes has more than five edges, meaning that if these edges were to
become part of a different cycle through expanding $H_{1}$ or $H_{2}$
gadgets, it would have to be a cycle of length at least $8$. These
edges are protected (Definition \ref{t6}) like those discussed in Section
\ref{expandh1}
\begin{prop} \label{h2prop}
A cycle, $C$, of length $x$ has at most $\frac{x}{5}$ affiliated
$6$-cycles that were formed during the expansion of a $H_{2}$ gadget.\end{prop}
\begin{proof}
Suppose for the sake of contradiction that this cycle has $y>\frac{x}{5}$
affiliated $6$-cycles that were formed during the expansion of a
$H_{2}$ gadget. Each of these $6$-cycles was formed during a distinct
expansion operation (no expansion operation introduces more than one
organic $6$-cycle), so the protected edges for each of these cycles
are disjoint. Each of these $6$-cycles has $5$ protected edges,
so $5y>x$ of the edges in $C$ are protected edges affliated with
$6$-cycles that were formed during the expansion of a $H_{2}$ gadget.
This is a contradiction because $C$ has fewer than $5y$ total edges.
\end{proof}
This ``bad'' expansion, then, is very similar to the ``bad'' expansion
of $H_{1}$ gadgets. In fact, the main difference is that expanding
these $H_{2}$ gadgets is less costly because such an expansion protects
more edges than the corresponding $H_{1}$ gadget expansion. Consequently,
in our worst-case analysis, we will tend to discuss the $H_{1}$ gadget
expansion as this will be sufficient to analyze worst-case performance
of the algorithm.

\subsection{Expanding $H_{3}$ gadgets\label{expandh3}}

Appendix D in Section \ref{apdxh3} documents, in detail for all cases, the process
of winding the $2$-factor through a $H_{3}$ after the algorithm
has expanded a $H'_{3}$, the $H_{3}$'s replacement gadget. An examination
of this appendix confirms that the example shown in Figures \ref{h3badex1} and
\ref{h3badex2} is the only type of operation involving the $H_{3}$ configuration
that can introduce an organic $6$-cycle into the $2$-factor during
an expansion.\\

The instances where expanding a $H'_{3}$ can split off a $6$-cycle
are substantially different than those we examined for $H'_{1}$ and
$H'_{2}$s. This is because $H_{3}$s are replaced by super-edges
rather than super-vertices. Figures \ref{h3badex1} and
\ref{h3badex2} demonstrate this operation:

\begin{figure}[H]
\begin{minipage}[t]{0.45\columnwidth}%
\begin{center}
\tikzstyle{node}=[circle, draw, fill=black!40,                         inner sep=0pt, minimum width=4pt]
\begin{tikzpicture}[thick,scale=1]
	\node [node] (b1) at (-1.5,1) {$B_1$}; 	\node [node] (b2) at (1.5,1) {$B_2$}; 	\node [node] (s) at (0,-1) {$S$};
	\draw [line width=.08cm] (b1) to (s) to (b2); 	\draw [dashed, line width=.08cm] (b2) to[distance=.8cm]  node [label={[xshift=-.0cm]x}] {} (b1);
\end{tikzpicture}
\par\end{center}

\begin{center}
\caption{A cycle in $F_{i}$, a $2$-factor over the condensed graph $G_{i}$.
The edges $(S,B_{1})$ and $(S,B_{2})$ are two of the three super-edges
that replaced a $H_{3}$. The dashed line is a path of length $x$,
where $x\geq4$ and even.\label{h3badex1}}

\par\end{center}%
\end{minipage}\hfill{}%
\begin{minipage}[t]{0.45\columnwidth}%
\begin{center}
\tikzstyle{node}=[circle, draw, fill=black!50,                         inner sep=0pt, minimum width=4pt]
\begin{tikzpicture}[thick,scale=1] 	\node [node] (a) at (-1,.5) {}; 	\node [node] (b) at (0,1.2) {}; 	\node [node] (c) at (1,.5) {}; 	\node [node] (d) at (1,-.5) {}; 	\node [node] (e) at (0,-1.2) {}; 	\node [node] (f) at (-1,-.5) {}; 	\node [node] (b2) at (0,.5) {}; 	\node [node] (a2) at (0,-.5) {}; 	\node [node] (a1) at (-.5,1.7) {}; 	\node [node] (b1) at (.5, 1.7) {}; 	\node [node] (s1) at (2.5,0) {$S$}; 	\node [node] (t1) at (-2.2,.5) {$B_1$}; 	\node [node] (t2) at (-2.2,-.5) {$B_2$};
	\draw (a) -- (b) -- (c) -- (d) -- (e) -- (f) -- (a); 	\draw (b) -- (b2) -- (a2) -- (e); 	\draw (a) to (a1) to (b1) to (s1); 	\draw (b1) to[in=0,out=300] (d); 	\draw (s1) to[in=-30,out=220] (b2); 	\draw (a2) to[out=220, in=300] (t2); 	\draw (t1) to (a1);
	\draw [line width = .08cm] (t1) to (a1) to  node [label={[xshift=.3cm,yshift=-.15cm]$P$}] {} (a) to  node [label={[xshift=-.2cm,yshift=-.1cm]$P$}] {} (f) to  node [label={[xshift=.0cm,yshift=.0cm]$P$}] {} (e) to  node [label={[xshift=.25cm,yshift=-.18cm]$P$}] {} (a2); 	\draw [line width = .08cm] (a2) to[out=220, in=300] (t2); 	\draw [dashed, line width = .08cm] (t2) to[in=180,out=180] node[label={[xshift=-.3cm,yshift=-.2cm]x}] {} (t1);
	\draw [line width = .08cm] (s1) to[in=-30,out=220] (b2); 	\draw [line width = .08cm] (b2) to (b) to (c) to (d); 	\draw [line width = .08cm] (b1) to (s1); 	\draw [line width = .08cm] (b1) to[in=0,out=300] (d); \end{tikzpicture}
\par\end{center}

\begin{center}
\caption{The cycle from Figure \ref{h3badex1}, after expanding the two super-edges\label{h3badex2}}

\par\end{center}%
\end{minipage}
\end{figure}

In these figures, we see that this expansion requires two super-edges
to be directly neighboring each other in a cycle of the $2$-factor.
When this expansion is performed, the two super-edges are split off
and form a $6$-cycle while the larger cycle they came from increases
in length by four. These four new edges are protected (Definition
3). 
\begin{prop} \label{h3prop}
A cycle, $C$, of length $x$ has at most $\frac{x}{4}$ affiliated
$6$-cycles that were formed during the expansion of a $H_{3}$ gadget.\end{prop}
\begin{proof}
Suppose for the sake of contradiction that this cycle has $y>\frac{x}{4}$
affiliated $6$-cycles that were formed during the expansion of a
$H_{3}$ gadget. Each of these $6$-cycles was formed during a distinct
expansion operation (no expansion operation introduces more than one
organic $6$-cycle), so the protected edges for each of these cycles
are disjoint. Each of these $6$-cycles has $4$ protected edges,
so $4y>x$ of the edges in $C$ are protected edges affliated with
$6$-cycles that were formed during the expansion of a $H_{3}$ gadget.
This is a contradiction because $C$ has fewer than $4y$ total edges.
\end{proof}
\subsection{Expanding gadgets that replaced squares\label{expandsquare}}

\begin{figure}[ht]
\begin{minipage}[t]{0.45\columnwidth}%
\begin{center}
\tikzstyle{node}=[circle, draw, fill=black!50,                         inner sep=0pt, minimum width=4pt]
\begin{tikzpicture}[thick,scale=0.25] 	\node [node] (Ax) at (-8,8.5) {$A_1$}; 	\node [node] (Bx) at (-8,-8.5) {$B_1$}; 	\draw (Ax) to (Bx);
\node [node] (a) at (-2,2) {}; 	\node [node] (b) at (2,2) {}; 	\node [node] (c) at (2,-2) {}; 	\node [node] (d) at (-2,-2) {}; 	 \node [node] (e) at (0,5.5) {}; 	\node [node] (f) at (0,-5.5) {}; 	\draw (a) -- (b) -- (c) -- (d) -- (a); 	\draw (a) -- (e); 	 \draw (d) -- (f);
\node [node] (1) at  ($ (e)+ (0,3) $) {$A_1$}; 	\node [node] (2) at  ($ (f)+ (0,-3) $) {$B_1$};
	\draw [line width=.08cm] (e) to (a) to (d) to (f) to[out=0,in=-50,distance=5cm] (b) to (c) to[out=0,in=30, distance=5cm] (e);
	\draw (e) to (1); \draw (f) to (2);

\end{tikzpicture}
\par\end{center}

\caption{A $S'_{3}$ in $G_{i}$ which is not covered by the $2$-factor $F_{i}$ (left), and the $S_{3}$ in $G_{i-1}$ that replaced the $S'_{3}$ after expansion (right) \label{s3pbefore}}
\end{minipage}\hfill{}%
\begin{minipage}[t]{0.45\columnwidth}%
\begin{center}
\tikzstyle{node}=[circle, draw, fill=black!50,                         inner sep=0pt, minimum width=4pt]
\begin{tikzpicture}[thick,scale=0.25] 	\node [node] (Ax) at (-8,8.5) {$A_2$}; 	\node [node] (Bx) at (-8,-8.5) {$B_2$}; 	 \draw[line width=.08cm] (Ax) to (Bx); \draw [dashed, line width=.08cm] (Ax) to[out=180,in=180, distance=4cm] node[label={[xshift=-.2cm]x}] {} (Bx);

\node [node] (a) at (-2,2) {}; 	\node [node] (b) at (2,2) {}; 	\node [node] (c) at (2,-2) {}; 	\node [node] (d) at (-2,-2) {}; 	 \node [node] (e) at (0,5.5) {}; 	\node [node] (f) at (0,-5.5) {}; 	\draw (a) -- (b) -- (c) -- (d) -- (a); 	\draw (a) -- (e); 	 \draw (d) -- (f); 	\draw (c) to[out=0,in=50, distance=5cm] (e); 	\draw (b) to[out=-50,in=0, distance=5cm] (f); 	\node [node] (1) at  ($ (e)+ (0,3) $) {$A_2$}; 	\node [node] (2) at  ($ (f)+ (0,-3) $) {$B_2$};
	\draw [line width=.08cm] (1) to (e) to node[label={[xshift=-.1cm,yshift=-.1cm]$P$}] {} (a) to node[label={[xshift=.15cm,yshift=0cm]$P$}] {} (b) to node[label={[xshift=-.2cm,yshift=-.2cm]$P$}] {} (c) to node[label={[xshift=-.0cm,yshift=0cm]$P$}] {} (d) to node[label={[xshift=.3cm,yshift=-.25cm]$P$}] {} (f) to (2);
	\draw [dashed, line width=.08cm] (1) to[out=180,in=180, distance=4cm] node[label={[xshift=-.2cm]x}] {} (2); \end{tikzpicture}
\par\end{center}

\begin{center}
\caption{A $S'_{3}$ in $G_{i}$ that is covered by the $2$-factor $F_{i}$ (left), and the $S_{3}$ in $G_{i-1}$ that replaced the $S'_{3}$ after expansion (right). The $5$ bold edges labeled ``$P$'' will be the $6$-cycle
in Figure \ref{s3pbefore}'s protected edges. \label{s3p6cycle}}

\par\end{center}%
\end{minipage}
\end{figure}
Expanding $S'_{3}$s, not in the $2$-factor, of the type shown in
Figure \ref{s3pbefore} can introduce organic $6$-cycles into the $2$-factor.
However, we limit the number of $S'_{3}$s whose edges are not in the $2$-factor
by computing three disjoint perfect matchings and taking as our $2$-factor the union
of the two perfect matchings that contain the most edges in $S'_{3}$s.
Then, at most $\frac{1}{3}$ of the $S'_{3}$s will have edges not included in the $2$-factor.
Then, we can put the $6$-cycles introduced by expanding each $S'_{3}$ of this type
in correspondence with a larger cycle containing the protected edges of the same type shown
in Figure \ref{s3p6cycle}. These edges are labeled ``$P$'' in Figure \ref{s3p6cycle}.

\subsection{Expanding all other gadgets}

In the analysis that follows, we have identified all
expansions with the potential to introduce organic $6$-cycles cycles
into $F_{i}$ that were not present in $F_{i+1}$. Confirming this
fact requires checking all possible ways a $2$-factor can pass through
each gadget to ensure that all expansions that can introduce organic
$6$-cycles are properly analyzed. Diagrams documenting every one of these other
expansions are shown in the Section \ref{apdxsq}-\ref{apdxh6}, and a careful examinitaion
of these sections confirms that all other expansion operations are
not capable of introducing an organic $6$-cycle into the $2$-factor.
These other operations result either in converting one cycle into
one larger cycle, converting one cycle into two large (length of at
least $8$) cycles, or converting two cycles into one or two large
cycles.

\subsection{Analyzing the Worst Case} \label{analyzingworstcase}

We call the edges identified in Definition \ref{t6} ``protected'' because
they form organic paths in our $2$-factor, so these paths will remain
part of the $2$-factor regardless of the future expansion operations
performed. In Definition \ref{t8}, we established a correspondence between
protected edges and $6$-cycles in our $2$-factor. This relation allows us
to place every $6$-cycle in correspondence with some large cycle in our
$2$-factor, as stated in Definition \ref{t9}. We will use this correspondence to
analyze the performance of BIGCYCLE.

In Lemma \ref{oldlemma4} we prove that if a $6$-cycle's protected edges are in another $6$-cycle, then the second cycle's protected edges must be in a cycle of length at least $10$. This is helpful to us, as the average length of these three cycles is at least $\frac{22}{3}$ which is greater than $7$. In Lemma \ref{lem-corr}, we prove that all $8$-cycles have at most one corresponding $6$-cycle, so any $8$-cycle and its corresponding $6$-cycle have an average length of $7$. Lemma \ref{t12} generalizes the previous two Lemmas to show that any long cycle and its corresponding $6$-cycles have an average length of at least $7$. 

\begin{prop} \label{oldlemma3}
If a $6$-cycle, $C_{1}$, has its corresponding protected edges in
another $6$-cycle, $C_{2}$, then $C_{2}$ has $5$ corresponding
protected edges, and these protected edges are all in a cycle of length
at least $8$. Furthermore, $C_{1}$ and $C_{2}$ must have been introduced
into the $2$-factor during the expansion of a $H'_{1}$ and $H'_{2}$,
respectively.\end{prop}
\begin{proof}
In Section \ref{accounting}, we show all ways
an organic $6$-cycle can appear in the final $2$-factor. The appendices
validate this claim, as all possible expansions are examined in detail,
and all expansions not included in Section \ref{accounting} do not produce organic
$6$-cycles. Immediately following each of these three special expansion
operations, all newly identified protected edges are in cycles of
length at least $8$. \\

Then suppose, for the sake of contradiction, that in the final $2$-factor,
$C_{1}$ has a protected edge in another $6$-cycle, $C_{2}$, but
the lemma does not hold. The operation that brought $C_{1}$ into
the $2$-factor $F_{i-1}$ in graph $G_{i-1}$ is one that replaced
some gadget in a condensed graph, $G_{i}$, with a $H_{1}$, $H_{2}$,
$H_{3}$, or a small 2-cut as depicted in Figures \ref{badex1}, \ref{h2badex1}, \ref{h3badex1}, and
\ref{s3pbefore}, respectively. We now claim that such a situation cannot occur
if the expansion operation replaced the gadget with a $H_{2}$, $H_{3}$,
or a small 2-cut.\\

Consider the $H_{2}$ case first, as shown in Figures \ref{h2badex1}-\ref{h2badex2}. In this
case, $C_{1}$ is the left cycle in Figure \ref{h2badex2} and has five protected
edges (the path from $A_{2}$ to $B_{2}$). Furthermore, after this expansion the path from $A_{2}$
to $B_{2}$ in Figure \ref{h2badex2} is entirely organic, meaning none of the
nodes or edges (including $A_{2}$ and $B_{2}$) are part of gadgets.
The only way $C_{1}$'s protected edges could end up in a $6$-cycle
in the final $2$-factor is if another special expansion operation
split $C_{2}$ into two cycles. For this to happen, however, either
two super-vertices which together form a gadget must be a distance
of exactly $5$ edges away from each other in a cycle or two super-edges
must be directly next to each other in a cycle of the $2$-factor.
However, in this case, $C_{1}$'s protected edges form an organic
path of length $5$, which prevents this from occuring, as the two
closest possible super-vertices would be $A_{1}$ and $B_{1}$, which
are separated by $7$ edges and the two closest possible super-edges
would be $(A_{2},B_{1})$ and $(A_{1},B_{2})$, which are separated
by $5$ edges.\\

Similarly, consider if the expansion operation that first introduced
$C_{1}$ into the $2$-factor replaced a gadget with a $H_{3}$ or
a small $2$-cut of the type shown in Figure \ref{s3p6cycle}. The same reasoning
as in the previous paragraph applies here, too, as $C_{1}$ would
have at least six protected edges appearing consecutively in a $H_{3}$
or small $2$-cut in these cases. Any corresponding super-vertices
or super-edges in this cycle are separated by too many organic edges
and nodes to split off some of the protected edges into a new $6$-cycle.\\

The only remaining possibility is that $C_{1}$ was introduced through
an expansion operation that replaced a gadget with a $H_{1}$, as
shown in Figure \ref{badex2}. In this case, the only way these protected edges
could be split off into an organic $6$-cycle is if nodes $A_{1}$
and $B_{1}$ are a gadget. If this gadget was one that replaced a
$H_{1}$, then we see by Remark \ref{remark1} that the expansion of this gadget
cannot introduce an organic $6$-cycle into the $2$-factor. Then,
$A_{1}$ and $B_{1}$ must have been a $H'_{2}$. Thus, $C_{2}$ was
necessarily introduced into the $2$-factor through the expansion
of a $H'_{2}$, as depicted in Figures \ref{h2badex1}-\ref{h2badex2}. Then, $C_{2}$ has $5$
corresponding protected edges of its own. We conclude that these protected
edges will be in a cycle of length at least $8$ because we demonstrated
earlier in this proof that if protected edges are identified from
expanding a $H'_{2}$, these protected edges cannot be part of a $6$-cycle
in the final $2$-factor. The final $2$-factor does not contain any
odd cycles (due to bipartiteness of the original graph) and contains
no $4$-cycles so we conclude that $C_{2}$'s protected edges are
part of a cycle of length at least $8$, proving the lemma.\end{proof}

\begin{lemma} \label{oldlemma4}
\label{t10}
If a $6$-cycle, $C_{1}$, has its corresponding protected edges in
another $6$-cycle, $C_{2}$, then $C_{2}$ has $5$ corresponding
protected edges. These protected edges are all in a cycle which
either contains no other protected edges or is of length at least
$10$. Furthermore, $C_{1}$ and $C_{2}$ must have been introduced
into the $2$-factor during the expansion of a $H'_{1}$ and $H'_{2}$,
respectively.\end{lemma}
\begin{proof}
By Proposition \ref{oldlemma3}, we know that $C_{2}$ has $5$ protected edges in a third
cycle, $C_{3}$, of length at least $8$. Suppose then, for the sake
of contradiction, that this lemma is violated. Then, $C_{3}$ contains
$C_{2}$'s $5$ protected edges, some set of other protected edges,
and has length less than $10$. By Proposition \ref{oldlemma3} and the fact that the graph
is bipartite, we conclude that $C_{3}$ must have length $8$. Protected
edges from the same expansion cannot get separated from each other,
as there are no super-vertices or super-edges along a path of protected
edges, and protected edges come in sets of $3$, $4$, and $5$ edges,
depending on the expansion operation. $C_{3}$ contains the five protected
edges from $C_{2}$ as well as another set of protected edges, so
this additional set of protected edges must be a set of three edges.
This is only possible if the additional set of protected edges were
introduced by expanding a $H'_{1}$, under the circumstances depicted
in Figure \ref{badex1}-\ref{badex2}. Such an expansion would require the nodes corresponding
to $A_{2}$ and $B_{2}$ in Figure \ref{h2badex2} to be a gadget. However, by
Lemma \ref{oldlemma3}, $A_{2}$ and $B_{2}$ must be organic, as they are part of
the $H_{2}$ whose expansion introduced $C_{2}$ into the $2$-factor
and $C_{2}$'s protected edges into $C_{3}$. Since $A_{2}$ and $B_{2}$
are organic, they cannot be a $H'_{1}$, proving that $C_{3}$ contains
no protected edges other than $C_{2}$'s protected edges. This contradicts
our assumption, proving the lemma.
\end{proof}

Before we prove a lower bound on average cycle length, we need an additional lemma regarding
$8$-cycles in the final $2$-factor. The following propositions are needed for the upcoming proof of
Lemma \ref{lem-corr}.

\begin{prop} \label{t7}
Let $P_{1}$ and $P_{2}$ be the sets of protected edges corresponding
to two $6$-cycles, $C_{1}$ and $C_{2}$, respectively. Then $P_{1}\cap P_{2}=\emptyset$
and $V(P_{1})\cap V(P_{2})=\emptyset$.\end{prop}
\begin{proof}
Suppose $P_{1}\cap P_{2}\neq\emptyset$, for the sake of contradiction. Then there is some edge
$e$ such that $e\in P_{1}$ and $e\in P_{2}$. All four expansion
operations (described in detail in Sections \ref{expandh1}-\ref{expandsquare}) that introduce
organic $6$-cycles and identify protected edges are such that there
is some integer $i$ where $e\notin E_{i+1}$ but $e\in E_{i}$. Then,
the expansion operation from $G_{i+1}$ to $G_{i}$ will introduce
both $C_{1}$ and $C_{2}$ to $2$-factor $F_{i}$ as organic $6$-cycles.
This is not possible, as the expansion operations which introduce
organic $6$-cycles all introduce exactly one $6$-cycle into the
$2$-factor, proving the first claim.
Suppose$V(P_{1})\cap V(P_{2})\neq\emptyset$, for the sake of contradiction. Then there is some node
$v$ such that $v\in V(P_{1})$ and $v\in V(P_{2})$. All four expansion
operations (described in detail in Sections \ref{expandh1}-\ref{expandsquare}) that introduce
organic $6$-cycles and identify nodes that are endpoints of protected
edges are such that there is some integer $i$ where $v\notin V_{i+1}$
but $v\in V_{i}$. Then, the expansion operation from $G_{i+1}$ to
$G_{i}$ will introduce both $C_{1}$ and $C_{2}$ to $2$-factor
$F_{i}$ as organic $6$-cycles. This is not possible, as the expansion
operations which introduce organic $6$-cycles all introduce exactly
one $6$-cycle into the $2$-factor, proving the lemma.\end{proof}

\begin{prop} \label{oldprop6}
Suppose there is a cycle $C$ in a preliminary $2$-factor $F_{i}$
which is the union of four edge disjoint paths: $P_{1}$, $X$, $P_{2}$,
$Y$. Furthermore, suppose that $P_{1}$ and $P_{2}$ are organic,
$P_{1}$ shares its endpoints with endpoints of $X$ and $Y$, and
$P_{2}$ shares its endpoints with the remaining endpoints of $X$
and $Y$. Then the final $2$-factor $F$ does not contain a cycle
$C'$ which is the union of $P_{1}$, $e_{1}$, $P_{2}$, $e_{2}$,
where $e_{1}$ and $e_{2}$ are single edges separating $P_{1}$ and
$P_{2}$ in $C'$, unless both $X$ and $Y$ are single edges.
\end{prop}
\begin{figure}[H]
\begin{minipage}[t]{0.45\columnwidth}%
\begin{center}
\tikzstyle{node}=[circle, draw, fill=black!50,                         inner sep=0pt, minimum width=4pt]
\begin{tikzpicture}[thick,scale=.8]
	\node [node] (a) at (-1,1) {}; 	\node [node] (b) at (1,1) {}; 	\node [node] (c) at (1,-1) {}; 	\node [node] (d) at (-1,-1) {};
	\draw [dashed, line width=.08cm] (a) to[distance=.8cm]  node [label={[xshift=-.0cm]$P_1$}] {} (b); 	\draw [dashed, line width=.08cm] (b) to[out=-30,in=30,distance=.5cm]  node [label={[xshift=.4cm]$X$}] {} (c); 	\draw [dashed, line width=.08cm] (c) to[out=270,in=270,distance=.5cm]  node [label={[xshift=-.0cm]$P_2$}] {} (d); 	\draw [dashed, line width=.08cm] (d) to[out=150,in=210,distance=.5cm]  node [label={[xshift=-.4cm]$Y$}] {} (a);
\end{tikzpicture}
\par\end{center}

\begin{center}
\caption{A cycle $C$, composed of four edge disjoint paths $P_{1}$, $X$,
$P_{2}$, $Y$}

\par\end{center}%
\end{minipage}\hfill{}%
\begin{minipage}[t]{0.45\columnwidth}%
\begin{center}
\tikzstyle{node}=[circle, draw, fill=black!50,                         inner sep=0pt, minimum width=4pt]
\begin{tikzpicture}[thick,scale=0.8]
	\node [node] (a) at (-1,.5) {}; 	\node [node] (b) at (1,.5) {}; 	\node [node] (c) at (1,-.5) {}; 	\node [node] (d) at (-1,-.5) {};
	\draw [dashed, line width=.08cm] (a) to[distance=.8cm]  node [label={[xshift=-.0cm]$P_1$}] {} (b); 	\draw [line width=.08cm] (b) to  node [label={[xshift=.4cm,yshift=-.3cm]$e_1$}] {} (c); 	\draw [dashed, line width=.08cm] (c) to[out=270,in=270,distance=.5cm]  node [label={[xshift=-.0cm]$P_2$}] {} (d); 	\draw [line width=.08cm] (d) to  node [label={[xshift=-.4cm,yshift=-.3cm]$e_2$}] {} (a);
\end{tikzpicture}
\par\end{center}

\begin{center}
\caption{A cycle $C'$, composed of $P_{1}$, $e_{1}$, $P_{2}$, $e_{2}$.
Proposition \ref{oldprop6} proves that if $C$ is a cycle in the computed $2$-factor
in a compressed graph $G_{i}$, then $C'$ will never be a cycle in
the final $2$-factor $F$.}

\par\end{center}%
\end{minipage}
\end{figure}

\begin{proof}
Suppose, for the sake of contradiction, that $F$ contained such a
cycle $C'$ where $C\neq C'$. The algorithm does not perform any
expansion operations that combine two organic paths in separate cycles
into a single cycle connected by two single edges. Then, none of the
final $i$ expansions the algorithm performs can separate $P_{1}$
and $P_{2}$ into two different cycles. $|C|>|C'|$, so, for $F$
to contain $C'$, one of these final $i$ expansions must have shortened
one of the paths connecting $P_{1}$ to $P_{2}$ to just a single edge while keeping
both $P_{1}$ and $P_{2}$ in the same segment. Without loss of generality we assume that $X$ has length at least $2$. The algorithm does not perform any expansion operations that result in a new cycle containing both $P_1$ and $P_2$ such that the endpoints of $P_1$ and $P_2$ that were previously incident on $X$ are now connected by a single edge. Then there is no sequence of expansion operations that could result in $C$ being contained in $F$, contradicting our assumption.\end{proof}
\begin{lemma} \label{lem-corr}
Every cycle $C$ of length $8$ in the final $2$-factor $F$ has
at most one corresponding $6$-cycle.\end{lemma}
\begin{proof}
Suppose, for the sake of contradiction, that there exists an $8$-cycle
$C$ in $F$ which has $k$ corresponding $6$-cycles, $C_{1}$, $C_{2}$,
..., $C_{k}$, where $k\geq2$. By Definition 6, for $1\leq i\leq k$,
$C_{i}$ has protected edges, which are located either in $C$ or
in $C_{j}$ for some $j\neq i$, $1\leq j\leq k$.\\

First, let's consider when there are values $i,j$ such that $C_{i}$'s
protected edges are in $C_{j}$. $C$ is of length $8$, so by Lemma
\ref{oldlemma4}, it must be the case that $k=2$ and without loss of generality,
$i=1$ and $j=2$. Additionally, we know due to Lemma \ref{oldlemma3} that $C_{1}$
was introduced through the expansion of a $H'_{1}$ and $C_{2}$ was
introduced through the expansion of a $H'_{2}$.\\

If the expansion of the $H'_{2}$ that introduced $C_{2}$ into a
preliminary $2$-factor $F_{x}$ also introduced $C$ into $F_{x}$,
then $C$ could not be an $8$-cycle, which is a contradiction. To
see this, see Figure \ref{h2badex2}, which describes this class of expansion operation.
If $C$ were an $8$-cycle, then the dashed path connecting nodes
$A_{1}$ and $B_{1}$ in this figure is of length $1$, but this would
mean that the algorithm compressed a $H_{2}$ at a time when there
was a square present in graph, which is also a contradiction. Then,
the remaining possibility is that immediately following the expansion
of the $H'_{2}$ that introduced $C_{2}$ into $F_{x}$, the cycle's
protected edges were in a non-organic cycle $C'$, of the form shown
on the right side of Figure \ref{h2badex2}. A later expansion operation must introduce
$C$, resulting in the nodes corresponding to $A_{2}$ and $B_{2}$
in Figure \ref{h2badex2} being connected by a path of length $3$. This expansion
cannot be one where a square is expanded because none of these operations
shorten the length of the cycles involved. Then any other potential
expansion would contradict the algorithm, as a path of length $3$
from $A_{2}$ to $B_{2}$ would form a square, and the algorithm would
have previously contracted a $H_{i}$ for some $i\geq2$ when a square
is present in the graph. \\

We have now ruled out the possibility of cycle $C_{i}$ having protected
edges in a cycle other than $C$. By Lemma \ref{t7}, we know that the protected
edges of two cycles $C_{i}$ and $C_{j}$ are disjoint. Then, we can
easily see that $k<3$. Each $6$-cycle has at least $3$ protected
edges in $C$, and $C$ is an $8$-cycle, so if $k\geq3$, then either
$C$ would need to have more than $8$ edges or the $6$-cycles would
need to share protected edges, which is not possible.\\

We must now also show that we obtain a contradiction when $k=2$.
By Lemma \ref{t7}, we know that any two sets of protected edges do not share
any vertices. Then if $C$ contains two sets of protected edges, they
must be separated by at least one edge on each side of the cycle.
Then, $8=|C|\geq|P_{1}|+|P_{2}|+2$, where $P_{1}$ and $P_{2}$ are
the protected edges of $C_{1}$ and $C_{2}$, respectively. $P_{i}\geq3$,
and if either $P_{1}$ or $P_{2}$ is at least $4$ then $|P_{1}|+|P_{2}|+2\geq9$,
so the only possibility we need to consider is when $|P_{1}|=|P_{2}|=3$.
This would require that both $6$-cycles are introduced from expanding
$H'_{1}$s. All other possibilities would result in $C$ containing
more than $8$ protected edges, which is not possible. We now demonstrate
that this case results in a contradiction.\\

We now know that both $6$-cycles are introduced through the expansion
of two $H'_{1}$s, as all other cases result in a contradiction. The
specific expansion of this type that can introduce organic $6$-cycles
is described in detail in Section \ref{expandh1} and shown in Figures \ref{badex1}-\ref{badex2}.
If $P_{1}$ and $P_{2}$ are ever contained in different cycles of
a preliminary $2$-factor $F_{x}$, then some expansion operation
will eventually bring these two sets of protected edges into the same
cycle. However, observe that the algorithm does not perform any expansion
operations that combine two organic paths in separate cycles into
a single cycle, where the two paths share an endnode or are separated
by a single edge on both sides. Then $P_{1}$ and $P_{2}$ will necessarily
be separated by at least two edges on one side and at least one edge
on the other side. By Proposition \ref{oldprop6}, no future expansions could result
in $P_{1}$ and $P_{2}$ being contained in a single $8$-cycle. This
would contradict the assumptions that $C$ is an $8$-cycle and contains
two sets of protected edges introduced through the expansion of two
$H'_{1}$s.\\

The other possibility is that the expansion of the second $H'_{1}$
introduces $P_{2}$ directly into a cycle that contains $P_{1}$.
By Remark \ref{remark1}, we know that the nodes corresponding to $A_{1}$ and
$B_{1}$ shown Figure \ref{badex2} cannot be super-vertices of a $H'_{1}$ whose
expansion introduces an organic $6$-cycle into a preliminary $2$-factor
$F_{x}$. Then, there must be at least $7$ edges between any two
super-vertices of a $H'_{1}$ whose expansion introduces an organic
$6$-cycle into $F_{x}$. If $A_{1}$ or $B_{1}$ or either edge $(A_{2},B_{1})$
or $(B_{2},A_{1})$ are non-organic and get expanded before the expansion
that introduces the $6$-cycle into $F_{x}$, then these expansions
will replace these non-organic subgraphs with organic subgraphs, so
the new nodes in the place of $A_{1}$ or $B_{1}$ will never be two
super-vertices of a $H'_{1}$ whose expansion introduces a $6$-cycle
into $F_{x}$. Then, after the expansion introducing the second $6$-cycle,
the cycle containing both sets of protected edges in $F_{x}$ will
have at least $10$ edges, the $7$ edges in between the two super-vertices
that replaced the second $H_{1}$ and the second $6$-cycle's set
of $3$ protected edges. By Proposition \ref{oldprop6}, no future expansions can
result in both $P_{1}$ and $P_{2}$ being contained in a single $8$-cycle.
This contradicts the assumptions that $C$ is an $8$-cycle and contains
two sets of protected edges introduced through the expansion of two
$H'_{1}$s.\\

We have proved that all cases that could result in the existance of
an $8$-cycle $C$ in the final $2$-factor $F$ that has two or more
corresponding $6$-cycles leads to a contradiction, proving the lemma.
\end{proof}

We are now prepared to prove the next lemma, regarding average cycle
length of a large cycle and its set of corresponding $6$-cycles:
\begin{lemma} \label{t12}
For any cycle $C$ in the final $2$-factor $F$ of length $l$ such that $l\geq8$
and its set of corresponding $6$-cycles, the average length of this
set of cycles is at least $7$.\end{lemma}
\begin{proof}
First, consider the simple case where $C$ has no corresponding $6$-cycles.
The set of cycles we are considering in this case is just a single
cycle of length $l\geq8$. $l\geq8>7$, so in this case, the only cycle in the
set of cycles has length at least $7$.

Now, consider the case when all of the corresponding $6$-cycles have
their protected edges contained in the large cycle $C$ (to be clear,
the only way this condition could be violated is if some $6$-cycle
has its protected edges in another $6$-cycle, whose protected edges are contained
in $C$). Each of the expansion operations that included one of the
corresponding $6$-cycles in the $2$-factor protects at least $3$
edges, and these protected edges are contained in $C$, so at most
$\frac{l}{3}$ $6$-cycles can correspond to cycle $C$. If $l\geq10$,
then the average cycle length among cycle $C$ and its corresponding
$6$-cycles is at least
\begin{align*}
\frac{\lfloor\frac{l}{3}\rfloor\times6+l}{\lfloor\frac{l}{3}\rfloor+1} & \geq7
\end{align*}

If $l=8$ then by Lemma \ref{lem-corr}, $C$ has at most one corresponding $6$-cycle,
so the average length of $C$ and its corresponding $6$-cycle is
also $7$. We must consider the case when at least one corresponding
$6$-cycle, $C_{1}$, has its protected edges in another
$6$-cycle, $C_{2}$. If $l=8$ and C contains $C_{2}$'s protected
edges then $C$ has at least two corresponding $6$-cycles, $C_{1}$
and $C_{2}$, contradicting Lemma \ref{lem-corr}.

Next, consider if $l=10$ and $C$ contains $C_{2}$'s $5$ protected
edges in the final $2$-factor. $C$ cannot contain another set of
$5$ protected edges, due to Lemma \ref{t7}, because this would require these
protected edges to share a node with $C_{2}$'s protected edges. Then,
in addition to $C_{1}$ and $C_{2}$, $C$ can have at most one additional
corresponding $6$-cycle, otherwise $C$ would contain more than $10$
protected edges. In this case, $C$ has at most $3$ corresponding
$6$-cycles, so the average length of $C$ and its corresponding $6$-cycles
is at most $\frac{10+3\times6}{4}=7$.

The only remaining case is when $l\geq12$ and at least one corresponding
$6$-cycle, $C_{1}$, has its protected edges contained in another
$6$-cycle, $C_{2}$. By Lemma \ref{t10}, each $6$-cycle has at least $3$
protected edges in $C$ or its protected edges are in another $6$-cycle
whose $5$ protected edges are in $C$. So, if a corresponding $6$-cycle's
protected edges are not in $C$, then there is another $6$-cycle
corresponding to $C$ for which these two $6$-cycles contribute $5$
protected edges to $C$. Then, each $6$-cycle on average contributes
at least $\frac{5}{2}$ protected edges to $C$, so there are at most
$\frac{2l}{5}$ $6$-cycles corresponding to $C$. Then, the average
cycle length among cycle $C$ and its corresponding $6$-cycles is
at least
\begin{align*}
\frac{\lfloor\frac{2l}{5}\rfloor\times6+l}{\lfloor\frac{2l}{5}\rfloor+1} & \geq7\;(\text{because}\; l\geq12)
\end{align*}

In all possible cases, $C$ and its corresponding $6$-cycles have
average length of at least $7$.
\end{proof}

\subsection{Main Theorems}
\begin{theorem} \label{t13}
Given a cubic bipartite graph $G$ with $n>6$ vertices, there is
a polynomial time algorithm that computes a $2$-factor with at most
$\frac{n}{7}$ cycles \end{theorem}
\begin{proof}
It follows directly from Lemma \ref{t12} that the average cycle length in
the $2$-factor produced by BIGCYCLE is at least $7$. Then, it must
be the case that BIGCYCLE produces a $2$-factor with at most $\frac{n}{7}$
cycles.

The BIGCYCLE algorithm performs $O(n)$ contractions and expansions.
In between each contraction, the algorithm will search for other subgraphs
to contract and will compute a $2$-factor in the current graph. Classical algorithms can compute $2$-factors with $O(n^{\frac{3}{2}})$ operations in the worse case \cite{hopcroft1973n}, so the contraction
phase of the algorithm runs in $O(n^{\frac{5}{2}})$ time. The expansion
phase of the algorithm takes $O(n)$ time in the worst case, since
there are at most $O(n)$ expansions and each one is performed in
constant time. Then, BIGCYCLE finds a $2$-factor with at most $\frac{n}{7}$
cycles in $O(n^{\frac{5}{2}})$.\end{proof}

\addtocounter{theorem}{-13}

We can now restate our main theorem from Section 1.1:

\begin{theorem}
Given a cubic bipartite connected graph $G$ with $n$ vertices, there is a polynomial
time algorithm that computes a spanning Eulerian multigraph $H$ in
$G$ with at most $\frac{9}{7}n$ edges. \end{theorem}
\begin{proof}
Theorem \ref{t13} proves that the COMPRESS and EXPAND phases
of BIGCYCLE produce a $2$-factor with at most $\frac{n}{7}$ cycles
for the required class of graphs. Proposition \ref{t3} demonstrates that
the DOUBLETREE phase in BIGCYCLE successfully extends the $2$-factor
into a spanning Eulerian multigraph with at most $\frac{9}{7}n-2$
edges.

From Theorem \ref{t13} that we can compute the required $2$-factor in $O(n^{\frac{5}{2}})$ time.
Once we have done this, we can compute the doubled spanning tree in $O(n)$ time as well, as the graph	
has $O(n)$ edges. The total running time of the algorithm is $O(n^{\frac{5}{2}})$ in the worst case.
\end{proof}

\section{An Extension to $k$-regular bipartite graphs}

In this section, we demonstrate how the $BIGCYCLE$ algorithm can
be used as a subroutine to produce an improved approximation algorithm
for $k$-regular bipartite graphs. The main idea in this algorithm
is that $k$-regular bipartite graphs contain cubic subgraphs on which
we can run $BIGCYCLE$ to obtain solutions to the cubic subgraphs,
which will also be solutions to the original $k$-regular bipartite
graphs. If the cubic subgraph we find is composed entirely of connected
components of size $8$ and larger, then we will get a solution with
at most $\frac{9}{7}n-2$ edges. However, if some of the components
are of size $6$ (in a cubic bipartite graph these will be $K_{3,3}$s),
then the $2$-factor we compute may have between $\frac{n}{6}$ and
$\frac{n}{7}$ cycles, which gives us a solution of size $x$ where
$\frac{9}{7}n-2\leq x\leq\frac{4}{3}n-2$. Algorithm \ref{oldalgo6} provides the
pseudo-code for selecting cubic subgraph from a $k$-regular bipartite
graph containing a small number of $K_{3,3}$s. In the analysis that
follows, we will bound the number of $K_{3,3}$s in the cubic subgraph
computed by Algorithm \ref{oldalgo6}, allowing us to prove a specific approximation
factor. The subroutine CountK33, used in Algorithm \ref{oldalgo6}, takes a graph
as input and returns the number of connected components of the graph
that are $K_{3,3}$s.

\begin{algorithm}[H] 
\caption{An algorithm to find a cubic subgraph from a $k$-regular bipartite
graph: CUBIC\label{oldalgo6}}

\begin{lyxcode}
Input:~A~connected,~undirected,~unweighted,~$k$-regular,~bipartite~graph,~$G_{0}=(V,E)$~\\
For~$i=1$~to~$k$:~\\
~~$M_{i}\leftarrow$FindPerfectMatching($G_{i-1}$)~\\
~~$G_{i}\leftarrow G_{i-1}\backslash M_{i}$~\\
~~If~$i=3$:~\\
~~~~$G_{cubic}\leftarrow(V,M_{1}\cup M_{2}\cup M_{3})$~\\
~~If~$i>3$:~\\
~~~~$G_{temp}\leftarrow(V,M_{1}\cup M_{2}\cup M_{i})$~\\
~~~~If~CountK33($G_{temp}$)<CountK33($G_{cubic}$):~\\
~~~~~~$G_{cubic}\leftarrow G_{temp}$~\\
End~Loop~\\
Return~$G_{cubic}$\end{lyxcode}
\end{algorithm}

\begin{lemma} \label{oldlemma7}
For any $k$-regular bipartite graph where $k\geq4$, $G$, at most
$\frac{n}{6(k-2)}$ $K_{3,3}$s are contained in CUBIC($G$), the
cubic bipartite graph output by Algorithm \ref{oldalgo6}.\end{lemma}
\begin{proof}
Consider an arbitrary $k$-regular bipartite graph $G$. Note that
the nodes of any $6$-cycle in $(V,M_{1}\cup M_{2})$ will form a
$K_{3,3}$ in $(V,M_{1}\cup M_{2}\cup M_{i})$ for at most $1$ value
of $i$ because the $k$ matchings $M_{1},\ldots,M_{k}$ are edge-disjoint.
There are at most $\frac{n}{6}$ $6$-cycles in $M_{1}\cup M_{2}$.
In the worst case, the nodes of each of these $6$-cycles can form
a $K_{3,3}$ in $(V,M_{1}\cup M_{2}\cup M_{i})$ for exactly one value
of $i$ where $3\leq i\leq k$. Then, by the pigeonhole principle,
there is some value $i$ where $3\leq i\leq k$ where $(V,M_{1}\cup M_{2}\cup M_{i})$
contains at most $\frac{\frac{n}{6}}{k-2}=\frac{n}{6(k-2)}$ $K_{3,3}$s.
Algorithm \ref{oldalgo6} finds $M_{i}$ where $(V,M_{1}\cup M_{2}\cup M_{i})$
contains the fewest $K_{3,3}$s, so the algorithm will necessarily
output three edge-disjoint matchings with at most $\frac{n}{6(k-2)}$
$K_{3,3}$s, proving the lemma.\end{proof}
\begin{theorem}
Given a $k$-regular bipartite $G$ with $n$ vertices where $k\geq4$,
there is a polynomial time algorithm that computes a spanning Eulerian
multigraph $H$ in $G$ with at most $(\frac{9}{7}+\frac{1}{21(k-2)})n-2$
edges.\end{theorem}
\begin{proof}
First, we will find a cubic subraph of $G$, $G_{cubic}$, by running
Algorithm \ref{oldalgo6} on $G$. In each component of $G_{cubic}$ that is a $K_{3,3}$
we will find a $6$-cycle covering these nodes. This can be done in
constant time by taking any walk through this component that does
not visit a node twice as long as this is possible, then returning
to the first node. In every other connected component, run the Contract
and Expand phases of the $BIGCYCLE$ algorithm, which will find a
$2$-factor over this component containing at most $\frac{n_{i}}{7}$
cycles, where $n_{i}$ is the number of nodes in the connected component.
The upper bound of $\frac{n_{i}}{7}$ cycles is proven by Theorem
\ref{t13}. By Lemma \ref{oldlemma7}, there are $x_{1}$ nodes in $K_{3,3}$s within $G_{cubic}$,
where $x_{1}\leq\frac{n}{k-2}$. These nodes are covered by $\frac{x_{1}}{6}$
cycles. Then, there $x_{2}$ nodes in the remaining components and
$x_{1}+x_{2}=n$. These $x_{2}$ nodes are covered by at most $\frac{x_{2}}{7}$
cycles. Then, the overall $2$-factor of $G$ has at most $\frac{x_{1}}{6}+\frac{x_{2}}{7}$
cycles. The following calculations compute an upper bound on these
cycles in terms of $n$:
\begin{align*}
\frac{x_{1}}{6}+\frac{x_{2}}{7} & =\frac{7x_{1}+6x_{2}}{42}\\
 & =\frac{6n+x_{1}}{42}\\
 & =\frac{n}{7}+\frac{x_{1}}{42}\\
 & \leq\frac{n}{7}+\frac{n}{42(k-2)}
\end{align*}

By Proposition \ref{t3}, this $2$-factor can be extended into a spanning
Eulerian multigraph in $G$ with at most $n+2(\frac{n}{7}+\frac{n}{42(k-2)}-1)=(\frac{9}{7}+\frac{1}{21(k-2)})n-2$
edges, proving the theorem.
\end{proof}

\subsection*{Acknowledgment}

We would like to thank Satoru Iwata and Alantha Newman for useful discussions during this project.

\bibliography{biblio}

\section{Appendix A: Squares}\label{apdxsq}

The purpose of the appendices is to demonstrate in detail how the
algorithm $BIGCYCLE$ ``winds'' a $2$-factor through the gadgets
($4$-cycles, $H_{1}$s, $H_{2}$s, $H_{3}$s, $H_{4}$s, $H_{5}$s,
and $H_{6}$s) as it expands the condensed graph back to its original
state.

\subsection{Gadget is covered by three edges of a single cycle}

If a gadget that replaced a square is covered by two disjoint cycles
of $F_{i}$, then the internal edge of the gadget must be included
from $F_{i}$. Then, $F_{i}$ must include either edge 1 or 3 and
either edge 2 or 4. However, while there are four orientations to
consider, they are all symmetric to each other, so there is only one
case to consider. In this case, we start with a cycle of length $x+3$
in $F_{i}$ and are returned a single cycle of length $x+5$ in $F_{i-1}$.
$x+3\geq6$, so $x+5\geq8$, meaning this class of expansions cannot
introduce an organic $6$-cycle into the $2$-factor.

\begin{figure}[H]
\begin{minipage}[t]{0.45\columnwidth}%
\begin{center}
\tikzstyle{node}=[circle, draw, fill=black!50,                         inner sep=0pt, minimum width=4pt]
\begin{tikzpicture}[auto,thick, scale=.4]	 	\node [node] (a) at (0,2) {}; 	\node [node] (b) at (0,-2) {}; 	\draw (a) -- (b); 	\coordinate (1) at  ($ (a)+ (-1,1) $); 	\draw (a) to node [label={[xshift=-.1cm]1}] {} (1); 	\coordinate (3) at  ($ (a)+ (1,1) $); 	\draw (a) to node [label={[xshift=.5cm, yshift=-.5cm]3}] {} (3); 	\coordinate (2) at  ($ (b)+ (-1,-1) $); 	\draw (b) to node [label={[xshift=-.5cm, yshift = .5cm]2}] {} (2); 	\coordinate (4) at  ($ (b)+ (1,-1) $); 	\draw (b) to node [label={4}] {} (4);
	\draw [line width=.08cm] (1) to (a) to (b) to (2); 	\draw [dashed, line width=.08cm] (1) to[out=180,in=180] node[label=x] {} (2);
\end{tikzpicture}
\par\end{center}

\begin{center}
\caption{A cycle of length $x+3$ that passes through a gadget that replaced
a square. }

\par\end{center}%
\end{minipage}\hfill{}%
\begin{minipage}[t]{0.45\columnwidth}%
\begin{center}
\tikzstyle{node}=[circle, draw, fill=black!50,                         inner sep=0pt, minimum width=4pt]
\begin{tikzpicture}[thick,scale=0.4] 	\node [node] (a) at (-2,2) {}; 	\node [node] (b) at (2,2) {}; 	\node [node] (c) at (2,-2) {}; 	\node [node] (d) at (-2,-2) {}; 	\draw (a) -- (b) -- (c) -- (d) -- (a); 	\coordinate (1) at  ($ (a)+ (-1,1) $); 	\coordinate (2) at  ($ (b)+ (1,1) $); 	\coordinate (3) at  ($ (c)+ (1,-1) $); 	\coordinate (4) at  ($ (d)+ (-1,-1) $); 	\draw (a) to node [label={[xshift=.5cm]1}] {} (1); 	\draw (b) to node [label={[xshift=-.5cm]2}] {} (2); 	\draw (c) to node [label={[xshift=-.5cm,yshift=-.5cm]3}] {} (3); 	\draw (d) to node [label={[xshift=.5cm,yshift=-.5cm]4}] {} (4);
	\draw [line width=.08cm] (1) to (a) to (d) to (c) to (b) to (2); 	\draw [dashed, line width=.08cm] (1) to[out=90,in=90] node[label=x] {} (2);
\end{tikzpicture}
\par\end{center}

\begin{center}
\caption{The cycle from the previous figure, after expanding the gadget, is
now of length $x+5$.}

\par\end{center}%
\end{minipage}
\end{figure}

\subsection{Gadget is covered by two cycles}

If a gadget that replaced a square is covered by two disjoint cycles
of $F_{i}$, then the internal edge of the gadget must be excluded
from $F_{i}$. Then, there is only one possible orientation in which
$F_{i}$ could cover the nodes of this gadget. In this case, we start
with cycles of lengths $x+2$ and $y+2$ in $F_{i}$ and are returned
a single cycle of length $x+y+6$ in $F_{i-1}$. $x+2\geq6$ and $y+2\geq6$,
so $x+y+6>8$, meaning this class of expansions cannot introduce an
organic $6$-cycle into the $2$-factor.

\begin{figure}[H]
\begin{minipage}[t]{0.45\columnwidth}%
\begin{center}
\tikzstyle{node}=[circle, draw, fill=black!50,                         inner sep=0pt, minimum width=4pt]
\begin{tikzpicture}[auto,thick, scale=.4]	 	\node [node] (a) at (0,2) {}; 	\node [node] (b) at (0,-2) {}; 	\draw (a) -- (b); 	\coordinate (1) at  ($ (a)+ (-1,1) $); 	\draw (a) to node [label={[xshift=-.1cm]1}] {} (1); 	\coordinate (3) at  ($ (a)+ (1,1) $); 	\draw (a) to node [label={[xshift=.5cm, yshift=-.5cm]3}] {} (3); 	\coordinate (2) at  ($ (b)+ (-1,-1) $); 	\draw (b) to node [label={[xshift=-.5cm, yshift = .5cm]2}] {} (2); 	\coordinate (4) at  ($ (b)+ (1,-1) $); 	\draw (b) to node [label={4}] {} (4);
	\draw [line width=.08cm] (1) to (a) to (3); 	\draw [line width=.08cm] (2) to (b) to (4); 	\draw [dashed, line width=.08cm] (1) to[out=90,in=90] node[label=x] {} (3); 	\draw [dashed, line width=.08cm] (2) to[out=-90,in=-90] node[label=y] {} (4);
\end{tikzpicture}
\par\end{center}

\begin{center}
\caption{Two cycles of lengths $x+2$ and $y+2$ that pass through a gadget
that replaced a square.}

\par\end{center}%
\end{minipage}\hfill{}%
\begin{minipage}[t]{0.45\columnwidth}%
\begin{center}
\tikzstyle{node}=[circle, draw, fill=black!50,                         inner sep=0pt, minimum width=4pt]
\begin{tikzpicture}[thick,scale=0.4]
\clip (-3,-5) rectangle (5,5);
\node [node] (a) at (-2,2) {}; 	\node [node] (b) at (2,2) {}; 	\node [node] (c) at (2,-2) {}; 	\node [node] (d) at (-2,-2) {}; 	\draw (a) -- (b) -- (c) -- (d) -- (a); 	\coordinate (1) at  ($ (a)+ (-1,1) $); 	\coordinate (2) at  ($ (b)+ (1,1) $); 	\coordinate (3) at  ($ (c)+ (1,-1) $); 	\coordinate (4) at  ($ (d)+ (-1,-1) $); 	\draw (a) to node [label={[xshift=.5cm]1}] {} (1); 	\draw (b) to node [label={[xshift=-.5cm]2}] {} (2); 	\draw (c) to node [label={[xshift=-.5cm,yshift=-.5cm]3}] {} (3); 	\draw (d) to node [label={[xshift=.5cm,yshift=-.5cm]4}] {} (4);
	\draw [line width=.08cm] (1) to (a) to (b) to (2); 	\draw [line width=.08cm] (4) to (d) to (c) to (3); 	\draw [dashed, line width=.08cm] (1) to[out=45,in=45, distance=7cm] node[label=x] {} (3); 	\draw [dashed, line width=.08cm] (2) to[out=-45,in=-45, distance=7cm] node[label=y] {} (4);
\end{tikzpicture}
\par\end{center}

\begin{center}
\caption{The cycles from the previous figure, after expanding the gadget, now
form a single cycle of length $x+y+6$.}

\par\end{center}%
\end{minipage}
\end{figure}

\subsection{Gadget is covered by four edges of a single cycle}

If a gadget that replaced a square is covered by four edges of a single
cycle of $F_{i}$, then there are two orientations in which $F_{i}$
could pass through these edges. However, these two cases are symmetric,
so there is only one case to consider. In this case, we start with
a cycle of length $x+y+4$ in $F_{i}$ and are returned a single cycle
of length $x+y+6$ in $F_{i-1}$. $x+y+4\geq6$, so $x+y+6\geq8$,
meaning this class of expansions cannot introduce an organic $6$-cycle
into the $2$-factor.

\begin{figure}[H]
\begin{minipage}[t]{0.45\columnwidth}%
\begin{center}
\tikzstyle{node}=[circle, draw, fill=black!50,                         inner sep=0pt, minimum width=4pt]
\begin{tikzpicture}[auto,thick, scale=.4]	 	\node [node] (a) at (0,2) {}; 	\node [node] (b) at (0,-2) {}; 	\draw (a) -- (b); 	\coordinate (1) at  ($ (a)+ (-1,1) $); 	\draw (a) to node [label={[xshift=-.1cm]1}] {} (1); 	\coordinate (3) at  ($ (a)+ (1,1) $); 	\draw (a) to node [label={[xshift=.5cm, yshift=-.5cm]3}] {} (3); 	\coordinate (2) at  ($ (b)+ (-1,-1) $); 	\draw (b) to node [label={[xshift=-.5cm, yshift = .5cm]2}] {} (2); 	\coordinate (4) at  ($ (b)+ (1,-1) $); 	\draw (b) to node [label={4}] {} (4);
	\draw [line width=.08cm] (1) to (a) to (3); 	\draw [line width=.08cm] (2) to (b) to (4); 	\draw [dashed, line width=.08cm] (1) to[out=180,in=180] node[label=x] {} (2); 	\draw [dashed, line width=.08cm] (3) to[out=0,in=-0] node[label=y] {} (4);
\end{tikzpicture} 
\par\end{center}

\begin{center}
\caption{A cycle of length $x+y+4$ that passes through a gadget that replaced
a square.}

\par\end{center}%
\end{minipage}\hfill{}%
\begin{minipage}[t]{0.45\columnwidth}%
\begin{center}
\tikzstyle{node}=[circle, draw, fill=black!50,                         inner sep=0pt, minimum width=4pt]
\begin{tikzpicture}[thick,scale=0.4] 	\node [node] (a) at (-2,2) {}; 	\node [node] (b) at (2,2) {}; 	\node [node] (c) at (2,-2) {}; 	\node [node] (d) at (-2,-2) {}; 	\draw (a) -- (b) -- (c) -- (d) -- (a); 	\coordinate (1) at  ($ (a)+ (-1,1) $); 	\coordinate (2) at  ($ (b)+ (1,1) $); 	\coordinate (3) at  ($ (c)+ (1,-1) $); 	\coordinate (4) at  ($ (d)+ (-1,-1) $); 	\draw (a) to node [label={[xshift=.5cm]1}] {} (1); 	\draw (b) to node [label={[xshift=-.5cm]2}] {} (2); 	\draw (c) to node [label={[xshift=-.5cm,yshift=-.5cm]3}] {} (3); 	\draw (d) to node [label={[xshift=.5cm,yshift=-.5cm]4}] {} (4);
	\draw [line width=.08cm] (1) to (a) to (d) to (4); 	\draw [line width=.08cm] (2) to (b) to (c) to (3); 	\draw [dashed, line width=.08cm] (1) to[out=45,in=135] node[label=x] {} (2); 	\draw [dashed, line width=.08cm] (3) to[out=-45,in=225] node[label=y] {} (4);
\end{tikzpicture}
\par\end{center}

\begin{center}
\caption{The cycle from the previous figure, after expanding the gadget, is
now of length $x+y+6$.}

\par\end{center}%
\end{minipage}
\end{figure}

\subsection{Super-vertex replaces a square}

A super-vertex that replaced a square is necessarily covered by $F_{i}$.
Then, there are three ways we can select the edge in each of the super-vertices
to exclude from $F_{i}$. However, the cases where edges 2 and 3 are
excluded from $F_{i}$ are symmetric so we will examine only the second
of these cases. Then, there are only two case to consider. In all
cases, we start with a cycle of length $x+2$ in $F_{i}$ and are
returned a single cycle of length $x+6$ in $F_{i-1}$. $x+2\geq6$,
so $x+6>8$, meaning this class of expansions cannot introduce an
organic $6$-cycle into the $2$-factor.

\begin{figure}[H]
\begin{minipage}[t]{0.45\columnwidth}%
\begin{center}
\tikzstyle{node}=[circle, draw, fill=black!50,                         inner sep=0pt, minimum width=4pt]
\begin{tikzpicture}[thick,scale=1.2] 	\node [node] (a) at (-1,.5) {}; 	\coordinate (1) at  ($ (a)+ (-.75,.75) $); 	\coordinate (2) at  ($ (a)+ (.75,.75) $); 	\coordinate (3) at  ($ (a)+ (0,-.8) $); 	\draw (a) to node [label={[xshift=.05cm]1}] {} (1); 	\draw (a) to node [label={[xshift=-.1cm]2}] {} (2); 	\draw (a) to node [label={[xshift=-.2cm,yshift=-.3cm]3}] {} (3);
	\draw [line width=.08cm] (1) to (a) to (2); 	\draw [dashed, line width=.08cm] (1) to[out=45,in=135] node[label=x] {} (2); \end{tikzpicture}
\par\end{center}

\begin{center}
\caption{A cycle of length $x+2$ that passes through a gadget that replaced
a square.}

\par\end{center}%
\end{minipage}\hfill{}%
\begin{minipage}[t]{0.45\columnwidth}%
\begin{center}
\tikzstyle{node}=[circle, draw, fill=black!50,                         inner sep=0pt, minimum width=4pt]
\begin{tikzpicture}[thick,scale=0.4] 	\node [node] (a) at (-2,2) {}; 	\node [node] (b) at (2,2) {}; 	\node [node] (c) at (2,-2) {}; 	\node [node] (d) at (-2,-2) {}; 	\node [node] (e) at (0,4) {};
	\draw (a) -- (b) -- (c) -- (d) -- (a); 	\draw (a) -- (e);
	\draw (c) to[out=0,in=50, distance=4.5cm] (e); 	\coordinate (1) at  ($ (e)+ (0,1.5) $); 	\coordinate (2) at  ($ (b)+ (1,1) $); 	\coordinate (3) at  ($ (d)+ (-1,-1) $); 	\draw (e) to node [label={[xshift=-.2cm]1}] {} (1); 	\draw (b) to node [label={[xshift=-.10cm]2}] {} (2); 	\draw (d) to node [label={[xshift=.1cm,yshift=-.5cm]3}] {} (3);
	\draw [line width=.08cm] (1) to (e) to (a) to (d) to (c) to (b) to (2); 	\draw [dashed, line width=.08cm] (1) to[out=45,in=45] node[label=x] {} (2);
\end{tikzpicture}
\par\end{center}

\begin{center}
\caption{The cycle from the previous figure, after expanding the gadget, is
now of length $x+6$.}

\par\end{center}%
\end{minipage}
\end{figure}
\begin{figure}[H]
\begin{minipage}[t]{0.45\columnwidth}%
\begin{center}
\tikzstyle{node}=[circle, draw, fill=black!50,                         inner sep=0pt, minimum width=4pt]
\begin{tikzpicture}[thick,scale=1.2] 	\node [node] (a) at (-1,.5) {}; 	\coordinate (1) at  ($ (a)+ (-.75,.75) $); 	\coordinate (2) at  ($ (a)+ (.75,.75) $); 	\coordinate (3) at  ($ (a)+ (0,-.8) $); 	\draw (a) to node [label={[xshift=.05cm]1}] {} (1); 	\draw (a) to node [label={[xshift=-.1cm]2}] {} (2); 	\draw (a) to node [label={[xshift=-.2cm,yshift=-.3cm]3}] {} (3);
	\draw [line width=.08cm] (2) to (a) to (3); 	\draw [dashed, line width=.08cm] (2) to[out=270,in=0] node[label={[xshift=-.1cm]x}] {} (3); \end{tikzpicture}
\par\end{center}

\begin{center}
\caption{A cycle of length $x+2$ that passes through a gadget that replaced
a square.}

\par\end{center}%
\end{minipage}\hfill{}%
\begin{minipage}[t]{0.45\columnwidth}%
\begin{center}
\tikzstyle{node}=[circle, draw, fill=black!50,                         inner sep=0pt, minimum width=4pt]
\begin{tikzpicture}[thick,scale=0.4] 	\node [node] (a) at (-2,2) {}; 	\node [node] (b) at (2,2) {}; 	\node [node] (c) at (2,-2) {}; 	\node [node] (d) at (-2,-2) {}; 	\node [node] (e) at (0,4) {};
	\draw (a) -- (b) -- (c) -- (d) -- (a); 	\draw (a) -- (e);
	\draw (c) to[out=0,in=50, distance=4.5cm] (e); 	\coordinate (1) at  ($ (e)+ (0,1.5) $); 	\coordinate (2) at  ($ (b)+ (1,1) $); 	\coordinate (3) at  ($ (d)+ (-1,-1) $); 	\draw (e) to node [label={[xshift=-.2cm]1}] {} (1); 	\draw (b) to node [label={[xshift=.10cm,yshift=-.6cm]2}] {} (2); 	\draw (d) to node [label={[xshift=.1cm,yshift=-.5cm]3}] {} (3);
	\draw [line width=.08cm] (3) to (d) to (c) to[out=0,in=50,distance=4.5cm] (e); 	\draw [line width=.08cm] (e) to (a) to (b) to (2); 	\draw [dashed, line width=.08cm] (2) to[out=160,in=135, distance=4.7cm] node[label={[xshift=0cm]x}] {} (3);
\end{tikzpicture}
\par\end{center}

\begin{center}
\caption{The cycle from the previous figure, after expanding the gadget, is
now of length $x+6$.}

\par\end{center}%
\end{minipage}
\end{figure}

\subsection{Super-edge replaces a square}

\begin{figure}[H]
\begin{minipage}[t]{0.45\columnwidth}%
\begin{center}
\tikzstyle{node}=[circle, draw, fill=black!50,                         inner sep=0pt, minimum width=4pt]
\begin{tikzpicture}[thick,scale=0.4] 	\node [node] (A) at (0,5.5) {A}; 	\node [node] (B) at (0,-5.5) {B}; 	\draw (A) to (B);
\end{tikzpicture}
\par\end{center}

\begin{center}
\caption{The super-edge is not included in the $2$-factor.}

\par\end{center}%
\end{minipage}\hfill{}%
\begin{minipage}[t]{0.45\columnwidth}%
\begin{center}
\tikzstyle{node}=[circle, draw, fill=black!50,                         inner sep=0pt, minimum width=4pt]
\begin{tikzpicture}[thick,scale=0.4] 	\node [node] (a) at (-2,2) {}; 	\node [node] (b) at (2,2) {}; 	\node [node] (c) at (2,-2) {}; 	\node [node] (d) at (-2,-2) {}; 	\node [node] (e) at (0,4) {}; 	\node [node] (f) at (0,-4) {}; 	\draw (a) -- (b) -- (c) -- (d) -- (a); 	\draw (a) -- (e); 	\draw (d) -- (f); 	\draw (c) to[out=0,in=50, distance=5cm] (e); 	\draw (b) to[out=-50,in=0, distance=5cm] (f); 	\node [node] (1) at  ($ (e)+ (0,1.5) $) {A}; 	\node [node] (2) at  ($ (f)+ (0,-1.5) $) {B}; 	\draw (e) to (1); 	\draw (f) to (2);
	\draw [line width=.08cm] (e) to (a) to (d) to (f) to[out=0,in=-50,distance=5cm] (b) to (c) to[out=0,in=50, distance=5cm] (e);
\end{tikzpicture}
\par\end{center}

\begin{center}
\caption{A cycle of length $6$ passes through the square after the super-edge
in the previous figure is expanded. The impact of these $6$-cycles
on the algorithm's result is analyzed in Sections \ref{expandsquare} and \ref{analyzingworstcase}}

\par\end{center}%
\end{minipage}
\end{figure}
\begin{figure}[H]
\begin{minipage}[t]{0.45\columnwidth}%
\begin{center}
\tikzstyle{node}=[circle, draw, fill=black!50,                         inner sep=0pt, minimum width=4pt]
\begin{tikzpicture}[thick,scale=0.4] 	\node [node] (A) at (0,5.5) {A}; 	\node [node] (B) at (0,-5.5) {B}; 	\draw [line width=.08cm] (A) to (B); 	\draw [dashed, line width=.08cm] (A) to[out=0,in=0, distance=4cm] node[label={[xshift=-.2cm]x}] {} (B);
\end{tikzpicture}
\par\end{center}

\begin{center}
\caption{A cycle of length $x+1$ that passes through a gadget that replaced
a square.}

\par\end{center}%
\end{minipage}\hfill{}%
\begin{minipage}[t]{0.45\columnwidth}%
\begin{center}
\tikzstyle{node}=[circle, draw, fill=black!50,                         inner sep=0pt, minimum width=4pt]
\begin{tikzpicture}[thick,scale=0.4] 	\node [node] (a) at (-2,2) {}; 	\node [node] (b) at (2,2) {}; 	\node [node] (c) at (2,-2) {}; 	\node [node] (d) at (-2,-2) {}; 	\node [node] (e) at (0,4) {}; 	\node [node] (f) at (0,-4) {}; 	\draw (a) -- (b) -- (c) -- (d) -- (a); 	\draw (a) -- (e); 	\draw (d) -- (f); 	\draw (c) to[out=0,in=50, distance=5cm] (e); 	\draw (b) to[out=-50,in=0, distance=5cm] (f); 	\node [node] (1) at  ($ (e)+ (0,1.5) $) {A}; 	\node [node] (2) at  ($ (f)+ (0,-1.5) $) {B}; 	\draw (e) to (1); 	\draw (f) to (2);
	\draw [line width=.08cm] (1) to (e) to (a) to (b) to (c) to (d) to (f) to (2);
	\draw [dashed, line width=.08cm] (1) to[out=180,in=180, distance=4cm] node[label={[xshift=-.2cm]x}] {} (2); \end{tikzpicture}
\par\end{center}

\begin{center}
\caption{The cycle from the previous figure, after expanding the gadget, is
now of length $x+7$.}

\par\end{center}%
\end{minipage}
\end{figure}

\section{Appendix B: $H_{1}$s} \label{apdxh1}

\subsection{Gadget is covered by two cycles}

If a $H'_{1}$ is covered by two disjoint cycle in $F_{i}$, then,
there are three ways we can select the edge in each of the super-vertices
to exclude from $F_{i}$. However, without loss of generality, we
can fix the edge of the first super-vertex to exclude from $F_{i}$.
Then, there are only three cases to consider. In each of these cases,
we start with cycles of lengths $x+2$ and $y+2$ in $F_{i}$ and
are returned either a single cycle of length $x+y+8$ or two cycles
of lengths $x+4$ and $y+4$ in $F_{i-1}$. $x+2\geq6$ and $y+2\geq6$,
so $x+y+8>8$, meaning the first case cannot introduce an organic
$6$-cycle into the $2$-factor. Similarly, we conclude in the second
case that $x+4\geq8$ and $y+4\geq8$ so neither the $x+4$ or $y+4$
cycles can be organic $6$-cycles.

\begin{figure}[H]
\begin{minipage}[t]{0.45\columnwidth}%
\begin{center}
\tikzstyle{node}=[circle, draw, fill=black!50,                         inner sep=0pt, minimum width=4pt]
\begin{tikzpicture}[thick,scale=1.2] 	\node [node] (a) at (-1,.5) {};
	\node [node] (b) at (1,.5) {};
	\coordinate (1) at  ($ (a)+ (-.75,.75) $); 	\coordinate (2) at  ($ (b)+ (-.75,.75) $); 	\coordinate (3) at  ($ (a)+ (.75,.75) $); 	\coordinate (4) at  ($ (b)+ (.75,.75) $); 	\coordinate (5) at  ($ (a)+ (0,-.8) $); 	\coordinate (6) at  ($ (b) + (0,-.8) $); 	\draw (a) to node [label={[xshift=.05cm]1}] {} (1); 	\draw (b) to node [label={[xshift=.05cm]2}] {} (2); 	\draw (a) to node [label={[xshift=-.1cm]3}] {} (3); 	\draw (b) to node [label={[xshift=-.1cm]4}] {} (4); 	\draw (a) to node [label={[xshift=-.2cm,yshift=-.3cm]5}] {} (5); 	\draw (b) to node [label={[xshift=-.2cm,yshift=-.3cm]6}] {} (6);
	\draw [line width=.08cm] (1) to (a) to (3); 	\draw [line width=.08cm] (2) to (b) to (4); 	\draw [dashed, line width=.08cm] (1) to[out=90,in=90] node[label={[xshift=0cm]x}] {} (3);
	\draw [dashed, line width=.08cm] (2) to[out=90,in=90] node[label={[xshift=0cm]y}] {} (4); 
\end{tikzpicture}
\par\end{center}

\begin{center}
\caption{Two cycles of lengths $x+2$ and $y+2$ that pass through a $H'_{1}$.}

\par\end{center}%
\end{minipage}\hfill{}%
\begin{minipage}[t]{0.45\columnwidth}%
\begin{center}
\tikzstyle{node}=[circle, draw, fill=black!50,                         inner sep=0pt, minimum width=4pt]
\begin{tikzpicture}[thick,scale=.8] 	\node [node] (a) at (-1,.5) {}; 	\node [node] (b) at (0,1.2) {}; 	\node [node] (c) at (1,.5) {}; 	\node [node] (d) at (1,-.5) {}; 	\node [node] (e) at (0,-1.2) {}; 	\node [node] (f) at (-1,-.5) {};
	\draw (a) -- (b) -- (c) -- (d) -- (e) -- (f) -- (a); 	\coordinate (1) at  ($ (a)+ (-.75,.75) $); 	\coordinate (2) at  ($ (b)+ (0,.8) $); 	\coordinate (3) at  ($ (c)+ (.75,.75) $); 	\coordinate (4) at  ($ (d)+ (.75,-.75) $); 	\coordinate (5) at  ($ (e)+ (0,-.8) $); 	\coordinate (6) at  ($ (f) + (-.75,-.75) $); 	\draw (a) to node [label={[xshift=.1cm]1}] {} (1); 	\draw (b) to node [label={[xshift=-.2cm]2}] {} (2); 	\draw (c) to node [label={[xshift=-.2cm]3}] {} (3); 	\draw (d) to node [label={[xshift=.3cm,yshift=-.2cm]4}] {} (4); 	\draw (e) to node [label={[xshift=-.2cm,yshift=-.3cm]5}] {} (5); 	\draw (f) to node [label={[xshift=-.1cm]6}] {} (6);
	\draw [line width=.08cm] (a) to (1); 	\draw [line width=.08cm] (b) to (2); 	\draw [line width=.08cm] (c) to (3); 	\draw [line width=.08cm] (d) to (4); 	\draw [line width=.08cm] (b) to (c); 	\draw [line width=.08cm] (a) to (f) to (e) to (d);
	\draw [dashed, line width=.08cm] (1) to[distance=2cm] node[label=x] {} (3); 	\draw [dashed, line width=.08cm] (2) to[in=50,distance=2.5cm] node[label=y] {} (4); \end{tikzpicture}
\par\end{center}

\begin{center}
\caption{The cycles from the previous figure, after expanding the gadget, now
form a single cycle of length $x+y+8$.}

\par\end{center}%
\end{minipage}
\end{figure}
\begin{figure}[H]
\begin{minipage}[t]{0.45\columnwidth}%
\begin{center}
\tikzstyle{node}=[circle, draw, fill=black!50,                         inner sep=0pt, minimum width=4pt]
\begin{tikzpicture}[thick,scale=1.2] 	\node [node] (a) at (-1,.5) {};
	\node [node] (b) at (1,.5) {};
	\coordinate (1) at  ($ (a)+ (-.75,.75) $); 	\coordinate (2) at  ($ (b)+ (-.75,.75) $); 	\coordinate (3) at  ($ (a)+ (.75,.75) $); 	\coordinate (4) at  ($ (b)+ (.75,.75) $); 	\coordinate (5) at  ($ (a)+ (0,-.8) $); 	\coordinate (6) at  ($ (b) + (0,-.8) $); 	\draw (a) to node [label={[xshift=.05cm]1}] {} (1); 	\draw (b) to node [label={[xshift=.05cm]2}] {} (2); 	\draw (a) to node [label={[xshift=-.1cm]3}] {} (3); 	\draw (b) to node [label={[xshift=-.1cm]4}] {} (4); 	\draw (a) to node [label={[xshift=-.2cm,yshift=-.3cm]5}] {} (5); 	\draw (b) to node [label={[xshift=-.2cm,yshift=-.3cm]6}] {} (6);
	\draw [line width=.08cm] (1) to (a) to (3); 	\draw [line width=.08cm] (2) to (b) to (6); 	\draw [dashed, line width=.08cm] (1) to[out=90,in=90] node[label={[xshift=0cm]x}] {} (3);
	\draw [dashed, line width=.08cm] (2) to[out=200,in=240] node[label={[xshift=-.3cm]y}] {} (6); 
\end{tikzpicture}
\par\end{center}

\begin{center}
\caption{Two cycles of lengths $x+2$ and $y+2$ that pass through a $H'_{1}$.}

\par\end{center}%
\end{minipage}\hfill{}%
\begin{minipage}[t]{0.45\columnwidth}%
\begin{center}
\tikzstyle{node}=[circle, draw, fill=black!50,                         inner sep=0pt, minimum width=4pt]
\begin{tikzpicture}[thick,scale=.8] 	\node [node] (a) at (-1,.5) {}; 	\node [node] (b) at (0,1.2) {}; 	\node [node] (c) at (1,.5) {}; 	\node [node] (d) at (1,-.5) {}; 	\node [node] (e) at (0,-1.2) {}; 	\node [node] (f) at (-1,-.5) {};
	\draw (a) -- (b) -- (c) -- (d) -- (e) -- (f) -- (a); 	\coordinate (1) at  ($ (a)+ (-.75,.75) $); 	\coordinate (2) at  ($ (b)+ (0,.8) $); 	\coordinate (3) at  ($ (c)+ (.75,.75) $); 	\coordinate (4) at  ($ (d)+ (.75,-.75) $); 	\coordinate (5) at  ($ (e)+ (0,-.8) $); 	\coordinate (6) at  ($ (f) + (-.75,-.75) $); 	\draw (a) to node [label={[xshift=.1cm]1}] {} (1); 	\draw (b) to node [label={[xshift=-.2cm]2}] {} (2); 	\draw (c) to node [label={[xshift=-.2cm]3}] {} (3); 	\draw (d) to node [label={[xshift=.3cm,yshift=-.2cm]4}] {} (4); 	\draw (e) to node [label={[xshift=-.2cm,yshift=-.3cm]5}] {} (5); 	\draw (f) to node [label={[xshift=-.1cm]6}] {} (6);
	\draw [line width=.08cm] (2) to (b) to (a) to (1); 	\draw [line width=.08cm] (6) to (f) to (e) to (d) to (c) to (3); 	\draw [dashed, line width=.08cm] (1) to[distance=2cm] node[label=x] {} (3); 	\draw [dashed, line width=.08cm] (2) to[out=120,in=135,distance=2.5cm] node[label=y] {} (6); \end{tikzpicture}
\par\end{center}

\begin{center}
\caption{The cycles from the previous figure, after expanding the gadget, now
form a single cycle of length $x+y+8$.}

\par\end{center}%
\end{minipage}
\end{figure}
\begin{figure}[H]
\begin{minipage}[t]{0.45\columnwidth}%
\begin{center}
 \tikzstyle{node}=[circle, draw, fill=black!50,                         inner sep=0pt, minimum width=4pt]
\begin{tikzpicture}[thick,scale=1.2] 	\node [node] (a) at (-1,.5) {};
	\node [node] (b) at (1,.5) {};
	\coordinate (1) at  ($ (a)+ (-.75,.75) $); 	\coordinate (2) at  ($ (b)+ (-.75,.75) $); 	\coordinate (3) at  ($ (a)+ (.75,.75) $); 	\coordinate (4) at  ($ (b)+ (.75,.75) $); 	\coordinate (5) at  ($ (a)+ (0,-.8) $); 	\coordinate (6) at  ($ (b) + (0,-.8) $); 	\draw (a) to node [label={[xshift=.05cm]1}] {} (1); 	\draw (b) to node [label={[xshift=.05cm]2}] {} (2); 	\draw (a) to node [label={[xshift=-.1cm]3}] {} (3); 	\draw (b) to node [label={[xshift=-.1cm]4}] {} (4); 	\draw (a) to node [label={[xshift=-.2cm,yshift=-.3cm]5}] {} (5); 	\draw (b) to node [label={[xshift=-.2cm,yshift=-.3cm]6}] {} (6);
	\draw [line width=.08cm] (1) to (a) to (3); 	\draw [line width=.08cm] (4) to (b) to (6); 	\draw [dashed, line width=.08cm] (1) to[out=90,in=90] node[label={[xshift=0cm]x}] {} (3);
	\draw [dashed, line width=.08cm] (4) to[out=20,in=0] node[label={[xshift=.3cm]y}] {} (6); 
\end{tikzpicture}
\par\end{center}

\begin{center}
\caption{Two cycles of lengths $x+2$ and $y+2$ that pass through a $H'_{1}$.}

\par\end{center}%
\end{minipage}\hfill{}%
\begin{minipage}[t]{0.45\columnwidth}%
\begin{center}
 \tikzstyle{node}=[circle, draw, fill=black!50,                         inner sep=0pt, minimum width=4pt]
\begin{tikzpicture}[thick,scale=.8] 	\node [node] (a) at (-1,.5) {}; 	\node [node] (b) at (0,1.2) {}; 	\node [node] (c) at (1,.5) {}; 	\node [node] (d) at (1,-.5) {}; 	\node [node] (e) at (0,-1.2) {}; 	\node [node] (f) at (-1,-.5) {};
	\draw (a) -- (b) -- (c) -- (d) -- (e) -- (f) -- (a); 	\coordinate (1) at  ($ (a)+ (-.75,.75) $); 	\coordinate (2) at  ($ (b)+ (0,.8) $); 	\coordinate (3) at  ($ (c)+ (.75,.75) $); 	\coordinate (4) at  ($ (d)+ (.75,-.75) $); 	\coordinate (5) at  ($ (e)+ (0,-.8) $); 	\coordinate (6) at  ($ (f) + (-.75,-.75) $); 	\draw (a) to node [label={[xshift=.1cm]1}] {} (1); 	\draw (b) to node [label={[xshift=-.2cm]2}] {} (2); 	\draw (c) to node [label={[xshift=-.2cm]3}] {} (3); 	\draw (d) to node [label={[xshift=.3cm,yshift=-.2cm]4}] {} (4); 	\draw (e) to node [label={[xshift=-.2cm,yshift=-.3cm]5}] {} (5); 	\draw (f) to node [label={[xshift=-.1cm]6}] {} (6);
	\draw [line width=.08cm] (3) to (c) to (b) to (a) to (1); 	\draw [line width=.08cm] (6) to (f) to (e) to (d) to (4); 	\draw [dashed, line width=.08cm] (1) to[distance=2cm] node[label=x] {} (3); 	\draw [dashed, line width=.08cm] (4) to[out=240,in=300,distance=2cm] node[label=y] {} (6); \end{tikzpicture}
\par\end{center}

\begin{center}
\caption{The cycles from the previous figure, after expanding the gadget, now
form two cycles of lengths $x+4$ and $y+4$, respectively.}

\par\end{center}%
\end{minipage}
\end{figure}

\subsection{Gadget is covered by one cycle}

If a $H'_{1}$ is covered by a single cycle in $F_{i}$, then, there
are three ways we can select the edge in each of the super-vertices
to exclude from $F_{i}$. However, without loss of generality, we
can fix the edge of the first super-vertex to exclude from $F_{i}$.
Then, in each of these three configurations, we examine the two orientations
in which the cycle can pass through the two super-vertices. In all
$6$ cases, we start with a cycle of lengths $x+y+4$ in $F_{i}$
and are returned either a single cycle of length $x+y+8$, two cycles
of lengths $x+3$ and $y+5$, or two cycles of lengths $x+5$ and
$y+3$ in $F_{i-1}$. $x+y+4\geq6$, so $x+y+8>8$, meaning the first
case cannot introduce an organic $6$-cycle into the $2$-factor.
In the later two cases the $x+3$ or $y+3$ cycle can be an organic
$6$-cycle, but this is the expansion examined in detail in Sections
\ref{expandh1} and \ref{analyzingworstcase}.

\begin{figure}[H]
\begin{minipage}[t]{0.45\columnwidth}%
\begin{center}
\tikzstyle{node}=[circle, draw, fill=black!50,                         inner sep=0pt, minimum width=4pt]
\begin{tikzpicture}[thick,scale=.8] 	\node [node] (a) at (0,1.5) {};
	\node [node] (b) at (0,-1.5) {};
	\node [node] (t1) at  ($ (a)+ (-.5,-1) $) {}; 	\node [node] (t2) at  ($ (a)+ (.5,-1) $) {}; 	\coordinate (t3) at  ($ (a)+ (0,.75) $); 	\node [node] (b1) at  ($ (b)+ (-.5,1) $) {}; 	\node [node] (b2) at  ($ (b)+ (.5,1) $) {}; 	\coordinate (b3) at  ($ (b) + (0,-.75) $); 	\draw (a) to node [label={[xshift=-.1cm]1}] {} (t1); 	\draw (a) to node [label={[xshift=.1cm]3}] {} (t2); 	\draw (a) to node [label={[xshift=-.2cm]5}] {} (t3); 	\draw (b) to node [label={[xshift=-.3cm,yshift=-.2cm]2}] {} (b1); 	\draw (b) to node [label={[xshift=.3cm,yshift=-.20cm]4}] {} (b2); 	\draw (b) to node [label={[xshift=-.2cm,yshift=-.3cm]6}] {} (b3);
	\draw [line width=.08cm] (t1) to (a) to (t2); 	\draw [line width=.08cm] (b1) to (b) to (b2); 	\draw [dashed, line width=.08cm] (t1) to node[label={[xshift=-.2cm,yshift=-.25cm]x}] {} (b1); 	\draw [dashed, line width=.08cm] (t2) to node[label={[xshift=.2cm,yshift=-.25cm]y}] {} (b2); 
\end{tikzpicture}
\par\end{center}

\begin{center}
\caption{A cycle of length $x+y+4$ that passes through a $H'_{1}$.}

\par\end{center}%
\end{minipage}\hfill{}%
\begin{minipage}[t]{0.45\columnwidth}%
\begin{center}
\tikzstyle{node}=[circle, draw, fill=black!50,                         inner sep=0pt, minimum width=4pt]
\begin{tikzpicture}[thick,scale=.8] 	\node [node] (a) at (-1,.5) {}; 	\node [node] (b) at (0,1.2) {}; 	\node [node] (c) at (1,.5) {}; 	\node [node] (d) at (1,-.5) {}; 	\node [node] (e) at (0,-1.2) {}; 	\node [node] (f) at (-1,-.5) {};
	\draw (a) -- (b) -- (c) -- (d) -- (e) -- (f) -- (a); 	\coordinate (1) at  ($ (a)+ (-.75,.75) $); 	\coordinate (2) at  ($ (b)+ (0,.8) $); 	\coordinate (3) at  ($ (c)+ (.75,.75) $); 	\coordinate (4) at  ($ (d)+ (.75,-.75) $); 	\coordinate (5) at  ($ (e)+ (0,-.8) $); 	\coordinate (6) at  ($ (f) + (-.75,-.75) $); 	\draw (a) to node [label={[xshift=.1cm]1}] {} (1); 	\draw (b) to node [label={[xshift=-.2cm]2}] {} (2); 	\draw (c) to node [label={[xshift=-.2cm]3}] {} (3); 	\draw (d) to node [label={[xshift=.3cm,yshift=-.2cm]4}] {} (4); 	\draw (e) to node [label={[xshift=-.2cm,yshift=-.3cm]5}] {} (5); 	\draw (f) to node [label={[xshift=-.1cm]6}] {} (6);
	\draw [line width=.08cm] (a) to (1); 	\draw [line width=.08cm] (b) to (2); 	\draw [line width=.08cm] (c) to (3); 	\draw [line width=.08cm] (d) to (4); 	\draw [line width=.08cm] (b) to (c); 	\draw [line width=.08cm] (a) to (f) to (e) to (d);
	\draw [dashed, line width=.08cm] (1) to[out=90,in=110] node[label=x] {} (2); 	\draw [dashed, line width=.08cm] (3) to[out=20,in=50] node[label={[xshift=.2cm]y}] {} (4); \end{tikzpicture}
\par\end{center}

\begin{center}
\caption{The cycle from the previous figure, after expanding the gadget, is
now of length $x+y+8$.}

\par\end{center}%
\end{minipage}
\end{figure}
\begin{figure}[H]
\begin{minipage}[t]{0.45\columnwidth}%
\begin{center}
\tikzstyle{node}=[circle, draw, fill=black!50,                         inner sep=0pt, minimum width=4pt]
\begin{tikzpicture}[thick,scale=.8] 	\node [node] (a) at (0,1.5) {};
	\node [node] (b) at (0,-1.5) {};
	\node [node] (t1) at  ($ (a)+ (-.5,-1) $) {}; 	\node [node] (t2) at  ($ (a)+ (.5,-1) $) {}; 	\coordinate (t3) at  ($ (a)+ (0,.75) $); 	\node [node] (b1) at  ($ (b)+ (-.5,1) $) {}; 	\node [node] (b2) at  ($ (b)+ (.5,1) $) {}; 	\coordinate (b3) at  ($ (b) + (0,-.75) $); 	\draw (a) to node [label={[xshift=-.1cm]1}] {} (t1); 	\draw (a) to node [label={[xshift=.1cm]3}] {} (t2); 	\draw (a) to node [label={[xshift=-.2cm]5}] {} (t3); 	\draw (b) to node [label={[xshift=-.3cm,yshift=-.2cm]2}] {} (b1); 	\draw (b) to node [label={[xshift=.3cm,yshift=-.20cm]6}] {} (b2); 	\draw (b) to node [label={[xshift=-.2cm,yshift=-.3cm]4}] {} (b3);
	\draw [line width=.08cm] (t1) to (a) to (t2); 	\draw [line width=.08cm] (b1) to (b) to (b2); 	\draw [dashed, line width=.08cm] (t1) to node[label={[xshift=-.2cm,yshift=-.25cm]x}] {} (b1); 	\draw [dashed, line width=.08cm] (t2) to node[label={[xshift=.2cm,yshift=-.25cm]y}] {} (b2); 
\end{tikzpicture}
\par\end{center}

\begin{center}
\caption{A cycle of length $x+y+4$ that passes through a $H'_{1}$.}

\par\end{center}%
\end{minipage}\hfill{}%
\begin{minipage}[t]{0.45\columnwidth}%
\begin{center}
 \tikzstyle{node}=[circle, draw, fill=black!50,                         inner sep=0pt, minimum width=4pt]
\begin{tikzpicture}[thick,scale=.8] 	\node [node] (a) at (-1,.5) {}; 	\node [node] (b) at (0,1.2) {}; 	\node [node] (c) at (1,.5) {}; 	\node [node] (d) at (1,-.5) {}; 	\node [node] (e) at (0,-1.2) {}; 	\node [node] (f) at (-1,-.5) {};
	\draw (a) -- (b) -- (c) -- (d) -- (e) -- (f) -- (a); 	\coordinate (1) at  ($ (a)+ (-.75,.75) $); 	\coordinate (2) at  ($ (b)+ (0,.8) $); 	\coordinate (3) at  ($ (c)+ (.75,.75) $); 	\coordinate (4) at  ($ (d)+ (.75,-.75) $); 	\coordinate (5) at  ($ (e)+ (0,-.8) $); 	\coordinate (6) at  ($ (f) + (-.75,-.75) $); 	\draw (a) to node [label={[xshift=.1cm]1}] {} (1); 	\draw (b) to node [label={[xshift=-.2cm]2}] {} (2); 	\draw (c) to node [label={[xshift=-.2cm]3}] {} (3); 	\draw (d) to node [label={[xshift=.3cm,yshift=-.2cm]4}] {} (4); 	\draw (e) to node [label={[xshift=-.2cm,yshift=-.3cm]5}] {} (5); 	\draw (f) to node [label={[xshift=-.1cm]6}] {} (6);
	\draw [line width=.08cm] (2) to (b) to (a) to (1); 	\draw [line width=.08cm] (6) to (f) to (e) to (d) to (c) to (3); 	\draw [dashed, line width=.08cm] (1) to[out=90,in=110] node[label=x] {} (2); 	\draw [dashed, line width=.08cm] (3) to[out=-60,in=-45,distance=3cm] node[label={[xshift=0cm,yshift=-.6cm]y}] {} (6); \end{tikzpicture}
\par\end{center}

\begin{center}
\caption{The cycle from the previous figure, after expanding the gadget, is
now two cycles, of lengths $x+3$ and $y+5$, respectively. This expansion
can produce an organic $6$-cycle if $x=3$ and the cycle is organic.
The impact of these $6$-cycles on the algorithm's result is analyzed
in Sections \ref{expandh1} and \ref{analyzingworstcase}.}

\par\end{center}%
\end{minipage}
\end{figure}
\begin{figure}[H]
\begin{minipage}[t]{0.45\columnwidth}%
\begin{center}
\tikzstyle{node}=[circle, draw, fill=black!50,                         inner sep=0pt, minimum width=4pt]
\begin{tikzpicture}[thick,scale=.8] 	\node [node] (a) at (0,1.5) {};
	\node [node] (b) at (0,-1.5) {};
	\node [node] (t1) at  ($ (a)+ (-.5,-1) $) {}; 	\node [node] (t2) at  ($ (a)+ (.5,-1) $) {}; 	\coordinate (t3) at  ($ (a)+ (0,.75) $); 	\node [node] (b1) at  ($ (b)+ (-.5,1) $) {}; 	\node [node] (b2) at  ($ (b)+ (.5,1) $) {}; 	\coordinate (b3) at  ($ (b) + (0,-.75) $); 	\draw (a) to node [label={[xshift=-.1cm]1}] {} (t1); 	\draw (a) to node [label={[xshift=.1cm]3}] {} (t2); 	\draw (a) to node [label={[xshift=-.2cm]5}] {} (t3); 	\draw (b) to node [label={[xshift=-.3cm,yshift=-.2cm]4}] {} (b1); 	\draw (b) to node [label={[xshift=.3cm,yshift=-.20cm]2}] {} (b2); 	\draw (b) to node [label={[xshift=-.2cm,yshift=-.3cm]6}] {} (b3);
	\draw [line width=.08cm] (t1) to (a) to (t2); 	\draw [line width=.08cm] (b1) to (b) to (b2); 	\draw [dashed, line width=.08cm] (t1) to node[label={[xshift=-.2cm,yshift=-.25cm]x}] {} (b1); 	\draw [dashed, line width=.08cm] (t2) to node[label={[xshift=.2cm,yshift=-.25cm]y}] {} (b2); 
\end{tikzpicture}
\par\end{center}

\begin{center}
\caption{A cycle of length $x+y+4$ that passes through a $H'_{1}$.}

\par\end{center}%
\end{minipage}\hfill{}%
\begin{minipage}[t]{0.45\columnwidth}%
\begin{center}
 \tikzstyle{node}=[circle, draw, fill=black!50,                         inner sep=0pt, minimum width=4pt]
\begin{tikzpicture}[thick,scale=.8]
\clip (-3.5,-3) rectangle (3,3);
\node [node] (a) at (-1,.5) {}; 	\node [node] (b) at (0,1.2) {}; 	\node [node] (c) at (1,.5) {}; 	\node [node] (d) at (1,-.5) {}; 	\node [node] (e) at (0,-1.2) {}; 	\node [node] (f) at (-1,-.5) {};
	\draw (a) -- (b) -- (c) -- (d) -- (e) -- (f) -- (a); 	\coordinate (1) at  ($ (a)+ (-.75,.75) $); 	\coordinate (2) at  ($ (b)+ (0,.8) $); 	\coordinate (3) at  ($ (c)+ (.75,.75) $); 	\coordinate (4) at  ($ (d)+ (.75,-.75) $); 	\coordinate (5) at  ($ (e)+ (0,-.8) $); 	\coordinate (6) at  ($ (f) + (-.75,-.75) $); 	\draw (a) to node [label={[xshift=.1cm]1}] {} (1); 	\draw (b) to node [label={[xshift=-.2cm]2}] {} (2); 	\draw (c) to node [label={[xshift=-.2cm]3}] {} (3); 	\draw (d) to node [label={[xshift=.3cm,yshift=-.2cm]4}] {} (4); 	\draw (e) to node [label={[xshift=-.2cm,yshift=-.3cm]5}] {} (5); 	\draw (f) to node [label={[xshift=-.1cm]6}] {} (6);
	\draw [line width=.08cm] (a) to (1); 	\draw [line width=.08cm] (b) to (2); 	\draw [line width=.08cm] (c) to (3); 	\draw [line width=.08cm] (d) to (4); 	\draw [line width=.08cm] (b) to (c); 	\draw [line width=.08cm] (a) to (f) to (e) to (d);
	\draw [dashed, line width=.08cm] (1) to[in=260,out=180, distance=4cm] node[label={[xshift=-.5cm,yshift=-.2cm]x}] {} (4); 	\draw [dashed, line width=.08cm] (3) to[in=70,out=60] node[label=y] {} (2); \end{tikzpicture} 
\par\end{center}

\begin{center}
\caption{The cycle from the previous figure, after expanding the gadget, is
now two cycles, of lengths $y+3$ and $x+5$, respectively. This expansion
can produce an organic $6$-cycle if $y=3$ and the cycle is organic.
The impact of these $6$-cycles on the algorithm's result is analyzed
in Sections \ref{expandh1} and \ref{analyzingworstcase}.}

\par\end{center}%
\end{minipage}
\end{figure}
\begin{figure}[H]
\begin{minipage}[t]{0.45\columnwidth}%
\begin{center}
\tikzstyle{node}=[circle, draw, fill=black!50,                         inner sep=0pt, minimum width=4pt]
\begin{tikzpicture}[thick,scale=.8] 	\node [node] (a) at (0,1.5) {};
	\node [node] (b) at (0,-1.5) {};
	\node [node] (t1) at  ($ (a)+ (-.5,-1) $) {}; 	\node [node] (t2) at  ($ (a)+ (.5,-1) $) {}; 	\coordinate (t3) at  ($ (a)+ (0,.75) $); 	\node [node] (b1) at  ($ (b)+ (-.5,1) $) {}; 	\node [node] (b2) at  ($ (b)+ (.5,1) $) {}; 	\coordinate (b3) at  ($ (b) + (0,-.75) $); 	\draw (a) to node [label={[xshift=-.1cm]1}] {} (t1); 	\draw (a) to node [label={[xshift=.1cm]3}] {} (t2); 	\draw (a) to node [label={[xshift=-.2cm]5}] {} (t3); 	\draw (b) to node [label={[xshift=-.3cm,yshift=-.2cm]4}] {} (b1); 	\draw (b) to node [label={[xshift=.3cm,yshift=-.20cm]6}] {} (b2); 	\draw (b) to node [label={[xshift=-.2cm,yshift=-.3cm]2}] {} (b3);
	\draw [line width=.08cm] (t1) to (a) to (t2); 	\draw [line width=.08cm] (b1) to (b) to (b2); 	\draw [dashed, line width=.08cm] (t1) to node[label={[xshift=-.2cm,yshift=-.25cm]x}] {} (b1); 	\draw [dashed, line width=.08cm] (t2) to node[label={[xshift=.2cm,yshift=-.25cm]y}] {} (b2); 
\end{tikzpicture}
\par\end{center}

\begin{center}
\caption{A cycle of length $x+y+4$ that passes through a $H'_{1}$.}

\par\end{center}%
\end{minipage}\hfill{}%
\begin{minipage}[t]{0.45\columnwidth}%
\begin{center}
 \tikzstyle{node}=[circle, draw, fill=black!50,                         inner sep=0pt, minimum width=4pt]
\begin{tikzpicture}[thick,scale=.8]
\clip (-3,-3) rectangle (5,4);
\node [node] (a) at (-1,.5) {}; 	\node [node] (b) at (0,1.2) {}; 	\node [node] (c) at (1,.5) {}; 	\node [node] (d) at (1,-.5) {}; 	\node [node] (e) at (0,-1.2) {}; 	\node [node] (f) at (-1,-.5) {};
	\draw (a) -- (b) -- (c) -- (d) -- (e) -- (f) -- (a); 	\coordinate (1) at  ($ (a)+ (-.75,.75) $); 	\coordinate (2) at  ($ (b)+ (0,.8) $); 	\coordinate (3) at  ($ (c)+ (.75,.75) $); 	\coordinate (4) at  ($ (d)+ (.75,-.75) $); 	\coordinate (5) at  ($ (e)+ (0,-.8) $); 	\coordinate (6) at  ($ (f) + (-.75,-.75) $); 	\draw (a) to node [label={[xshift=.1cm]1}] {} (1); 	\draw (b) to node [label={[xshift=-.2cm]2}] {} (2); 	\draw (c) to node [label={[xshift=-.2cm]3}] {} (3); 	\draw (d) to node [label={[xshift=.3cm,yshift=-.2cm]4}] {} (4); 	\draw (e) to node [label={[xshift=-.2cm,yshift=-.3cm]5}] {} (5); 	\draw (f) to node [label={[xshift=-.1cm]6}] {} (6);
	\draw [line width=.08cm] (3) to (c) to (b) to (a) to (1); 	\draw [line width=.08cm] (6) to (f) to (e) to (d) to (4); 	\draw [dashed, line width=.08cm] (1) to[out=70,in=30,distance=3.5cm] node[label=x] {} (4); 	\draw [dashed, line width=.08cm] (3) to[out=-30,in=300,distance=3.5cm] node[label={[yshift=-.6cm]y}] {} (6); \end{tikzpicture}
\par\end{center}

\begin{center}
\caption{The cycle from the previous figure, after expanding the gadget, is
now of length $x+y+8$.}

\par\end{center}%
\end{minipage}
\end{figure}
\begin{figure}[H]
\begin{minipage}[t]{0.45\columnwidth}%
\begin{center}
\tikzstyle{node}=[circle, draw, fill=black!50,                         inner sep=0pt, minimum width=4pt]
\begin{tikzpicture}[thick,scale=.8] 	\node [node] (a) at (0,1.5) {};
	\node [node] (b) at (0,-1.5) {};
	\node [node] (t1) at  ($ (a)+ (-.5,-1) $) {}; 	\node [node] (t2) at  ($ (a)+ (.5,-1) $) {}; 	\coordinate (t3) at  ($ (a)+ (0,.75) $); 	\node [node] (b1) at  ($ (b)+ (-.5,1) $) {}; 	\node [node] (b2) at  ($ (b)+ (.5,1) $) {}; 	\coordinate (b3) at  ($ (b) + (0,-.75) $); 	\draw (a) to node [label={[xshift=-.1cm]1}] {} (t1); 	\draw (a) to node [label={[xshift=.1cm]3}] {} (t2); 	\draw (a) to node [label={[xshift=-.2cm]5}] {} (t3); 	\draw (b) to node [label={[xshift=-.3cm,yshift=-.2cm]6}] {} (b1); 	\draw (b) to node [label={[xshift=.3cm,yshift=-.20cm]2}] {} (b2); 	\draw (b) to node [label={[xshift=-.2cm,yshift=-.3cm]4}] {} (b3);
	\draw [line width=.08cm] (t1) to (a) to (t2); 	\draw [line width=.08cm] (b1) to (b) to (b2); 	\draw [dashed, line width=.08cm] (t1) to node[label={[xshift=-.2cm,yshift=-.25cm]x}] {} (b1); 	\draw [dashed, line width=.08cm] (t2) to node[label={[xshift=.2cm,yshift=-.25cm]y}] {} (b2); 
\end{tikzpicture}
\par\end{center}

\begin{center}
\caption{A cycle of length $x+y+4$ that passes through a $H'_{1}$.}

\par\end{center}%
\end{minipage}\hfill{}%
\begin{minipage}[t]{0.45\columnwidth}%
\begin{center}
\tikzstyle{node}=[circle, draw, fill=black!50,                         inner sep=0pt, minimum width=4pt]
\begin{tikzpicture}[thick,scale=.8] 	\node [node] (a) at (-1,.5) {}; 	\node [node] (b) at (0,1.2) {}; 	\node [node] (c) at (1,.5) {}; 	\node [node] (d) at (1,-.5) {}; 	\node [node] (e) at (0,-1.2) {}; 	\node [node] (f) at (-1,-.5) {};
	\draw (a) -- (b) -- (c) -- (d) -- (e) -- (f) -- (a); 	\coordinate (1) at  ($ (a)+ (-.75,.75) $); 	\coordinate (2) at  ($ (b)+ (0,.8) $); 	\coordinate (3) at  ($ (c)+ (.75,.75) $); 	\coordinate (4) at  ($ (d)+ (.75,-.75) $); 	\coordinate (5) at  ($ (e)+ (0,-.8) $); 	\coordinate (6) at  ($ (f) + (-.75,-.75) $); 	\draw (a) to node [label={[xshift=.1cm]1}] {} (1); 	\draw (b) to node [label={[xshift=-.2cm]2}] {} (2); 	\draw (c) to node [label={[xshift=-.2cm]3}] {} (3); 	\draw (d) to node [label={[xshift=.3cm,yshift=-.2cm]4}] {} (4); 	\draw (e) to node [label={[xshift=-.2cm,yshift=-.3cm]5}] {} (5); 	\draw (f) to node [label={[xshift=-.1cm]6}] {} (6);
	\draw [line width=.08cm] (2) to (b) to (a) to (1); 	\draw [line width=.08cm] (6) to (f) to (e) to (d) to (c) to (3); 	\draw [dashed, line width=.08cm] (1) to[out=180,in=180] node[label={[xshift=.2cm]x}] {} (6); 	\draw [dashed, line width=.08cm] (2) to[out=80,in=45] node[label=y] {} (3); \end{tikzpicture}

\par\end{center}

\begin{center}
\caption{The cycle from the previous figure, after expanding the gadget, is
now of length $x+y+8$.}

\par\end{center}%
\end{minipage}
\end{figure}
\begin{figure}[H]
\begin{minipage}[t]{0.45\columnwidth}%
\begin{center}
\tikzstyle{node}=[circle, draw, fill=black!50,                         inner sep=0pt, minimum width=4pt]
\begin{tikzpicture}[thick,scale=.8] 	\node [node] (a) at (0,1.5) {};
	\node [node] (b) at (0,-1.5) {};
	\node [node] (t1) at  ($ (a)+ (-.5,-1) $) {}; 	\node [node] (t2) at  ($ (a)+ (.5,-1) $) {}; 	\coordinate (t3) at  ($ (a)+ (0,.75) $); 	\node [node] (b1) at  ($ (b)+ (-.5,1) $) {}; 	\node [node] (b2) at  ($ (b)+ (.5,1) $) {}; 	\coordinate (b3) at  ($ (b) + (0,-.75) $); 	\draw (a) to node [label={[xshift=-.1cm]1}] {} (t1); 	\draw (a) to node [label={[xshift=.1cm]3}] {} (t2); 	\draw (a) to node [label={[xshift=-.2cm]5}] {} (t3); 	\draw (b) to node [label={[xshift=-.3cm,yshift=-.2cm]6}] {} (b1); 	\draw (b) to node [label={[xshift=.3cm,yshift=-.20cm]4}] {} (b2); 	\draw (b) to node [label={[xshift=-.2cm,yshift=-.3cm]2}] {} (b3);
	\draw [line width=.08cm] (t1) to (a) to (t2); 	\draw [line width=.08cm] (b1) to (b) to (b2); 	\draw [dashed, line width=.08cm] (t1) to node[label={[xshift=-.2cm,yshift=-.25cm]x}] {} (b1); 	\draw [dashed, line width=.08cm] (t2) to node[label={[xshift=.2cm,yshift=-.25cm]y}] {} (b2); 
\end{tikzpicture}
\par\end{center}

\begin{center}
\caption{A cycle of length $x+y+4$ that passes through a $H'_{1}$.}

\par\end{center}%
\end{minipage}\hfill{}%
\begin{minipage}[t]{0.45\columnwidth}%
\begin{center}
\tikzstyle{node}=[circle, draw, fill=black!50,                         inner sep=0pt, minimum width=4pt]
\begin{tikzpicture}[thick,scale=.8] 	\node [node] (a) at (-1,.5) {}; 	\node [node] (b) at (0,1.2) {}; 	\node [node] (c) at (1,.5) {}; 	\node [node] (d) at (1,-.5) {}; 	\node [node] (e) at (0,-1.2) {}; 	\node [node] (f) at (-1,-.5) {};
	\draw (a) -- (b) -- (c) -- (d) -- (e) -- (f) -- (a); 	\coordinate (1) at  ($ (a)+ (-.75,.75) $); 	\coordinate (2) at  ($ (b)+ (0,.8) $); 	\coordinate (3) at  ($ (c)+ (.75,.75) $); 	\coordinate (4) at  ($ (d)+ (.75,-.75) $); 	\coordinate (5) at  ($ (e)+ (0,-.8) $); 	\coordinate (6) at  ($ (f) + (-.75,-.75) $); 	\draw (a) to node [label={[xshift=.1cm]1}] {} (1); 	\draw (b) to node [label={[xshift=-.2cm]2}] {} (2); 	\draw (c) to node [label={[xshift=-.2cm]3}] {} (3); 	\draw (d) to node [label={[xshift=.3cm,yshift=-.2cm]4}] {} (4); 	\draw (e) to node [label={[xshift=-.2cm,yshift=-.3cm]5}] {} (5); 	\draw (f) to node [label={[xshift=-.1cm]6}] {} (6);
	\draw [line width=.08cm] (3) to (c) to (b) to (a) to (1); 	\draw [line width=.08cm] (6) to (f) to (e) to (d) to (4); 	\draw [dashed, line width=.08cm] (1) to[out=180,in=180] node[label={[xshift=.2cm]x}] {} (6); 	\draw [dashed, line width=.08cm] (3) to[out=0,in=0] node[label={[xshift=-.2cm]y}] {} (4); \end{tikzpicture}

\par\end{center}

\begin{center}
\caption{The cycle from the previous figure, after expanding the gadget, is
now of length $x+y+8$.}

\par\end{center}%
\end{minipage}
\end{figure}

\section{Appendix C: $H_{2}$s} \label{apdxh2}

\subsection{Gadget is covered by two cycles}

If a $H'_{2}$ is covered by a two cycles in $F_{i}$, then, there
are three ways we can select the edge in each of the super-vertices
to exclude from $F_{i}$. However, without loss of generality, we
can fix the edge of the first super-vertex to exclude from $F_{i}$.
Then, we only have to consider three cases. In each of these cases,
we start with two cycles of lengths $x+2$ and $y+2$ in $F_{i}$
and are returned a single cycle of length $x+y+10$ in $F_{i-1}$.
$x+2\geq6$ and $y+2\geq6$, so $x+y+10>8$, meaning that none of
these expansions can introduce an organic $6$-cycle into the $2$-factor.

\begin{figure}[H]
\begin{minipage}[t]{0.45\columnwidth}%
\begin{center}
\tikzstyle{node}=[circle, draw, fill=black!50,                         inner sep=0pt, minimum width=4pt]
\begin{tikzpicture}[thick,scale=1] 	\node [node] (a) at (-1,.5) {};
	\node [node] (b) at (1,.5) {};
	\coordinate (1) at  ($ (a)+ (-.75,.75) $); 	\coordinate (2) at  ($ (b)+ (-.75,.75) $); 	\coordinate (3) at  ($ (a)+ (.75,.75) $); 	\coordinate (4) at  ($ (b)+ (.75,.75) $); 	\coordinate (5) at  ($ (a)+ (0,-.8) $); 	\coordinate (6) at  ($ (b) + (0,-.8) $); 	\draw (a) to node [label={[xshift=.05cm]1}] {} (1); 	\draw (a) to node [label={[xshift=-.1cm]2}] {} (3); 	\draw (a) to node [label={[xshift=-.2cm,yshift=-.3cm]3}] {} (5);
	\draw (b) to node [label={[xshift=.05cm]4}] {} (2); 	\draw (b) to node [label={[xshift=-.1cm]5}] {} (4); 	\draw (b) to node [label={[xshift=-.2cm,yshift=-.3cm]6}] {} (6);
	\draw [line width=.08cm] (1) to (a) to (3); 	\draw [line width=.08cm] (2) to (b) to (4); 	\draw [dashed, line width=.08cm] (1) to[out=90,in=90] node[label={[xshift=0cm]x}] {} (3);
	\draw [dashed, line width=.08cm] (2) to[out=90,in=90] node[label={[xshift=0cm]y}] {} (4); 
\end{tikzpicture}
\par\end{center}

\begin{center}
\caption{Two cycles of lengths $x+2$ and $y+2$ pass through a $H'_{2}$.}

\par\end{center}%
\end{minipage}\hfill{}%
\begin{minipage}[t]{0.45\columnwidth}%
\begin{center}
\tikzstyle{node}=[circle, draw, fill=black!50,                         inner sep=0pt, minimum width=4pt]
\begin{tikzpicture}[thick,scale=.75] 	\node [node] (a) at (-1,.5) {}; 	\node [node] (b) at (0,1.2) {}; 	\node [node] (c) at (1,.5) {}; 	\node [node] (d) at (1,-.5) {}; 	\node [node] (e) at (0,-1.2) {}; 	\node [node] (f) at (-1,-.5) {}; 	\node [node] (g) at (0,.5) {}; 	\node [node] (h) at (0,-.5) {};
	\draw (a) -- (b) -- (c) -- (d) -- (e) -- (f) -- (a); 	\draw (b) -- (g) -- (h) -- (e); 	\coordinate (1) at  ($ (a)+ (-.75,.75) $); 	\coordinate (2) at  ($ (g)+ (-1.2, 1.2) $); 	\coordinate (3) at  ($ (c)+ (.75,.75) $); 	\coordinate (6) at  ($ (d)+ (.75,-.75) $); 	\coordinate (5) at  ($ (h)+ (1.2,-1.2) $); 	\coordinate (4) at  ($ (f) + (-.75,-.75) $); 	\draw (a) to node [label={[xshift=.1cm]1}] {} (1); 	\draw (g) to node [label={[xshift=.2cm]2}] {} (2); 	\draw (c) to node [label={[xshift=-.2cm]3}] {} (3); 	\draw (d) to node [label={[xshift=.3cm,yshift=-.2cm]6}] {} (6); 	\draw (h) to node [label={[xshift=.15cm, yshift=-.1cm]5}] {} (5); 	\draw (f) to node [label={[xshift=-.1cm]4}] {} (4);
	\draw [line width=.08cm] (1) to (a) to (f) to (4); 	\draw [line width=.08cm] (2) to (g) to (b) to (c) to (d) to (e) to (h) to (5);
	\draw [dashed, line width=.08cm] (1) to[out=180,in=70,distance=1.5cm] node[label=x] {} (2); 	\draw [dashed, line width=.08cm] (4) to[out=220,in=-45,distance=.8cm] node[label={[yshift=-.5cm]y}] {} (5); \end{tikzpicture}
\par\end{center}

\begin{center}
\caption{The cycles from the previous figure, after expanding the gadget, are
now a single cycle of length $x+y+10$.}

\par\end{center}%
\end{minipage}
\end{figure}
\begin{figure}[H]
\begin{minipage}[t]{0.45\columnwidth}%
\begin{center}
\tikzstyle{node}=[circle, draw, fill=black!50,                         inner sep=0pt, minimum width=4pt]
\begin{tikzpicture}[thick,scale=1] 	\node [node] (a) at (-1,.5) {};
	\node [node] (b) at (1,.5) {};
	\coordinate (1) at  ($ (a)+ (-.75,.75) $); 	\coordinate (2) at  ($ (b)+ (-.75,.75) $); 	\coordinate (3) at  ($ (a)+ (.75,.75) $); 	\coordinate (4) at  ($ (b)+ (.75,.75) $); 	\coordinate (5) at  ($ (a)+ (0,-.8) $); 	\coordinate (6) at  ($ (b) + (0,-.8) $); 	\draw (a) to node [label={[xshift=.05cm]1}] {} (1); 	\draw (b) to node [label={[xshift=.05cm]4}] {} (2); 	\draw (a) to node [label={[xshift=-.1cm]2}] {} (3); 	\draw (b) to node [label={[xshift=-.1cm]5}] {} (4); 	\draw (a) to node [label={[xshift=-.2cm,yshift=-.3cm]3}] {} (5); 	\draw (b) to node [label={[xshift=-.2cm,yshift=-.3cm]6}] {} (6);
	\draw [line width=.08cm] (1) to (a) to (3); 	\draw [line width=.08cm] (2) to (b) to (6); 	\draw [dashed, line width=.08cm] (1) to[out=90,in=90] node[label={[xshift=0cm]x}] {} (3);
	\draw [dashed, line width=.08cm] (2) to[out=200,in=240] node[label={[xshift=-.3cm]y}] {} (6); 
\end{tikzpicture}
\par\end{center}

\begin{center}
\caption{Two cycles of lengths $x+2$ and $y+2$ pass through a $H'_{2}$.}

\par\end{center}%
\end{minipage}\hfill{}%
\begin{minipage}[t]{0.45\columnwidth}%
\begin{center}
\tikzstyle{node}=[circle, draw, fill=black!50,                         inner sep=0pt, minimum width=4pt]
\begin{tikzpicture}[thick,scale=.75] 	\node [node] (a) at (-1,.5) {}; 	\node [node] (b) at (0,1.2) {}; 	\node [node] (c) at (1,.5) {}; 	\node [node] (d) at (1,-.5) {}; 	\node [node] (e) at (0,-1.2) {}; 	\node [node] (f) at (-1,-.5) {}; 	\node [node] (g) at (0,.5) {}; 	\node [node] (h) at (0,-.5) {};
	\draw (a) -- (b) -- (c) -- (d) -- (e) -- (f) -- (a); 	\draw (b) -- (g) -- (h) -- (e); 	\coordinate (1) at  ($ (a)+ (-.75,.75) $); 	\coordinate (2) at  ($ (g)+ (-1.2, 1.2) $); 	\coordinate (3) at  ($ (c)+ (.75,.75) $); 	\coordinate (6) at  ($ (d)+ (.75,-.75) $); 	\coordinate (5) at  ($ (h)+ (1.2,-1.2) $); 	\coordinate (4) at  ($ (f) + (-.75,-.75) $); 	\draw (a) to node [label={[xshift=.1cm]1}] {} (1); 	\draw (g) to node [label={[xshift=.2cm]2}] {} (2); 	\draw (c) to node [label={[xshift=-.2cm]3}] {} (3); 	\draw (d) to node [label={[xshift=.3cm,yshift=-.2cm]6}] {} (6); 	\draw (h) to node [label={[xshift=.15cm, yshift=-.1cm]5}] {} (5); 	\draw (f) to node [label={[xshift=-.1cm]4}] {} (4);
	\draw [line width=.08cm] (2) to (g) to (h) to (e) to (f) to (4); 	\draw [line width=.08cm] (1) to (a) to (b) to (c) to (d) to (6);
	\draw [dashed, line width=.08cm] (1) to[out=180,in=70,distance=1.5cm] node[label=x] {} (2); 	\draw [dashed, line width=.08cm] (4) to[out=220,in=-45,distance=1.3cm] node[label={[yshift=-.5cm]y}] {} (6); \end{tikzpicture}
\par\end{center}

\begin{center}
\caption{The cycles from the previous figure, after expanding the gadget, are
now a single cycle of length $x+y+10$.}

\par\end{center}%
\end{minipage}
\end{figure}
\begin{figure}[H]
\begin{minipage}[t]{0.45\columnwidth}%
\begin{center}
\tikzstyle{node}=[circle, draw, fill=black!50,                         inner sep=0pt, minimum width=4pt]
\begin{tikzpicture}[thick,scale=1] 	\node [node] (a) at (-1,.5) {};
	\node [node] (b) at (1,.5) {};
	\coordinate (1) at  ($ (a)+ (-.75,.75) $); 	\coordinate (2) at  ($ (b)+ (-.75,.75) $); 	\coordinate (3) at  ($ (a)+ (.75,.75) $); 	\coordinate (4) at  ($ (b)+ (.75,.75) $); 	\coordinate (5) at  ($ (a)+ (0,-.8) $); 	\coordinate (6) at  ($ (b) + (0,-.8) $); 	\draw (a) to node [label={[xshift=.05cm]1}] {} (1); 	\draw (b) to node [label={[xshift=.05cm]4}] {} (2); 	\draw (a) to node [label={[xshift=-.1cm]2}] {} (3); 	\draw (b) to node [label={[xshift=-.1cm]5}] {} (4); 	\draw (a) to node [label={[xshift=-.2cm,yshift=-.3cm]3}] {} (5); 	\draw (b) to node [label={[xshift=-.2cm,yshift=-.3cm]6}] {} (6);
	\draw [line width=.08cm] (1) to (a) to (3); 	\draw [line width=.08cm] (4) to (b) to (6); 	\draw [dashed, line width=.08cm] (1) to[out=90,in=90] node[label={[xshift=0cm]x}] {} (3);
	\draw [dashed, line width=.08cm] (4) to[out=20,in=0] node[label={[xshift=.3cm]y}] {} (6); 
\end{tikzpicture}
\par\end{center}

\begin{center}
\caption{Two cycles of lengths $x+2$ and $y+2$ pass through a $H'_{2}$.}

\par\end{center}%
\end{minipage}\hfill{}%
\begin{minipage}[t]{0.45\columnwidth}%
\begin{center}
\tikzstyle{node}=[circle, draw, fill=black!50,                         inner sep=0pt, minimum width=4pt]
\begin{tikzpicture}[thick,scale=.75] 	\node [node] (a) at (-1,.5) {}; 	\node [node] (b) at (0,1.2) {}; 	\node [node] (c) at (1,.5) {}; 	\node [node] (d) at (1,-.5) {}; 	\node [node] (e) at (0,-1.2) {}; 	\node [node] (f) at (-1,-.5) {}; 	\node [node] (g) at (0,.5) {}; 	\node [node] (h) at (0,-.5) {};
	\draw (a) -- (b) -- (c) -- (d) -- (e) -- (f) -- (a); 	\draw (b) -- (g) -- (h) -- (e); 	\coordinate (1) at  ($ (a)+ (-.75,.75) $); 	\coordinate (2) at  ($ (g)+ (-1.2, 1.2) $); 	\coordinate (3) at  ($ (c)+ (.75,.75) $); 	\coordinate (6) at  ($ (d)+ (.75,-.75) $); 	\coordinate (5) at  ($ (h)+ (1.2,-1.2) $); 	\coordinate (4) at  ($ (f) + (-.75,-.75) $); 	\draw (a) to node [label={[xshift=.1cm]1}] {} (1); 	\draw (g) to node [label={[xshift=.2cm]2}] {} (2); 	\draw (c) to node [label={[xshift=-.2cm]3}] {} (3); 	\draw (d) to node [label={[xshift=.3cm,yshift=-.2cm]6}] {} (6); 	\draw (h) to node [label={[xshift=.15cm, yshift=-.1cm]5}] {} (5); 	\draw (f) to node [label={[xshift=-.1cm]4}] {} (4);
	\draw [line width=.08cm] (2) to (g) to (b) to (c) to (d) to (6); 	\draw [line width=.08cm] (1) to (a) to (f) to (e) to (h) to (5);
	\draw [dashed, line width=.08cm] (1) to[out=180,in=70,distance=1.5cm] node[label=x] {} (2); 	\draw [dashed, line width=.08cm] (5) to[out=240,in=20,distance=1.5cm] node[label={[yshift=-.5cm]y}] {} (6); \end{tikzpicture}
\par\end{center}

\begin{center}
\caption{The cycles from the previous figure, after expanding the gadget, are
now a single cycle of length $x+y+10$.}

\par\end{center}%
\end{minipage}
\end{figure}

\subsection{Gadget is covered by one cycle}

If a $H'_{2}$ is covered by a single cycle in $F_{i}$, then, there
are three ways we can select the edge in each of the super-vertices
to exclude from $F_{i}$. However, without loss of generality, we
can fix the edge of the first super-vertex to exclude from $F_{i}$.
Then, in each of these three configurations, we examine the two orientations
in which the cycle can pass through the two super-vertices. In all
$6$ cases, we start with a cycle of lengths $x+y+4$ in $F_{i}$
and are returned a single cycle of length $x+y+10$, two cycles of
lengths $x+5$ and $y+5$, or two cycles of lengths $x+3$ and $y+7$
in $F_{i-1}$. $x+y+4\geq6$, so $x+y+10>8$, meaning the first case
cannot introduce an organic $6$-cycle into the $2$-factor. In the
second case, neither the $x+5$ and $y+5$ cycles can be organic $6$-cycles,
otherwise the $H_{2}$ the gadget replaced would have been part of
an organic $H_{3}$, which would have been contracted instead of the
$H_{2}$. In the third case, the $x+3$ cycle can be an organic $6$-cycle,
but this is the expansion examined in detail in Sections \ref{expandh2} and \ref{analyzingworstcase}.

\begin{figure}[H]
\begin{minipage}[t]{0.45\columnwidth}%
\begin{center}
\tikzstyle{node}=[circle, draw, fill=black!50,                         inner sep=0pt, minimum width=4pt]
\begin{tikzpicture}[thick,scale=1] 	\node [node] (a) at (0,1.5) {};
	\node [node] (b) at (0,-1.5) {};
	\node [node] (t1) at  ($ (a)+ (-.5,-1) $) {}; 	\node [node] (t2) at  ($ (a)+ (.5,-1) $) {}; 	\coordinate (t3) at  ($ (a)+ (0,.75) $); 	\node [node] (b1) at  ($ (b)+ (-.5,1) $) {}; 	\node [node] (b2) at  ($ (b)+ (.5,1) $) {}; 	\coordinate (b3) at  ($ (b) + (0,-.75) $); 	\draw (a) to node [label={[xshift=-.1cm]1}] {} (t1); 	\draw (a) to node [label={[xshift=.1cm]2}] {} (t2); 	\draw (a) to node [label={[xshift=-.2cm]3}] {} (t3); 	\draw (b) to node [label={[xshift=-.3cm,yshift=-.2cm]4}] {} (b1); 	\draw (b) to node [label={[xshift=.3cm,yshift=-.20cm]5}] {} (b2); 	\draw (b) to node [label={[xshift=-.2cm,yshift=-.3cm]6}] {} (b3);
	\draw [line width=.08cm] (t1) to (a) to (t2); 	\draw [line width=.08cm] (b1) to (b) to (b2); 	\draw [dashed, line width=.08cm] (t1) to node[label={[xshift=-.2cm,yshift=-.25cm]x}] {} (b1); 	\draw [dashed, line width=.08cm] (t2) to node[label={[xshift=.2cm,yshift=-.25cm]y}] {} (b2); 
\end{tikzpicture}
\par\end{center}

\begin{center}
\caption{A cycle of length $x+y+4$ passes through a $H'_{2}$.}

\par\end{center}%
\end{minipage}\hfill{}%
\begin{minipage}[t]{0.45\columnwidth}%
\begin{center}
\tikzstyle{node}=[circle, draw, fill=black!50,                         inner sep=0pt, minimum width=4pt]
\begin{tikzpicture}[thick,scale=.75] 	\node [node] (a) at (-1,.5) {}; 	\node [node] (b) at (0,1.2) {}; 	\node [node] (c) at (1,.5) {}; 	\node [node] (d) at (1,-.5) {}; 	\node [node] (e) at (0,-1.2) {}; 	\node [node] (f) at (-1,-.5) {}; 	\node [node] (g) at (0,.5) {}; 	\node [node] (h) at (0,-.5) {};
	\draw (a) -- (b) -- (c) -- (d) -- (e) -- (f) -- (a); 	\draw (b) -- (g) -- (h) -- (e); 	\coordinate (1) at  ($ (a)+ (-.75,.75) $); 	\coordinate (2) at  ($ (g)+ (-1.2, 1.2) $); 	\coordinate (3) at  ($ (c)+ (.75,.75) $); 	\coordinate (6) at  ($ (d)+ (.75,-.75) $); 	\coordinate (5) at  ($ (h)+ (1.2,-1.2) $); 	\coordinate (4) at  ($ (f) + (-.75,-.75) $); 	\draw (a) to node [label={[xshift=.1cm]1}] {} (1); 	\draw (g) to node [label={[xshift=.2cm]2}] {} (2); 	\draw (c) to node [label={[xshift=-.2cm]3}] {} (3); 	\draw (d) to node [label={[xshift=.3cm,yshift=-.2cm]6}] {} (6); 	\draw (h) to node [label={[xshift=.15cm, yshift=-.1cm]5}] {} (5); 	\draw (f) to node [label={[xshift=-.1cm]4}] {} (4);
	\draw [line width=.08cm] (1) to (a) to (f) to (4); 	\draw [line width=.08cm] (2) to (g) to (b) to (c) to (d) to (e) to (h) to (5);
	\draw [dashed, line width=.08cm] (1) to[out=160,in=200,distance=.8cm] node[label={[xshift=-.2cm]x}] {} (4); 	\draw [dashed, line width=.08cm] (2) to[out=75,in=0,distance=4cm] node[label={[xshift=.1cm,yshift=-.0cm]y}] {} (5); \end{tikzpicture}
\par\end{center}

\begin{center}
\caption{The cycle from the previous figure, after expanding the gadget, is
now two cycles, of lengths $x+3$ and $y+7$, respectively. This expansion
can produce an organic $6$-cycle if $x=3$ and the cycle is organic.
The impact of these $6$-cycles on the algorithm's result is analyzed
in Sections \ref{expandh2} and \ref{analyzingworstcase}.}

\par\end{center}%
\end{minipage}
\end{figure}
\begin{figure}[H]
\begin{minipage}[t]{0.45\columnwidth}%
\begin{center}
\tikzstyle{node}=[circle, draw, fill=black!50,                         inner sep=0pt, minimum width=4pt]
\begin{tikzpicture}[thick,scale=1] 	\node [node] (a) at (0,1.5) {};
	\node [node] (b) at (0,-1.5) {};
	\node [node] (t1) at  ($ (a)+ (-.5,-1) $) {}; 	\node [node] (t2) at  ($ (a)+ (.5,-1) $) {}; 	\coordinate (t3) at  ($ (a)+ (0,.75) $); 	\node [node] (b1) at  ($ (b)+ (-.5,1) $) {}; 	\node [node] (b2) at  ($ (b)+ (.5,1) $) {}; 	\coordinate (b3) at  ($ (b) + (0,-.75) $); 	\draw (a) to node [label={[xshift=-.1cm]1}] {} (t1); 	\draw (a) to node [label={[xshift=.1cm]2}] {} (t2); 	\draw (a) to node [label={[xshift=-.2cm]3}] {} (t3); 	\draw (b) to node [label={[xshift=-.3cm,yshift=-.2cm]4}] {} (b1); 	\draw (b) to node [label={[xshift=.3cm,yshift=-.20cm]6}] {} (b2); 	\draw (b) to node [label={[xshift=-.2cm,yshift=-.3cm]5}] {} (b3);
	\draw [line width=.08cm] (t1) to (a) to (t2); 	\draw [line width=.08cm] (b1) to (b) to (b2); 	\draw [dashed, line width=.08cm] (t1) to node[label={[xshift=-.2cm,yshift=-.25cm]x}] {} (b1); 	\draw [dashed, line width=.08cm] (t2) to node[label={[xshift=.2cm,yshift=-.25cm]y}] {} (b2); 
\end{tikzpicture}
\par\end{center}

\begin{center}
\caption{A cycle of length $x+y+4$ passes through a $H'_{2}$.}

\par\end{center}%
\end{minipage}\hfill{}%
\begin{minipage}[t]{0.45\columnwidth}%
\begin{center}
\tikzstyle{node}=[circle, draw, fill=black!50,                         inner sep=0pt, minimum width=4pt]
\begin{tikzpicture}[thick,scale=.75] 	\node [node] (a) at (-1,.5) {}; 	\node [node] (b) at (0,1.2) {}; 	\node [node] (c) at (1,.5) {}; 	\node [node] (d) at (1,-.5) {}; 	\node [node] (e) at (0,-1.2) {}; 	\node [node] (f) at (-1,-.5) {}; 	\node [node] (g) at (0,.5) {}; 	\node [node] (h) at (0,-.5) {};
	\draw (a) -- (b) -- (c) -- (d) -- (e) -- (f) -- (a); 	\draw (b) -- (g) -- (h) -- (e); 	\coordinate (1) at  ($ (a)+ (-.75,.75) $); 	\coordinate (2) at  ($ (g)+ (-1.2, 1.2) $); 	\coordinate (3) at  ($ (c)+ (.75,.75) $); 	\coordinate (6) at  ($ (d)+ (.75,-.75) $); 	\coordinate (5) at  ($ (h)+ (1.2,-1.2) $); 	\coordinate (4) at  ($ (f) + (-.75,-.75) $); 	\draw (a) to node [label={[xshift=.1cm]1}] {} (1); 	\draw (g) to node [label={[xshift=.2cm]2}] {} (2); 	\draw (c) to node [label={[xshift=-.2cm]3}] {} (3); 	\draw (d) to node [label={[xshift=.3cm,yshift=-.2cm]6}] {} (6); 	\draw (h) to node [label={[xshift=.15cm, yshift=-.1cm]5}] {} (5); 	\draw (f) to node [label={[xshift=-.1cm]4}] {} (4);
	\draw [line width=.08cm] (2) to (g) to (h) to (e) to (f) to (4); 	\draw [line width=.08cm] (1) to (a) to (b) to (c) to (d) to (6);
	\draw [dashed, line width=.08cm] (1) to[out=160,in=200,distance=.8cm] node[label={[xshift=-.2cm]x}] {} (4); 	\draw [dashed, line width=.08cm] (2) to[out=45,in=30,distance=2.7cm] node[label={[yshift=-.0cm]y}] {} (6); \end{tikzpicture}
\par\end{center}

\begin{center}
\caption{The cycle from the previous figure, after expanding the gadget, is
now a cycle of length $x+y+10$.}

\par\end{center}%
\end{minipage}
\end{figure}
\begin{figure}[H]
\begin{minipage}[t]{0.45\columnwidth}%
\begin{center}
\tikzstyle{node}=[circle, draw, fill=black!50,                         inner sep=0pt, minimum width=4pt]
\begin{tikzpicture}[thick,scale=1] 	\node [node] (a) at (0,1.5) {};
	\node [node] (b) at (0,-1.5) {};
	\node [node] (t1) at  ($ (a)+ (-.5,-1) $) {}; 	\node [node] (t2) at  ($ (a)+ (.5,-1) $) {}; 	\coordinate (t3) at  ($ (a)+ (0,.75) $); 	\node [node] (b1) at  ($ (b)+ (-.5,1) $) {}; 	\node [node] (b2) at  ($ (b)+ (.5,1) $) {}; 	\coordinate (b3) at  ($ (b) + (0,-.75) $); 	\draw (a) to node [label={[xshift=-.1cm]1}] {} (t1); 	\draw (a) to node [label={[xshift=.1cm]2}] {} (t2); 	\draw (a) to node [label={[xshift=-.2cm]3}] {} (t3); 	\draw (b) to node [label={[xshift=-.3cm,yshift=-.2cm]5}] {} (b1); 	\draw (b) to node [label={[xshift=.3cm,yshift=-.20cm]4}] {} (b2); 	\draw (b) to node [label={[xshift=-.2cm,yshift=-.3cm]6}] {} (b3);
	\draw [line width=.08cm] (t1) to (a) to (t2); 	\draw [line width=.08cm] (b1) to (b) to (b2); 	\draw [dashed, line width=.08cm] (t1) to node[label={[xshift=-.2cm,yshift=-.25cm]x}] {} (b1); 	\draw [dashed, line width=.08cm] (t2) to node[label={[xshift=.2cm,yshift=-.25cm]y}] {} (b2); 
\end{tikzpicture}
\par\end{center}

\begin{center}
\caption{A cycle of length $x+y+4$ passes through a $H'_{2}$.}

\par\end{center}%
\end{minipage}\hfill{}%
\begin{minipage}[t]{0.45\columnwidth}%
\begin{center}
 \tikzstyle{node}=[circle, draw, fill=black!50,                         inner sep=0pt, minimum width=4pt]
\begin{tikzpicture}[thick,scale=.75] 	\node [node] (a) at (-1,.5) {}; 	\node [node] (b) at (0,1.2) {}; 	\node [node] (c) at (1,.5) {}; 	\node [node] (d) at (1,-.5) {}; 	\node [node] (e) at (0,-1.2) {}; 	\node [node] (f) at (-1,-.5) {}; 	\node [node] (g) at (0,.5) {}; 	\node [node] (h) at (0,-.5) {};
	\draw (a) -- (b) -- (c) -- (d) -- (e) -- (f) -- (a); 	\draw (b) -- (g) -- (h) -- (e); 	\coordinate (1) at  ($ (a)+ (-.75,.75) $); 	\coordinate (2) at  ($ (g)+ (-1.2, 1.2) $); 	\coordinate (3) at  ($ (c)+ (.75,.75) $); 	\coordinate (6) at  ($ (d)+ (.75,-.75) $); 	\coordinate (5) at  ($ (h)+ (1.2,-1.2) $); 	\coordinate (4) at  ($ (f) + (-.75,-.75) $); 	\draw (a) to node [label={[xshift=.1cm]1}] {} (1); 	\draw (g) to node [label={[xshift=.2cm]2}] {} (2); 	\draw (c) to node [label={[xshift=-.2cm]3}] {} (3); 	\draw (d) to node [label={[xshift=.3cm,yshift=-.2cm]6}] {} (6); 	\draw (h) to node [label={[xshift=.15cm, yshift=-.1cm]5}] {} (5); 	\draw (f) to node [label={[xshift=-.1cm]4}] {} (4);
	\draw [line width=.08cm] (1) to (a) to (f) to (4); 	\draw [line width=.08cm] (2) to (g) to (b) to (c) to (d) to (e) to (h) to (5);
	\draw [dashed, line width=.08cm] (1) to[out=80,in=20,distance=4cm] node[label=x] {} (5); 	\draw [dashed, line width=.08cm] (4) to[out=150,in=200,distance=2cm] node[label={[xshift=-.3cm,yshift=-.5cm]y}] {} (2); \end{tikzpicture}
\par\end{center}

\begin{center}
\caption{The cycle from the previous figure, after expanding the gadget, is
now a cycle of length $x+y+10$.}

\par\end{center}%
\end{minipage}
\end{figure}
\begin{figure}[H]
\begin{minipage}[t]{0.45\columnwidth}%
\begin{center}
\tikzstyle{node}=[circle, draw, fill=black!50,                         inner sep=0pt, minimum width=4pt]
\begin{tikzpicture}[thick,scale=1] 	\node [node] (a) at (0,1.5) {};
	\node [node] (b) at (0,-1.5) {};
	\node [node] (t1) at  ($ (a)+ (-.5,-1) $) {}; 	\node [node] (t2) at  ($ (a)+ (.5,-1) $) {}; 	\coordinate (t3) at  ($ (a)+ (0,.75) $); 	\node [node] (b1) at  ($ (b)+ (-.5,1) $) {}; 	\node [node] (b2) at  ($ (b)+ (.5,1) $) {}; 	\coordinate (b3) at  ($ (b) + (0,-.75) $); 	\draw (a) to node [label={[xshift=-.1cm]1}] {} (t1); 	\draw (a) to node [label={[xshift=.1cm]2}] {} (t2); 	\draw (a) to node [label={[xshift=-.2cm]3}] {} (t3); 	\draw (b) to node [label={[xshift=-.3cm,yshift=-.2cm]5}] {} (b1); 	\draw (b) to node [label={[xshift=.3cm,yshift=-.20cm]6}] {} (b2); 	\draw (b) to node [label={[xshift=-.2cm,yshift=-.3cm]4}] {} (b3);
	\draw [line width=.08cm] (t1) to (a) to (t2); 	\draw [line width=.08cm] (b1) to (b) to (b2); 	\draw [dashed, line width=.08cm] (t1) to node[label={[xshift=-.2cm,yshift=-.25cm]x}] {} (b1); 	\draw [dashed, line width=.08cm] (t2) to node[label={[xshift=.2cm,yshift=-.25cm]y}] {} (b2); 
\end{tikzpicture}
\par\end{center}

\begin{center}
\caption{A cycle of length $x+y+4$ passes through a $H'_{2}$.}

\par\end{center}%
\end{minipage}\hfill{}%
\begin{minipage}[t]{0.45\columnwidth}%
\begin{center}
\tikzstyle{node}=[circle, draw, fill=black!50,                         inner sep=0pt, minimum width=4pt]
\begin{tikzpicture}[thick,scale=.75] 	\node [node] (a) at (-1,.5) {}; 	\node [node] (b) at (0,1.2) {}; 	\node [node] (c) at (1,.5) {}; 	\node [node] (d) at (1,-.5) {}; 	\node [node] (e) at (0,-1.2) {}; 	\node [node] (f) at (-1,-.5) {}; 	\node [node] (g) at (0,.5) {}; 	\node [node] (h) at (0,-.5) {};
	\draw (a) -- (b) -- (c) -- (d) -- (e) -- (f) -- (a); 	\draw (b) -- (g) -- (h) -- (e); 	\coordinate (1) at  ($ (a)+ (-.75,.75) $); 	\coordinate (2) at  ($ (g)+ (-1.2, 1.2) $); 	\coordinate (3) at  ($ (c)+ (.75,.75) $); 	\coordinate (6) at  ($ (d)+ (.75,-.75) $); 	\coordinate (5) at  ($ (h)+ (1.2,-1.2) $); 	\coordinate (4) at  ($ (f) + (-.75,-.75) $); 	\draw (a) to node [label={[xshift=.1cm]1}] {} (1); 	\draw (g) to node [label={[xshift=.2cm]2}] {} (2); 	\draw (c) to node [label={[xshift=-.2cm]3}] {} (3); 	\draw (d) to node [label={[xshift=.3cm,yshift=-.2cm]6}] {} (6); 	\draw (h) to node [label={[xshift=.15cm, yshift=-.1cm]5}] {} (5); 	\draw (f) to node [label={[xshift=-.1cm]4}] {} (4);
	\draw [line width=.08cm] (2) to (g) to (b) to (c) to (d) to (6); 	\draw [line width=.08cm] (1) to (a) to (f) to (e) to (h) to (5);
	\draw [dashed, line width=.08cm] (1) to[out=180,in=-110,distance=3cm] node[label={[yshift=-.5cm]x}] {} (5); 	\draw [dashed, line width=.08cm] (2) to[out=80,in=20,distance=3cm] node[label={[yshift=-.0cm]y}] {} (6); \end{tikzpicture}
\par\end{center}

\begin{center}
\caption{The cycle from the previous figure, after expanding the gadget, is
now two cycles, of lengths $x+5$ and $y+5$, respectively. Note that
$x$ and $y$ cannot be $1$, otherwise this $H_{2}$ would be part
of a $H_{3}$. This is not possible, otherwise the $H_{3}$ would
have been contracted instead of this $H_{2}$.}

\par\end{center}%
\end{minipage}
\end{figure}
\begin{figure}[H]
\begin{minipage}[t]{0.45\columnwidth}%
\begin{center}
\tikzstyle{node}=[circle, draw, fill=black!50,                         inner sep=0pt, minimum width=4pt]
\begin{tikzpicture}[thick,scale=1] 	\node [node] (a) at (0,1.5) {};
	\node [node] (b) at (0,-1.5) {};
	\node [node] (t1) at  ($ (a)+ (-.5,-1) $) {}; 	\node [node] (t2) at  ($ (a)+ (.5,-1) $) {}; 	\coordinate (t3) at  ($ (a)+ (0,.75) $); 	\node [node] (b1) at  ($ (b)+ (-.5,1) $) {}; 	\node [node] (b2) at  ($ (b)+ (.5,1) $) {}; 	\coordinate (b3) at  ($ (b) + (0,-.75) $); 	\draw (a) to node [label={[xshift=-.1cm]1}] {} (t1); 	\draw (a) to node [label={[xshift=.1cm]2}] {} (t2); 	\draw (a) to node [label={[xshift=-.2cm]3}] {} (t3); 	\draw (b) to node [label={[xshift=-.3cm,yshift=-.2cm]6}] {} (b1); 	\draw (b) to node [label={[xshift=.3cm,yshift=-.20cm]4}] {} (b2); 	\draw (b) to node [label={[xshift=-.2cm,yshift=-.3cm]5}] {} (b3);
	\draw [line width=.08cm] (t1) to (a) to (t2); 	\draw [line width=.08cm] (b1) to (b) to (b2); 	\draw [dashed, line width=.08cm] (t1) to node[label={[xshift=-.2cm,yshift=-.25cm]x}] {} (b1); 	\draw [dashed, line width=.08cm] (t2) to node[label={[xshift=.2cm,yshift=-.25cm]y}] {} (b2); 
\end{tikzpicture}
\par\end{center}

\begin{center}
\caption{A cycle of length $x+y+4$ passes through a $H'_{2}$.}

\par\end{center}%
\end{minipage}\hfill{}%
\begin{minipage}[t]{0.45\columnwidth}%
\begin{center}
\tikzstyle{node}=[circle, draw, fill=black!50,                         inner sep=0pt, minimum width=4pt]
\begin{tikzpicture}[thick,scale=.75] 	\node [node] (a) at (-1,.5) {}; 	\node [node] (b) at (0,1.2) {}; 	\node [node] (c) at (1,.5) {}; 	\node [node] (d) at (1,-.5) {}; 	\node [node] (e) at (0,-1.2) {}; 	\node [node] (f) at (-1,-.5) {}; 	\node [node] (g) at (0,.5) {}; 	\node [node] (h) at (0,-.5) {};
	\draw (a) -- (b) -- (c) -- (d) -- (e) -- (f) -- (a); 	\draw (b) -- (g) -- (h) -- (e); 	\coordinate (1) at  ($ (a)+ (-.75,.75) $); 	\coordinate (2) at  ($ (g)+ (-1.2, 1.2) $); 	\coordinate (3) at  ($ (c)+ (.75,.75) $); 	\coordinate (6) at  ($ (d)+ (.75,-.75) $); 	\coordinate (5) at  ($ (h)+ (1.2,-1.2) $); 	\coordinate (4) at  ($ (f) + (-.75,-.75) $); 	\draw (a) to node [label={[xshift=.1cm]1}] {} (1); 	\draw (g) to node [label={[xshift=.2cm]2}] {} (2); 	\draw (c) to node [label={[xshift=-.2cm]3}] {} (3); 	\draw (d) to node [label={[xshift=.3cm,yshift=-.2cm]6}] {} (6); 	\draw (h) to node [label={[xshift=.15cm, yshift=-.1cm]5}] {} (5); 	\draw (f) to node [label={[xshift=-.1cm]4}] {} (4);
	\draw [line width=.08cm] (2) to (g) to (h) to (e) to (f) to (4); 	\draw [line width=.08cm] (1) to (a) to (b) to (c) to (d) to (6);
	\draw [dashed, line width=.08cm] (1) to[out=90,in=30,distance=3.5cm] node[label=x] {} (6); 	\draw [dashed, line width=.08cm] (4) to[out=200,in=160,distance=1.5cm] node[label={[xshift=-.2cm,yshift=-.0cm]y}] {} (2); \end{tikzpicture}
\par\end{center}

\begin{center}
\caption{The cycle from the previous figure, after expanding the gadget, is
now two cycles, of lengths $x+5$ and $y+5$, respectively. Note that
$x$ and $y$ cannot be $1$, otherwise this $H_{2}$ would be part
of a $H_{3}$. This is not possible, otherwise the $H_{3}$ would
have been contracted instead of this $H_{2}$.}

\par\end{center}%
\end{minipage}
\end{figure}
\begin{figure}[H]
\begin{minipage}[t]{0.45\columnwidth}%
\begin{center}
\tikzstyle{node}=[circle, draw, fill=black!50,                         inner sep=0pt, minimum width=4pt]
\begin{tikzpicture}[thick,scale=1] 	\node [node] (a) at (0,1.5) {};
	\node [node] (b) at (0,-1.5) {};
	\node [node] (t1) at  ($ (a)+ (-.5,-1) $) {}; 	\node [node] (t2) at  ($ (a)+ (.5,-1) $) {}; 	\coordinate (t3) at  ($ (a)+ (0,.75) $); 	\node [node] (b1) at  ($ (b)+ (-.5,1) $) {}; 	\node [node] (b2) at  ($ (b)+ (.5,1) $) {}; 	\coordinate (b3) at  ($ (b) + (0,-.75) $); 	\draw (a) to node [label={[xshift=-.1cm]1}] {} (t1); 	\draw (a) to node [label={[xshift=.1cm]2}] {} (t2); 	\draw (a) to node [label={[xshift=-.2cm]3}] {} (t3); 	\draw (b) to node [label={[xshift=-.3cm,yshift=-.2cm]6}] {} (b1); 	\draw (b) to node [label={[xshift=.3cm,yshift=-.20cm]5}] {} (b2); 	\draw (b) to node [label={[xshift=-.2cm,yshift=-.3cm]4}] {} (b3);
	\draw [line width=.08cm] (t1) to (a) to (t2); 	\draw [line width=.08cm] (b1) to (b) to (b2); 	\draw [dashed, line width=.08cm] (t1) to node[label={[xshift=-.2cm,yshift=-.25cm]x}] {} (b1); 	\draw [dashed, line width=.08cm] (t2) to node[label={[xshift=.2cm,yshift=-.25cm]y}] {} (b2); 
\end{tikzpicture}
\par\end{center}

\begin{center}
\caption{A cycle of length $x+y+4$ passes through a $H'_{2}$.}

\par\end{center}%
\end{minipage}\hfill{}%
\begin{minipage}[t]{0.45\columnwidth}%
\begin{center}
\tikzstyle{node}=[circle, draw, fill=black!50,                         inner sep=0pt, minimum width=4pt]
\begin{tikzpicture}[thick,scale=.75]
\clip (-3.5,-3) rectangle (3,3);
\node [node] (a) at (-1,.5) {}; 	\node [node] (b) at (0,1.2) {}; 	\node [node] (c) at (1,.5) {}; 	\node [node] (d) at (1,-.5) {}; 	\node [node] (e) at (0,-1.2) {}; 	\node [node] (f) at (-1,-.5) {}; 	\node [node] (g) at (0,.5) {}; 	\node [node] (h) at (0,-.5) {};
	\draw (a) -- (b) -- (c) -- (d) -- (e) -- (f) -- (a); 	\draw (b) -- (g) -- (h) -- (e); 	\coordinate (1) at  ($ (a)+ (-.75,.75) $); 	\coordinate (2) at  ($ (g)+ (-1.2, 1.2) $); 	\coordinate (3) at  ($ (c)+ (.75,.75) $); 	\coordinate (6) at  ($ (d)+ (.75,-.75) $); 	\coordinate (5) at  ($ (h)+ (1.2,-1.2) $); 	\coordinate (4) at  ($ (f) + (-.75,-.75) $); 	\draw (a) to node [label={[xshift=.1cm]1}] {} (1); 	\draw (g) to node [label={[xshift=.2cm]2}] {} (2); 	\draw (c) to node [label={[xshift=-.2cm]3}] {} (3); 	\draw (d) to node [label={[xshift=.3cm,yshift=-.2cm]6}] {} (6); 	\draw (h) to node [label={[xshift=.15cm, yshift=-.1cm]5}] {} (5); 	\draw (f) to node [label={[xshift=-.1cm]4}] {} (4);
	\draw [line width=.08cm] (2) to (g) to (b) to (c) to (d) to (6); 	\draw [line width=.08cm] (1) to (a) to (f) to (e) to (h) to (5);
	\draw [dashed, line width=.08cm] (1) to[out=180,in=-90,distance=4.5cm] node[label={[yshift=-.5cm]x}] {} (6); 	\draw [dashed, line width=.08cm] (2) to[out=80,in=-20,distance=4.5cm] node[label={[xshift=.2cm,yshift=-.0cm]y}] {} (5); \end{tikzpicture}
\par\end{center}

\begin{center}
\caption{The cycle from the previous figure, after expanding the gadget, is
now a cycle of length $x+y+10$.}

\par\end{center}%
\end{minipage}
\end{figure}

\section{Appendix D: $H_{3}$s} \label{apdxh3}

\subsection{Zero super-edges are covered by $2$-factor}

If none of the super-edges of a $H'_{3}$ are covered by $F_{i}$,
then expanding this gadget returns a cycle of length $10$.

\begin{figure}[H]
\begin{minipage}[t]{0.45\columnwidth}%
\begin{center}
\tikzstyle{node}=[circle, draw, fill=black!50,                         inner sep=0pt, minimum width=4pt]
\begin{tikzpicture}[thick,scale=0.8] 	\coordinate (4) at (-2,1); 	\coordinate (3) at (-2,-1); 	\draw (4) to node [label={[xshift=.15cm,yshift=.3cm]4}] {} (3); 	\draw (4) to node [label={[xshift=.15cm,yshift=-.7cm]3}] {} (3); 	\coordinate (5) at (0,1); 	\coordinate (1) at (0,-1); 	\draw (5) to node [label={[xshift=.15cm,yshift=.3cm]5}] {} (1); 	\draw (5) to node [label={[xshift=.15cm,yshift=-.7cm]1}] {} (1); 	\coordinate (6) at (2,1); 	\coordinate (2) at (2,-1); 	\draw (6) to node [label={[xshift=.15cm,yshift=.3cm]6}] {} (2); 	\draw (6) to node [label={[xshift=.15cm,yshift=-.7cm]2}] {} (2); \end{tikzpicture}
\par\end{center}

\begin{center}
\caption{None of the super-edges in a $H'_{3}$ are included in the $2$-factor.}

\par\end{center}%
\end{minipage}\hfill{}%
\begin{minipage}[t]{0.45\columnwidth}%
\begin{center}
\tikzstyle{node}=[circle, draw, fill=black!50,                         inner sep=0pt, minimum width=4pt]
\begin{tikzpicture}[thick,scale=.75] 	\node [node] (a) at (-1,.5) {}; 	\node [node] (b) at (0,1.2) {}; 	\node [node] (c) at (1,.5) {}; 	\node [node] (d) at (1,-.5) {}; 	\node [node] (e) at (0,-1.2) {}; 	\node [node] (f) at (-1,-.5) {}; 	\node [node] (g) at (0,.5) {}; 	\node [node] (h) at (0,-.5) {}; 	\node [node] (i) at  ($ (a)+ (.5,1.5) $) {}; 	\node [node] (j) at  ($ (d)+ (-.5,2.5) $) {};
	\draw (a) -- (b) -- (c) -- (d) -- (e) -- (f) -- (a); 	\draw (b) -- (g) -- (h) -- (e); 	\draw (a) -- (i) -- (j); 	\draw (d) to[out=0,in=0, distance=.5cm] (j); 	\coordinate (1) at  ($ (j) + (.75,.75) $); 	\coordinate (2) at  ($ (g)+ (-1.2, 1.2) $); 	\coordinate (3) at  ($ (c)+ (.75,.75) $); 	\coordinate (4) at  ($ (i)+ (-.75,.75) $); 	\coordinate (5) at  ($ (h)+ (1.2,-1.2) $); 	\coordinate (6) at  ($ (f) + (-.75,-.75) $); 	\draw (i) to node [label=6] {} (4); 	\draw (j) to node [label=1] {} (1); 	\draw (g) to node [label={[xshift=.2cm]2}] {} (2); 	\draw (c) to node [label={[xshift=-.1cm]3}] {} (3); 	\draw (h) to node [label={[xshift=.15cm, yshift=-.1cm]5}] {} (5); 	\draw (f) to node [label={[xshift=-.1cm]4}] {} (6);
	\draw [line width=.08cm] (a) to (f) to (e) to (h) to (g) to (b) to (c) to (d) to[out=0,in=0, distance=.5cm] (j); 	\draw [line width=.08cm] (j) to (i) to (a); \end{tikzpicture}
\par\end{center}

\begin{center}
\caption{After expanding the gadget, the $H_{3}$ is covered by a cycle of
length $10$.}

\par\end{center}%
\end{minipage}
\end{figure}

\subsection{One super-edge is covered by $2$-factor}

If exactly one edge of a $H_{3}$ gadget is covered by a cycle of
$F_{i}$, then, to ensure we examine every case, we examine all three
ways we can select the super-edge to include in $F_{i}$. In all three
cases, we start with a cycle of length $x+1$ in $F_{i}$ and are
returned either a single cycle of length $x+11$ in $F_{i-1}$. $x+1\geq6$,
so $x+11\geq8$, meaning the resulting cycle in $F_{i-1}$ cannot
be a cycle of length $6$.

\begin{figure}[H]
\begin{minipage}[t]{0.45\columnwidth}%
\begin{center}
\tikzstyle{node}=[circle, draw, fill=black!50,                         inner sep=0pt, minimum width=4pt]
\begin{tikzpicture}[thick,scale=0.8] 	\coordinate (4) at (-2,1); 	\coordinate (3) at (-2,-1); 	\draw (4) to node [label={[xshift=.15cm,yshift=.3cm]4}] {} (3); 	\draw (4) to node [label={[xshift=.15cm,yshift=-.7cm]3}] {} (3); 	\coordinate (5) at (0,1); 	\coordinate (1) at (0,-1); 	\draw (5) to node [label={[xshift=.15cm,yshift=.3cm]5}] {} (1); 	\draw (5) to node [label={[xshift=.15cm,yshift=-.7cm]1}] {} (1); 	\coordinate (6) at (2,1); 	\coordinate (2) at (2,-1); 	\draw (6) to node [label={[xshift=.15cm,yshift=.3cm]6}] {} (2); 	\draw (6) to node [label={[xshift=.15cm,yshift=-.7cm]2}] {} (2);
	\draw [line width=.08cm] (4) to (3); 	\draw [dashed, line width=.08cm] (4) to[out=110,in=-110, distance=2cm] node [label={[xshift=-.2cm]x}] {} (3); \end{tikzpicture}
\par\end{center}

\begin{center}
\caption{A cycle of length $x+1$ passes through a $H'_{3}$.}

\par\end{center}%
\end{minipage}\hfill{}%
\begin{minipage}[t]{0.45\columnwidth}%
\begin{center}
\tikzstyle{node}=[circle, draw, fill=black!50,                         inner sep=0pt, minimum width=4pt]
\begin{tikzpicture}[thick,scale=.75] 	\node [node] (a) at (-1,.5) {}; 	\node [node] (b) at (0,1.2) {}; 	\node [node] (c) at (1,.5) {}; 	\node [node] (d) at (1,-.5) {}; 	\node [node] (e) at (0,-1.2) {}; 	\node [node] (f) at (-1,-.5) {}; 	\node [node] (g) at (0,.5) {}; 	\node [node] (h) at (0,-.5) {}; 	\node [node] (i) at  ($ (a)+ (.5,1.5) $) {}; 	\node [node] (j) at  ($ (d)+ (-.5,2.5) $) {};
	\draw (a) -- (b) -- (c) -- (d) -- (e) -- (f) -- (a); 	\draw (b) -- (g) -- (h) -- (e); 	\draw (a) -- (i) -- (j); 	\draw (d) to[out=0,in=0, distance=.5cm] (j); 	\coordinate (1) at  ($ (j) + (.75,.75) $); 	\coordinate (2) at  ($ (g)+ (-1.2, 1.2) $); 	\coordinate (3) at  ($ (c)+ (.75,.75) $); 	\coordinate (6) at  ($ (i)+ (-.75,.75) $); 	\coordinate (5) at  ($ (h)+ (1.2,-1.2) $); 	\coordinate (4) at  ($ (f) + (-.75,-.75) $); 	\draw (i) to node [label=6] {} (6); 	\draw (j) to node [label=1] {} (1); 	\draw (g) to node [label={[xshift=.2cm]2}] {} (2); 	\draw (c) to node [label={[xshift=.0cm]3}] {} (3); 	\draw (h) to node [label={[xshift=.15cm, yshift=-.1cm]5}] {} (5); 	\draw (f) to node [label={[xshift=-.1cm]4}] {} (4);
	\draw [line width=.08cm] (3) to (c) to (d) to [out=0,in=0, distance=.5cm] (j) to (j) to (i) to (a) to (b) to (g) to (h) to (e) to (f) to (4);
	\draw [dashed, line width=.08cm] (4) to[out=300,in=-30,distance=3cm] node [label={[xshift=.4cm,yshift=-.2cm]x}] {} (3); \end{tikzpicture}
\par\end{center}

\begin{center}
\caption{The cycle from the previous figure, after expanding the gadget, is
now a cycle of length $x+11$.}

\par\end{center}%
\end{minipage}
\end{figure}
\begin{figure}[H]
\begin{minipage}[t]{0.45\columnwidth}%
\begin{center}
\tikzstyle{node}=[circle, draw, fill=black!50,                         inner sep=0pt, minimum width=4pt]
\begin{tikzpicture}[thick,scale=0.8] 	\coordinate (4) at (-2,1); 	\coordinate (3) at (-2,-1); 	\draw (4) to node [label={[xshift=.15cm,yshift=.3cm]4}] {} (3); 	\draw (4) to node [label={[xshift=.15cm,yshift=-.7cm]3}] {} (3); 	\coordinate (5) at (0,1); 	\coordinate (1) at (0,-1); 	\draw (5) to node [label={[xshift=.15cm,yshift=.3cm]5}] {} (1); 	\draw (5) to node [label={[xshift=.15cm,yshift=-.7cm]1}] {} (1); 	\coordinate (6) at (2,1); 	\coordinate (2) at (2,-1); 	\draw (6) to node [label={[xshift=.15cm,yshift=.3cm]6}] {} (2); 	\draw (6) to node [label={[xshift=.15cm,yshift=-.7cm]2}] {} (2);
	\draw [line width=.08cm] (5) to (1); 	\draw [dashed, line width=.08cm] (5) to[out=110,in=-110, distance=2cm] node [label={[xshift=-.2cm]x}] {} (1); \end{tikzpicture}
\par\end{center}

\begin{center}
\caption{A cycle of length $x+1$ passes through a $H'_{3}$.}

\par\end{center}%
\end{minipage}\hfill{}%
\begin{minipage}[t]{0.45\columnwidth}%
\begin{center}
\tikzstyle{node}=[circle, draw, fill=black!50,                         inner sep=0pt, minimum width=4pt]
\begin{tikzpicture}[thick,scale=.75] 	\node [node] (a) at (-1,.5) {}; 	\node [node] (b) at (0,1.2) {}; 	\node [node] (c) at (1,.5) {}; 	\node [node] (d) at (1,-.5) {}; 	\node [node] (e) at (0,-1.2) {}; 	\node [node] (f) at (-1,-.5) {}; 	\node [node] (g) at (0,.5) {}; 	\node [node] (h) at (0,-.5) {}; 	\node [node] (i) at  ($ (a)+ (.5,1.5) $) {}; 	\node [node] (j) at  ($ (d)+ (-.5,2.5) $) {};
	\draw (a) -- (b) -- (c) -- (d) -- (e) -- (f) -- (a); 	\draw (b) -- (g) -- (h) -- (e); 	\draw (a) -- (i) -- (j); 	\draw (d) to[out=0,in=0, distance=.5cm] (j); 	\coordinate (1) at  ($ (j) + (.75,.75) $); 	\coordinate (2) at  ($ (g)+ (-1.2, 1.2) $); 	\coordinate (3) at  ($ (c)+ (.75,.75) $); 	\coordinate (6) at  ($ (i)+ (-.75,.75) $); 	\coordinate (5) at  ($ (h)+ (1.2,-1.2) $); 	\coordinate (4) at  ($ (f) + (-.75,-.75) $); 	\draw (i) to node [label=6] {} (6); 	\draw (j) to node [label=1] {} (1); 	\draw (g) to node [label={[xshift=.2cm]2}] {} (2); 	\draw (c) to node [label={[xshift=-.1cm]3}] {} (3); 	\draw (h) to node [label={[xshift=.15cm, yshift=-.1cm]5}] {} (5); 	\draw (f) to node [label={[xshift=-.1cm]4}] {} (4);
	\draw [line width=.08cm] (5) to (h) to (g) to (b) to (c) to (d) to (e) to (f) to (a) to (i) to (j) to (1);
	\draw [dashed, line width=.08cm] (5) to[out=30,in=-30,distance=1.3cm] node [label={[xshift=.2cm]x}] {} (1); \end{tikzpicture}
\par\end{center}

\begin{center}
\caption{The cycle from the previous figure, after expanding the gadget, is
now a cycle of length $x+11$.}

\par\end{center}%
\end{minipage}
\end{figure}
\begin{figure}[H]
\begin{minipage}[t]{0.45\columnwidth}%
\begin{center}
\tikzstyle{node}=[circle, draw, fill=black!50,                         inner sep=0pt, minimum width=4pt]
\begin{tikzpicture}[thick,scale=0.8] 	\coordinate (4) at (-2,1); 	\coordinate (3) at (-2,-1); 	\draw (4) to node [label={[xshift=.15cm,yshift=.3cm]4}] {} (3); 	\draw (4) to node [label={[xshift=.15cm,yshift=-.7cm]3}] {} (3); 	\coordinate (5) at (0,1); 	\coordinate (1) at (0,-1); 	\draw (5) to node [label={[xshift=.15cm,yshift=.3cm]5}] {} (1); 	\draw (5) to node [label={[xshift=.15cm,yshift=-.7cm]1}] {} (1); 	\coordinate (6) at (2,1); 	\coordinate (2) at (2,-1); 	\draw (6) to node [label={[xshift=.15cm,yshift=.3cm]6}] {} (2); 	\draw (6) to node [label={[xshift=.15cm,yshift=-.7cm]2}] {} (2);
	\draw [line width=.08cm] (6) to (2); 	\draw [dashed, line width=.08cm] (6) to[out=110,in=-110, distance=2cm] node [label={[xshift=-.2cm]x}] {} (2); \end{tikzpicture}
\par\end{center}

\begin{center}
\caption{A cycle of length $x+1$ passes through a $H'_{3}$.}

\par\end{center}%
\end{minipage}\hfill{}%
\begin{minipage}[t]{0.45\columnwidth}%
\begin{center}
\tikzstyle{node}=[circle, draw, fill=black!50,                         inner sep=0pt, minimum width=4pt]
\begin{tikzpicture}[thick,scale=.75] 	\node [node] (a) at (-1,.5) {}; 	\node [node] (b) at (0,1.2) {}; 	\node [node] (c) at (1,.5) {}; 	\node [node] (d) at (1,-.5) {}; 	\node [node] (e) at (0,-1.2) {}; 	\node [node] (f) at (-1,-.5) {}; 	\node [node] (g) at (0,.5) {}; 	\node [node] (h) at (0,-.5) {}; 	\node [node] (i) at  ($ (a)+ (.5,1.5) $) {}; 	\node [node] (j) at  ($ (d)+ (-.5,2.5) $) {};
	\draw (a) -- (b) -- (c) -- (d) -- (e) -- (f) -- (a); 	\draw (b) -- (g) -- (h) -- (e); 	\draw (a) -- (i) -- (j); 	\draw (d) to[out=0,in=0, distance=.5cm] (j); 	\coordinate (1) at  ($ (j) + (.75,.75) $); 	\coordinate (2) at  ($ (g)+ (-1.2, 1.2) $); 	\coordinate (3) at  ($ (c)+ (.75,.75) $); 	\coordinate (6) at  ($ (i)+ (-.75,.75) $); 	\coordinate (5) at  ($ (h)+ (1.2,-1.2) $); 	\coordinate (4) at  ($ (f) + (-.75,-.75) $); 	\draw (i) to node [label=6] {} (6); 	\draw (j) to node [label=1] {} (1); 	\draw (g) to node [label={[xshift=.2cm]2}] {} (2); 	\draw (c) to node [label={[xshift=-.05cm]3}] {} (3); 	\draw (h) to node [label={[xshift=.15cm, yshift=-.1cm]5}] {} (5); 	\draw (f) to node [label={[xshift=-.1cm]4}] {} (4);
	\draw [line width=.08cm] (6) to (i) to (j) to [out=0,in=0, distance=.5cm] (d) to (d) to (c) to (b) to (a) to (f) to (e) to (h) to (g) to (2);
	\draw [dashed, line width=.08cm] (6) to[out=150,in=210,distance=.5cm] node [label={[xshift=-.2cm]x}] {} (2); \end{tikzpicture}
\par\end{center}

\begin{center}
\caption{The cycle from the previous figure, after expanding the gadget, is
now a cycle of length $x+11$.}

\par\end{center}%
\end{minipage}
\end{figure}

\subsection{Two super-edges are covered by $2$-factor}

\subsubsection{Two cycles pass through gadget}

If two edges of a $H'_{3}$ are covered by two disjoint cycles in
$F_{i}$, then, to ensure we examine every case, we examine all three
ways we can select the super-edge to exclude from $F_{i}$. In all
three cases, we start with cycles of lengths $x+1$ and $y+1$ in
$F_{i}$ and are returned either a single cycle of length $x+y+12$
or two cycles of lengths $x+5$ and $y+7$ in $F_{i-1}$. $x+1\geq6$,
and $y+1\geq6$, so $x+5\geq8$, $y+7\geq8$, and $x+y+12\geq8$,
meaning the resulting cycle or cycles in $F_{i-1}$ cannot be a cycle
of length $6$.

\begin{figure}[H]
\begin{minipage}[t]{0.45\columnwidth}%
\begin{center}
\tikzstyle{node}=[circle, draw, fill=black!50,                         inner sep=0pt, minimum width=4pt]
\begin{tikzpicture}[thick,scale=0.8] 	\coordinate (4) at (-2,1); 	\coordinate (3) at (-2,-1); 	\draw (4) to node [label={[xshift=.15cm,yshift=.3cm]4}] {} (3); 	\draw (4) to node [label={[xshift=.15cm,yshift=-.7cm]3}] {} (3); 	\coordinate (5) at (0,1); 	\coordinate (1) at (0,-1); 	\draw (5) to node [label={[xshift=.15cm,yshift=.3cm]5}] {} (1); 	\draw (5) to node [label={[xshift=.15cm,yshift=-.7cm]1}] {} (1); 	\coordinate (6) at (2,1); 	\coordinate (2) at (2,-1); 	\draw (6) to node [label={[xshift=.15cm,yshift=.3cm]6}] {} (2); 	\draw (6) to node [label={[xshift=.15cm,yshift=-.7cm]2}] {} (2);
	\draw [line width=.08cm] (4) to (3); 	\draw [dashed, line width=.08cm] (4) to[out=110,in=-110, distance=2cm] node [label={[xshift=-.2cm]x}] {} (3);
	\draw [line width=.08cm] (5) to (1); 	\draw [dashed, line width=.08cm] (5) to[out=110,in=-110, distance=2cm] node [label={[xshift=-.2cm]y}] {} (1); \end{tikzpicture}
\par\end{center}

\begin{center}
\caption{Two cycles of lengths $x+1$ and $y+1$ pass through a $H'_{3}$.}

\par\end{center}%
\end{minipage}\hfill{}%
\begin{minipage}[t]{0.45\columnwidth}%
\begin{center}
\tikzstyle{node}=[circle, draw, fill=black!50,                         inner sep=0pt, minimum width=4pt]
\begin{tikzpicture}[thick,scale=.75] 	\node [node] (a) at (-1,.5) {}; 	\node [node] (b) at (0,1.2) {}; 	\node [node] (c) at (1,.5) {}; 	\node [node] (d) at (1,-.5) {}; 	\node [node] (e) at (0,-1.2) {}; 	\node [node] (f) at (-1,-.5) {}; 	\node [node] (g) at (0,.5) {}; 	\node [node] (h) at (0,-.5) {}; 	\node [node] (i) at  ($ (a)+ (.5,1.5) $) {}; 	\node [node] (j) at  ($ (d)+ (-.5,2.5) $) {};
	\draw (a) -- (b) -- (c) -- (d) -- (e) -- (f) -- (a); 	\draw (b) -- (g) -- (h) -- (e); 	\draw (a) -- (i) -- (j); 	\draw (d) to[out=0,in=0, distance=.5cm] (j); 	\coordinate (1) at  ($ (j) + (.75,.75) $); 	\coordinate (2) at  ($ (g)+ (-1.2, 1.2) $); 	\coordinate (3) at  ($ (c)+ (.75,.75) $); 	\coordinate (6) at  ($ (i)+ (-.75,.75) $); 	\coordinate (5) at  ($ (h)+ (1.2,-1.2) $); 	\coordinate (4) at  ($ (f) + (-.75,-.75) $); 	\draw (i) to node [label=6] {} (6); 	\draw (j) to node [label=1] {} (1); 	\draw (g) to node [label={[xshift=.2cm]2}] {} (2); 	\draw (c) to node [label={[xshift=-.1cm]3}] {} (3); 	\draw (h) to node [label={[xshift=.15cm, yshift=-.1cm]5}] {} (5); 	\draw (f) to node [label={[xshift=-.1cm]4}] {} (4);
	\draw [line width=.08cm] (1) to (j) to (i) to (a) to (b) to (g) to (h) to (5); 	\draw [line width=.08cm] (4) to (f) to (e) to (d) to (c) to (3);
	\draw [dashed, line width=.08cm] (4) to[out=310,in=-40,distance=3.2cm] node [label={[xshift=.4cm,yshift=-.2cm]x}] {} (3); 	\draw [dashed, line width=.08cm] (5) to[out=30,in=-30,distance=1.3cm] node [label={[xshift=.2cm,yshift=.5cm]y}] {} (1); \end{tikzpicture}
\par\end{center}

\begin{center}
\caption{The cycles from the previous figure, after expanding the gadget, are
now cycles of lengths $x+5$ and $y+7$, respectively.}

\par\end{center}%
\end{minipage}
\end{figure}
\begin{figure}[H]
\begin{minipage}[t]{0.45\columnwidth}%
\begin{center}
\tikzstyle{node}=[circle, draw, fill=black!50,                         inner sep=0pt, minimum width=4pt]
\begin{tikzpicture}[thick,scale=0.8] 	\coordinate (4) at (-2,1); 	\coordinate (3) at (-2,-1); 	\draw (4) to node [label={[xshift=.15cm,yshift=.3cm]4}] {} (3); 	\draw (4) to node [label={[xshift=.15cm,yshift=-.7cm]3}] {} (3); 	\coordinate (5) at (0,1); 	\coordinate (1) at (0,-1); 	\draw (5) to node [label={[xshift=.15cm,yshift=.3cm]5}] {} (1); 	\draw (5) to node [label={[xshift=.15cm,yshift=-.7cm]1}] {} (1); 	\coordinate (6) at (2,1); 	\coordinate (2) at (2,-1); 	\draw (6) to node [label={[xshift=.15cm,yshift=.3cm]6}] {} (2); 	\draw (6) to node [label={[xshift=.15cm,yshift=-.7cm]2}] {} (2);
	\draw [line width=.08cm] (4) to (3); 	\draw [dashed, line width=.08cm] (4) to[out=110,in=-110, distance=2cm] node [label={[xshift=-.2cm]x}] {} (3);
	\draw [line width=.08cm] (6) to (2); 	\draw [dashed, line width=.08cm] (6) to[out=110,in=-110, distance=2cm] node [label={[xshift=-.2cm]y}] {} (2); \end{tikzpicture} 
\par\end{center}

\begin{center}
\caption{Two cycles of lengths $x+1$ and $y+1$ pass through a $H'_{3}$.}

\par\end{center}%
\end{minipage}\hfill{}%
\begin{minipage}[t]{0.45\columnwidth}%
\begin{center}
\tikzstyle{node}=[circle, draw, fill=black!50,                         inner sep=0pt, minimum width=4pt]
\begin{tikzpicture}[thick,scale=.75] 	\node [node] (a) at (-1,.5) {}; 	\node [node] (b) at (0,1.2) {}; 	\node [node] (c) at (1,.5) {}; 	\node [node] (d) at (1,-.5) {}; 	\node [node] (e) at (0,-1.2) {}; 	\node [node] (f) at (-1,-.5) {}; 	\node [node] (g) at (0,.5) {}; 	\node [node] (h) at (0,-.5) {}; 	\node [node] (i) at  ($ (a)+ (.5,1.5) $) {}; 	\node [node] (j) at  ($ (d)+ (-.5,2.5) $) {};
	\draw (a) -- (b) -- (c) -- (d) -- (e) -- (f) -- (a); 	\draw (b) -- (g) -- (h) -- (e); 	\draw (a) -- (i) -- (j); 	\draw (d) to[out=0,in=0, distance=.5cm] (j); 	\coordinate (1) at  ($ (j) + (.75,.75) $); 	\coordinate (2) at  ($ (g)+ (-1.2, 1.2) $); 	\coordinate (3) at  ($ (c)+ (.75,.75) $); 	\coordinate (6) at  ($ (i)+ (-.75,.75) $); 	\coordinate (5) at  ($ (h)+ (1.2,-1.2) $); 	\coordinate (4) at  ($ (f) + (-.75,-.75) $); 	\draw (i) to node [label=6] {} (6); 	\draw (j) to node [label=1] {} (1); 	\draw (g) to node [label={[xshift=.2cm]2}] {} (2); 	\draw (c) to node [label={[xshift=-.05cm]3}] {} (3); 	\draw (h) to node [label={[xshift=.15cm, yshift=-.1cm]5}] {} (5); 	\draw (f) to node [label={[xshift=-.1cm]4}] {} (4);
	\draw [line width=.08cm] (6) to (i) to (j) to[out=0,in=0,distance=.5cm] (d) to (d) to (e) to (h) to (g) to (2); 	\draw [line width=.08cm] (4) to (f) to (a) to (b) to (c) to (3);
	\draw [dashed, line width=.08cm] (4) to[out=310,in=-40,distance=3.2cm] node [label={[xshift=.4cm,yshift=-.2cm]x}] {} (3); 	\draw [dashed, line width=.08cm] (6) to[out=150,in=210,distance=.5cm] node [label={[xshift=-.2cm]y}] {} (2); \end{tikzpicture}
\par\end{center}

\begin{center}
\caption{The cycles from the previous figure, after expanding the gadget, are
now cycles of lengths $x+5$ and $y+7$, respectively.}

\par\end{center}%
\end{minipage}
\end{figure}
\begin{figure}[H]
\begin{minipage}[t]{0.45\columnwidth}%
\begin{center}
\tikzstyle{node}=[circle, draw, fill=black!50,                         inner sep=0pt, minimum width=4pt]
\begin{tikzpicture}[thick,scale=0.8] 	\coordinate (4) at (-2,1); 	\coordinate (3) at (-2,-1); 	\draw (4) to node [label={[xshift=.15cm,yshift=.3cm]4}] {} (3); 	\draw (4) to node [label={[xshift=.15cm,yshift=-.7cm]3}] {} (3); 	\coordinate (5) at (0,1); 	\coordinate (1) at (0,-1); 	\draw (5) to node [label={[xshift=.15cm,yshift=.3cm]5}] {} (1); 	\draw (5) to node [label={[xshift=.15cm,yshift=-.7cm]1}] {} (1); 	\coordinate (6) at (2,1); 	\coordinate (2) at (2,-1); 	\draw (6) to node [label={[xshift=.15cm,yshift=.3cm]6}] {} (2); 	\draw (6) to node [label={[xshift=.15cm,yshift=-.7cm]2}] {} (2);
	\draw [line width=.08cm] (5) to (1); 	\draw [dashed, line width=.08cm] (5) to[out=110,in=-110, distance=2cm] node [label={[xshift=-.2cm]x}] {} (1);
	\draw [line width=.08cm] (6) to (2); 	\draw [dashed, line width=.08cm] (6) to[out=110,in=-110, distance=2cm] node [label={[xshift=-.2cm]y}] {} (2); \end{tikzpicture}
\par\end{center}

\begin{center}
\caption{Two cycles of lengths $x+1$ and $y+1$ pass through a $H'_{3}$.}

\par\end{center}%
\end{minipage}\hfill{}%
\begin{minipage}[t]{0.45\columnwidth}%
\begin{center}
\tikzstyle{node}=[circle, draw, fill=black!50,                         inner sep=0pt, minimum width=4pt]
\begin{tikzpicture}[thick,scale=.75] 	\node [node] (a) at (-1,.5) {}; 	\node [node] (b) at (0,1.2) {}; 	\node [node] (c) at (1,.5) {}; 	\node [node] (d) at (1,-.5) {}; 	\node [node] (e) at (0,-1.2) {}; 	\node [node] (f) at (-1,-.5) {}; 	\node [node] (g) at (0,.5) {}; 	\node [node] (h) at (0,-.5) {}; 	\node [node] (i) at  ($ (a)+ (.5,1.5) $) {}; 	\node [node] (j) at  ($ (d)+ (-.5,2.5) $) {};
	\draw (a) -- (b) -- (c) -- (d) -- (e) -- (f) -- (a); 	\draw (b) -- (g) -- (h) -- (e); 	\draw (a) -- (i) -- (j); 	\draw (d) to[out=0,in=0, distance=.5cm] (j); 	\coordinate (1) at  ($ (j) + (.75,.75) $); 	\coordinate (2) at  ($ (g)+ (-1.2, 1.2) $); 	\coordinate (3) at  ($ (c)+ (.75,.75) $); 	\coordinate (6) at  ($ (i)+ (-.75,.75) $); 	\coordinate (5) at  ($ (h)+ (1.2,-1.2) $); 	\coordinate (4) at  ($ (f) + (-.75,-.75) $); 	\draw (i) to node [label=6] {} (6); 	\draw (j) to node [label=1] {} (1); 	\draw (g) to node [label={[xshift=.2cm]2}] {} (2); 	\draw (c) to node [label={[xshift=-.05cm]3}] {} (3); 	\draw (h) to node [label={[xshift=.15cm, yshift=-.1cm]5}] {} (5); 	\draw (f) to node [label={[xshift=-.1cm]4}] {} (4);
	\draw [line width=.08cm] (1) to (j) to[out=0,in=0,distance=.5cm] (d) to (d) to (c) to (b) to (g) to (2); 	\draw [line width=.08cm] (6) to (i) to (a) to (f) to (e) to (h) to (5);
	\draw [dashed, line width=.08cm] (5) to[out=30,in=-30,distance=1.3cm] node [label={[xshift=.2cm,yshift=.5cm]x}] {} (1); 	\draw [dashed, line width=.08cm] (6) to[out=150,in=210,distance=.5cm] node [label={[xshift=-.2cm]y}] {} (2); \end{tikzpicture}
\par\end{center}

\begin{center}
\caption{The cycles from the previous figure, after expanding the gadget, are
now a single cycle of length $x+y+12$.}

\par\end{center}%
\end{minipage}
\end{figure}

\subsubsection{One cycle passes through gadget}

If two edges of a $H'_{3}$ are covered by a single cycle in $F_{i}$,
then, to ensure we examine every case, we examine all three ways we
can select the super-edge to exclude from $F_{i}$ and then within
each of these arrangements, we examine both orientations in which
the two super-edges in the same cycle can be connected. In all six
cases, we start with a cycle of length $x+y+2$ in $F_{i}$ and are
returned a single cycle of length $x+y+12$ in $F_{i-1}$, except
for one exception, which is analyzed in detail
in Sections \ref{expandh3} and \ref{analyzingworstcase}. $x+y+2\geq6$, so $x+y+12>8$, meaning
the resulting cycle in $F_{i-1}$ cannot be a cycle of length $6$,
except for in the special case we have identified.

\begin{figure}[H]
\begin{minipage}[t]{0.45\columnwidth}%
\begin{center}
\tikzstyle{node}=[circle, draw, fill=black!50,                         inner sep=0pt, minimum width=4pt]
\begin{tikzpicture}[thick,scale=0.8] 	\coordinate (4) at (-2,1); 	\coordinate (3) at (-2,-1); 	\draw (4) to node [label={[xshift=.15cm,yshift=.3cm]4}] {} (3); 	\draw (4) to node [label={[xshift=.15cm,yshift=-.7cm]3}] {} (3); 	\coordinate (5) at (0,1); 	\coordinate (1) at (0,-1); 	\draw (5) to node [label={[xshift=.15cm,yshift=.3cm]5}] {} (1); 	\draw (5) to node [label={[xshift=.15cm,yshift=-.7cm]1}] {} (1); 	\coordinate (6) at (2,1); 	\coordinate (2) at (2,-1); 	\draw (6) to node [label={[xshift=.15cm,yshift=.3cm]6}] {} (2); 	\draw (6) to node [label={[xshift=.15cm,yshift=-.7cm]2}] {} (2);
	\draw [line width=.08cm] (4) to (3); 	\draw [dashed, line width=.08cm] (4) to[out=230,in=-140, distance=2.2cm] node [label={[xshift=-.4cm]x}] {} (1);
	\draw [line width=.08cm] (5) to (1); 	\draw [dashed, line width=.08cm] (3) to[out=130,in=140, distance=2.2cm] node [label={[xshift=-.0cm]y}] {} (5); \end{tikzpicture}
\par\end{center}

\begin{center}
\caption{A cycle of length $x+y+2$ passes through a $H'_{3}$.}

\par\end{center}%
\end{minipage}\hfill{}%
\begin{minipage}[t]{0.45\columnwidth}%
\begin{center}
\tikzstyle{node}=[circle, draw, fill=black!50,                         inner sep=0pt, minimum width=4pt]
\begin{tikzpicture}[thick,scale=.75] 	\node [node] (a) at (-1,.5) {}; 	\node [node] (b) at (0,1.2) {}; 	\node [node] (c) at (1,.5) {}; 	\node [node] (d) at (1,-.5) {}; 	\node [node] (e) at (0,-1.2) {}; 	\node [node] (f) at (-1,-.5) {}; 	\node [node] (g) at (0,.5) {}; 	\node [node] (h) at (0,-.5) {}; 	\node [node] (i) at  ($ (a)+ (.5,1.5) $) {}; 	\node [node] (j) at  ($ (d)+ (-.5,2.5) $) {};
	\draw (a) -- (b) -- (c) -- (d) -- (e) -- (f) -- (a); 	\draw (b) -- (g) -- (h) -- (e); 	\draw (a) -- (i) -- (j); 	\draw (d) to[out=0,in=0, distance=.5cm] (j); 	\coordinate (1) at  ($ (j) + (.75,.75) $); 	\coordinate (2) at  ($ (g)+ (-1.2, 1.2) $); 	\coordinate (3) at  ($ (c)+ (.75,.75) $); 	\coordinate (6) at  ($ (i)+ (-.75,.75) $); 	\coordinate (5) at  ($ (h)+ (1.2,-1.2) $); 	\coordinate (4) at  ($ (f) + (-.75,-.75) $); 	\draw (i) to node [label=6] {} (6); 	\draw (j) to node [label=1] {} (1); 	\draw (g) to node [label={[xshift=.2cm]2}] {} (2); 	\draw (c) to node [label={[xshift=-.1cm]3}] {} (3); 	\draw (h) to node [label={[xshift=.15cm, yshift=-.1cm]5}] {} (5); 	\draw (f) to node [label={[xshift=-.1cm]4}] {} (4);
	\draw [line width=.08cm] (1) to (j) to (i) to (a) to (b) to (g) to (h) to (5); 	\draw [line width=.08cm] (4) to (f) to (e) to (d) to (c) to (3);
	\draw [dashed, line width=.08cm] (4) to[out=130,in=120,distance=3.2cm] node [label={[xshift=-.1cm,yshift=-.0cm]x}] {} (1); 	\draw [dashed, line width=.08cm] (5) to[out=30,in=-30,distance=1.3cm] node [label={[xshift=.2cm,yshift=.5cm]y}] {} (3); \end{tikzpicture}
\par\end{center}

\begin{center}
\caption{The cycle from the previous figure, after expanding the gadget, is
now a cycle of length $x+y+12$.}

\par\end{center}%
\end{minipage}
\end{figure}
\begin{figure}[H]
\begin{minipage}[t]{0.45\columnwidth}%
\begin{center}
\tikzstyle{node}=[circle, draw, fill=black!50,                         inner sep=0pt, minimum width=4pt]
\begin{tikzpicture}[thick,scale=0.8] 	\coordinate (4) at (-2,1); 	\coordinate (3) at (-2,-1); 	\draw (4) to node [label={[xshift=.15cm,yshift=.3cm]4}] {} (3); 	\draw (4) to node [label={[xshift=.15cm,yshift=-.7cm]3}] {} (3); 	\coordinate (5) at (0,1); 	\coordinate (1) at (0,-1); 	\draw (5) to node [label={[xshift=.15cm,yshift=.3cm]5}] {} (1); 	\draw (5) to node [label={[xshift=.15cm,yshift=-.7cm]1}] {} (1); 	\coordinate (6) at (2,1); 	\coordinate (2) at (2,-1); 	\draw (6) to node [label={[xshift=.15cm,yshift=.3cm]6}] {} (2); 	\draw (6) to node [label={[xshift=.15cm,yshift=-.7cm]2}] {} (2);
	\draw [line width=.08cm] (4) to (3); 	\draw [dashed, line width=.08cm] (4) to[out=230,in=-140, distance=3cm] node [label={[xshift=-.2cm,yshift=-.5cm]x}] {} (2);
	\draw [line width=.08cm] (6) to (2); 	\draw [dashed, line width=.08cm] (3) to[out=130,in=130, distance=3cm] node [label={[xshift=-.2cm]y}] {} (6); \end{tikzpicture}
\par\end{center}

\begin{center}
\caption{A cycle of length $x+y+2$ passes through a $H'_{3}$.}

\par\end{center}%
\end{minipage}\hfill{}%
\begin{minipage}[t]{0.45\columnwidth}%
\begin{center}
\tikzstyle{node}=[circle, draw, fill=black!50,                         inner sep=0pt, minimum width=4pt]
\begin{tikzpicture}[thick,scale=.75] 	\node [node] (a) at (-1,.5) {}; 	\node [node] (b) at (0,1.2) {}; 	\node [node] (c) at (1,.5) {}; 	\node [node] (d) at (1,-.5) {}; 	\node [node] (e) at (0,-1.2) {}; 	\node [node] (f) at (-1,-.5) {}; 	\node [node] (g) at (0,.5) {}; 	\node [node] (h) at (0,-.5) {}; 	\node [node] (i) at  ($ (a)+ (.5,1.5) $) {}; 	\node [node] (j) at  ($ (d)+ (-.5,2.5) $) {};
	\draw (a) -- (b) -- (c) -- (d) -- (e) -- (f) -- (a); 	\draw (b) -- (g) -- (h) -- (e); 	\draw (a) -- (i) -- (j); 	\draw (d) to[out=0,in=0, distance=.5cm] (j); 	\coordinate (1) at  ($ (j) + (.75,.75) $); 	\coordinate (2) at  ($ (g)+ (-1.2, 1.2) $); 	\coordinate (3) at  ($ (c)+ (.75,.75) $); 	\coordinate (6) at  ($ (i)+ (-.75,.75) $); 	\coordinate (5) at  ($ (h)+ (1.2,-1.2) $); 	\coordinate (4) at  ($ (f) + (-.75,-.75) $); 	\draw (i) to node [label=6] {} (6); 	\draw (j) to node [label=1] {} (1); 	\draw (g) to node [label={[xshift=.2cm]2}] {} (2); 	\draw (c) to node [label={[xshift=-.05cm]3}] {} (3); 	\draw (h) to node [label={[xshift=.15cm, yshift=-.1cm]5}] {} (5); 	\draw (f) to node [label={[xshift=-.1cm]4}] {} (4);
	\draw [line width=.08cm] (6) to (i) to (j) to[out=0,in=0,distance=.5cm] (d) to (d) to (e) to (h) to (g) to (2); 	\draw [line width=.08cm] (4) to (f) to (a) to (b) to (c) to (3);
	\draw [dashed, line width=.08cm] (4) to[out=160,in=200,distance=1cm] node [label={[xshift=-.2cm,yshift=-.0cm]x}] {} (2); 	\draw [dashed, line width=.08cm] (6) to[out=50,in=30,distance=2cm] node [label={[xshift=-.2cm,yshift=.1cm]y}] {} (3); \end{tikzpicture}
\par\end{center}

\begin{center}
\caption{The cycle from the previous figure, after expanding the gadget, is
now a cycle of length $x+y+12$.}

\par\end{center}%
\end{minipage}
\end{figure}
\begin{figure}[H]
\begin{minipage}[t]{0.45\columnwidth}%
\begin{center}
\tikzstyle{node}=[circle, draw, fill=black!50,                         inner sep=0pt, minimum width=4pt]
\begin{tikzpicture}[thick,scale=0.8] 	\coordinate (4) at (-2,1); 	\coordinate (3) at (-2,-1); 	\draw (4) to node [label={[xshift=.15cm,yshift=.3cm]4}] {} (3); 	\draw (4) to node [label={[xshift=.15cm,yshift=-.7cm]3}] {} (3); 	\coordinate (5) at (0,1); 	\coordinate (1) at (0,-1); 	\draw (5) to node [label={[xshift=.15cm,yshift=.3cm]5}] {} (1); 	\draw (5) to node [label={[xshift=.15cm,yshift=-.7cm]1}] {} (1); 	\coordinate (6) at (2,1); 	\coordinate (2) at (2,-1); 	\draw (6) to node [label={[xshift=.15cm,yshift=.3cm]6}] {} (2); 	\draw (6) to node [label={[xshift=.15cm,yshift=-.7cm]2}] {} (2);
	\draw [line width=.08cm] (4) to (3); 	\draw [dashed, line width=.08cm] (4) to[out=90,in=90, distance=1cm] node [label={[xshift=-.0cm]x}] {} (5);
	\draw [line width=.08cm] (5) to (1); 	\draw [dashed, line width=.08cm] (3) to[out=-90,in=-90, distance=1cm] node [label={[xshift=-.0cm]y}] {} (1); \end{tikzpicture}
\par\end{center}

\begin{center}
\caption{A cycle of length $x+y+2$ passes through a $H'_{3}$.}

\par\end{center}%
\end{minipage}\hfill{}%
\begin{minipage}[t]{0.45\columnwidth}%
\begin{center}
\tikzstyle{node}=[circle, draw, fill=black!50,                         inner sep=0pt, minimum width=4pt]
\begin{tikzpicture}[thick,scale=.75] 	\node [node] (a) at (-1,.5) {}; 	\node [node] (b) at (0,1.2) {}; 	\node [node] (c) at (1,.5) {}; 	\node [node] (d) at (1,-.5) {}; 	\node [node] (e) at (0,-1.2) {}; 	\node [node] (f) at (-1,-.5) {}; 	\node [node] (g) at (0,.5) {}; 	\node [node] (h) at (0,-.5) {}; 	\node [node] (i) at  ($ (a)+ (.5,1.5) $) {}; 	\node [node] (j) at  ($ (d)+ (-.5,2.5) $) {};
	\draw (a) -- (b) -- (c) -- (d) -- (e) -- (f) -- (a); 	\draw (b) -- (g) -- (h) -- (e); 	\draw (a) -- (i) -- (j); 	\draw (d) to[out=0,in=0, distance=.5cm] (j); 	\coordinate (1) at  ($ (j) + (.75,.75) $); 	\coordinate (2) at  ($ (g)+ (-1.2, 1.2) $); 	\coordinate (3) at  ($ (c)+ (.75,.75) $); 	\coordinate (6) at  ($ (i)+ (-.75,.75) $); 	\coordinate (5) at  ($ (h)+ (1.2,-1.2) $); 	\coordinate (4) at  ($ (f) + (-.75,-.75) $); 	\draw (i) to node [label=6] {} (6); 	\draw (j) to node [label=1] {} (1); 	\draw (g) to node [label={[xshift=.2cm]2}] {} (2); 	\draw (c) to node [label={[xshift=-.1cm]3}] {} (3); 	\draw (h) to node [label={[xshift=.15cm, yshift=-.1cm]5}] {} (5); 	\draw (f) to node [label={[xshift=-.1cm]4}] {} (4);
	\draw [line width=.08cm] (1) to (j) to (i) to (a) to (b) to (g) to (h) to (5); 	\draw [line width=.08cm] (4) to (f) to (e) to (d) to (c) to (3);
	\draw [dashed, line width=.08cm] (4) to[out=260,in=-80,distance=1cm] node [label={[xshift=.0cm,yshift=-.5cm]x}] {} (5); 	\draw [dashed, line width=.08cm] (3) to[out=30,in=-30,distance=.5cm] node [label={[xshift=.2cm,yshift=.0cm]y}] {} (1); \end{tikzpicture}
\par\end{center}

\begin{center}
\caption{The cycle from the previous figure, after expanding the gadget, is
now a cycle of length $x+y+12$.}

\par\end{center}%
\end{minipage}
\end{figure}
\begin{figure}[H]
\begin{minipage}[t]{0.45\columnwidth}%
\begin{center}
\tikzstyle{node}=[circle, draw, fill=black!50,                         inner sep=0pt, minimum width=4pt]
\begin{tikzpicture}[thick,scale=0.8] 	\coordinate (4) at (-2,1); 	\coordinate (3) at (-2,-1); 	\draw (4) to node [label={[xshift=.15cm,yshift=.3cm]4}] {} (3); 	\draw (4) to node [label={[xshift=.15cm,yshift=-.7cm]3}] {} (3); 	\coordinate (5) at (0,1); 	\coordinate (1) at (0,-1); 	\draw (5) to node [label={[xshift=.15cm,yshift=.3cm]5}] {} (1); 	\draw (5) to node [label={[xshift=.15cm,yshift=-.7cm]1}] {} (1); 	\coordinate (6) at (2,1); 	\coordinate (2) at (2,-1); 	\draw (6) to node [label={[xshift=.15cm,yshift=.3cm]6}] {} (2); 	\draw (6) to node [label={[xshift=.15cm,yshift=-.7cm]2}] {} (2);
	\draw [line width=.08cm] (4) to (3); 	\draw [dashed, line width=.08cm] (4) to[out=90,in=90, distance=1cm] node [label={[xshift=-.0cm]x}] {} (6);
	\draw [line width=.08cm] (6) to (2); 	\draw [dashed, line width=.08cm] (3) to[out=-90,in=-90, distance=1cm] node [label={[xshift=-.0cm]y}] {} (2); \end{tikzpicture}
\par\end{center}

\begin{center}
\caption{A cycle of length $x+y+2$ passes through a $H'_{3}$.}

\par\end{center}%
\end{minipage}\hfill{}%
\begin{minipage}[t]{0.45\columnwidth}%
\begin{center}
\tikzstyle{node}=[circle, draw, fill=black!50,                         inner sep=0pt, minimum width=4pt]
\begin{tikzpicture}[thick,scale=.8] 	\node [node] (a) at (-1,.5) {}; 	\node [node] (b) at (0,1.2) {}; 	\node [node] (c) at (1,.5) {}; 	\node [node] (d) at (1,-.5) {}; 	\node [node] (e) at (0,-1.2) {}; 	\node [node] (f) at (-1,-.5) {}; 	\node [node] (g) at (0,.5) {}; 	\node [node] (h) at (0,-.5) {}; 	\node [node] (i) at  ($ (a)+ (.5,1.5) $) {}; 	\node [node] (j) at  ($ (d)+ (-.5,2.5) $) {};
	\draw (a) -- (b) -- (c) -- (d) -- (e) -- (f) -- (a); 	\draw (b) -- (g) -- (h) -- (e); 	\draw (a) -- (i) -- (j); 	\draw (d) to[out=0,in=0, distance=.5cm] (j); 	\coordinate (1) at  ($ (j) + (.75,.75) $); 	\coordinate (2) at  ($ (g)+ (-1.2, 1.2) $); 	\coordinate (3) at  ($ (c)+ (.75,.75) $); 	\coordinate (6) at  ($ (i)+ (-.75,.75) $); 	\coordinate (5) at  ($ (h)+ (1.2,-1.2) $); 	\coordinate (4) at  ($ (f) + (-.75,-.75) $); 	\draw (i) to node [label=6] {} (6); 	\draw (j) to node [label=1] {} (1); 	\draw (g) to node [label={[xshift=.2cm]2}] {} (2); 	\draw (c) to node [label={[xshift=-.1cm]3}] {} (3); 	\draw (h) to node [label={[xshift=.15cm, yshift=-.1cm]5}] {} (5); 	\draw (f) to node [label={[xshift=-.1cm]4}] {} (4);
	\draw [line width=.08cm] (6) to (i) to (j) to[out=0,in=0,distance=.5cm] (d) to (d) to (e) to (h) to (g) to (2); 	\draw [line width=.08cm] (4) to (f) to (a) to (b) to (c) to (3);
	\draw [dashed, line width=.08cm] (4) to[out=160,in=200,distance=1.3cm] node [label={[xshift=-.2cm,yshift=-.0cm]x}] {} (6); 	\draw [dashed, line width=.08cm] (3) to[out=60,in=130,distance=3cm] node [label={[xshift=-.2cm]y}] {} (2); \end{tikzpicture}
\par\end{center}

\begin{center}
\caption{The cycle from the previous figure, after expanding the gadget, is
now a cycle of length $x+y+12$.}

\par\end{center}%
\end{minipage}
\end{figure}
\begin{figure}[H]
\begin{minipage}[t]{0.45\columnwidth}%
\begin{center}
\tikzstyle{node}=[circle, draw, fill=black!50,                         inner sep=0pt, minimum width=4pt]
\begin{tikzpicture}[thick,scale=0.8] 	\coordinate (4) at (-2,1); 	\coordinate (3) at (-2,-1); 	\draw (4) to node [label={[xshift=.15cm,yshift=.3cm]4}] {} (3); 	\draw (4) to node [label={[xshift=.15cm,yshift=-.7cm]3}] {} (3); 	\coordinate (5) at (0,1); 	\coordinate (1) at (0,-1); 	\draw (5) to node [label={[xshift=.15cm,yshift=.3cm]5}] {} (1); 	\draw (5) to node [label={[xshift=.15cm,yshift=-.7cm]1}] {} (1); 	\coordinate (6) at (2,1); 	\coordinate (2) at (2,-1); 	\draw (6) to node [label={[xshift=.15cm,yshift=.3cm]6}] {} (2); 	\draw (6) to node [label={[xshift=.15cm,yshift=-.7cm]2}] {} (2);
	\draw [line width=.08cm] (5) to (1); 	\draw [dashed, line width=.08cm] (5) to[out=230,in=-140, distance=2.2cm] node [label={[xshift=-.4cm]x}] {} (2);
	\draw [line width=.08cm] (6) to (2); 	\draw [dashed, line width=.08cm] (1) to[out=130,in=140, distance=2.2cm] node [label={[xshift=-.0cm]y}] {} (6); \end{tikzpicture}
\par\end{center}

\begin{center}
\caption{A cycle of length $x+y+2$ passes through a $H'_{3}$.}

\par\end{center}%
\end{minipage}\hfill{}%
\begin{minipage}[t]{0.45\columnwidth}%
\begin{center}
\tikzstyle{node}=[circle, draw, fill=black!50,                         inner sep=0pt, minimum width=4pt]
\begin{tikzpicture}[thick,scale=.75] 	\node [node] (a) at (-1,.5) {}; 	\node [node] (b) at (0,1.2) {}; 	\node [node] (c) at (1,.5) {}; 	\node [node] (d) at (1,-.5) {}; 	\node [node] (e) at (0,-1.2) {}; 	\node [node] (f) at (-1,-.5) {}; 	\node [node] (g) at (0,.5) {}; 	\node [node] (h) at (0,-.5) {}; 	\node [node] (i) at  ($ (a)+ (.5,1.5) $) {}; 	\node [node] (j) at  ($ (d)+ (-.5,2.5) $) {};
	\draw (a) -- (b) -- (c) -- (d) -- (e) -- (f) -- (a); 	\draw (b) -- (g) -- (h) -- (e); 	\draw (a) -- (i) -- (j); 	\draw (d) to[out=0,in=0, distance=.5cm] (j); 	\coordinate (1) at  ($ (j) + (.75,.75) $); 	\coordinate (2) at  ($ (g)+ (-1.2, 1.2) $); 	\coordinate (3) at  ($ (c)+ (.75,.75) $); 	\coordinate (6) at  ($ (i)+ (-.75,.75) $); 	\coordinate (5) at  ($ (h)+ (1.2,-1.2) $); 	\coordinate (4) at  ($ (f) + (-.75,-.75) $); 	\draw (i) to node [label=6] {} (6); 	\draw (j) to node [label=1] {} (1); 	\draw (g) to node [label={[xshift=.2cm]2}] {} (2); 	\draw (c) to node [label={[xshift=-.05cm]3}] {} (3); 	\draw (h) to node [label={[xshift=.15cm, yshift=-.1cm]5}] {} (5); 	\draw (f) to node [label={[xshift=-.1cm]4}] {} (4);
	\draw [line width=.08cm] (1) to (j) to[out=0,in=0,distance=.5cm] (d) to (d) to (c) to (b) to (g) to (2); 	\draw [line width=.08cm] (6) to (i) to (a) to (f) to (e) to (h) to (5);
	\draw [dashed, line width=.08cm] (5) to[out=210,in=200,distance=3.1cm] node [label={[xshift=-.4cm,yshift=.0cm]x}] {} (2); 	\draw [dashed, line width=.08cm] (6) to[out=90,in=90,distance=.6cm] node [label={[xshift=-.0cm]y}] {} (1); \end{tikzpicture}
\par\end{center}

\begin{center}
\caption{The cycle from the previous figure, after expanding the gadget, is
now a cycle of length $x+y+12$.}

\par\end{center}%
\end{minipage}
\end{figure}
\begin{figure}[H]
\begin{minipage}[t]{0.45\columnwidth}%
\begin{center}
\tikzstyle{node}=[circle, draw, fill=black!50,                         inner sep=0pt, minimum width=4pt]
\begin{tikzpicture}[thick,scale=0.8] 	\coordinate (4) at (-2,1); 	\coordinate (3) at (-2,-1); 	\draw (4) to node [label={[xshift=.15cm,yshift=.3cm]4}] {} (3); 	\draw (4) to node [label={[xshift=.15cm,yshift=-.7cm]3}] {} (3); 	\coordinate (5) at (0,1); 	\coordinate (1) at (0,-1); 	\draw (5) to node [label={[xshift=.15cm,yshift=.3cm]5}] {} (1); 	\draw (5) to node [label={[xshift=.15cm,yshift=-.7cm]1}] {} (1); 	\coordinate (6) at (2,1); 	\coordinate (2) at (2,-1); 	\draw (6) to node [label={[xshift=.15cm,yshift=.3cm]6}] {} (2); 	\draw (6) to node [label={[xshift=.15cm,yshift=-.7cm]2}] {} (2);
	\draw [line width=.08cm] (5) to (1); 	\draw [dashed, line width=.08cm] (5) to[out=90,in=90, distance=1cm] node [label={[xshift=-.0cm]x}] {} (6);
	\draw [line width=.08cm] (6) to (2); 	\draw [dashed, line width=.08cm] (1) to[out=-90,in=-90, distance=1cm] node [label={[xshift=-.0cm]y}] {} (2); \end{tikzpicture}
\par\end{center}

\begin{center}
\caption{A cycle of length $x+y+2$ passes through a $H'_{3}$.}

\par\end{center}%
\end{minipage}\hfill{}%
\begin{minipage}[t]{0.45\columnwidth}%
\begin{center}
\tikzstyle{node}=[circle, draw, fill=black!50,                         inner sep=0pt, minimum width=4pt]
\begin{tikzpicture}[thick,scale=.8] 	\node [node] (a) at (-1,.5) {}; 	\node [node] (b) at (0,1.2) {}; 	\node [node] (c) at (1,.5) {}; 	\node [node] (d) at (1,-.5) {}; 	\node [node] (e) at (0,-1.2) {}; 	\node [node] (f) at (-1,-.5) {}; 	\node [node] (g) at (0,.5) {}; 	\node [node] (h) at (0,-.5) {}; 	\node [node] (i) at  ($ (a)+ (.5,1.5) $) {}; 	\node [node] (j) at  ($ (d)+ (-.5,2.5) $) {};
	\draw (a) -- (b) -- (c) -- (d) -- (e) -- (f) -- (a); 	\draw (b) -- (g) -- (h) -- (e); 	\draw (a) -- (i) -- (j); 	\draw (d) to[out=0,in=0, distance=.5cm] (j); 	\coordinate (1) at  ($ (j) + (.75,.75) $); 	\coordinate (2) at  ($ (g)+ (-1.2, 1.2) $); 	\coordinate (3) at  ($ (c)+ (.75,.75) $); 	\coordinate (6) at  ($ (i)+ (-.75,.75) $); 	\coordinate (5) at  ($ (h)+ (1.2,-1.2) $); 	\coordinate (4) at  ($ (f) + (-.75,-.75) $); 	\draw (i) to node [label=6] {} (6); 	\draw (j) to node [label=1] {} (1); 	\draw (g) to node [label={[xshift=.2cm]2}] {} (2); 	\draw (c) to node [label={[xshift=-.1cm]3}] {} (3); 	\draw (h) to node [label={[xshift=.15cm, yshift=-.1cm]5}] {} (5); 	\draw (f) to node [label={[xshift=-.1cm]4}] {} (4);
	\draw [line width=.08cm] (1) to (j) to[out=0,in=0,distance=.5cm] (d) to (d) to (c) to (b) to (g) to (2); 	\draw [line width=.08cm] (6) to (i) to (a) to (f) to (e) to (h) to (5);
	\draw [dashed, line width=.08cm] (5) to[out=210,in=220, distance=3.3cm] node [label={[xshift=-.3cm,yshift=.0cm]x}] {} (6); 	\draw [dashed, line width=.08cm] (1) to[out=140,in=150,distance=2cm] node [label={[xshift=-.2cm]y}] {} (2); \end{tikzpicture}
\par\end{center}

\begin{center}
\caption{The cycle from the previous figure, after expanding the gadget, is
now two cycles of lengths $x+6$ and $y+6$, respectively. Both $x$
and $y$ cannot have length $0$, otherwise this $H_{3}$ would be
part of a $H_{4}$. This expansion can produce an organic $6$-cycle
if either $x$ or $y$ is a $0$-length path and the cycle is organic.
The impact of these $6$-cycles on the algorithm's result is analyzed
in Sections \ref{expandh3} and \ref{analyzingworstcase}.}

\par\end{center}%
\end{minipage}
\end{figure}

\subsection{Three super-edges are covered by $2$-factor}

\subsubsection{Three cycles pass through gadget}

\begin{figure}[H]
\begin{minipage}[t]{0.45\columnwidth}%
\begin{center}
\tikzstyle{node}=[circle, draw, fill=black!50,                         inner sep=0pt, minimum width=4pt]
\begin{tikzpicture}[thick,scale=0.8] 	\coordinate (4) at (-2,1); 	\coordinate (3) at (-2,-1); 	\draw (4) to node [label={[xshift=.15cm,yshift=.3cm]4}] {} (3); 	\draw (4) to node [label={[xshift=.15cm,yshift=-.7cm]3}] {} (3); 	\coordinate (5) at (0,1); 	\coordinate (1) at (0,-1); 	\draw (5) to node [label={[xshift=.15cm,yshift=.3cm]5}] {} (1); 	\draw (5) to node [label={[xshift=.15cm,yshift=-.7cm]1}] {} (1); 	\coordinate (6) at (2,1); 	\coordinate (2) at (2,-1); 	\draw (6) to node [label={[xshift=.15cm,yshift=.3cm]6}] {} (2); 	\draw (6) to node [label={[xshift=.15cm,yshift=-.7cm]2}] {} (2);
	\draw [line width=.08cm] (4) to (3); 	\draw [dashed, line width=.08cm] (4) to[out=110,in=-110, distance=2cm] node [label={[xshift=-.2cm]x}] {} (3);
	\draw [line width=.08cm] (5) to (1); 	\draw [dashed, line width=.08cm] (5) to[out=110,in=-110, distance=2cm] node [label={[xshift=-.2cm]y}] {} (1);
	\draw [line width=.08cm] (6) to (2); 	\draw [dashed, line width=.08cm] (6) to[out=110,in=-110, distance=2cm] node [label={[xshift=-.2cm]z}] {} (2); \end{tikzpicture}
\par\end{center}

\begin{center}
\caption{Three cycles of lengths $x+1$, $y+1$, and $z+1$ pass through a
$H'_{3}$.}

\par\end{center}%
\end{minipage}\hfill{}%
\begin{minipage}[t]{0.45\columnwidth}%
\begin{center}
\tikzstyle{node}=[circle, draw, fill=black!50,                         inner sep=0pt, minimum width=4pt]
\begin{tikzpicture}[thick,scale=.75] 	\node [node] (a) at (-1,.5) {}; 	\node [node] (b) at (0,1.2) {}; 	\node [node] (c) at (1,.5) {}; 	\node [node] (d) at (1,-.5) {}; 	\node [node] (e) at (0,-1.2) {}; 	\node [node] (f) at (-1,-.5) {}; 	\node [node] (g) at (0,.5) {}; 	\node [node] (h) at (0,-.5) {}; 	\node [node] (i) at  ($ (a)+ (.5,1.5) $) {}; 	\node [node] (j) at  ($ (d)+ (-.5,2.5) $) {};
	\draw (a) -- (b) -- (c) -- (d) -- (e) -- (f) -- (a); 	\draw (b) -- (g) -- (h) -- (e); 	\draw (a) -- (i) -- (j); 	\draw (d) to[out=0,in=0, distance=.5cm] (j); 	\coordinate (1) at  ($ (j) + (.75,.75) $); 	\coordinate (2) at  ($ (g)+ (-1.2, 1.2) $); 	\coordinate (3) at  ($ (c)+ (.75,.75) $); 	\coordinate (6) at  ($ (i)+ (-.75,.75) $); 	\coordinate (5) at  ($ (h)+ (1.2,-1.2) $); 	\coordinate (4) at  ($ (f) + (-.75,-.75) $); 	\draw (i) to node [label=6] {} (6); 	\draw (j) to node [label=1] {} (1); 	\draw (g) to node [label={[xshift=.2cm]2}] {} (2); 	\draw (c) to node [label={[xshift=-.1cm]3}] {} (3); 	\draw (h) to node [label={[xshift=.15cm, yshift=-.1cm]5}] {} (5); 	\draw (f) to node [label={[xshift=-.1cm]4}] {} (4);
	\draw [line width=.08cm] (1) to (j) to (i) to (6); 	\draw [line width=.08cm] (4) to (f) to (a) to (b) to (g) to (2); 	\draw [line width=.08cm] (3) to (c) to (d) to (e) to (h) to (5);
	\draw [dashed, line width=.08cm] (4) to[out=310,in=-40,distance=3.2cm] node [label={[xshift=.4cm,yshift=-.2cm]x}] {} (3); 	\draw [dashed, line width=.08cm] (5) to[out=30,in=-30,distance=1.3cm] node [label={[xshift=.2cm,yshift=.5cm]y}] {} (1); 	\draw [dashed, line width=.08cm] (6) to[out=150,in=210,distance=.5cm] node [label={[xshift=-.2cm]z}] {} (2); \end{tikzpicture}
\par\end{center}

\begin{center}
\caption{The cycles from the previous figure, after expanding the gadget, are
now a single cycle of length $x+y+z+13$.}

\par\end{center}%
\end{minipage}
\end{figure}

\subsubsection{Two cycles pass through gadget}

If a $H'_{3}$ is covered by two disjoint cycles in $F_{i}$, then
two of the super-edges must be part of the same cycle, with the third
super-edge in a different cycle. To ensure we examine every case,
we examine all three ways we can select the super-edge to appear in
a separate cycle and then within each of these arrangements, we examine
both orientations in which the two super-edges in the same cycle can
be connected. In all cases we start with cycles of lengths $x+y+2$
and $z+1$ in $F_{i}$ and are returned a single cycle of length $x+y+z+13$
in $F_{i-1}$. $x+y+2\geq6$ and $z+1\geq6$, so $x+y+z+13>8$, meaning
the resulting cycle in $F_{i-1}$ cannot be a cycle of length $6$.

\begin{figure}[H]
\begin{minipage}[t]{0.45\columnwidth}%
\begin{center}
\tikzstyle{node}=[circle, draw, fill=black!50,                         inner sep=0pt, minimum width=4pt]
\begin{tikzpicture}[thick,scale=0.8] 	\coordinate (4) at (-2,1); 	\coordinate (3) at (-2,-1); 	\draw (4) to node [label={[xshift=.15cm,yshift=.3cm]4}] {} (3); 	\draw (4) to node [label={[xshift=.15cm,yshift=-.7cm]3}] {} (3); 	\coordinate (5) at (0,1); 	\coordinate (1) at (0,-1); 	\draw (5) to node [label={[xshift=.15cm,yshift=.3cm]5}] {} (1); 	\draw (5) to node [label={[xshift=.15cm,yshift=-.7cm]1}] {} (1); 	\coordinate (6) at (2,1); 	\coordinate (2) at (2,-1); 	\draw (6) to node [label={[xshift=.15cm,yshift=.3cm]6}] {} (2); 	\draw (6) to node [label={[xshift=.15cm,yshift=-.7cm]2}] {} (2);
	\draw [line width=.08cm] (4) to (3); 	\draw [dashed, line width=.08cm] (4) to[out=230,in=-140, distance=2.2cm] node [label={[xshift=-.4cm]x}] {} (1);
	\draw [line width=.08cm] (5) to (1); 	\draw [dashed, line width=.08cm] (3) to[out=130,in=140, distance=2.2cm] node [label={[xshift=-.0cm]y}] {} (5);
	\draw [line width=.08cm] (6) to (2); 	\draw [dashed, line width=.08cm] (6) to[out=110,in=-110, distance=2cm] node [label={[xshift=-.2cm]z}] {} (2); \end{tikzpicture}
\par\end{center}

\begin{center}
\caption{Two cycles of lengths $x+y+2$ and $z+1$ pass through a $H'_{3}$.}

\par\end{center}%
\end{minipage}\hfill{}%
\begin{minipage}[t]{0.45\columnwidth}%
\begin{center}
\tikzstyle{node}=[circle, draw, fill=black!50,                         inner sep=0pt, minimum width=4pt]
\begin{tikzpicture}[thick,scale=.75] 	\node [node] (a) at (-1,.5) {}; 	\node [node] (b) at (0,1.2) {}; 	\node [node] (c) at (1,.5) {}; 	\node [node] (d) at (1,-.5) {}; 	\node [node] (e) at (0,-1.2) {}; 	\node [node] (f) at (-1,-.5) {}; 	\node [node] (g) at (0,.5) {}; 	\node [node] (h) at (0,-.5) {}; 	\node [node] (i) at  ($ (a)+ (.5,1.5) $) {}; 	\node [node] (j) at  ($ (d)+ (-.5,2.5) $) {};
	\draw (a) -- (b) -- (c) -- (d) -- (e) -- (f) -- (a); 	\draw (b) -- (g) -- (h) -- (e); 	\draw (a) -- (i) -- (j); 	\draw (d) to[out=0,in=0, distance=.5cm] (j); 	\coordinate (1) at  ($ (j) + (.75,.75) $); 	\coordinate (2) at  ($ (g)+ (-1.2, 1.2) $); 	\coordinate (3) at  ($ (c)+ (.75,.75) $); 	\coordinate (6) at  ($ (i)+ (-.75,.75) $); 	\coordinate (5) at  ($ (h)+ (1.2,-1.2) $); 	\coordinate (4) at  ($ (f) + (-.75,-.75) $); 	\draw (i) to node [label=6] {} (6); 	\draw (j) to node [label=1] {} (1); 	\draw (g) to node [label={[xshift=.2cm]2}] {} (2); 	\draw (c) to node [label={[xshift=-.05cm]3}] {} (3); 	\draw (h) to node [label={[xshift=.15cm, yshift=-.1cm]5}] {} (5); 	\draw (f) to node [label={[xshift=-.1cm]4}] {} (4);
	\draw [line width=.08cm] (1) to (j) to[out=0,in=0,distance=.5cm] (d) to (d) to (e) to (h) to (5); 	\draw [line width=.08cm] (4) to (f) to (a) to (i) to (6); 	\draw [line width=.08cm] (3) to (c) to (b) to (g) to (2);
	\draw [dashed, line width=.08cm] (4) to[out=130,in=120,distance=3.7cm] node [label={[xshift=-.1cm,yshift=-.0cm]x}] {} (1); 	\draw [dashed, line width=.08cm] (5) to[out=30,in=-30,distance=1.3cm] node [label={[xshift=.2cm,yshift=.5cm]y}] {} (3);
	\draw [dashed, line width=.08cm] (6) to[out=150,in=210,distance=.5cm] node [label={[xshift=-.2cm]z}] {} (2); \end{tikzpicture}
\par\end{center}

\begin{center}
\caption{The cycles from the previous figure, after expanding the gadget, are
now a single cycle of length $x+y+z+13$.}

\par\end{center}%
\end{minipage}
\end{figure}
\begin{figure}[H]
\begin{minipage}[t]{0.45\columnwidth}%
\begin{center}
\tikzstyle{node}=[circle, draw, fill=black!50,                         inner sep=0pt, minimum width=4pt]
\begin{tikzpicture}[thick,scale=0.8] 	\coordinate (4) at (-2,1); 	\coordinate (3) at (-2,-1); 	\draw (4) to node [label={[xshift=.15cm,yshift=.3cm]4}] {} (3); 	\draw (4) to node [label={[xshift=.15cm,yshift=-.7cm]3}] {} (3); 	\coordinate (5) at (0,1); 	\coordinate (1) at (0,-1); 	\draw (5) to node [label={[xshift=.15cm,yshift=.3cm]5}] {} (1); 	\draw (5) to node [label={[xshift=.15cm,yshift=-.7cm]1}] {} (1); 	\coordinate (6) at (2,1); 	\coordinate (2) at (2,-1); 	\draw (6) to node [label={[xshift=.15cm,yshift=.3cm]6}] {} (2); 	\draw (6) to node [label={[xshift=.15cm,yshift=-.7cm]2}] {} (2);
	\draw [line width=.08cm] (4) to (3); 	\draw [dashed, line width=.08cm] (4) to[out=230,in=-140, distance=3cm] node [label={[xshift=-.2cm,yshift=-.5cm]x}] {} (2);
	\draw [line width=.08cm] (6) to (2); 	\draw [dashed, line width=.08cm] (3) to[out=130,in=130, distance=3cm] node [label={[xshift=-.2cm]y}] {} (6);
	\draw [line width=.08cm] (5) to (1); 	\draw [dashed, line width=.08cm] (5) to[out=110,in=-110, distance=2cm] node [label={[xshift=-.2cm]z}] {} (1); \end{tikzpicture}
\par\end{center}

\begin{center}
\caption{Two cycles of lengths $x+y+2$ and $z+1$ pass through a $H'_{3}$.}

\par\end{center}%
\end{minipage}\hfill{}%
\begin{minipage}[t]{0.45\columnwidth}%
\begin{center}
\tikzstyle{node}=[circle, draw, fill=black!50,                         inner sep=0pt, minimum width=4pt]
\begin{tikzpicture}[thick,scale=.75] 	\node [node] (a) at (-1,.5) {}; 	\node [node] (b) at (0,1.2) {}; 	\node [node] (c) at (1,.5) {}; 	\node [node] (d) at (1,-.5) {}; 	\node [node] (e) at (0,-1.2) {}; 	\node [node] (f) at (-1,-.5) {}; 	\node [node] (g) at (0,.5) {}; 	\node [node] (h) at (0,-.5) {}; 	\node [node] (i) at  ($ (a)+ (.5,1.5) $) {}; 	\node [node] (j) at  ($ (d)+ (-.5,2.5) $) {};
	\draw (a) -- (b) -- (c) -- (d) -- (e) -- (f) -- (a); 	\draw (b) -- (g) -- (h) -- (e); 	\draw (a) -- (i) -- (j); 	\draw (d) to[out=0,in=0, distance=.5cm] (j); 	\coordinate (1) at  ($ (j) + (.75,.75) $); 	\coordinate (2) at  ($ (g)+ (-1.2, 1.2) $); 	\coordinate (3) at  ($ (c)+ (.75,.75) $); 	\coordinate (6) at  ($ (i)+ (-.75,.75) $); 	\coordinate (5) at  ($ (h)+ (1.2,-1.2) $); 	\coordinate (4) at  ($ (f) + (-.75,-.75) $); 	\draw (i) to node [label=6] {} (6); 	\draw (j) to node [label=1] {} (1); 	\draw (g) to node [label={[xshift=.2cm]2}] {} (2); 	\draw (c) to node [label={[xshift=-.05cm]3}] {} (3); 	\draw (h) to node [label={[xshift=.15cm, yshift=-.1cm]5}] {} (5); 	\draw (f) to node [label={[xshift=-.1cm]4}] {} (4);
	\draw [line width=.08cm] (1) to (j) to[out=0,in=0,distance=.5cm] (d) to (d) to (c) to (3); 	\draw [line width=.08cm] (4) to (f) to (e) to (h) to (5); 	\draw [line width=.08cm] (6) to (i) to (a) to (b) to (g) to (2);
	\draw [dashed, line width=.08cm] (4) to[out=160,in=200,distance=1cm] node [label={[xshift=-.2cm,yshift=-.0cm]x}] {} (2); 	\draw [dashed, line width=.08cm] (6) to[out=50,in=30,distance=2cm] node [label={[xshift=-.2cm,yshift=.1cm]y}] {} (3); 	\draw [dashed, line width=.08cm] (5) to[out=30,in=-30,distance=1.3cm] node [label={[xshift=.2cm]z}] {} (1); \end{tikzpicture}
\par\end{center}

\begin{center}
\caption{The cycles from the previous figure, after expanding the gadget, are
now a single cycle of length $x+y+z+13$.}

\par\end{center}%
\end{minipage}
\end{figure}
\begin{figure}[H]
\begin{minipage}[t]{0.45\columnwidth}%
\begin{center}
\tikzstyle{node}=[circle, draw, fill=black!50,                         inner sep=0pt, minimum width=4pt]
\begin{tikzpicture}[thick,scale=0.8] 	\coordinate (4) at (-2,1); 	\coordinate (3) at (-2,-1); 	\draw (4) to node [label={[xshift=.15cm,yshift=.3cm]4}] {} (3); 	\draw (4) to node [label={[xshift=.15cm,yshift=-.7cm]3}] {} (3); 	\coordinate (5) at (0,1); 	\coordinate (1) at (0,-1); 	\draw (5) to node [label={[xshift=.15cm,yshift=.3cm]5}] {} (1); 	\draw (5) to node [label={[xshift=.15cm,yshift=-.7cm]1}] {} (1); 	\coordinate (6) at (2,1); 	\coordinate (2) at (2,-1); 	\draw (6) to node [label={[xshift=.15cm,yshift=.3cm]6}] {} (2); 	\draw (6) to node [label={[xshift=.15cm,yshift=-.7cm]2}] {} (2);
	\draw [line width=.08cm] (4) to (3); 	\draw [dashed, line width=.08cm] (4) to[out=90,in=90, distance=1cm] node [label={[xshift=-.0cm]x}] {} (5);
	\draw [line width=.08cm] (5) to (1); 	\draw [dashed, line width=.08cm] (3) to[out=-90,in=-90, distance=1cm] node [label={[xshift=-.0cm]y}] {} (1);
	\draw [line width=.08cm] (6) to (2); 	\draw [dashed, line width=.08cm] (6) to[out=110,in=-110, distance=2cm] node [label={[xshift=-.2cm]z}] {} (2); \end{tikzpicture}
\par\end{center}

\begin{center}
\caption{Two cycles of lengths $x+y+2$ and $z+1$ pass through a $H'_{3}$.}

\par\end{center}%
\end{minipage}\hfill{}%
\begin{minipage}[t]{0.45\columnwidth}%
\begin{center}
\tikzstyle{node}=[circle, draw, fill=black!50,                         inner sep=0pt, minimum width=4pt]
\begin{tikzpicture}[thick,scale=.75] 	\node [node] (a) at (-1,.5) {}; 	\node [node] (b) at (0,1.2) {}; 	\node [node] (c) at (1,.5) {}; 	\node [node] (d) at (1,-.5) {}; 	\node [node] (e) at (0,-1.2) {}; 	\node [node] (f) at (-1,-.5) {}; 	\node [node] (g) at (0,.5) {}; 	\node [node] (h) at (0,-.5) {}; 	\node [node] (i) at  ($ (a)+ (.5,1.5) $) {}; 	\node [node] (j) at  ($ (d)+ (-.5,2.5) $) {};
	\draw (a) -- (b) -- (c) -- (d) -- (e) -- (f) -- (a); 	\draw (b) -- (g) -- (h) -- (e); 	\draw (a) -- (i) -- (j); 	\draw (d) to[out=0,in=0, distance=.5cm] (j); 	\coordinate (1) at  ($ (j) + (.75,.75) $); 	\coordinate (2) at  ($ (g)+ (-1.2, 1.2) $); 	\coordinate (3) at  ($ (c)+ (.75,.75) $); 	\coordinate (6) at  ($ (i)+ (-.75,.75) $); 	\coordinate (5) at  ($ (h)+ (1.2,-1.2) $); 	\coordinate (4) at  ($ (f) + (-.75,-.75) $); 	\draw (i) to node [label=6] {} (6); 	\draw (j) to node [label=1] {} (1); 	\draw (g) to node [label={[xshift=.2cm]2}] {} (2); 	\draw (c) to node [label={[xshift=-.05cm]3}] {} (3); 	\draw (h) to node [label={[xshift=.15cm, yshift=-.1cm]5}] {} (5); 	\draw (f) to node [label={[xshift=-.1cm]4}] {} (4);
	\draw [line width=.08cm] (1) to (j) to[out=0,in=0,distance=.5cm] (d) to (d) to (e) to (h) to (5); 	\draw [line width=.08cm] (4) to (f) to (a) to (i) to (6); 	\draw [line width=.08cm] (2) to (g) to (b) to (c) to (3);
	\draw [dashed, line width=.08cm] (4) to[out=260,in=-80,distance=1cm] node [label={[xshift=.0cm,yshift=-.5cm]x}] {} (5); 	\draw [dashed, line width=.08cm] (3) to[out=30,in=-30,distance=.5cm] node [label={[xshift=.2cm,yshift=.0cm]y}] {} (1); 	\draw [dashed, line width=.08cm] (6) to[out=150,in=210,distance=.5cm] node [label={[xshift=-.2cm]z}] {} (2); \end{tikzpicture}
\par\end{center}

\begin{center}
\caption{The cycles from the previous figure, after expanding the gadget, are
now a single cycle of length $x+y+z+13$.}

\par\end{center}%
\end{minipage}
\end{figure}
\begin{figure}[H]
\begin{minipage}[t]{0.45\columnwidth}%
\begin{center}
\tikzstyle{node}=[circle, draw, fill=black!50,                         inner sep=0pt, minimum width=4pt]
\begin{tikzpicture}[thick,scale=0.8] 	\coordinate (4) at (-2,1); 	\coordinate (3) at (-2,-1); 	\draw (4) to node [label={[xshift=.15cm,yshift=.3cm]4}] {} (3); 	\draw (4) to node [label={[xshift=.15cm,yshift=-.7cm]3}] {} (3); 	\coordinate (5) at (0,1); 	\coordinate (1) at (0,-1); 	\draw (5) to node [label={[xshift=.15cm,yshift=.3cm]5}] {} (1); 	\draw (5) to node [label={[xshift=.15cm,yshift=-.7cm]1}] {} (1); 	\coordinate (6) at (2,1); 	\coordinate (2) at (2,-1); 	\draw (6) to node [label={[xshift=.15cm,yshift=.3cm]6}] {} (2); 	\draw (6) to node [label={[xshift=.15cm,yshift=-.7cm]2}] {} (2);
	\draw [line width=.08cm] (4) to (3); 	\draw [dashed, line width=.08cm] (4) to[out=90,in=90, distance=1cm] node [label={[xshift=-.0cm]x}] {} (6);
	\draw [line width=.08cm] (6) to (2); 	\draw [dashed, line width=.08cm] (3) to[out=-90,in=-90, distance=1cm] node [label={[xshift=-.0cm,yshift=-.5cm]y}] {} (2);
	\draw [line width=.08cm] (5) to (1); 	\draw [dashed, line width=.08cm] (5) to[out=110,in=-110, distance=2cm] node [label={[xshift=-.2cm]z}] {} (1); \end{tikzpicture}
\par\end{center}

\begin{center}
\caption{Two cycles of lengths $x+y+2$ and $z+1$ pass through a $H'_{3}$.}

\par\end{center}%
\end{minipage}\hfill{}%
\begin{minipage}[t]{0.45\columnwidth}%
\begin{center}
\tikzstyle{node}=[circle, draw, fill=black!50,                         inner sep=0pt, minimum width=4pt]
\begin{tikzpicture}[thick,scale=.75] 	\node [node] (a) at (-1,.5) {}; 	\node [node] (b) at (0,1.2) {}; 	\node [node] (c) at (1,.5) {}; 	\node [node] (d) at (1,-.5) {}; 	\node [node] (e) at (0,-1.2) {}; 	\node [node] (f) at (-1,-.5) {}; 	\node [node] (g) at (0,.5) {}; 	\node [node] (h) at (0,-.5) {}; 	\node [node] (i) at  ($ (a)+ (.5,1.5) $) {}; 	\node [node] (j) at  ($ (d)+ (-.5,2.5) $) {};
	\draw (a) -- (b) -- (c) -- (d) -- (e) -- (f) -- (a); 	\draw (b) -- (g) -- (h) -- (e); 	\draw (a) -- (i) -- (j); 	\draw (d) to[out=0,in=0, distance=.5cm] (j); 	\coordinate (1) at  ($ (j) + (.75,.75) $); 	\coordinate (2) at  ($ (g)+ (-1.2, 1.2) $); 	\coordinate (3) at  ($ (c)+ (.75,.75) $); 	\coordinate (6) at  ($ (i)+ (-.75,.75) $); 	\coordinate (5) at  ($ (h)+ (1.2,-1.2) $); 	\coordinate (4) at  ($ (f) + (-.75,-.75) $); 	\draw (i) to node [label=6] {} (6); 	\draw (j) to node [label=1] {} (1); 	\draw (g) to node [label={[xshift=.2cm]2}] {} (2); 	\draw (c) to node [label={[xshift=-.05cm]3}] {} (3); 	\draw (h) to node [label={[xshift=.15cm, yshift=-.1cm]5}] {} (5); 	\draw (f) to node [label={[xshift=-.1cm]4}] {} (4);
	\draw [line width=.08cm] (6) to (i) to (j) to (1); 	\draw [line width=.08cm] (4) to (f) to (a) to (b) to (g) to (2); 	\draw [line width=.08cm] (5) to (h) to (e) to (d) to (c) to (3);
	\draw [dashed, line width=.08cm] (4) to[out=160,in=200,distance=1.3cm] node [label={[xshift=-.2cm,yshift=-.0cm]x}] {} (6); 	\draw [dashed, line width=.08cm] (3) to[out=60,in=130,distance=3cm] node [label={[xshift=-.2cm]y}] {} (2); 	\draw [dashed, line width=.08cm] (5) to[out=30,in=-30,distance=1.3cm] node [label={[xshift=.2cm]z}] {} (1); \end{tikzpicture}
\par\end{center}

\begin{center}
\caption{The cycles from the previous figure, after expanding the gadget, are
now a single cycle of length $x+y+z+13$.}

\par\end{center}%
\end{minipage}
\end{figure}
\begin{figure}[H]
\begin{minipage}[t]{0.45\columnwidth}%
\begin{center}
\tikzstyle{node}=[circle, draw, fill=black!50,                         inner sep=0pt, minimum width=4pt]
\begin{tikzpicture}[thick,scale=0.8] 	\coordinate (4) at (-2,1); 	\coordinate (3) at (-2,-1); 	\draw (4) to node [label={[xshift=.15cm,yshift=.3cm]4}] {} (3); 	\draw (4) to node [label={[xshift=.15cm,yshift=-.7cm]3}] {} (3); 	\coordinate (5) at (0,1); 	\coordinate (1) at (0,-1); 	\draw (5) to node [label={[xshift=.15cm,yshift=.3cm]5}] {} (1); 	\draw (5) to node [label={[xshift=.15cm,yshift=-.7cm]1}] {} (1); 	\coordinate (6) at (2,1); 	\coordinate (2) at (2,-1); 	\draw (6) to node [label={[xshift=.15cm,yshift=.3cm]6}] {} (2); 	\draw (6) to node [label={[xshift=.15cm,yshift=-.7cm]2}] {} (2);
	\draw [line width=.08cm] (5) to (1); 	\draw [dashed, line width=.08cm] (5) to[out=230,in=-140, distance=2.2cm] node [label={[xshift=-.4cm]x}] {} (2);
	\draw [line width=.08cm] (6) to (2); 	\draw [dashed, line width=.08cm] (1) to[out=130,in=140, distance=2.2cm] node [label={[xshift=-.0cm]y}] {} (6);
	\draw [line width=.08cm] (4) to (3); 	\draw [dashed, line width=.08cm] (4) to[out=110,in=-110, distance=2cm] node [label={[xshift=-.2cm]z}] {} (3); \end{tikzpicture}
\par\end{center}

\begin{center}
\caption{Two cycles of lengths $x+y+2$ and $z+1$ pass through a $H'_{3}$.}

\par\end{center}%
\end{minipage}\hfill{}%
\begin{minipage}[t]{0.45\columnwidth}%
\begin{center}
\tikzstyle{node}=[circle, draw, fill=black!50,                         inner sep=0pt, minimum width=4pt]
\begin{tikzpicture}[thick,scale=.75] 	\node [node] (a) at (-1,.5) {}; 	\node [node] (b) at (0,1.2) {}; 	\node [node] (c) at (1,.5) {}; 	\node [node] (d) at (1,-.5) {}; 	\node [node] (e) at (0,-1.2) {}; 	\node [node] (f) at (-1,-.5) {}; 	\node [node] (g) at (0,.5) {}; 	\node [node] (h) at (0,-.5) {}; 	\node [node] (i) at  ($ (a)+ (.5,1.5) $) {}; 	\node [node] (j) at  ($ (d)+ (-.5,2.5) $) {};
	\draw (a) -- (b) -- (c) -- (d) -- (e) -- (f) -- (a); 	\draw (b) -- (g) -- (h) -- (e); 	\draw (a) -- (i) -- (j); 	\draw (d) to[out=0,in=0, distance=.5cm] (j); 	\coordinate (1) at  ($ (j) + (.75,.75) $); 	\coordinate (2) at  ($ (g)+ (-1.2, 1.2) $); 	\coordinate (3) at  ($ (c)+ (.75,.75) $); 	\coordinate (6) at  ($ (i)+ (-.75,.75) $); 	\coordinate (5) at  ($ (h)+ (1.2,-1.2) $); 	\coordinate (4) at  ($ (f) + (-.75,-.75) $); 	\draw (i) to node [label=6] {} (6); 	\draw (j) to node [label=1] {} (1); 	\draw (g) to node [label={[xshift=.2cm]2}] {} (2); 	\draw (c) to node [label={[xshift=-.05cm]3}] {} (3); 	\draw (h) to node [label={[xshift=.15cm, yshift=-.1cm]5}] {} (5); 	\draw (f) to node [label={[xshift=-.1cm]4}] {} (4);
	\draw [line width=.08cm] (1) to (j) to[out=0,in=0,distance=.5cm] (d) to (d) to (e) to (h) to (5); 	\draw [line width=.08cm] (6) to (i) to (a) to (f) to (4); 	\draw [line width=.08cm] (2) to (g) to (b) to (c) to (3);
	\draw [dashed, line width=.08cm] (5) to[out=225,in=200,distance=3.7cm] node [label={[xshift=-.4cm,yshift=.0cm]x}] {} (2); 	\draw [dashed, line width=.08cm] (6) to[out=90,in=90,distance=.6cm] node [label={[xshift=-.0cm]y}] {} (1); 	\draw [dashed, line width=.08cm] (4) to[out=280,in=-30,distance=3.7cm] node [label={[xshift=.4cm,yshift=-.2cm]z}] {} (3); \end{tikzpicture}
\par\end{center}

\begin{center}
\caption{The cycles from the previous figure, after expanding the gadget, are
now a single cycle of length $x+y+z+13$.}

\par\end{center}%
\end{minipage}
\end{figure}
\begin{figure}[H]
\begin{minipage}[t]{0.45\columnwidth}%
\begin{center}
\tikzstyle{node}=[circle, draw, fill=black!50,                         inner sep=0pt, minimum width=4pt]
\begin{tikzpicture}[thick,scale=0.8] 	\coordinate (4) at (-2,1); 	\coordinate (3) at (-2,-1); 	\draw (4) to node [label={[xshift=.15cm,yshift=.3cm]4}] {} (3); 	\draw (4) to node [label={[xshift=.15cm,yshift=-.7cm]3}] {} (3); 	\coordinate (5) at (0,1); 	\coordinate (1) at (0,-1); 	\draw (5) to node [label={[xshift=.15cm,yshift=.3cm]5}] {} (1); 	\draw (5) to node [label={[xshift=.15cm,yshift=-.7cm]1}] {} (1); 	\coordinate (6) at (2,1); 	\coordinate (2) at (2,-1); 	\draw (6) to node [label={[xshift=.15cm,yshift=.3cm]6}] {} (2); 	\draw (6) to node [label={[xshift=.15cm,yshift=-.7cm]2}] {} (2);
	\draw [line width=.08cm] (5) to (1); 	\draw [dashed, line width=.08cm] (5) to[out=90,in=90, distance=1cm] node [label={[xshift=-.0cm]x}] {} (6);
	\draw [line width=.08cm] (6) to (2); 	\draw [dashed, line width=.08cm] (1) to[out=-90,in=-90, distance=1cm] node [label={[xshift=-.0cm]y}] {} (2);
	\draw [line width=.08cm] (4) to (3); 	\draw [dashed, line width=.08cm] (4) to[out=110,in=-110, distance=2cm] node [label={[xshift=-.2cm]z}] {} (3); \end{tikzpicture}
\par\end{center}

\begin{center}
\caption{Two cycles of lengths $x+y+2$ and $z+1$ pass through a $H'_{3}$.}

\par\end{center}%
\end{minipage}\hfill{}%
\begin{minipage}[t]{0.45\columnwidth}%
\begin{center}
\tikzstyle{node}=[circle, draw, fill=black!50,                         inner sep=0pt, minimum width=4pt]
\begin{tikzpicture}[thick,scale=.75] 	\node [node] (a) at (-1,.5) {}; 	\node [node] (b) at (0,1.2) {}; 	\node [node] (c) at (1,.5) {}; 	\node [node] (d) at (1,-.5) {}; 	\node [node] (e) at (0,-1.2) {}; 	\node [node] (f) at (-1,-.5) {}; 	\node [node] (g) at (0,.5) {}; 	\node [node] (h) at (0,-.5) {}; 	\node [node] (i) at  ($ (a)+ (.5,1.5) $) {}; 	\node [node] (j) at  ($ (d)+ (-.5,2.5) $) {};
	\draw (a) -- (b) -- (c) -- (d) -- (e) -- (f) -- (a); 	\draw (b) -- (g) -- (h) -- (e); 	\draw (a) -- (i) -- (j); 	\draw (d) to[out=0,in=0, distance=.5cm] (j); 	\coordinate (1) at  ($ (j) + (.75,.75) $); 	\coordinate (2) at  ($ (g)+ (-1.2, 1.2) $); 	\coordinate (3) at  ($ (c)+ (.75,.75) $); 	\coordinate (6) at  ($ (i)+ (-.75,.75) $); 	\coordinate (5) at  ($ (h)+ (1.2,-1.2) $); 	\coordinate (4) at  ($ (f) + (-.75,-.75) $); 	\draw (i) to node [label=6] {} (6); 	\draw (j) to node [label=1] {} (1); 	\draw (g) to node [label={[xshift=.2cm]2}] {} (2); 	\draw (c) to node [label={[xshift=-.05cm]3}] {} (3); 	\draw (h) to node [label={[xshift=.15cm, yshift=-.1cm]5}] {} (5); 	\draw (f) to node [label={[xshift=-.1cm]4}] {} (4);
	\draw [line width=.08cm] (6) to (i) to (j) to (1); 	\draw [line width=.08cm] (4) to (f) to (a) to (b) to (g) to (2); 	\draw [line width=.08cm] (5) to (h) to (e) to (d) to (c) to (3);
	\draw [dashed, line width=.08cm] (5) to[out=220,in=220, distance=3.7cm] node [label={[xshift=-.3cm,yshift=.0cm]x}] {} (6); 	\draw [dashed, line width=.08cm] (1) to[out=140,in=150,distance=2cm] node [label={[xshift=-.2cm]y}] {} (2); 	\draw [dashed, line width=.08cm] (4) to[out=300,in=-30,distance=3.5cm] node [label={[xshift=.4cm,yshift=-.2cm]z}] {} (3); \end{tikzpicture}
\par\end{center}

\begin{center}
\caption{The cycles from the previous figure, after expanding the gadget, are
now a single cycle of length $x+y+z+13$.}

\par\end{center}%
\end{minipage}
\end{figure}

\subsubsection{One cycle passes through gadget}

If a $H'_{3}$ is covered by a single cycle in $F_{i}$, then without
loss of generality, we can assume the cycle passes through the (4,3)
super-edge first. To ensure we examine every case, we examine all
four sides of the two remaining super-edges the $2$-factor could
enter after exiting edge 4. Then, after exiting the other side of
this second super-edge, we consider the two orientations in which
the $2$-factor could enter the third super-edge. In total, this gives
us eight cases to consider. In all cases we start with a cycle of
lengths $x+y+z+3$ in $F_{i}$ and are returned a single cycle of
length $x+y+z+13$ in $F_{i-1}$. $x+y+z+3\geq6$, so $x+y+z+13>8$,
meaning the resulting cycle in $F_{i-1}$ cannot be a cycle of length
$6$.

\begin{figure}[H]
\begin{minipage}[t]{0.45\columnwidth}%
\begin{center}
\tikzstyle{node}=[circle, draw, fill=black!50,                         inner sep=0pt, minimum width=4pt]
\begin{tikzpicture}[thick,scale=0.8] 	\coordinate (4) at (-2,1); 	\coordinate (3) at (-2,-1); 	\draw (4) to node [label={[xshift=.15cm,yshift=.3cm]4}] {} (3); 	\draw (4) to node [label={[xshift=.15cm,yshift=-.7cm]3}] {} (3); 	\coordinate (5) at (0,1); 	\coordinate (1) at (0,-1); 	\draw (5) to node [label={[xshift=.15cm,yshift=.3cm]5}] {} (1); 	\draw (5) to node [label={[xshift=.15cm,yshift=-.7cm]1}] {} (1); 	\coordinate (6) at (2,1); 	\coordinate (2) at (2,-1); 	\draw (6) to node [label={[xshift=.15cm,yshift=.3cm]6}] {} (2); 	\draw (6) to node [label={[xshift=.15cm,yshift=-.7cm]2}] {} (2);
	\draw [line width=.08cm] (4) to (3); 	\draw [line width=.08cm] (5) to (1); 	\draw [line width=.08cm] (6) to (2); 
	\draw [dashed,  line width=.08cm] (3) to[out=-90,in=-90] node [label={[xshift=-.0cm]x}] {} (1); 	\draw [dashed,  line width=.08cm] (5) to[out=0,in=180] node [label={[xshift=.1cm,yshift=-.0cm]y}] {} (2); 	\draw [dashed,  line width=.08cm] (6) to[out=90,in=90, distance=1cm] node [label=z] {} (4); 
\end{tikzpicture}
\par\end{center}

\begin{center}
\caption{A cycle of length $x+y+z+3$ passes through a $H'_{3}$.}

\par\end{center}%
\end{minipage}\hfill{}%
\begin{minipage}[t]{0.45\columnwidth}%
\begin{center}
\tikzstyle{node}=[circle, draw, fill=black!50,                         inner sep=0pt, minimum width=4pt]
\begin{tikzpicture}[thick,scale=.75]
\clip (-3.5,-3) rectangle (3,3);
\node [node] (a) at (-1,.5) {}; 	\node [node] (b) at (0,1.2) {}; 	\node [node] (c) at (1,.5) {}; 	\node [node] (d) at (1,-.5) {}; 	\node [node] (e) at (0,-1.2) {}; 	\node [node] (f) at (-1,-.5) {}; 	\node [node] (g) at (0,.5) {}; 	\node [node] (h) at (0,-.5) {}; 	\node [node] (i) at  ($ (a)+ (.5,1.5) $) {}; 	\node [node] (j) at  ($ (d)+ (-.5,2.5) $) {};
	\draw (a) -- (b) -- (c) -- (d) -- (e) -- (f) -- (a); 	\draw (b) -- (g) -- (h) -- (e); 	\draw (a) -- (i) -- (j); 	\draw (d) to[out=0,in=0, distance=.5cm] (j); 	\coordinate (1) at  ($ (j) + (.75,.75) $); 	\coordinate (2) at  ($ (g)+ (-1.2, 1.2) $); 	\coordinate (3) at  ($ (c)+ (.75,.75) $); 	\coordinate (6) at  ($ (i)+ (-.75,.75) $); 	\coordinate (5) at  ($ (h)+ (1.2,-1.2) $); 	\coordinate (4) at  ($ (f) + (-.75,-.75) $); 	\draw (i) to node [label=6] {} (6); 	\draw (j) to node [label=1] {} (1); 	\draw (g) to node [label={[xshift=.2cm]2}] {} (2); 	\draw (c) to node [label={[xshift=-.05cm]3}] {} (3); 	\draw (h) to node [label={[xshift=.15cm, yshift=-.1cm]5}] {} (5); 	\draw (f) to node [label={[xshift=-.1cm]4}] {} (4);
	\draw [dashed, line width=.08cm] (1) to[out=30,in=-30,distance=1.3cm] node [label={[xshift=.2cm,yshift=.0cm]x}] {} (3); 	\draw [dashed, line width=.08cm] (2) to[out=180,in=220,distance=4cm] node [label={[xshift=-.4cm]y}] {} (5); 	\draw [dashed, line width=.08cm] (4) to[out=160,in=200,distance=1cm] node [label={[xshift=.2cm,yshift=-.0cm]z}] {} (6);
	\draw [line width=.08cm] (1) to (j) to (i) to (6); 	\draw [line width=.08cm] (3) to (c) to (d) to (e) to (h) to (5); 	\draw [line width=.08cm] (2) to (g) to (b) to (a) to (f) to (4);

\end{tikzpicture}
\par\end{center}

\begin{center}
\caption{The cycle from the previous figure, after expanding the gadget, is
now a cycle of length $x+y+z+13$.}

\par\end{center}%
\end{minipage}
\end{figure}
\begin{figure}[H]
\begin{minipage}[t]{0.45\columnwidth}%
\begin{center}
\tikzstyle{node}=[circle, draw, fill=black!50,                         inner sep=0pt, minimum width=4pt]
\begin{tikzpicture}[thick,scale=0.8] 	\coordinate (4) at (-2,1); 	\coordinate (3) at (-2,-1); 	\draw (4) to node [label={[xshift=.15cm,yshift=.3cm]4}] {} (3); 	\draw (4) to node [label={[xshift=.15cm,yshift=-.7cm]3}] {} (3); 	\coordinate (5) at (0,1); 	\coordinate (1) at (0,-1); 	\draw (5) to node [label={[xshift=.15cm,yshift=.3cm]5}] {} (1); 	\draw (5) to node [label={[xshift=.15cm,yshift=-.7cm]1}] {} (1); 	\coordinate (6) at (2,1); 	\coordinate (2) at (2,-1); 	\draw (6) to node [label={[xshift=.15cm,yshift=.3cm]6}] {} (2); 	\draw (6) to node [label={[xshift=.15cm,yshift=-.7cm]2}] {} (2);
	\draw [line width=.08cm] (4) to (3); 	\draw [line width=.08cm] (5) to (1); 	\draw [line width=.08cm] (6) to (2); 
	\draw [dashed,  line width=.08cm] (3) to[out=-90,in=-90] node [label={[xshift=-.0cm]x}] {} (1); 	\draw [dashed,  line width=.08cm] (5) to[out=90,in=90] node [label={[xshift=-.0cm,yshift=-.5cm]y}] {} (6); 	\draw [dashed,  line width=.08cm] (2) to[out=20,in=50, distance=3.5cm] node [label=z] {} (4); 
\end{tikzpicture}
\par\end{center}

\begin{center}
\caption{A cycle of length $x+y+z+3$ passes through a $H'_{3}$.}

\par\end{center}%
\end{minipage}\hfill{}%
\begin{minipage}[t]{0.45\columnwidth}%
\begin{center}
\tikzstyle{node}=[circle, draw, fill=black!50,                         inner sep=0pt, minimum width=4pt]
\begin{tikzpicture}[thick,scale=.75]
\clip (-3.5,-3) rectangle (3,3);
\node [node] (a) at (-1,.5) {}; 	\node [node] (b) at (0,1.2) {}; 	\node [node] (c) at (1,.5) {}; 	\node [node] (d) at (1,-.5) {}; 	\node [node] (e) at (0,-1.2) {}; 	\node [node] (f) at (-1,-.5) {}; 	\node [node] (g) at (0,.5) {}; 	\node [node] (h) at (0,-.5) {}; 	\node [node] (i) at  ($ (a)+ (.5,1.5) $) {}; 	\node [node] (j) at  ($ (d)+ (-.5,2.5) $) {};
	\draw (a) -- (b) -- (c) -- (d) -- (e) -- (f) -- (a); 	\draw (b) -- (g) -- (h) -- (e); 	\draw (a) -- (i) -- (j); 	\draw (d) to[out=0,in=0, distance=.5cm] (j); 	\coordinate (1) at  ($ (j) + (.75,.75) $); 	\coordinate (2) at  ($ (g)+ (-1.2, 1.2) $); 	\coordinate (3) at  ($ (c)+ (.75,.75) $); 	\coordinate (6) at  ($ (i)+ (-.75,.75) $); 	\coordinate (5) at  ($ (h)+ (1.2,-1.2) $); 	\coordinate (4) at  ($ (f) + (-.75,-.75) $); 	\draw (i) to node [label=6] {} (6); 	\draw (j) to node [label=1] {} (1); 	\draw (g) to node [label={[xshift=.2cm]2}] {} (2); 	\draw (c) to node [label={[xshift=-.05cm]3}] {} (3); 	\draw (h) to node [label={[xshift=.15cm, yshift=-.1cm]5}] {} (5); 	\draw (f) to node [label={[xshift=-.1cm]4}] {} (4);
	\draw [dashed, line width=.08cm] (1) to[out=30,in=-30,distance=.8cm] node [label={[xshift=.2cm,yshift=.0cm]x}] {} (3); 	\draw [dashed, line width=.08cm] (6) to[out=200,in=220,distance=4cm] node [label={[xshift=-.4cm]y}] {} (5); 	\draw [dashed, line width=.08cm] (4) to[out=160,in=200,distance=1cm] node [label={[xshift=-.2cm,yshift=-.0cm]z}] {} (2);
	\draw [line width=.08cm] (2) to (g) to (h) to (5); 	\draw [line width=.08cm] (6) to (i) to (a) to (b) to (c) to (3); 	\draw [line width=.08cm] (4) to (f) to (e) to (d) to[out=0,in=0,distance=.5cm] (j) to (1);

\end{tikzpicture}
\par\end{center}

\begin{center}
\caption{The cycle from the previous figure, after expanding the gadget, is
now a cycle of length $x+y+z+13$.}

\par\end{center}%
\end{minipage}
\end{figure}
\begin{figure}[H]
\begin{minipage}[t]{0.45\columnwidth}%
\begin{center}
\tikzstyle{node}=[circle, draw, fill=black!50,                         inner sep=0pt, minimum width=4pt]
\begin{tikzpicture}[thick,scale=0.8] 	\coordinate (4) at (-2,1); 	\coordinate (3) at (-2,-1); 	\draw (4) to node [label={[xshift=.15cm,yshift=.3cm]4}] {} (3); 	\draw (4) to node [label={[xshift=.15cm,yshift=-.7cm]3}] {} (3); 	\coordinate (5) at (0,1); 	\coordinate (1) at (0,-1); 	\draw (5) to node [label={[xshift=.15cm,yshift=.3cm]5}] {} (1); 	\draw (5) to node [label={[xshift=.15cm,yshift=-.7cm]1}] {} (1); 	\coordinate (6) at (2,1); 	\coordinate (2) at (2,-1); 	\draw (6) to node [label={[xshift=.15cm,yshift=.3cm]6}] {} (2); 	\draw (6) to node [label={[xshift=.15cm,yshift=-.7cm]2}] {} (2);
	\draw [line width=.08cm] (4) to (3); 	\draw [line width=.08cm] (5) to (1); 	\draw [line width=.08cm] (6) to (2); 
	\draw [dashed,  line width=.08cm] (3) to[out=0,in=180] node [label={[xshift=-.1cm]x}] {} (5); 	\draw [dashed,  line width=.08cm] (1) to[out=-90,in=-90] node [label={[xshift=-.0cm]y}] {} (2); 	\draw [dashed,  line width=.08cm] (6) to[out=90,in=90, distance=1cm] node [label=z] {} (4); 
\end{tikzpicture}
\par\end{center}

\begin{center}
\caption{A cycle of length $x+y+z+3$ passes through a $H'_{3}$.}

\par\end{center}%
\end{minipage}\hfill{}%
\begin{minipage}[t]{0.45\columnwidth}%
\begin{center}
\tikzstyle{node}=[circle, draw, fill=black!50,                         inner sep=0pt, minimum width=4pt]
\begin{tikzpicture}[thick,scale=.75] 	\node [node] (a) at (-1,.5) {}; 	\node [node] (b) at (0,1.2) {}; 	\node [node] (c) at (1,.5) {}; 	\node [node] (d) at (1,-.5) {}; 	\node [node] (e) at (0,-1.2) {}; 	\node [node] (f) at (-1,-.5) {}; 	\node [node] (g) at (0,.5) {}; 	\node [node] (h) at (0,-.5) {}; 	\node [node] (i) at  ($ (a)+ (.5,1.5) $) {}; 	\node [node] (j) at  ($ (d)+ (-.5,2.5) $) {};
	\draw (a) -- (b) -- (c) -- (d) -- (e) -- (f) -- (a); 	\draw (b) -- (g) -- (h) -- (e); 	\draw (a) -- (i) -- (j); 	\draw (d) to[out=0,in=0, distance=.5cm] (j); 	\coordinate (1) at  ($ (j) + (.75,.75) $); 	\coordinate (2) at  ($ (g)+ (-1.2, 1.2) $); 	\coordinate (3) at  ($ (c)+ (.75,.75) $); 	\coordinate (6) at  ($ (i)+ (-.75,.75) $); 	\coordinate (5) at  ($ (h)+ (1.2,-1.2) $); 	\coordinate (4) at  ($ (f) + (-.75,-.75) $); 	\draw (i) to node [label=6] {} (6); 	\draw (j) to node [label=1] {} (1); 	\draw (g) to node [label={[xshift=.2cm]2}] {} (2); 	\draw (c) to node [label={[xshift=-.05cm]3}] {} (3); 	\draw (h) to node [label={[xshift=.15cm, yshift=-.1cm]5}] {} (5); 	\draw (f) to node [label={[xshift=-.1cm]4}] {} (4);
	\draw [dashed, line width=.08cm] (5) to[out=30,in=-30,distance=1.3cm] node [label={[xshift=.2cm,yshift=.5cm]x}] {} (3); 	\draw [dashed, line width=.08cm] (2) to[out=90,in=90,distance=.6cm] node [label={[xshift=-.0cm]y}] {} (1); 	\draw [dashed, line width=.08cm] (4) to[out=160,in=200,distance=1cm] node [label={[xshift=-.2cm,yshift=-.0cm]z}] {} (6);
	\draw [line width=.08cm] (2) to (g) to (b) to (a) to (i) to (6); 	\draw [line width=.08cm] (4) to (f) to (e) to (h) to (5); 	\draw [line width=.08cm] (3) to (c) to (d) to[out=0,in=0,distance=.5cm] (j) to (j) to (1);

\end{tikzpicture}
\par\end{center}

\begin{center}
\caption{The cycle from the previous figure, after expanding the gadget, is
now a cycle of length $x+y+z+13$.}

\par\end{center}%
\end{minipage}
\end{figure}
\begin{figure}[H]
\begin{minipage}[t]{0.45\columnwidth}%
\begin{center}
\tikzstyle{node}=[circle, draw, fill=black!50,                         inner sep=0pt, minimum width=4pt]
\begin{tikzpicture}[thick,scale=0.8] 	\coordinate (4) at (-2,1); 	\coordinate (3) at (-2,-1); 	\draw (4) to node [label={[xshift=.15cm,yshift=.3cm]4}] {} (3); 	\draw (4) to node [label={[xshift=.15cm,yshift=-.7cm]3}] {} (3); 	\coordinate (5) at (0,1); 	\coordinate (1) at (0,-1); 	\draw (5) to node [label={[xshift=.15cm,yshift=.3cm]5}] {} (1); 	\draw (5) to node [label={[xshift=.15cm,yshift=-.7cm]1}] {} (1); 	\coordinate (6) at (2,1); 	\coordinate (2) at (2,-1); 	\draw (6) to node [label={[xshift=.15cm,yshift=.3cm]6}] {} (2); 	\draw (6) to node [label={[xshift=.15cm,yshift=-.7cm]2}] {} (2);
	\draw [line width=.08cm] (4) to (3); 	\draw [line width=.08cm] (5) to (1); 	\draw [line width=.08cm] (6) to (2); 
	\draw [dashed,  line width=.08cm] (3) to[out=0,in=180] node [label={[xshift=-.1cm]x}] {} (5); 	\draw [dashed,  line width=.08cm] (1) to[out=0,in=180] node [label={[xshift=-.1cm]y}] {} (6); 	\draw [dashed,  line width=.08cm] (2) to[out=20,in=50, distance=3.2cm] node [label=z] {} (4); 
\end{tikzpicture}
\par\end{center}

\begin{center}
\caption{A cycle of length $x+y+z+3$ passes through a $H'_{3}$.}

\par\end{center}%
\end{minipage}\hfill{}%
\begin{minipage}[t]{0.45\columnwidth}%
\begin{center}
\tikzstyle{node}=[circle, draw, fill=black!50,                         inner sep=0pt, minimum width=4pt]
\begin{tikzpicture}[thick,scale=.75] 	\node [node] (a) at (-1,.5) {}; 	\node [node] (b) at (0,1.2) {}; 	\node [node] (c) at (1,.5) {}; 	\node [node] (d) at (1,-.5) {}; 	\node [node] (e) at (0,-1.2) {}; 	\node [node] (f) at (-1,-.5) {}; 	\node [node] (g) at (0,.5) {}; 	\node [node] (h) at (0,-.5) {}; 	\node [node] (i) at  ($ (a)+ (.5,1.5) $) {}; 	\node [node] (j) at  ($ (d)+ (-.5,2.5) $) {};
	\draw (a) -- (b) -- (c) -- (d) -- (e) -- (f) -- (a); 	\draw (b) -- (g) -- (h) -- (e); 	\draw (a) -- (i) -- (j); 	\draw (d) to[out=0,in=0, distance=.5cm] (j); 	\coordinate (1) at  ($ (j) + (.75,.75) $); 	\coordinate (2) at  ($ (g)+ (-1.2, 1.2) $); 	\coordinate (3) at  ($ (c)+ (.75,.75) $); 	\coordinate (6) at  ($ (i)+ (-.75,.75) $); 	\coordinate (5) at  ($ (h)+ (1.2,-1.2) $); 	\coordinate (4) at  ($ (f) + (-.75,-.75) $); 	\draw (i) to node [label=6] {} (6); 	\draw (j) to node [label=1] {} (1); 	\draw (g) to node [label={[xshift=.2cm]2}] {} (2); 	\draw (c) to node [label={[xshift=-.05cm]3}] {} (3); 	\draw (h) to node [label={[xshift=.15cm, yshift=-.1cm]5}] {} (5); 	\draw (f) to node [label={[xshift=-.1cm]4}] {} (4);
	\draw [dashed, line width=.08cm] (5) to[out=30,in=-30,distance=1.3cm] node [label={[xshift=.2cm,yshift=.5cm]x}] {} (3); 	\draw [dashed, line width=.08cm] (6) to[out=90,in=90,distance=.6cm] node [label={[xshift=-.0cm]y}] {} (1); 	\draw [dashed, line width=.08cm] (4) to[out=160,in=200,distance=1cm] node [label={[xshift=-.2cm,yshift=-.0cm]z}] {} (2);
	\draw [line width=.08cm] (2) to (g) to (h) to (5); 	\draw [line width=.08cm] (6) to (i) to (a) to (b) to (c) to (3); 	\draw [line width=.08cm] (4) to (f) to (e) to (d) to[out=0,in=0,distance=.5cm] (j) to (1);

\end{tikzpicture}
\par\end{center}

\begin{center}
\caption{The cycle from the previous figure, after expanding the gadget, is
now a cycle of length $x+y+z+13$.}

\par\end{center}%
\end{minipage}
\end{figure}

\begin{figure}[H]
\begin{minipage}[t]{0.45\columnwidth}%
\begin{center}
 \tikzstyle{node}=[circle, draw, fill=black!50,                         inner sep=0pt, minimum width=4pt]
\begin{tikzpicture}[thick,scale=0.8] 	\coordinate (4) at (-2,1); 	\coordinate (3) at (-2,-1); 	\draw (4) to node [label={[xshift=.15cm,yshift=.3cm]4}] {} (3); 	\draw (4) to node [label={[xshift=.15cm,yshift=-.7cm]3}] {} (3); 	\coordinate (5) at (0,1); 	\coordinate (1) at (0,-1); 	\draw (5) to node [label={[xshift=.15cm,yshift=.3cm]5}] {} (1); 	\draw (5) to node [label={[xshift=.15cm,yshift=-.7cm]1}] {} (1); 	\coordinate (6) at (2,1); 	\coordinate (2) at (2,-1); 	\draw (6) to node [label={[xshift=.15cm,yshift=.3cm]6}] {} (2); 	\draw (6) to node [label={[xshift=.15cm,yshift=-.7cm]2}] {} (2);
	\draw [line width=.08cm] (4) to (3); 	\draw [line width=.08cm] (5) to (1); 	\draw [line width=.08cm] (6) to (2); 
	\draw [dashed,  line width=.08cm] (3) to[out=-50,in=-130] node [label={[xshift=-.1cm]x}] {} (2); 	\draw [dashed,  line width=.08cm] (1) to[out=0,in=180] node [label={[xshift=-.0cm]y}] {} (6); 	\draw [dashed,  line width=.08cm] (5) to[out=180,in=0] node [label=z] {} (4); 
\end{tikzpicture}
\par\end{center}

\begin{center}
\caption{A cycle of length $x+y+z+3$ passes through a $H'_{3}$.}

\par\end{center}%
\end{minipage}\hfill{}%
\begin{minipage}[t]{0.45\columnwidth}%
\begin{center}
\tikzstyle{node}=[circle, draw, fill=black!50,                         inner sep=0pt, minimum width=4pt]
\begin{tikzpicture}[thick,scale=.75] 	\node [node] (a) at (-1,.5) {}; 	\node [node] (b) at (0,1.2) {}; 	\node [node] (c) at (1,.5) {}; 	\node [node] (d) at (1,-.5) {}; 	\node [node] (e) at (0,-1.2) {}; 	\node [node] (f) at (-1,-.5) {}; 	\node [node] (g) at (0,.5) {}; 	\node [node] (h) at (0,-.5) {}; 	\node [node] (i) at  ($ (a)+ (.5,1.5) $) {}; 	\node [node] (j) at  ($ (d)+ (-.5,2.5) $) {};
	\draw (a) -- (b) -- (c) -- (d) -- (e) -- (f) -- (a); 	\draw (b) -- (g) -- (h) -- (e); 	\draw (a) -- (i) -- (j); 	\draw (d) to[out=0,in=0, distance=.5cm] (j); 	\coordinate (1) at  ($ (j) + (.75,.75) $); 	\coordinate (2) at  ($ (g)+ (-1.2, 1.2) $); 	\coordinate (3) at  ($ (c)+ (.75,.75) $); 	\coordinate (6) at  ($ (i)+ (-.75,.75) $); 	\coordinate (5) at  ($ (h)+ (1.2,-1.2) $); 	\coordinate (4) at  ($ (f) + (-.75,-.75) $); 	\draw (i) to node [label=6] {} (6); 	\draw (j) to node [label=1] {} (1); 	\draw (g) to node [label={[xshift=.2cm]2}] {} (2); 	\draw (c) to node [label={[xshift=-.05cm]3}] {} (3); 	\draw (h) to node [label={[xshift=.15cm, yshift=-.1cm]5}] {} (5); 	\draw (f) to node [label={[xshift=-.1cm]4}] {} (4);
	\draw [dashed, line width=.08cm] (3) to[out=60,in=130,distance=4cm] node [label={[xshift=-.2cm]x}] {} (2); 	\draw [dashed, line width=.08cm] (6) to[out=90,in=90,distance=.6cm] node [label={[xshift=-.0cm]y}] {} (1); 	\draw [dashed, line width=.08cm] (4) to[out=300,in=250,distance=1cm] node [label={[xshift=-.2cm,yshift=-.0cm]z}] {} (5);
	\draw [line width=.08cm] (2) to (g) to (h) to (5); 	\draw [line width=.08cm] (4) to (f) to (e) to (d) to[out=0,in=0, distance=.5cm] (j) to (j) to (1); 	\draw [line width=.08cm] (3) to (c) to (b) to (a) to (i) to (6);

\end{tikzpicture}
\par\end{center}

\begin{center}
\caption{The cycle from the previous figure, after expanding the gadget, is
now a cycle of length $x+y+z+13$.}

\par\end{center}%
\end{minipage}
\end{figure}

\begin{figure}[H]
\begin{minipage}[t]{0.45\columnwidth}%
\begin{center}
\tikzstyle{node}=[circle, draw, fill=black!50,                         inner sep=0pt, minimum width=4pt]
\begin{tikzpicture}[thick,scale=0.8] 	\coordinate (4) at (-2,1); 	\coordinate (3) at (-2,-1); 	\draw (4) to node [label={[xshift=.15cm,yshift=.3cm]4}] {} (3); 	\draw (4) to node [label={[xshift=.15cm,yshift=-.7cm]3}] {} (3); 	\coordinate (5) at (0,1); 	\coordinate (1) at (0,-1); 	\draw (5) to node [label={[xshift=.15cm,yshift=.3cm]5}] {} (1); 	\draw (5) to node [label={[xshift=.15cm,yshift=-.7cm]1}] {} (1); 	\coordinate (6) at (2,1); 	\coordinate (2) at (2,-1); 	\draw (6) to node [label={[xshift=.15cm,yshift=.3cm]6}] {} (2); 	\draw (6) to node [label={[xshift=.15cm,yshift=-.7cm]2}] {} (2);
	\draw [line width=.08cm] (4) to (3); 	\draw [line width=.08cm] (5) to (1); 	\draw [line width=.08cm] (6) to (2); 
	\draw [dashed,  line width=.08cm] (3) to[out=-50,in=-130] node [label={[xshift=-.1cm]x}] {} (2); 	\draw [dashed,  line width=.08cm] (6) to[out=180,in=0] node [label={[xshift=-.0cm]y}] {} (5); 	\draw [dashed,  line width=.08cm] (1) to[out=180,in=0, distance=1cm] node [label={[xshift=.1cm]z}] {} (4); 
\end{tikzpicture}
\par\end{center}

\begin{center}
\caption{A cycle of length $x+y+z+3$ passes through a $H'_{3}$.}

\par\end{center}%
\end{minipage}\hfill{}%
\begin{minipage}[t]{0.45\columnwidth}%
\begin{center}
\tikzstyle{node}=[circle, draw, fill=black!50,                         inner sep=0pt, minimum width=4pt]
\begin{tikzpicture}[thick,scale=.75] 	\node [node] (a) at (-1,.5) {}; 	\node [node] (b) at (0,1.2) {}; 	\node [node] (c) at (1,.5) {}; 	\node [node] (d) at (1,-.5) {}; 	\node [node] (e) at (0,-1.2) {}; 	\node [node] (f) at (-1,-.5) {}; 	\node [node] (g) at (0,.5) {}; 	\node [node] (h) at (0,-.5) {}; 	\node [node] (i) at  ($ (a)+ (.5,1.5) $) {}; 	\node [node] (j) at  ($ (d)+ (-.5,2.5) $) {};
	\draw (a) -- (b) -- (c) -- (d) -- (e) -- (f) -- (a); 	\draw (b) -- (g) -- (h) -- (e); 	\draw (a) -- (i) -- (j); 	\draw (d) to[out=0,in=0, distance=.5cm] (j); 	\coordinate (1) at  ($ (j) + (.75,.75) $); 	\coordinate (2) at  ($ (g)+ (-1.2, 1.2) $); 	\coordinate (3) at  ($ (c)+ (.75,.75) $); 	\coordinate (6) at  ($ (i)+ (-.75,.75) $); 	\coordinate (5) at  ($ (h)+ (1.2,-1.2) $); 	\coordinate (4) at  ($ (f) + (-.75,-.75) $); 	\draw (i) to node [label=6] {} (6); 	\draw (j) to node [label=1] {} (1); 	\draw (g) to node [label={[xshift=.2cm]2}] {} (2); 	\draw (c) to node [label={[xshift=-.05cm]3}] {} (3); 	\draw (h) to node [label={[xshift=.15cm, yshift=-.1cm]5}] {} (5); 	\draw (f) to node [label={[xshift=-.1cm]4}] {} (4);
	\draw [dashed, line width=.08cm] (3) to[out=60,in=130,distance=3cm] node [label={[xshift=-.2cm]x}] {} (2); 	\draw [dashed, line width=.08cm] (5) to[out=220,in=220, distance=3.7cm] node [label={[xshift=-.3cm,yshift=.0cm]y}] {} (6); 	\draw [dashed, line width=.08cm] (4) to[out=130,in=120,distance=3.7cm] node [label={[xshift=-.1cm,yshift=-.0cm]z}] {} (1);
	\draw [line width=.08cm] (2) to (g) to (b) to (a) to (f) to (4); 	\draw [line width=.08cm] (3) to (c) to (d) to (e) to (h) to (5); 	\draw [line width=.08cm] (6) to (i) to (j) to (1);

\end{tikzpicture}
\par\end{center}

\begin{center}
\caption{The cycle from the previous figure, after expanding the gadget, is
now a cycle of length $x+y+z+13$.}

\par\end{center}%
\end{minipage}
\end{figure}

\begin{figure}[H]
\begin{minipage}[t]{0.45\columnwidth}%
\begin{center}
\tikzstyle{node}=[circle, draw, fill=black!50,                         inner sep=0pt, minimum width=4pt]
\begin{tikzpicture}[thick,scale=0.8] 	\coordinate (4) at (-2,1); 	\coordinate (3) at (-2,-1); 	\draw (4) to node [label={[xshift=.15cm,yshift=.3cm]4}] {} (3); 	\draw (4) to node [label={[xshift=.15cm,yshift=-.7cm]3}] {} (3); 	\coordinate (5) at (0,1); 	\coordinate (1) at (0,-1); 	\draw (5) to node [label={[xshift=.15cm,yshift=.3cm]5}] {} (1); 	\draw (5) to node [label={[xshift=.15cm,yshift=-.7cm]1}] {} (1); 	\coordinate (6) at (2,1); 	\coordinate (2) at (2,-1); 	\draw (6) to node [label={[xshift=.15cm,yshift=.3cm]6}] {} (2); 	\draw (6) to node [label={[xshift=.15cm,yshift=-.7cm]2}] {} (2);
	\draw [line width=.08cm] (4) to (3); 	\draw [line width=.08cm] (5) to (1); 	\draw [line width=.08cm] (6) to (2); 
	\draw [dashed,  line width=.08cm] (3) to[out=130,in=120, distance=3cm] node [label={[xshift=-.1cm]x}] {} (6); 	\draw [dashed,  line width=.08cm] (2) to[out=180,in=0] node [label={[xshift=-.0cm]y}] {} (1); 	\draw [dashed,  line width=.08cm] (5) to[out=180,in=0] node [label=z] {} (4); 
\end{tikzpicture}
\par\end{center}

\begin{center}
\caption{A cycle of length $x+y+z+3$ passes through a $H'_{3}$.}

\par\end{center}%
\end{minipage}\hfill{}%
\begin{minipage}[t]{0.45\columnwidth}%
\begin{center}
\tikzstyle{node}=[circle, draw, fill=black!50,                         inner sep=0pt, minimum width=4pt]
\begin{tikzpicture}[thick,scale=.75] 	\node [node] (a) at (-1,.5) {}; 	\node [node] (b) at (0,1.2) {}; 	\node [node] (c) at (1,.5) {}; 	\node [node] (d) at (1,-.5) {}; 	\node [node] (e) at (0,-1.2) {}; 	\node [node] (f) at (-1,-.5) {}; 	\node [node] (g) at (0,.5) {}; 	\node [node] (h) at (0,-.5) {}; 	\node [node] (i) at  ($ (a)+ (.5,1.5) $) {}; 	\node [node] (j) at  ($ (d)+ (-.5,2.5) $) {};
	\draw (a) -- (b) -- (c) -- (d) -- (e) -- (f) -- (a); 	\draw (b) -- (g) -- (h) -- (e); 	\draw (a) -- (i) -- (j); 	\draw (d) to[out=0,in=0, distance=.5cm] (j); 	\coordinate (1) at  ($ (j) + (.75,.75) $); 	\coordinate (2) at  ($ (g)+ (-1.2, 1.2) $); 	\coordinate (3) at  ($ (c)+ (.75,.75) $); 	\coordinate (6) at  ($ (i)+ (-.75,.75) $); 	\coordinate (5) at  ($ (h)+ (1.2,-1.2) $); 	\coordinate (4) at  ($ (f) + (-.75,-.75) $); 	\draw (i) to node [label=6] {} (6); 	\draw (j) to node [label=1] {} (1); 	\draw (g) to node [label={[xshift=.2cm]2}] {} (2); 	\draw (c) to node [label={[xshift=-.05cm]3}] {} (3); 	\draw (h) to node [label={[xshift=.15cm, yshift=-.1cm]5}] {} (5); 	\draw (f) to node [label={[xshift=-.1cm]4}] {} (4);
	\draw [dashed, line width=.08cm] (6) to[out=50,in=30,distance=2cm] node [label={[xshift=-.2cm,yshift=.1cm]x}] {} (3); 	\draw [dashed, line width=.08cm] (1) to[out=120,in=150,distance=2cm] node [label={[xshift=-.2cm]y}] {} (2); 	\draw [dashed, line width=.08cm] (4) to[out=260,in=-80,distance=1cm] node [label={[xshift=.0cm,yshift=-.5cm]z}] {} (5);
	\draw [line width=.08cm] (1) to (j) to (i) to (6); 	\draw [line width=.08cm] (4) to (f) to (a) to (b) to (g) to (2); 	\draw [line width=.08cm] (3) to (c) to (d) to (e) to (h) to (5);

\end{tikzpicture}
\par\end{center}

\begin{center}
\caption{The cycle from the previous figure, after expanding the gadget, is
now a cycle of length $x+y+z+13$.}

\par\end{center}%
\end{minipage}
\end{figure}

\begin{figure}[H]
\begin{minipage}[t]{0.45\columnwidth}%
\begin{center}
\tikzstyle{node}=[circle, draw, fill=black!50,                         inner sep=0pt, minimum width=4pt]
\begin{tikzpicture}[thick,scale=0.8] 	\coordinate (4) at (-2,1); 	\coordinate (3) at (-2,-1); 	\draw (4) to node [label={[xshift=.15cm,yshift=.3cm]4}] {} (3); 	\draw (4) to node [label={[xshift=.15cm,yshift=-.7cm]3}] {} (3); 	\coordinate (5) at (0,1); 	\coordinate (1) at (0,-1); 	\draw (5) to node [label={[xshift=.15cm,yshift=.3cm]5}] {} (1); 	\draw (5) to node [label={[xshift=.15cm,yshift=-.7cm]1}] {} (1); 	\coordinate (6) at (2,1); 	\coordinate (2) at (2,-1); 	\draw (6) to node [label={[xshift=.15cm,yshift=.3cm]6}] {} (2); 	\draw (6) to node [label={[xshift=.15cm,yshift=-.7cm]2}] {} (2);
	\draw [line width=.08cm] (4) to (3); 	\draw [line width=.08cm] (5) to (1); 	\draw [line width=.08cm] (6) to (2); 
	\draw [dashed,  line width=.08cm] (3) to[out=130,in=120, distance=3cm] node [label={[xshift=-.1cm]x}] {} (6); 	\draw [dashed,  line width=.08cm] (2) to[out=180,in=0] node [label={[xshift=.1cm]y}] {} (5); 	\draw [dashed,  line width=.08cm] (1) to[out=180,in=0] node [label={[xshift=.1cm]z}] {} (4); 
\end{tikzpicture}
\par\end{center}

\begin{center}
\caption{A cycle of length $x+y+z+3$ passes through a $H'_{3}$.}

\par\end{center}%
\end{minipage}\hfill{}%
\begin{minipage}[t]{0.45\columnwidth}%
\begin{center}
\tikzstyle{node}=[circle, draw, fill=black!50,                         inner sep=0pt, minimum width=4pt]
\begin{tikzpicture}[thick,scale=.75] 	\node [node] (a) at (-1,.5) {}; 	\node [node] (b) at (0,1.2) {}; 	\node [node] (c) at (1,.5) {}; 	\node [node] (d) at (1,-.5) {}; 	\node [node] (e) at (0,-1.2) {}; 	\node [node] (f) at (-1,-.5) {}; 	\node [node] (g) at (0,.5) {}; 	\node [node] (h) at (0,-.5) {}; 	\node [node] (i) at  ($ (a)+ (.5,1.5) $) {}; 	\node [node] (j) at  ($ (d)+ (-.5,2.5) $) {};
	\draw (a) -- (b) -- (c) -- (d) -- (e) -- (f) -- (a); 	\draw (b) -- (g) -- (h) -- (e); 	\draw (a) -- (i) -- (j); 	\draw (d) to[out=0,in=0, distance=.5cm] (j); 	\coordinate (1) at  ($ (j) + (.75,.75) $); 	\coordinate (2) at  ($ (g)+ (-1.2, 1.2) $); 	\coordinate (3) at  ($ (c)+ (.75,.75) $); 	\coordinate (6) at  ($ (i)+ (-.75,.75) $); 	\coordinate (5) at  ($ (h)+ (1.2,-1.2) $); 	\coordinate (4) at  ($ (f) + (-.75,-.75) $); 	\draw (i) to node [label=6] {} (6); 	\draw (j) to node [label=1] {} (1); 	\draw (g) to node [label={[xshift=.2cm]2}] {} (2); 	\draw (c) to node [label={[xshift=-.05cm]3}] {} (3); 	\draw (h) to node [label={[xshift=.15cm, yshift=-.1cm]5}] {} (5); 	\draw (f) to node [label={[xshift=-.1cm]4}] {} (4);
	\draw [dashed, line width=.08cm] (6) to[out=90,in=45,distance=2cm] node [label={[xshift=.2cm]x}] {} (3); 	\draw [dashed, line width=.08cm] (2) to[out=210,in=220,distance=3.2cm] node [label={[xshift=-.4cm]y}] {} (5); 	\draw [dashed, line width=.08cm] (4) to[out=320,in=0,distance=4cm] node [label={[xshift=-.2cm,yshift=-.0cm]z}] {} (1);
	\draw [line width=.08cm] (2) to (g) to (b) to (a) to (f) to (4); 	\draw [line width=.08cm] (3) to (c) to (d) to (e) to (h) to (5); 	\draw [line width=.08cm] (6) to (i) to (j) to (1);

\end{tikzpicture}
\par\end{center}

\begin{center}
\caption{The cycle from the previous figure, after expanding the gadget, is
now a cycle of length $x+y+z+13$.}

\par\end{center}%
\end{minipage}
\end{figure}

\section{Appendix E: $H_{4}$s} \label{apdxh4}

\subsection{Gadget is covered by two cycles}

If a $H'_{4}$ is covered by two disjoint cycles in $F_{i}$, then
the internal edge cannot be part of the $2$-factor. This leaves only
the single possibility depicted in Figures 143-144, which takes two
cycles of lengths $x+2$ and $y+2$ in $F_{i}$ and returns a single
cycle of length $x+y+14$ in $F_{i-1}$ after the expansion. $x+2$
and $y+2$ are both at least $6$, so $x+y+14>8$, meaning the resulting
cycle in $F_{i-1}$ cannot be a cycle of length $6$.

\begin{figure}[H]
\begin{minipage}[t]{0.45\columnwidth}%
\begin{center}
\tikzstyle{node}=[circle, draw, fill=black!50,                         inner sep=0pt, minimum width=4pt]
\begin{tikzpicture}[auto,thick, scale=.5]	 	\node [node] (a) at (0,2) {}; 	\node [node] (b) at (0,-2) {}; 	\draw (a) -- (b); 	\coordinate (1) at  ($ (a)+ (-1,1) $); 	\draw (a) to node [label={[xshift=-.1cm]1}] {} (1); 	\coordinate (3) at  ($ (a)+ (1,1) $); 	\draw (a) to node [label={[xshift=.5cm, yshift=-.5cm]3}] {} (3); 	\coordinate (2) at  ($ (b)+ (-1,-1) $); 	\draw (b) to node [label={[xshift=-.5cm, yshift = .5cm]2}] {} (2); 	\coordinate (4) at  ($ (b)+ (1,-1) $); 	\draw (b) to node [label={4}] {} (4);
	\draw [line width=.08cm] (1) to (a) to (3); 	\draw [line width=.08cm] (2) to (b) to (4);
	\draw [dashed, line width=.08cm] (1) to[out=90,in=90,distance=.8cm] node [label=x] {} (3); 	\draw [dashed, line width=.08cm] (2) to[out=-90,in=-90,distance=.8cm] node [label=y] {} (4); \end{tikzpicture}
\par\end{center}

\begin{center}
\caption{Two cycles of lengths $x+2$ and $y+2$ pass through a $H'_{4}$.}

\par\end{center}%
\end{minipage}\hfill{}%
\begin{minipage}[t]{0.45\columnwidth}%
\begin{center}
\tikzstyle{node}=[circle, draw, fill=black!50,                         inner sep=0pt, minimum width=4pt]
\begin{tikzpicture}[thick,scale=.75] 	\node [node] (a) at (-1,.5) {}; 	\node [node] (b) at (0,1.2) {}; 	\node [node] (c) at (1,.5) {}; 	\node [node] (d) at (1,-.5) {}; 	\node [node] (e) at (0,-1.2) {}; 	\node [node] (f) at (-1,-.5) {}; 	\node [node] (g) at (0,.5) {}; 	\node [node] (h) at (0,-.5) {}; 	\node [node] (i) at  ($ (a)+ (.5,1.5) $) {}; 	\node [node] (j) at  ($ (d)+ (-.5,2.5) $) {}; 	\node [node] (w1) at ($ (j) + (2.5,-.7) $) {}; 	\node [node] (w2) at ($ (i) + (-2.5,-.7) $) {};
	\draw (a) -- (b) -- (c) -- (d) -- (e) -- (f) -- (a); 	\draw (b) -- (g) -- (h) -- (e); 	\draw (a) -- (i) -- (j) -- (w1) -- (g); 	\draw (i) -- (w2) -- (h); 	\draw (d) to[out=0,in=0, distance=.5cm] (j); 	\coordinate (1) at  ($ (w2) + (-1,0) $); 	\coordinate (2) at  ($ (w1)+ (1, 0) $); 	\coordinate (3) at  ($ (c)+ (1.5,-.75) $); 	\coordinate (4) at  ($ (f)+ (-1.5,.0) $); 	\draw (w2) to node [label=1] {} (1); 	\draw (w1) to node [label=2] {} (2); 	\draw (c) to node [label={[xshift=.1cm]3}] {} (3); 	\draw (f) to node [label={[xshift=-.1cm]4}] {} (4);
	\draw [dashed, line width=.08cm] (1) to[out=260,in=240, distance=3cm] node [label={[xshift=-1.2cm]x}] {} (3); 	\draw [dashed, line width=.08cm] (2) to[out=280,in=300, distance=3cm] node [label={[xshift=1.2cm]y}] {} (4);
	\draw [line width=.08cm] (1) to (w2) to (i) to (j) to (w1) to (2); 	\draw [line width=.08cm] (4) to (f) to (a) to (b) to (g) to (h) to (e) to (d) to (c) to (3); \end{tikzpicture}
\par\end{center}

\begin{center}
\caption{The cycles from the previous figure, after expanding the gadget, are
now a single cycle of length $x+y+14$.}

\par\end{center}%
\end{minipage}
\end{figure}

\subsection{Gadget is covered by one cycle}

If a $H'_{4}$ is covered by a single cycle in $F_{i}$, then we consider
cases when the internal edge is part of $F_{i}$ and those when the
internal edge is not included. First consider the cases when the internal
edge is included in $F_{i}$. Then $F_{i}$ passes through either
edge 1 or 3 and either edge 2 or 4. Each of these possibilities takes
a cycle of length $x+3$ in $F_{i}$ and returns a cycle of length
$x+13$ in $F_{i-1}$. $x+3\geq6$, so $x+13>8$, meaning the resulting
cycle in $F_{i-1}$ cannot be a cycle of length $6$. The case when
exiting edges 1 and 4 are included in $F_{i}$ is symmetric to the
case when exiting edges 2 and 3, so only the first of these cases
is included in this appendix. Then, there are three unique cases to
consider where the internal edge is included in $F_{i}$. We consider
these cases in the first six figures of this subsection.

Next, consider when the internal edge is not included in $F_{i}$.
Then all four exiting edges of the gadget must be used. There are
two cases to consider where $F_{i}$ can pass through these four edges,
because after exiting edge 1, the cycle can re-enter the gadget at
either edge 2 or 4. Each of these possibilities takes a cycle of length
$x+y+4$ in $F_{i}$ and returns a cycle of length $x+y+14$ in $F_{i-1}$.
$x+y+4\geq6$, so $x+y+14>8$, meaning the resulting cycle in $F_{i-1}$
cannot be a cycle of length $6$. We consider both cases in the last four figures of this subsection.

\begin{figure}[H]
\begin{minipage}[t]{0.45\columnwidth}%
\begin{center}
\tikzstyle{node}=[circle, draw, fill=black!50,                         inner sep=0pt, minimum width=4pt]
\begin{tikzpicture}[auto,thick, scale=.5]	 	\node [node] (a) at (0,2) {}; 	\node [node] (b) at (0,-2) {}; 	\draw (a) -- (b); 	\coordinate (1) at  ($ (a)+ (-1,1) $); 	\draw (a) to node [label={[xshift=-.1cm]1}] {} (1); 	\coordinate (3) at  ($ (a)+ (1,1) $); 	\draw (a) to node [label={[xshift=.5cm, yshift=-.5cm]3}] {} (3); 	\coordinate (2) at  ($ (b)+ (-1,-1) $); 	\draw (b) to node [label={[xshift=-.5cm, yshift = .5cm]2}] {} (2); 	\coordinate (4) at  ($ (b)+ (1,-1) $); 	\draw (b) to node [label={4}] {} (4);
	\draw [line width=.08cm] (1) to (a) to (b) to (2); 	\draw [dashed, line width=.08cm] (1) to[out=180,in=180, distance=1cm] node [label={[xshift=-.5cm]x}] {} (2); \end{tikzpicture}
\par\end{center}

\begin{center}
\caption{A cycle of length $x+3$ passes through a $H'_{4}$.}

\par\end{center}%
\end{minipage}\hfill{}%
\begin{minipage}[t]{0.45\columnwidth}%
\begin{center}
\tikzstyle{node}=[circle, draw, fill=black!50,                         inner sep=0pt, minimum width=4pt]
\begin{tikzpicture}[thick,scale=.75] 	\node [node] (a) at (-1,.5) {}; 	\node [node] (b) at (0,1.2) {}; 	\node [node] (c) at (1,.5) {}; 	\node [node] (d) at (1,-.5) {}; 	\node [node] (e) at (0,-1.2) {}; 	\node [node] (f) at (-1,-.5) {}; 	\node [node] (g) at (0,.5) {}; 	\node [node] (h) at (0,-.5) {}; 	\node [node] (i) at  ($ (a)+ (.5,1.5) $) {}; 	\node [node] (j) at  ($ (d)+ (-.5,2.5) $) {}; 	\node [node] (w1) at ($ (j) + (2.5,-.7) $) {}; 	\node [node] (w2) at ($ (i) + (-2.5,-.7) $) {};
	\draw (a) -- (b) -- (c) -- (d) -- (e) -- (f) -- (a); 	\draw (b) -- (g) -- (h) -- (e); 	\draw (a) -- (i) -- (j) -- (w1) -- (g); 	\draw (i) -- (w2) -- (h); 	\draw (d) to[out=0,in=0, distance=.5cm] (j); 	\coordinate (1) at  ($ (w2) + (-1,0) $); 	\coordinate (2) at  ($ (w1)+ (1, 0) $); 	\coordinate (3) at  ($ (c)+ (1.5,-.75) $); 	\coordinate (4) at  ($ (f)+ (-1.5,.0) $); 	\draw (w2) to node [label=1] {} (1); 	\draw (w1) to node [label=2] {} (2); 	\draw (c) to node [label={[xshift=.1cm]3}] {} (3); 	\draw (f) to node [label={[xshift=-.1cm]4}] {} (4);
	\draw [dashed, line width=.08cm] (1) to[out=90,in=90, distance=2cm] node [label=x] {} (2); 	\draw [line width=.08cm] (1) to (w2) to (h) to (e) to (f) to (a) to (i) to (j) to[out=0,in=0,distance=.5cm] (d) to (d) to (c) to (b) to (g) to (w1) to (2); \end{tikzpicture}
\par\end{center}

\begin{center}
\caption{The cycle from the previous figure, after expanding the gadget, is
now a cycle of length $x+13$.}

\par\end{center}%
\end{minipage}
\end{figure}
\begin{figure}[H]
\begin{minipage}[t]{0.45\columnwidth}%
\begin{center}
\tikzstyle{node}=[circle, draw, fill=black!50,                         inner sep=0pt, minimum width=4pt]
\begin{tikzpicture}[auto,thick, scale=.5]	 	\node [node] (a) at (0,2) {}; 	\node [node] (b) at (0,-2) {}; 	\draw (a) -- (b); 	\coordinate (1) at  ($ (a)+ (-1,1) $); 	\draw (a) to node [label={[xshift=-.1cm]1}] {} (1); 	\coordinate (3) at  ($ (a)+ (1,1) $); 	\draw (a) to node [label={[xshift=.5cm, yshift=-.5cm]3}] {} (3); 	\coordinate (2) at  ($ (b)+ (-1,-1) $); 	\draw (b) to node [label={[xshift=-.5cm, yshift = .5cm]2}] {} (2); 	\coordinate (4) at  ($ (b)+ (1,-1) $); 	\draw (b) to node [label={4}] {} (4);
	\draw [line width=.08cm] (1) to (a) to (b) to (4); 	\draw [dashed, line width=.08cm] (1) to[out=45,in=40, distance=3.5cm] node [label={[xshift=-.0cm]x}] {} (4); \end{tikzpicture} 
\par\end{center}

\begin{center}
\caption{A cycle of length $x+3$ passes through a $H'_{4}$.}

\par\end{center}%
\end{minipage}\hfill{}%
\begin{minipage}[t]{0.45\columnwidth}%
\begin{center}
\tikzstyle{node}=[circle, draw, fill=black!50,                         inner sep=0pt, minimum width=4pt]
\begin{tikzpicture}[thick,scale=.8] 	\node [node] (a) at (-1,.5) {}; 	\node [node] (b) at (0,1.2) {}; 	\node [node] (c) at (1,.5) {}; 	\node [node] (d) at (1,-.5) {}; 	\node [node] (e) at (0,-1.2) {}; 	\node [node] (f) at (-1,-.5) {}; 	\node [node] (g) at (0,.5) {}; 	\node [node] (h) at (0,-.5) {}; 	\node [node] (i) at  ($ (a)+ (.5,1.5) $) {}; 	\node [node] (j) at  ($ (d)+ (-.5,2.5) $) {}; 	\node [node] (w1) at ($ (j) + (2.5,-.7) $) {}; 	\node [node] (w2) at ($ (i) + (-2.5,-.7) $) {};
	\draw (a) -- (b) -- (c) -- (d) -- (e) -- (f) -- (a); 	\draw (b) -- (g) -- (h) -- (e); 	\draw (a) -- (i) -- (j) -- (w1) -- (g); 	\draw (i) -- (w2) -- (h); 	\draw (d) to[out=0,in=0, distance=.5cm] (j); 	\coordinate (1) at  ($ (w2) + (-1,0) $); 	\coordinate (2) at  ($ (w1)+ (1, 0) $); 	\coordinate (3) at  ($ (c)+ (1.5,-.75) $); 	\coordinate (4) at  ($ (f)+ (-1.5,.0) $); 	\draw (w2) to node [label=1] {} (1); 	\draw (w1) to node [label=2] {} (2); 	\draw (c) to node [label={[xshift=.1cm]3}] {} (3); 	\draw (f) to node [label={[xshift=-.1cm]4}] {} (4);
	\draw [dashed, line width=.08cm] (1) to[out=180,in=180, distance=1cm] node [label={[yshift=-.6cm]x}] {} (4); 	\draw [line width=.08cm] (1) to (w2) to (i) to (j) to (w1) to (g) to (h) to (e) to (d) to (c) to (b) to (a) to (f) to (4); \end{tikzpicture}
\par\end{center}

\begin{center}
\caption{The cycle from the previous figure, after expanding the gadget, is
now a cycle of length $x+13$.}

\par\end{center}%
\end{minipage}
\end{figure}
\begin{figure}[H]
\begin{minipage}[t]{0.45\columnwidth}%
\begin{center}
\tikzstyle{node}=[circle, draw, fill=black!50,                         inner sep=0pt, minimum width=4pt]
\begin{tikzpicture}[auto,thick, scale=.5]	 	\node [node] (a) at (0,2) {}; 	\node [node] (b) at (0,-2) {}; 	\draw (a) -- (b); 	\coordinate (1) at  ($ (a)+ (-1,1) $); 	\draw (a) to node [label={[xshift=-.1cm]1}] {} (1); 	\coordinate (3) at  ($ (a)+ (1,1) $); 	\draw (a) to node [label={[xshift=.5cm, yshift=-.5cm]3}] {} (3); 	\coordinate (2) at  ($ (b)+ (-1,-1) $); 	\draw (b) to node [label={[xshift=-.5cm, yshift = .5cm]2}] {} (2); 	\coordinate (4) at  ($ (b)+ (1,-1) $); 	\draw (b) to node [label={4}] {} (4);
	\draw [line width=.08cm] (3) to (a) to (b) to (4); 	\draw [dashed, line width=.08cm] (3) to[out=0,in=0, distance=1cm] node [label={[xshift=-.0cm]x}] {} (4); \end{tikzpicture}
\par\end{center}

\begin{center}
\caption{A cycle of length $x+3$ passes through a $H'_{4}$.}

\par\end{center}%
\end{minipage}\hfill{}%
\begin{minipage}[t]{0.45\columnwidth}%
\begin{center}
\tikzstyle{node}=[circle, draw, fill=black!50,                         inner sep=0pt, minimum width=4pt]
\begin{tikzpicture}[thick,scale=.8] 	\node [node] (a) at (-1,.5) {}; 	\node [node] (b) at (0,1.2) {}; 	\node [node] (c) at (1,.5) {}; 	\node [node] (d) at (1,-.5) {}; 	\node [node] (e) at (0,-1.2) {}; 	\node [node] (f) at (-1,-.5) {}; 	\node [node] (g) at (0,.5) {}; 	\node [node] (h) at (0,-.5) {}; 	\node [node] (i) at  ($ (a)+ (.5,1.5) $) {}; 	\node [node] (j) at  ($ (d)+ (-.5,2.5) $) {}; 	\node [node] (w1) at ($ (j) + (2.5,-.7) $) {}; 	\node [node] (w2) at ($ (i) + (-2.5,-.7) $) {};
	\draw (a) -- (b) -- (c) -- (d) -- (e) -- (f) -- (a); 	\draw (b) -- (g) -- (h) -- (e); 	\draw (a) -- (i) -- (j) -- (w1) -- (g); 	\draw (i) -- (w2) -- (h); 	\draw (d) to[out=0,in=0, distance=.5cm] (j); 	\coordinate (1) at  ($ (w2) + (-1,0) $); 	\coordinate (2) at  ($ (w1)+ (1, 0) $); 	\coordinate (3) at  ($ (c)+ (1.5,-.75) $); 	\coordinate (4) at  ($ (f)+ (-1.5,.0) $); 	\draw (w2) to node [label=1] {} (1); 	\draw (w1) to node [label=2] {} (2); 	\draw (c) to node [label={[xshift=.1cm]3}] {} (3); 	\draw (f) to node [label={[xshift=-.1cm]4}] {} (4);
	\draw [dashed, line width=.08cm] (3) to[out=-90,in=-90, distance=2cm] node [label={[yshift=-.6cm]x}] {} (4); 	\draw [line width=.08cm] (4) to (f) to (a) to (b) to (g) to (w1) to (j) to (i) to (w2) to (h) to (e) to (d) to (c) to (3); \end{tikzpicture}
\par\end{center}

\begin{center}
\caption{The cycle from the previous figure, after expanding the gadget, is
now a cycle of length $x+13$.}

\par\end{center}%
\end{minipage}
\end{figure}
\begin{figure}[H]
\begin{minipage}[t]{0.45\columnwidth}%
\begin{center}
\tikzstyle{node}=[circle, draw, fill=black!50,                         inner sep=0pt, minimum width=4pt]
\begin{tikzpicture}[auto,thick, scale=.5]	 	\node [node] (a) at (0,2) {}; 	\node [node] (b) at (0,-2) {}; 	\draw (a) -- (b); 	\coordinate (1) at  ($ (a)+ (-1,1) $); 	\draw (a) to node [label={[xshift=-.1cm]1}] {} (1); 	\coordinate (3) at  ($ (a)+ (1,1) $); 	\draw (a) to node [label={[xshift=.5cm, yshift=-.5cm]3}] {} (3); 	\coordinate (2) at  ($ (b)+ (-1,-1) $); 	\draw (b) to node [label={[xshift=-.5cm, yshift = .5cm]2}] {} (2); 	\coordinate (4) at  ($ (b)+ (1,-1) $); 	\draw (b) to node [label={4}] {} (4);
	\draw [line width=.08cm] (1) to (a) to (3); 	\draw [line width=.08cm] (2) to (b) to (4); 	\draw [dashed, line width=.08cm] (1) to[out=180,in=180, distance=1cm] node [label={[xshift=-.5cm]x}] {} (2); 	\draw [dashed, line width=.08cm] (3) to [out=0,in=0,distance=1cm] node [label=y] {} (4); \end{tikzpicture}
\par\end{center}

\begin{center}
\caption{A cycle of length $x+4$ passes through a $H'_{4}$.}

\par\end{center}%
\end{minipage}\hfill{}%
\begin{minipage}[t]{0.45\columnwidth}%
\begin{center}
\tikzstyle{node}=[circle, draw, fill=black!50,                         inner sep=0pt, minimum width=4pt]
\begin{tikzpicture}[thick,scale=.75] 	\node [node] (a) at (-1,.5) {}; 	\node [node] (b) at (0,1.2) {}; 	\node [node] (c) at (1,.5) {}; 	\node [node] (d) at (1,-.5) {}; 	\node [node] (e) at (0,-1.2) {}; 	\node [node] (f) at (-1,-.5) {}; 	\node [node] (g) at (0,.5) {}; 	\node [node] (h) at (0,-.5) {}; 	\node [node] (i) at  ($ (a)+ (.5,1.5) $) {}; 	\node [node] (j) at  ($ (d)+ (-.5,2.5) $) {}; 	\node [node] (w1) at ($ (j) + (2.5,-.7) $) {}; 	\node [node] (w2) at ($ (i) + (-2.5,-.7) $) {};
	\draw (a) -- (b) -- (c) -- (d) -- (e) -- (f) -- (a); 	\draw (b) -- (g) -- (h) -- (e); 	\draw (a) -- (i) -- (j) -- (w1) -- (g); 	\draw (i) -- (w2) -- (h); 	\draw (d) to[out=0,in=0, distance=.5cm] (j); 	\coordinate (1) at  ($ (w2) + (-1,0) $); 	\coordinate (2) at  ($ (w1)+ (1, 0) $); 	\coordinate (3) at  ($ (c)+ (1.5,-.75) $); 	\coordinate (4) at  ($ (f)+ (-1.5,.0) $); 	\draw (w2) to node [label=1] {} (1); 	\draw (w1) to node [label=2] {} (2); 	\draw (c) to node [label={[xshift=.1cm]3}] {} (3); 	\draw (f) to node [label={[xshift=-.1cm]4}] {} (4);
	\draw [dashed, line width=.08cm] (1) to[out=90,in=90, distance=2cm] node [label=x] {} (2); 	\draw [dashed, line width=.08cm] (4) to [out=-90,in=-90, distance=1.5cm] node [label={[yshift=-.5cm]y}] {} (3); 	\draw [line width=.08cm] (1) to (w2) to (i) to (a) to (b) to (g) to (h) to (e) to (f) to (4); 	\draw [line width=.08cm] (3) to (c) to (d) to[out=0, in=0,distance=.5cm] (j) to (j) (w1) to (2); \end{tikzpicture}
\par\end{center}

\begin{center}
\caption{The cycle from the previous figure, after expanding the gadget, is
now a cycle of length $x+y+14$.}

\par\end{center}%
\end{minipage}
\end{figure}
\begin{figure}[H]
\begin{minipage}[t]{0.45\columnwidth}%
\begin{center}
\tikzstyle{node}=[circle, draw, fill=black!50,                         inner sep=0pt, minimum width=4pt]
\begin{tikzpicture}[auto,thick, scale=.5]	 	\node [node] (a) at (0,2) {}; 	\node [node] (b) at (0,-2) {}; 	\draw (a) -- (b); 	\coordinate (1) at  ($ (a)+ (-1,1) $); 	\draw (a) to node [label={[xshift=-.1cm]1}] {} (1); 	\coordinate (3) at  ($ (a)+ (1,1) $); 	\draw (a) to node [label={[xshift=.5cm, yshift=-.5cm]3}] {} (3); 	\coordinate (2) at  ($ (b)+ (-1,-1) $); 	\draw (b) to node [label={[xshift=-.5cm, yshift = .5cm]4}] {} (2); 	\coordinate (4) at  ($ (b)+ (1,-1) $); 	\draw (b) to node [label={2}] {} (4);
	\draw [line width=.08cm] (1) to (a) to (3); 	\draw [line width=.08cm] (2) to (b) to (4); 	\draw [dashed, line width=.08cm] (1) to[out=180,in=180, distance=1cm] node [label={[xshift=-.5cm]x}] {} (2); 	\draw [dashed, line width=.08cm] (3) to [out=0,in=0,distance=1cm] node [label=y] {} (4); \end{tikzpicture}
\par\end{center}

\begin{center}
\caption{A cycle of length $x+4$ passes through a $H'_{4}$.}

\par\end{center}%
\end{minipage}\hfill{}%
\begin{minipage}[t]{0.45\columnwidth}%
\begin{center}
\tikzstyle{node}=[circle, draw, fill=black!50,                         inner sep=0pt, minimum width=4pt]
\begin{tikzpicture}[thick,scale=.75] 	\node [node] (a) at (-1,.5) {}; 	\node [node] (b) at (0,1.2) {}; 	\node [node] (c) at (1,.5) {}; 	\node [node] (d) at (1,-.5) {}; 	\node [node] (e) at (0,-1.2) {}; 	\node [node] (f) at (-1,-.5) {}; 	\node [node] (g) at (0,.5) {}; 	\node [node] (h) at (0,-.5) {}; 	\node [node] (i) at  ($ (a)+ (.5,1.5) $) {}; 	\node [node] (j) at  ($ (d)+ (-.5,2.5) $) {}; 	\node [node] (w1) at ($ (j) + (2.5,-.7) $) {}; 	\node [node] (w2) at ($ (i) + (-2.5,-.7) $) {};
	\draw (a) -- (b) -- (c) -- (d) -- (e) -- (f) -- (a); 	\draw (b) -- (g) -- (h) -- (e); 	\draw (a) -- (i) -- (j) -- (w1) -- (g); 	\draw (i) -- (w2) -- (h); 	\draw (d) to[out=0,in=0, distance=.5cm] (j); 	\coordinate (1) at  ($ (w2) + (-1,0) $); 	\coordinate (2) at  ($ (w1)+ (1, 0) $); 	\coordinate (3) at  ($ (c)+ (1.5,-.75) $); 	\coordinate (4) at  ($ (f)+ (-1.5,.0) $); 	\draw (w2) to node [label=1] {} (1); 	\draw (w1) to node [label=2] {} (2); 	\draw (c) to node [label={[xshift=.1cm]3}] {} (3); 	\draw (f) to node [label={[xshift=-.1cm]4}] {} (4);
	\draw [dashed, line width=.08cm] (1) to[out=220,in=180] node [label={[xshift=-.5cm]x}] {} (4); 	\draw [dashed, line width=.08cm] (2) to [out=-40,in=0] node [label={[xshift=.5cm]y}] {} (3); 	\draw [line width=.08cm] (1) to (w2) to (i) to (j) to (w1) to (2); 	\draw [line width=.08cm] (3) to (c) to (d) to (e) to (h) to (g) to (b) to (a) to (f) to (4); \end{tikzpicture}
\par\end{center}

\begin{center}
\caption{The cycle from the previous figure, after expanding the gadget, is
now a cycle of length $x+y+14$.}

\par\end{center}%
\end{minipage}
\end{figure}

\section{Appendix F: $H_{5}$s} \label{apdxh5}

Two edges of each $H'_{5}$ is covered by a cycle in $F_{i}$. Then,
there are three cases to consider, when each of the gadget's edges
are excluded from $F_{i}$. In each of these cases, we start with
a cycle of length $x+2$ in $F_{i}$ and are returned a cycle of length
$x+14$. $x+2\geq6$, so $x+14>8$, meaning these expansion operations
cannot introduce an organic $6$-cycle into the $2$-factor. 

\begin{figure}[H]
\begin{minipage}[t]{0.45\columnwidth}%
\begin{center}
\tikzstyle{node}=[circle, draw, fill=black!50,                         inner sep=0pt, minimum width=4pt]
\begin{tikzpicture}[thick,scale=1.5] 	\node [node] (a) at (-1,.5) {}; 	\coordinate (1) at  ($ (a)+ (-.75,.75) $); 	\coordinate (2) at  ($ (a)+ (.75,.75) $); 	\coordinate (3) at  ($ (a)+ (0,-.8) $); 	\draw (a) to node [label={[xshift=.05cm]1}] {} (1); 	\draw (a) to node [label={[xshift=-.1cm]2}] {} (2); 	\draw (a) to node [label={[xshift=-.2cm,yshift=-.3cm]3}] {} (3);
	\draw [line width=.08cm] (1) to (a) to (2); 	\draw [dashed, line width=.08cm] (1) to[out=90,in=90] node [label=x] {} (2); \end{tikzpicture}
\par\end{center}

\begin{center}
\caption{A cycle of length $x+2$ passes through a $H'_{5}$.}

\par\end{center}%
\end{minipage}\hfill{}%
\begin{minipage}[t]{0.45\columnwidth}%
\begin{center}
\tikzstyle{node}=[circle, draw, fill=black!50,                         inner sep=0pt, minimum width=4pt]
\begin{tikzpicture}[thick,scale=.75] 	\node [node] (a) at (-1,.5) {}; 	\node [node] (b) at (0,1.2) {}; 	\node [node] (c) at (1,.5) {}; 	\node [node] (d) at (1,-.5) {}; 	\node [node] (e) at (0,-1.2) {}; 	\node [node] (f) at (-1,-.5) {}; 	\node [node] (g) at (0,.5) {}; 	\node [node] (h) at (0,-.5) {}; 	\node [node] (i) at  ($ (a)+ (.5,1.5) $) {}; 	\node [node] (j) at  ($ (d)+ (-.5,2.5) $) {}; 	\node [node] (w1) at ($ (j) + (2.5,-.7) $) {}; 	\node [node] (w2) at ($ (i) + (-2.5,-.7) $) {};
	\draw (a) -- (b) -- (c) -- (d) -- (e) -- (f) -- (a); 	\draw (b) -- (g) -- (h) -- (e); 	\draw (a) -- (i) -- (j) -- (w1) -- (g); 	\draw (i) -- (w2) -- (h); 	\draw (d) to[out=0,in=0, distance=.5cm] (j);
	\coordinate (2) at  ($ (w1)+ (1, 0) $); 	\coordinate (3) at  ($ (c)+ (1.5,-.75) $); 	\coordinate (4) at  ($ (f)+ (-1.5,.0) $);
	\draw (w1) to node [label=2] {} (2);
	\draw (f) to node [label={[xshift=-.1cm]3}] {} (4);
	\node [node] (1-3) at ($ (f) + (0,-1) $) {}; 	\coordinate (1) at  ($ (1-3) + (0,-1) $);
	\draw (w2) to[out=180,in=180] (1-3) to[out=-30,in=0, distance=2cm] (c); 	\draw (1-3) to node [label={[xshift=-.2cm,yshift=-.3cm]1}] {} (1);
	\draw [dashed, line width=.08cm] (1) to[out=330,in=280, distance=2cm] node [label={[xshift=-0cm]x}] {} (2);
	\draw [line width=.08cm] (1) to (1-3) to[in=180,out=180] (w2) to (w2) to (h) to (g) to (b) to (c) to (d) to (e) to (f) to (a) to (i) to (j) to (w1) to (2);
\end{tikzpicture}
\par\end{center}

\begin{center}
\caption{The cycle from the previous figure, after expanding the gadget, is
now a cycle of length $x+14$.}

\par\end{center}%
\end{minipage}
\end{figure}
\begin{figure}[H]
\begin{minipage}[t]{0.45\columnwidth}%
\begin{center}
\tikzstyle{node}=[circle, draw, fill=black!50,                         inner sep=0pt, minimum width=4pt]
\begin{tikzpicture}[thick,scale=1.5] 	\node [node] (a) at (-1,.5) {}; 	\coordinate (1) at  ($ (a)+ (-.75,.75) $); 	\coordinate (2) at  ($ (a)+ (.75,.75) $); 	\coordinate (3) at  ($ (a)+ (0,-.8) $); 	\draw (a) to node [label={[xshift=.05cm]1}] {} (1); 	\draw (a) to node [label={[xshift=-.1cm]2}] {} (2); 	\draw (a) to node [label={[xshift=-.2cm,yshift=-.3cm]3}] {} (3);
	\draw [line width=.08cm] (1) to (a) to (3); 	\draw [dashed, line width=.08cm] (1) to[out=180,in=180] node [label={[xshift=.2cm]x}] {} (3); \end{tikzpicture}
\par\end{center}

\begin{center}
\caption{A cycle of length $x+2$ passes through a $H'_{5}$.}

\par\end{center}%
\end{minipage}\hfill{}%
\begin{minipage}[t]{0.45\columnwidth}%
\begin{center}
\tikzstyle{node}=[circle, draw, fill=black!50,                         inner sep=0pt, minimum width=4pt]
\begin{tikzpicture}[thick,scale=.75]
\clip (-4,-3) rectangle (4.3,3);
\node [node] (a) at (-1,.5) {}; 	\node [node] (b) at (0,1.2) {}; 	\node [node] (c) at (1,.5) {}; 	\node [node] (d) at (1,-.5) {}; 	\node [node] (e) at (0,-1.2) {}; 	\node [node] (f) at (-1,-.5) {}; 	\node [node] (g) at (0,.5) {}; 	\node [node] (h) at (0,-.5) {}; 	\node [node] (i) at  ($ (a)+ (.5,1.5) $) {}; 	\node [node] (j) at  ($ (d)+ (-.5,2.5) $) {}; 	\node [node] (w1) at ($ (j) + (2.5,-.7) $) {}; 	\node [node] (w2) at ($ (i) + (-2.5,-.7) $) {};
	\draw (a) -- (b) -- (c) -- (d) -- (e) -- (f) -- (a); 	\draw (b) -- (g) -- (h) -- (e); 	\draw (a) -- (i) -- (j) -- (w1) -- (g); 	\draw (i) -- (w2) -- (h); 	\draw (d) to[out=0,in=0, distance=.5cm] (j);
	\coordinate (2) at  ($ (w1)+ (1, 0) $); 	\coordinate (3) at  ($ (c)+ (1.5,-.75) $); 	\coordinate (4) at  ($ (f)+ (-1.5,.0) $);
	\draw (w1) to node [label=2] {} (2);
	\draw (f) to node [label={[xshift=-.1cm]3}] {} (4);
	\node [node] (1-3) at ($ (f) + (0,-1) $) {}; 	\coordinate (1) at  ($ (1-3) + (0,-1) $);
	\draw (w2) to[out=180,in=180] (1-3) to[out=-30,in=0, distance=2cm] (c); 	\draw (1-3) to node [label={[xshift=-.2cm,yshift=-.3cm]1}] {} (1);
	\draw [dashed, line width=.08cm] (1) to[out=270,in=180, distance=1cm] node [label={[xshift=.1cm]x}] {} (4);
	\draw [line width=.08cm] (1) to (1-3) to[in=180,out=180] (w2) to (w2) to (i) to (j) to (w1) to (g) to (h) to (e) to (d) to (c) to (b) to (a) to (f) to (4);
\end{tikzpicture}
\par\end{center}

\begin{center}
\caption{The cycle from the previous figure, after expanding the gadget, is
now a cycle of length $x+14$.}

\par\end{center}%
\end{minipage}
\end{figure}
\begin{figure}[H]
\begin{minipage}[t]{0.45\columnwidth}%
\begin{center}
\tikzstyle{node}=[circle, draw, fill=black!50,                         inner sep=0pt, minimum width=4pt]
\begin{tikzpicture}[thick,scale=1.5] 	\node [node] (a) at (-1,.5) {}; 	\coordinate (1) at  ($ (a)+ (-.75,.75) $); 	\coordinate (2) at  ($ (a)+ (.75,.75) $); 	\coordinate (3) at  ($ (a)+ (0,-.8) $); 	\draw (a) to node [label={[xshift=.05cm]1}] {} (1); 	\draw (a) to node [label={[xshift=-.1cm]2}] {} (2); 	\draw (a) to node [label={[xshift=-.2cm,yshift=-.3cm]3}] {} (3);
	\draw [line width=.08cm] (2) to (a) to (3); 	\draw [dashed, line width=.08cm] (2) to[out=0,in=0] node [label={[xshift=.3cm]x}] {} (3); \end{tikzpicture}
\par\end{center}

\begin{center}
\caption{A cycle of length $x+2$ passes through a $H'_{5}$.}

\par\end{center}%
\end{minipage}\hfill{}%
\begin{minipage}[t]{0.45\columnwidth}%
\begin{center}
\tikzstyle{node}=[circle, draw, fill=black!50,                         inner sep=0pt, minimum width=4pt]
\begin{tikzpicture}[thick,scale=.75]
\clip (-4,-3) rectangle (4.3,3);
\node [node] (a) at (-1,.5) {}; 	\node [node] (b) at (0,1.2) {}; 	\node [node] (c) at (1,.5) {}; 	\node [node] (d) at (1,-.5) {}; 	\node [node] (e) at (0,-1.2) {}; 	\node [node] (f) at (-1,-.5) {}; 	\node [node] (g) at (0,.5) {}; 	\node [node] (h) at (0,-.5) {}; 	\node [node] (i) at  ($ (a)+ (.5,1.5) $) {}; 	\node [node] (j) at  ($ (d)+ (-.5,2.5) $) {}; 	\node [node] (w1) at ($ (j) + (2.5,-.7) $) {}; 	\node [node] (w2) at ($ (i) + (-2.5,-.7) $) {};
	\draw (a) -- (b) -- (c) -- (d) -- (e) -- (f) -- (a); 	\draw (b) -- (g) -- (h) -- (e); 	\draw (a) -- (i) -- (j) -- (w1) -- (g); 	\draw (i) -- (w2) -- (h); 	\draw (d) to[out=0,in=0, distance=.5cm] (j);
	\coordinate (2) at  ($ (w1)+ (1, 0) $); 	\coordinate (3) at  ($ (c)+ (1.5,-.75) $); 	\coordinate (4) at  ($ (f)+ (-1.5,.0) $);
	\draw (w1) to node [label=2] {} (2);
	\draw (f) to node [label={[xshift=-.1cm]3}] {} (4);
	\node [node] (1-3) at ($ (f) + (0,-1) $) {}; 	\coordinate (1) at  ($ (1-3) + (0,-1) $);
	\draw (w2) to[out=180,in=180] (1-3) to (1-3) to[out=-30,in=0, distance=2cm] (c); 	\draw (1-3) to node [label={[xshift=-.2cm,yshift=-.3cm]1}] {} (1);
	\draw [dashed, line width=.08cm] (2) to[out=-70,in=250, distance=4.1cm] node [label={[xshift=.1cm]x}] {} (4);
	\draw [line width=.08cm] (4) to (f) to (a) to (i) to (w2) to[out=180,in=180] (1-3) to (1-3) to[out=-30,in=0,distance=2cm] (c) to (c) to (b) to (g) to (h) to (e) to (d) to[out=0,in=0, distance=.5cm] (j) to (j) to (w1) to (2);
\end{tikzpicture}
\par\end{center}

\begin{center}
\caption{The cycle from the previous figure, after expanding the gadget, is
now a cycle of length $x+14$.}

\par\end{center}%
\end{minipage}
\end{figure}

\section{Appendix G: $H_{6}$s} \label{apdxh6}

There are two cases to consider for a $H'_{6}$, when the edge is
included in $F_{i}$ and when it is not. In both cases, expanding
the gadget cannot introduce an organic $6$-cycle to $F_{i-1}$. We
consider both cases in this section.

\begin{figure}[H]
\begin{minipage}[t]{0.45\columnwidth}%
\begin{center}
\tikzstyle{node}=[circle, draw, fill=black!50,                         inner sep=0pt, minimum width=4pt]
\begin{tikzpicture}[thick,scale=1.5] 	\coordinate (1) at (0,1); 	\coordinate (2) at (0,-1); 	\draw (1) to node [label={[xshift=.15cm,yshift=1cm]1}] {} (2); 	\draw (1) to node [label={[xshift=.15cm,yshift=-1.4cm]2}] {} (2);
\end{tikzpicture}
\par\end{center}

\begin{center}
\caption{A $H'_6$, which is not included in
the $2$-factor }

\par\end{center}%
\end{minipage}\hfill{}%
\begin{minipage}[t]{0.45\columnwidth}%
\begin{center}
\tikzstyle{node}=[circle, draw, fill=black!50,                         inner sep=0pt, minimum width=4pt]
\begin{tikzpicture}[thick,scale=.75] 	\node [node] (a) at (-1,.5) {}; 	\node [node] (b) at (0,1.2) {}; 	\node [node] (c) at (1,.5) {}; 	\node [node] (d) at (1,-.5) {}; 	\node [node] (e) at (0,-1.2) {}; 	\node [node] (f) at (-1,-.5) {}; 	\node [node] (g) at (0,.5) {}; 	\node [node] (h) at (0,-.5) {}; 	\node [node] (i) at  ($ (a)+ (.5,1.5) $) {}; 	\node [node] (j) at  ($ (d)+ (-.5,2.5) $) {}; 	\node [node] (w1) at ($ (j) + (2.5,-.7) $) {}; 	\node [node] (w2) at ($ (i) + (-2.5,-.7) $) {};
	\draw (a) -- (b) -- (c) -- (d) -- (e) -- (f) -- (a); 	\draw (b) -- (g) -- (h) -- (e); 	\draw (a) -- (i) -- (j) -- (w1) -- (g); 	\draw (i) -- (w2) -- (h); 	\draw (d) to[out=0,in=0, distance=.5cm] (j);
	\node [node] (left) at ($ (w2) + (0,-1) $) {}; 	\node [node] (right) at ($ (w1) + (0,-1) $) {};
	\draw (w2) to (left) to[out=290,in=-30,distance=4cm] (c); 	\draw (w1) to (right) to[out=280,in=250,distance=2.5cm] (f);
	\coordinate (1) at ($ (left) + (-1,0) $); 	\draw (left) to node [label={[xshift=.05cm]1}] {} (1);
	\coordinate (2) at ($ (right) + (1,0) $); 	\draw (right) to node [label={[xshift=.05cm]2}] {} (2);
	\draw[line width=.08cm] (left) to[out=290,in=-30,distance=4cm] (c) to (c) to (d) to (e) to (h) to (g) to (b) to (a) to (f) to[in=280,out=250,distance=2.5cm] (right) to (w1) to (j) to (i) to (w2) to (left); \end{tikzpicture}
\par\end{center}

\begin{center}
\caption{The cycle from the previous figure, after expanding the gadget, is
now a cycle of length $14$.}

\par\end{center}%
\end{minipage}
\end{figure}
\begin{figure}[H]
\begin{minipage}[t]{0.45\columnwidth}%
\begin{center}
\tikzstyle{node}=[circle, draw, fill=black!50,                         inner sep=0pt, minimum width=4pt]
\begin{tikzpicture}[thick,scale=1.5] 	\coordinate (1) at (0,1); 	\coordinate (2) at (0,-1); 	\draw (1) to node [label={[xshift=.15cm,yshift=1cm]1}] {} (2); 	\draw [line width=.08cm] (1) to node [label={[xshift=.15cm,yshift=-1.4cm]2}] {} (2); \draw [dashed, line width=.08cm] (1) to[out=0,in=0,distance=1cm] node[label={[xshift=.4cm]x}] {} (2);
\end{tikzpicture}
\par\end{center}

\begin{center}
\caption{A cycle of length $x+1$ passes through a $H'_{6}$.}

\par\end{center}%
\end{minipage}\hfill{}%
\begin{minipage}[t]{0.45\columnwidth}%
\begin{center}
\tikzstyle{node}=[circle, draw, fill=black!50,                         inner sep=0pt, minimum width=4pt]
\begin{tikzpicture}[thick,scale=.6]
\clip (-4.75,-3) rectangle (4.75,3.5);
\node [node] (a) at (-1,.5) {}; 	\node [node] (b) at (0,1.2) {}; 	\node [node] (c) at (1,.5) {}; 	\node [node] (d) at (1,-.5) {}; 	\node [node] (e) at (0,-1.2) {}; 	\node [node] (f) at (-1,-.5) {}; 	\node [node] (g) at (0,.5) {}; 	\node [node] (h) at (0,-.5) {}; 	\node [node] (i) at  ($ (a)+ (.5,1.5) $) {}; 	\node [node] (j) at  ($ (d)+ (-.5,2.5) $) {}; 	\node [node] (w1) at ($ (j) + (2.5,-.7) $) {}; 	\node [node] (w2) at ($ (i) + (-2.5,-.7) $) {};
	\draw (a) -- (b) -- (c) -- (d) -- (e) -- (f) -- (a); 	\draw (b) -- (g) -- (h) -- (e); 	\draw (a) -- (i) -- (j) -- (w1) -- (g); 	\draw (i) -- (w2) -- (h); 	\draw (d) to[out=0,in=0, distance=.5cm] (j);
	\node [node] (left) at ($ (w2) + (0,-1) $) {}; 	\node [node] (right) at ($ (w1) + (0,-1) $) {};
	\draw (w2) to (left) to[out=290,in=-30,distance=4cm] (c); 	\draw (w1) to (right) to[out=280,in=250,distance=2.5cm] (f);
	\coordinate (1) at ($ (left) + (-1,0) $); 	\draw (left) to node [label={[xshift=.05cm]1}] {} (1);
	\coordinate (2) at ($ (right) + (1,0) $); 	\draw (right) to node [label={[xshift=.05cm]2}] {} (2);
	\draw[line width=.08cm] (1) to (left) to[out=290,in=-30,distance=4cm] (c) to (c) to (d) to (e) to (f) to (a) to (b) to (g) to (h) to (w2) to (i) to (j) to (w1) to (right) to (2); 	\draw [dashed, line width=.08cm] (1) to[out=140,in=40,distance=5cm] node[label=x] {} (2); \end{tikzpicture}
\par\end{center}

\begin{center}
\caption{The cycle from the previous figure, after expanding the gadget, is
now a cycle of length $x+15$.}

\par\end{center}%
\end{minipage}
\end{figure}

%%STUFF FROM SAMPLE DOCUMENT

%%
%% Bibliography
%%

%% Either use bibtex (recommended), but commented out in this sample

\end{document}